\newcommand{\setpath}{}%{../}
\newcommand{\setpathfigs}{}%{figs/}%{../}%{figs/}%{}
\RequirePackage{fix-cm}
\documentclass[smallextended]{svjour3blank}       % onecolumn (second format)
\smartqed  % flush right qed marks, e.g. at end of proof
\usepackage{graphicx} 
\newcommand{\hidecalculustheorems}[2]{#1}

\usepackage{refcount}

\newcounter{mylabelcounter}

\makeatletter
\newcommand{\customlabel}[2]{%
   \protected@write \@auxout {}{\string \newlabel {#1}{{#2}{\thepage}{#2}{#1}{}} }%
   \hypertarget{#1}{#2}
}
\makeatother
  
\makeatletter
\newcommand{\labelText}[2]{%
  \refstepcounter{mylabelcounter}%
  \immediate\write\@auxout{%
    \string\newlabel{#1}{{#2}{\thepage}{{\unexpanded{#2}}}{mylabelcounter.\number\value{mylabelcounter}}{}}
  }%
}
\newcommand*{\storeuniversalcountertag}[2]{%
\labelText{universal:#1}{#2}
\labelText{store:universal:#1}{#2}
}
\makeatother

\makeatletter
\newcommand*{\storecounter}[2]{%
  \edef\@currentlabel{\the\value{#1}}% Store current counter value in \@currentlabel
  \label{#2}% Store label
}
\newcommand*{\storesubequationcountera}[1]{%
\setcounter{equationsection}{\value{section}}%
\setcounter{equationbuffer}{\value{subequationbuffera}}%
  \edef\@currentlabel{\theequationbuffer}% Store current counter value in \@currentlabel
  \label{store:equation:#1}% Store label
  \edef\@currentlabel{\theequationbuffer}% Store current counter value in \@currentlabel
  \label{equation:#1}% Store label
  \edef\@currentlabel{\theequationsection}% Store current counter value in \@currentlabel
  \label{store:section:#1}% Store label
  \edef\@currentlabel{\theequationsection}% Store current counter value in \@currentlabel
  \label{section:#1}% Store label
}
\newcommand*{\storesubequationdouble}[2]{%
\setcounter{equationsection}{\value{section}}%
\setcounter{equationbuffer}{\value{equation}}%
  \edef\@currentlabel{\theequationbuffer}% Store current counter value in \@currentlabel
  \label{store:subequation:parent:#1}% Store label
  \edef\@currentlabel{\theequationbuffer}% Store current counter value in \@currentlabel
  \label{subequation:parent:#1}% Store label
  \edef\@currentlabel{1}% Store current counter value in \@currentlabel
  \label{store:subequation:#1}% Store label
  \edef\@currentlabel{1}% Store current counter value in \@currentlabel
  \label{subequation:#1}% Store label
  \edef\@currentlabel{\theequationsection}% Store current counter value in \@currentlabel
  \label{store:section:#1}% Store label
  \edef\@currentlabel{\theequationsection}% Store current counter value in \@currentlabel
  \label{section:#1}% Store label
  \edef\@currentlabel{\theequationbuffer}% Store current counter value in \@currentlabel
  \label{store:subequation:parent:#2}% Store label
  \edef\@currentlabel{\theequationbuffer}% Store current counter value in \@currentlabel
  \label{subequation:parent:#2}% Store label
  \edef\@currentlabel{2}% Store current counter value in \@currentlabel
  \label{store:subequation:#2}% Store label
  \edef\@currentlabel{2}% Store current counter value in \@currentlabel
  \label{subequation:#2}% Store label
  \edef\@currentlabel{\theequationsection}% Store current counter value in \@currentlabel
  \label{store:section:#2}% Store label
  \edef\@currentlabel{\theequationsection}% Store current counter value in \@currentlabel
  \label{section:#2}% Store label
}
\newcommand*{\storesubequationtriple}[3]{%
\setcounter{equationsection}{\value{section}}%
\setcounter{equationbuffer}{\value{equation}}%
  \edef\@currentlabel{\theequationbuffer}% Store current counter value in \@currentlabel
  \label{store:subequation:parent:#1}% Store label
  \edef\@currentlabel{\theequationbuffer}% Store current counter value in \@currentlabel
  \label{subequation:parent:#1}% Store label
  \edef\@currentlabel{1}% Store current counter value in \@currentlabel
  \label{store:subequation:#1}% Store label
  \edef\@currentlabel{1}% Store current counter value in \@currentlabel
  \label{subequation:#1}% Store label
  \edef\@currentlabel{\theequationsection}% Store current counter value in \@currentlabel
  \label{store:section:#1}% Store label
  \edef\@currentlabel{\theequationsection}% Store current counter value in \@currentlabel
  \label{section:#1}% Store label
  \edef\@currentlabel{\theequationbuffer}% Store current counter value in \@currentlabel
  \label{store:subequation:parent:#2}% Store label
  \edef\@currentlabel{\theequationbuffer}% Store current counter value in \@currentlabel
  \label{subequation:parent:#2}% Store label
  \edef\@currentlabel{2}% Store current counter value in \@currentlabel
  \label{store:subequation:#2}% Store label
  \edef\@currentlabel{2}% Store current counter value in \@currentlabel
  \label{subequation:#2}% Store label
  \edef\@currentlabel{\theequationsection}% Store current counter value in \@currentlabel
  \label{store:section:#2}% Store label
  \edef\@currentlabel{\theequationsection}% Store current counter value in \@currentlabel
  \label{section:#2}% Store label
  \edef\@currentlabel{\theequationbuffer}% Store current counter value in \@currentlabel
  \label{store:subequation:parent:#3}% Store label
  \edef\@currentlabel{\theequationbuffer}% Store current counter value in \@currentlabel
  \label{subequation:parent:#3}% Store label
  \edef\@currentlabel{3}% Store current counter value in \@currentlabel
  \label{store:subequation:#3}% Store label
  \edef\@currentlabel{3}% Store current counter value in \@currentlabel
  \label{subequation:#3}% Store label
  \edef\@currentlabel{\theequationsection}% Store current counter value in \@currentlabel
  \label{store:section:#3}% Store label
  \edef\@currentlabel{\theequationsection}% Store current counter value in \@currentlabel
  \label{section:#3}% Store label
}
\newcommand*{\storeequationcounter}[1]{%
\setcounter{equationsection}{\value{section}}%
\setcounter{equationbuffer}{\value{equation}}%
  \edef\@currentlabel{\theequationbuffer}% Store current counter value in \@currentlabel
  \label{store:equation:#1}% Store label
  \edef\@currentlabel{\theequationbuffer}% Store current counter value in \@currentlabel
  \label{equation:#1}% Store label
  \edef\@currentlabel{\theequationsection}% Store current counter value in \@currentlabel
  \label{store:section:#1}% Store label
  \edef\@currentlabel{\theequationsection}% Store current counter value in \@currentlabel
  \label{section:#1}% Store label
}
\newcommand*{\storecompoundcounter}[2]{%
\setcounter{equationsection}{\value{section}}%
\setcounter{equationbuffer}{\value{#1}}%
  \edef\@currentlabel{\theequationbuffer}% Store current counter value in \@currentlabel
  \label{store:#1:#2}% Store label
  \edef\@currentlabel{\theequationbuffer}% Store current counter value in \@currentlabel
  \label{#1:#2}% Store label
  \edef\@currentlabel{\theequationsection}% Store current counter value in \@currentlabel
  \label{store:section:#2}% Store label
  \edef\@currentlabel{\theequationsection}% Store current counter value in \@currentlabel
  \label{section:#2}% Store label
}
\makeatother

\usepackage{ntheorem}
\usepackage{amsfonts}
\usepackage{amsmath}
\usepackage{etoolbox}
\usepackage{thmtools}
\usepackage{amssymb}
\usepackage{bbm}

\usepackage{tcolorbox}

\usepackage{xcolor}

  \usepackage{esint}
\usepackage{chngcntr}
\usepackage{upgreek}

\usepackage{graphicx}
\usepackage[shortlabels,inline]{enumitem}

\usepackage{subfigure}

\usepackage[noline,noend,linesnumbered]%[boxed,vlined,linesnumbered]
{algorithm2e}
\SetEndCharOfAlgoLine{}
\makeatletter
\newcommand{\removelatexerror}{\let\@latex@error\@gobble}
\makeatother

\usepackage{newfloat}
\usepackage{caption}

\usepackage{tikz}
\usetikzlibrary{calc,decorations.pathreplacing}
\usetikzlibrary{fadings}
\usetikzlibrary{shapes.arrows}
\usetikzlibrary {arrows.meta,bending,positioning}
\usepackage{pgfplots}
\pgfplotsset{compat=1.5}  %=newest}

\usepackage{gnuplottex}

\usepackage{pdfpages}
\newtheorem{assumption}{Assumption}
\newtheorem*{model}{Server Model}%[section]

\newtheorem{notation}{Notation}

\newenvironment{inlinenotation}[1]{\textbf{\textup{Notation~\notationlabel{#1}\thenotation{}.}}}{}
\newenvironment{inlinenotationwithtitle}[2]{\textbf{\textup{Notation~\notationlabel{#2}\thenotation{} (#1).}}}{}
\newcommand{\notationlabel}[1]{{\refstepcounter{notation}\label{#1}}}

\DeclareFloatingEnvironment[fileext=frm,placement={!ht},name=Algorithm]{algorithmenv}
\makeatletter
\renewcommand{\@algocf@capt@boxed}{above}% formerly {bottom}
\makeatother

\newcommand{\obsolete}[1]{}

\obsolete{

\newtheoremstyle{plainbigtitle}  % follow `plain` defaults but change HEADSPACE.
  {0.5\topsep}   % ABOVESPACE
  {0.5\topsep}   % BELOWSPACE
  {\itshape}  % BODYFONT
  {0pt}       % INDENT (empty value is the same as 0pt)
  {\bfseries} % HEADFONT %obsolete
  {.}         % HEADPUNCT
  {5pt plus 1pt minus 1pt}  % HEADSPACE. `plain` default: {5pt plus 1pt minus 1pt}  %%%% {\newline}
  {\bfseries\boldmath#1 #2\thmnote{ (#3)}}          % CUSTOM-HEAD-SPEC

\newtheoremstyle{plainbigtitlenonumber}  % follow `plain` defaults but change HEADSPACE.
  {0.5\topsep}   % ABOVESPACE
  {0.5\topsep}   % BELOWSPACE
  {\itshape}  % BODYFONT
  {0pt}       % INDENT (empty value is the same as 0pt)
  {\bfseries} % HEADFONT %obsolete
  {.}         % HEADPUNCT
  {5pt plus 1pt minus 1pt}  % HEADSPACE. `plain` default: {5pt plus 1pt minus 1pt}  %%%% {\newline}
  {\bfseries\boldmath#1\thmnote{ (#3)}}          % CUSTOM-HEAD-SPEC

\newtheoremstyle{definitionbigtitle}  % follow `plain` defaults but change HEADSPACE.
  {0.5\topsep}   % ABOVESPACE
  {0.5\topsep}   % BELOWSPACE
  {\upshape}  % BODYFONT
  {0pt}       % INDENT (empty value is the same as 0pt)
  {\bfseries} % HEADFONT %obsolete
  {.}         % HEADPUNCT
  {5pt plus 1pt minus 1pt}  % HEADSPACE. `plain` default: {5pt plus 1pt minus 1pt}  %%%% {\newline}
  {\bfseries\boldmath#1 #2\thmnote{ (#3)}}          % CUSTOM-HEAD-SPEC

\newtheoremstyle{definitionbigtitleseparator}  % follow `plain` defaults but change HEADSPACE.
  {0.5\topsep}   % ABOVESPACE
  {1.5\topsep}   % BELOWSPACE
  {\upshape}  % BODYFONT
  {0pt}       % INDENT (empty value is the same as 0pt)
  {\bfseries} % HEADFONT %obsolete
  {.}         % HEADPUNCT
  {5pt plus 1pt minus 1pt}  % HEADSPACE. `plain` default: {5pt plus 1pt minus 1pt}  %%%% {\newline}
  {\bfseries\boldmath#1 #2\thmnote{ (#3)}}          % CUSTOM-HEAD-SPEC

\newtheoremstyle{definitionbigtitleseparator}  % follow `plain` defaults but change HEADSPACE.
  {0.5\topsep}   % ABOVESPACE
  {1.5\topsep}   % BELOWSPACE
  {\upshape}  % BODYFONT
  {0pt}       % INDENT (empty value is the same as 0pt)
  {\bfseries} % HEADFONT %obsolete
  {.}         % HEADPUNCT
  {5pt plus 1pt minus 1pt}  % HEADSPACE. `plain` default: {5pt plus 1pt minus 1pt}  %%%% {\newline}
  {\bfseries\boldmath#1 #2\thmnote{ (#3)}}          % CUSTOM-HEAD-SPEC

\declaretheoremstyle[
  spaceabove=0.5\topsep, spacebelow=1.5\topsep,
  headfont=\bfseries,%\normalfont\scshape,
  headpunct={.},
  bodyfont=\upshape,%\normalfont,
  headspace={5pt plus 1pt minus 1pt},
  indent=0pt,
  custom-head-spec={\bfseries\boldmath#1 #2\thmnote{ (#3)}},
  qed=\diamond%\qedsymbol
]{examplethmstyle}

\newcommand{\newtheoremwithoutindent}[2]{\newtheorem{#1}{#2}\AfterEndEnvironment{#1}{\noindent\ignorespaces}}
\newcommand{\unnumberednewtheoremwithoutindent}[2]{\newtheorem*{#1}{#2}\AfterEndEnvironment{#1}{\noindent\ignorespaces}}
\newcommand{\newtheoremwithseparator}[2]{\newtheorem{#1}{#2}\AfterEndEnvironment{#1}{\indent}}

\theoremstyle{plainbigtitle}% default
\newtheoremwithoutindent{lemma}{Lemma}
\newtheoremwithoutindent{result}{Result}
\newtheoremwithoutindent{proposition}{Proposition}
\newtheoremwithoutindent{corollary}{Corollary}
\newtheoremwithoutindent{problem}{Problem}
\newtheoremwithoutindent{assumption}{Assumption}
\theoremstyle{plainbigtitlenonumber}
\unnumberednewtheoremwithoutindent{model}{Server Model}%[section]

\theoremstyle{definitionbigtitle}
\newtheoremwithoutindent{definition}{Definition}
\newtheoremwithoutindent{notation}{Notation}

\newenvironment{inlinenotationwithtitle}[2]{\textbf{\textup{Notation~\notationlabel{#2}\thenotation{} (#1).}}}{}
\newcommand{\notationlabel}[1]{{\refstepcounter{notation}\label{#1}}}

\theoremstyle{definitionbigtitleseparator}
\newtheoremwithseparator{remark}{Remark}
\theoremstyle{examplethmstyle}
\newtheorem{example}{Example}

\theoremstyle{comment}
\DeclareFloatingEnvironment[fileext=frm,placement={!ht},name=Algorithm]{algorithmenv}
\captionsetup[algorithmenv]{labelfont=bf}
\makeatletter
\renewcommand{\@algocf@capt@boxed}{above}% formerly {bottom}
\makeatother

}

\usepackage{hyperref} 
\hypersetup{
colorlinks=true,
linktocpage=true,
pdfstartview=FitH,
breaklinks=true,
pdfpagemode=UseNone,
pageanchor=true,
pdfpagemode=UseOutlines,
plainpages=false,
bookmarksnumbered,
bookmarksopen=false,
bookmarksopenlevel=1,
hypertexnames=true,
pdfhighlight=/O,
urlcolor=black,%charcoallight,%Maroon,
linkcolor=black,%charcoallight,%MyBlue!60!black,
citecolor=black,%charcoallight,%DarkGreen!70!black,	% for on-screen
pdftitle={},
pdfauthor={},
pdfsubject={},
pdfkeywords={},
pdfcreator={pdfLaTeX},
pdfproducer={LaTeX with hyperref}
}

\usepackage{mathtools}
\usepackage{acronym}		% for acronyms
\newcommand{\acli}[1]{\textit{\acl{#1}}}		% for italicized acro
\newcommand{\acdef}[1]{\textit{\acl{#1}} \textup{(\acs{#1})}\acused{#1}}		% for acro def
\newacro{wfunction}[CVF]{core function}
\newacro{Wfunction}[CVF]{Core function}
\acrodefplural{wfunction}[CVFs]{core functions}
\newacro{extwfunction}[$\extpurevaluefunction$-F]{extended $\purevaluefunction$-function}
\newacro{ROC}{region of convergence}
\acrodefplural{ROC}[ROCs]{regions of absolute convergence}
\newacro{PI}{policy iteration}
\newacro{FPI}{first-policy improvement}
\newacro{MDP}{Markov decision process}

\counterwithin{equation}{section}	
     \newcounter{equationbuffer}
     \newcounter{equationsection}
     
     \newcounter{subequationbuffera}
     \newcounter{subequationbufferb}
\counterwithin{buffer}{section}	
 \renewcommand{\thesection}{\Roman{section}} 
\numberwithin{subsubsection}{subsection}

\renewcommand{\qed}{\hfill%
$\square$%
}
\newcommand{\remarkmarker}{\hfill%
$\circ$%
}
\newcommand{\verticalshiftby}[2][0pt]{%
  \mathrel{\raisebox{#1}{${#2}$}}%
  }
\DeclareMathOperator\exponential{exp}
\newcommand{\makedot}[1]{\fill[fill=black]  {#1} ++ (0:\dotwidth*0.5) arc (0:360:\dotwidth*0.5) -- cycle ;}%
\newcommand{\makewhitedot}[1]{\draw[black,fill=white, very thin]  {#1} ++ (0:\dotwidth*0.5) arc (0:360:\dotwidth*0.5) -- cycle ;}%
\newcommand{\lengthmark}[6]{
\draw[very thin] (#1 - \gap*0.5 * cos #3 - \tickwidth/2 * sin #3,#2 - \gap*0.5 * sin #3  + \tickwidth/2 *cos #3) -- (#1 - \gap*0.5 * cos #3 + \tickwidth/2 * sin #3,#2 - \gap*0.5 * sin #3- \tickwidth/2 *cos #3);
\draw[very thin] (#1 + #4 * cos #3 + \gap*0.5 * cos #3 - \tickwidth/2 * sin #3,#2 + #4 * sin #3 + \gap*0.5 * sin #3  + \tickwidth/2 *cos #3) -- (#1 + #4 * cos #3 + \gap*0.5 * cos #3 + \tickwidth/2 * sin #3,#2 + #4 * sin #3 + \gap*0.5 * sin #3- \tickwidth/2 *cos #3);
\draw[densely dotted,very thin] (#1 - \gap*0.5 * cos #3 ,#2 - \gap*0.5 * sin #3  ) -- (#1 + #4 * cos #3 /2 ,#2 + #4 * sin #3  /2 )  -- (#1 + #4 * cos #3 + \gap*0.5 * cos #3 ,#2 + #4 * sin #3 + \gap*0.5 * sin #3 ) ;
\node  at (#1 + #4 * cos #3 /2 + #5 *  sin #3  ,#2 + #4 * sin #3  /2 - #5  * cos #3 ) {#6};
}%
\newcommand{\showhideproposition}[2]{#1}
\newcommand{\switchtoload}[2]{#2}

\newcommand{\scaley}[2]{#2}
\newcommand{\invariantmeasurefunction}{\nu}%
\newcommand{\invariantmeasure}[1]{\invariantmeasurefunction({#1})}%
\newcommand{\moveperiod}[2]{#2}%{(#1)-(#2)}

\newcommand{\separator}{\scaley{}{}&\scaley{}{}}
\newcommand{\beginline}{\,}
\input{\setpath%
QUESTA20_macros} 

\renewcommand{\newSing}{\mathcal{P}_{\Wt}}
\renewcommand{\binomialcoef}[2]{\big({}^{#1}_{#2}\big)}

\renewcommand{\inflexionpoint}[1]{\bl^\star(#1)}
\renewcommand{\admissionicostux}[4]{\admissionioperator{#1}\left(#2;#3,#4\right)}

\renewcommand{\decisionset}{\pi}

\renewcommand{\Valuetfunction}[1]{\Valuefunction}
 
\renewcommand{\polypurevaluenfunction}[1]{\polypurevaluefunction^{(#1)}}
\renewcommand{\Costix}[2]{\costfunction\left(#1,#2\right)}
\renewcommand{\tim}{n}
\renewcommand{\Tim}{N}
\renewcommand{\altbl}{t}
\renewcommand{\load}{\switchtoload{\ar\expectation{\stochst}}{\rho}}
\renewcommand{\loadzero}{\switchtoload{\ar\expectation{\stochstzero}}{\tilde{\load}}}
\renewcommand{\i}{i}
\renewcommand{\fLSTWtxs}[2]{\fLSTWtxfunction{#1}(#2)}

\renewcommand{\iLSTWtxs}[2]{\iLSTWtxfunction{#1}(#2)}

\renewcommand{\stochst}{X}
\renewcommand{\st}{x}
\renewcommand{\extpurevaluefunction}{\purevaluefunction}%{\tilde\purevaluefunction}

\renewcommand{\Upurevaluefunction}{\purevaluefunction}%{\makeU\purevaluefunction}

\begin{document}

\title{Dispatching to Parallel Servers %Solutions of Poisson's Equation for First-Policy Improvement%
\thanks{%
The author %gratefully
acknowledges %financial 
support from  the French National Research Agency (project ORA\-CLESS, ANR\textendash 16\textendash CE33\textendash 0004\textendash 01).
Part of this work was completed %during his stay 
at the Department of Communications and Networking, Aalto University, Espoo, Finland,  with support from the Academy of Finland in the project FQ4BD (Grant No. 296206).
%Grants or other notes about the article that should go on the front page should be placed here. General acknowledgments should be placed at the end of the article.
}
}
\subtitle{Solutions of Poisson's Equation for First-Policy Improvement}

\titlerunning{Solutions of Poisson's Equation for First-Policy Improvement}        % if too long for running head

\author{Olivier Bilenne%First Author         \and        Second Author %etc.
}

\authorrunning{Olivier Bilenne} % if too long for running head

\institute{%
Olivier Bilenne 
\at
Univ. Grenoble Alpes, CNRS, Inria, Grenoble INP, LIG, 38000 Grenoble, France\\%, and a visitor at Aalto University, Helsinki, Finland; 
\email{olivier.bilenne@inria.fr}
%    F. Author \at
%              first address \\
%              Tel.: +123-45-678910\\
%              Fax: +123-45-678910\\
%              \email{fauthor@example.com}           %  \\
}

\date{}%{Received: date / Accepted: date}

\maketitle

\begin{abstract}
Policy iteration techniques for multiple-server dispatching rely on the computation of value functions. In this context, we consider the continuous-space M/G/1-FCFS queue endowed with an arbitrarily-designed cost function for the waiting times of the incoming jobs. The associated value function is a solution of Poisson's equation for Markov chains, which in this work we solve in the Laplace transform domain by considering an ancillary, underlying stochastic process extended to (imaginary) negative backlog states. This construction enables us to issue closed-form value functions for polynomial and exponential cost functions and for piecewise compositions of the latter, in turn permitting the derivation of interval bounds for the value function in the form of power series or trigonometric sums. We review various cost approximation schemes and assess the convergence of the interval bounds these induce on the value function. Namely: Taylor expansions (divergent, except for a narrow class of entire functions with low orders of growth), and uniform approximation schemes (polynomials, trigonometric), which achieve optimal convergence rates over finite intervals. This study addresses all the steps to implementing dispatching policies for systems of parallel servers, from the specification of \emph{general} cost functions towards the computation of interval bounds for the value functions and the exact implementation of the first-policy improvement step.
% Dispatching;Policy iteration;First-policy improvement;Poisson equation;M/G/1 queue
%02.30.Lt Sequences, series, and summability
%02.30.Mv Approximations and expansions
%%02.30.Nw Fourier analysis
%%02.30.Sa Functional analysis
%02.30.Uu Integral transforms
%%02.50.−r Probability theory, stochastic processes, and statistics
%02.50.Ga Markov processes
%02.50.Le Decision theory and game theory
%
%%30B40: Analytic continuation
%40A30: Convergence and divergence of series and sequences of functions
%41A25: Rate of convergence, degree of approximation
%41A50: Best approximation, Chebyshev systems
%41A10: Approximation by polynomials
%42A10: Trigonometric approximation
%44A10: Laplace transform
%%60G10: Stationary processes
%%60J27: Markov chains with continuous parameter
%%60K15: Markov renewal processes, semi-Markov processes
%60K20: Applications of Markov renewal processes (reliability, queueing networks, etc.)
%60K30: Applications (congestion, allocation, storage, traffic, etc.)
%%62C99: None of the above, but in this section (Decision Theory)
%62E20: Asymptotic distribution theory
%%65D15: Algorithms for functional approximation
%65R10: Integral transforms
%%65T40: Trigonometric approximation and interpolation
%90B22: Queues and service
%90B80: Discrete location and assignment
%
\end{abstract}

%%%%%%%%%%%%%%%%%%%%%%

\section{Introduction}

An essential design aspect for systems of parallel servers resides in the allocation of the processing resources to the impending workload. 
% This allocation problem, commonly referred to as \emph{dispatching} (also: \emph{task assignment} or \emph{routing}), consists of assigning one of the parallel servers to each incoming job in a way so as to minimize a performance metric of interest: 
In the allocation problem, commonly referred to as \emph{dispatching} (also: \emph{task assignment} or \emph{routing}), one server must be assigned to each incoming job in a way so as to minimize a performance metric of interest: 
%
%typically, the waiting\,/\,sojourn times in the server queues, or overall power consumption. The dispatching problem is relevant in diverse domains of application including %, among others, 
parallel computing (mobile cloud computing, server clusters, supercomputers), industrial logistics (%manufacturing systems, 
customer service systems), and traffic congestion management (visitor queues, road tolls).

\begin{figure}
\centerline{
  \begin{tikzpicture}[xscale=0.51,yscale=0.51,font=\scriptsize]
%    \draw[thick,->] (-3.0,2.65)--(-2,2.65) node [above,xshift=-30,yshift=5] {\parbox{20mm}{Arriving jobs}} node [below,xshift=-30,yshift=-5] {$ \vect{\stochst,\stochstzero} $} node [below,xshift=-30,yshift=-15] {  $\stochTk{\tim}\sim\exponentialdistri{\ar} $};
    \draw[thick,->] (-3.0,2.65)--(-2,2.65) node [above,xshift=-30,yshift=5] {\parbox{20mm}{Incoming jobs}} node [below,xshift=-30,yshift=-8] {$ \vect{\stochstk{\tim};\stochTk{\tim}} $};
    \node at (0.8,2.65)  {Dispatcher $\initialpolicyxy{\bl}{\st}$};
    \foreach[evaluate=\x using 3.0-\j/4] \j in {1,...,8} 
                 \draw[thick] (-6,2.1) +(\x,0.1)--+(\x,0.9); % job #j
   % \node at (7,6.0) {Servers};
    \foreach \position/\jobs in {{(5,0.2)/5},{(5,1.5)/3},{(5,3)/4},{(5,4.5)/3}}
             { %\draw \position rectangle +(.25,.5);
               \draw[fill,color=gray!29] \position +(0.9,0) rectangle +(3,1);
               \draw \position +(0.9,0)--+(3,0)--+(3,1)--+(0.9,1);
               % jobs
               \foreach[evaluate=\x using 3.0-\j/4] \j in {1,...,\jobs} 
                 \draw[thick] \position +(\x,0.1)--+(\x,0.9); % job #j
               \draw[fill,color=gray!29] \position +(3.55,0.5) circle (0.5) ;
               \draw \position +(3.55,0.5) circle (0.5) ;
               % arriving arrow
               \draw[thick,<-] \position +(0.8,0.5) -- +(-0.1,0.5) -- (3.5,2.65);
             }
	\node at (11.5,5.2)  {FIFO Server $1$};
	\node at (11.5,4.7)  {(backlog $\bli{1}$)};
	\node at (11.5,3.7)  {FIFO Server $2$};
	\node at (11.5,3.2)  {(backlog $\bli{2}$)};
	\node at (11.5,2.2)  { $\vdots$};
	\node at (11.5,0.8)  {FIFO Server $\nbservers$};
	\node at (11.5,0.3)  {(backlog $\bli{\nbservers}$)};
     \draw[thick,fill] (3.35,2.45) rectangle +(0.4,0.4);
     %\draw (2,1.5) node {Dispatcher};
     %\draw (2.7,2) -- (3.2,2.4);
  \end{tikzpicture}
  } 
\caption{Size-aware dispatching with i.i.d. service times ($\stochstk{\tim}\eqdistribution \stochst$ for all~$\tim$)  and i.i.d. inter-arrival times ($\stochTk{\tim}\sim\exponentialdistri{\ar} $ for all~$\tim$) to~$\nbservers$ 
M/G/1-FCFS servers with backlog $\bl=\vect{\bli{1},\dots,\bli{\nbservers}}$.
\label{figure:dispatcher}}
\end{figure} 

  We are interested in sytems composed of several first-come, first-served (FCFS) queueing servers %~\cite{gross98} 
  operated in parallel, and fed with a sequence of jobs with Markovian arrival times. % (i.e., modulated by a Poisson process, \cite{gallager13}). 
  In our model, illustrated in Figure~\ref{figure:dispatcher}, every new job turning up at the dispatcher is instantly forwarded towards one of the servers, where a penalty is incurred
  %after every newly turned-up job is dispatched to a server. 
  as a function of the backlog (uncompleted work) at the %allotted
  server upon job arrival---server backlog thus coinciding with the waiting time of the job until processing begins. Our objective is to minimize the average cost experienced by the system over an infinite time horizon. 

A standard approach for solving this problem is through \acdef{PI}, \cite{howard60,bertsekas07}. Starting with an inital dispatching policy, PI proceeds in two steps, repeated in turn until a fixed policy is reached: (i) \emph{policy evaluation}, where the mean cost of the considered policy is computed, together with a \emph{value function} expressing state sensitivity with respect to the steady-state costs induced by the policy; followed by (ii) \emph{policy improvement}, where the value function is exploited to improve the current policy and derive a new, more cost-effective dispatching policy.
%Beyond finite-state settings---where dynamic programming may apply---, 
The policy evaluation step is difficult to implement  in continuous state spaces without extensive Monte Carlo simulations. 
Only the first PI iteration on a tractable, random initial  policy is easier to carry out, because
%Only the first PI iteration on a tractable Bernoulli-splitting initial  policy may be carried out straightforwardly, as %for the reason that 
%the Poisson process controlling the job arrivals then 
the job flow then decomposes into independent Poisson processes for the individual queueing servers, and the value function takes a separable form, solution of the so-called Poisson equation. %that can be derived, possibly in closed-form, as a 
%solution of the so-called Poisson equation.
The \acdef{FPI} approach (also known as \emph{one-step policy iteration}, and variants) consists of cutting short the policy iteration algorithm after the first iteration. 
%The \acdef{FPI} approach (also known as \emph{one-step policy iteration}, and variants) consists of cutting short the policy iteration algorithm after the first iteration. 
The motivation behind \ac{FPI} is twofold: it is known that a single iteration of the PI algorithm may produce fine heuristics (see e.g.~\cite{krishnan87,wijngaard79,ott92,sassen97,bhulai06} or \cite[\S{}7.5]{tijms03}) and, besides, the Poisson equation for Markov chains admits explicit solutions readily available for effortless~\ac{PI}.

\emph{Related work and our contribution.}
The existence of explicit solutions to the Poisson equation for the waiting times of the M/G/1 queue was pointed out in~\cite{glynn94}, 
where a general solution to Poisson's equation was proposed in the form of a fundamental kernel, whose application to the cost function produces solutions of the equation.
%through which the solutions of the equation can be derived by application of the kernel to the cost function.
%which, to our knowledge, remains the work of reference on the topic. Relying on the general theory of Poisson's equation for Markov chains~\cite{neveu72,nummelin91}, the author of~\cite{glynn94} derives such  solutions by application of a `solution kernel' to the cost function.
% a solution kernel that yields solutions to the Poisson equation by application to the cost function.  
These solutions proved, in particular, to take closed forms for cost functions given as moments of the waiting time, $\costx{\bl}=\bl^{\nn}$.  
There followed a list of derivations of explicit value functions for Markovian queueing systems: both in discrete-space settings where only the number of yet unprocessed jobs at the servers is known to the dispatcher and (typically) the expected sojourn times of the incoming jobs are penalized, \cite{krishnan87,sassen97,bhulai03,bhulai06}; and in `size-aware' continuous-space settings where the service times of the jobs become available to the dispatcher upon arrival and %where (early moments of) the  
the actual waiting or sojourn times are penalized, \cite{Aalto96NTS,%hyytia11DispatchingProblem,
hyytia11MM1PS,hyytia12OnTheVF,hyytia-ejor-2012,%Hyytia14EnergyAware,
hyytia14TaskAssignment%,hyytia15Fairness
}.
%More recent studies, witnessing renewed interest in solving the Poisson equation, have seen the class of explicit solutions 
%
Recent studies on size-aware dispatching renewed the interest in explicit Poisson solutions, extending their class in~\cite{hyytia-peva-2017} to the fixed-deadline cost functions $\costx{\bl}= \indicator{[\taufunction,\infty)}{\bl}$, %in the special case of identical jobs,
%and to exponential costs (with views on polynomials) in~\cite{%hyytia-vgf-itc-2017,hyytia-peva-2020}, where the solution to the Poisson equation for queueing systems reappears no longer in its kernel form, but as a (provably equivalent) integral expression involving the steady-state probability measure of the waiting times in the queues.
%We find the latter formulation of the solutions more eloquent from the perspective of complex analysis,  as it reduces in the transform domain to multiplying the costs by the Lapace-Stieltjes transform of the waiting time distribution, available through the Pollaczek–Khinchine formula.
%
and to exponential costs  in~\cite{%hyytia-vgf-itc-2017,
hyytia-peva-2020}, with views on polynomials.
In the discrete space setting, the forms~$\costn{\bl}=\bl^{\nn} \sing^{-\bl}$ and~$\costn{\bl}=\dirack{\bl-\sing}$ were identified in~\cite{turck12} as candidates for closed-form value functions, via transform-domain analysis (based on generating functions) of the general solution of Poisson's equation---a methodology in spirit similar to the approach we will use in our study.

In this work we extend the collection of explicit solutions of the continuous-space Poisson equation to $\costx{\bl}=\bl^{\nn} \exp{-\sing \bl}$, and we
% 
%  perspective of complex analysis,  as it reduces in the transform domain to multiplying the costs by the Lapace-Stieltjes transform of the waiting time distribution, available through the Pollaczek–Khinchine formula.
% This observation leads us in this work to extend  the collection of explicit solutions of the Poisson equation to $\costx{\bl}=\bl^{\nn} \exp{-\sing \bl}$%
% %characterized by meromorphic Laplace transforms
% ---the continuous counterpart of the discrete form~$\costn{\bl}=\bl^{\nn} \sing^{-\bl}$, previously identified as a candidate for closed-form value function in the discrete-space setting of~\cite{turck12}, where a transform-domain analysis (based on generating functions) of the general solution of Poisson's equation was taken, in spirit similar to the approach we will use in this work%
% ---, and to 
develop a methodology based on complex analysis for solving Poisson's equation that covers a more general class of piecewise continuous cost functions. Our motivation behind piecewise-definite functions is the possibility they offer to derive, under mild conditions for the cost function, %of the M/G/1 queue 
%(i.e, Lipschitz or Hölder continuity over finite intervals), 
tight bounds to the corresponding value function, which enable us to perform the \ac{FPI} step exactly. %Our developments are, to our knowledge, the first 
Our developments depart from previous studies by proposing 
%What distinguishes our developments fom    
a comprehensive implementation of \ac{FPI} in continuous spaces with cost functions of {any} general kind.
%, which enables us to implement \ac{FPI} with cost functions of any general kind.
%bounds to the corresponding value function as accurate as required by \ac{FPI}.

\emph{Outline.} The paper is structured as follows. Section~\ref{section:preliminaries} contains a more detailed presentation of~\ac{FPI} %(\ref{section:FPI}) 
and introduces the value function as the solution of  Poisson's equation.
%and surveys some identities linked to the value function (\ref{section:valuefunction}).  These identities lead us in Section~\ref{section:closedformVF} to reconsider the 
This equation is solved in Section~\ref{section:closedformVF} from the viewpoint of complex analysis (\ref{section:complexcharacterization}); complex analysis which allows us to derive the value function of the M/G/1 queue for cost functions of the type $\costx{\bl}=\bl^{\nn} \exp{-\sing \bl}$ (\ref{section:smoothcf}), and to provide a method of solution for piecewise-defined costs (\ref{section:piecewisedefined}). Various solutions previously reported in the literature
are then reconciled through basic case studies (Appendix~\ref{appendix:examples}).
In Section~\ref{section:approximatevaluefunctions} 
%we seek to extend the range of explicit value functions by considering 
we consider cost functions given as convergent series: successively, Taylor series  (\ref{section:entirecostfunctions}), Bernstein polynomials (\ref{section:bersteinpolynomials}),
trigonometric sums (\ref{section:trigonometricsum}), 
and near-optimal polynomials (\ref{section:nearoptimalpolynomials}); and we propose  an algorithm for computing \ac{FPI} policies based on approximations of the cost functions.  We conclude with a full implementation of the \ac{FPI} dispatcher for the cost function $ \costx{\bl}={\bl^2}/({\sing^2+\bl^2}) $, picked for illustrative purposes, in the case of a two-server system with exponentially distributed service times.

%Section~\ref{section:discussion}

\obsolete{
The design of multiple-server dispatching techniques based on policy iteration requires the accurate computation of a so-called {value function}. The \emph{value function} of a queueing system under a given dispatching policy assigns to each state an expression of the future costs to be expected when the system is initiated at that state. In a multiple-server setting, a given initial dispatching policy can be improved based on predictions, at each candidate server, of the variations in the value function under the considered policy. The value function is then recalculated for the improved policy, and the procedure is iterated towards the optimal dispatching policy~\cite{bertsekas07}.
It is known that fine dispatching policies may be obtained after a single iteration of the policy improvement procedure (see e.g.~\cite{krishnan87,wijngaard79,ott92,sassen97,bhulai06} or \cite[\S{}7.5]{tijms03}). %|this is the principle of {one-step policy improvement}\cite{ott92}. 
One interesting aspect of {one-step policy improvement} is its computational efficiency in certain settings where explicit expressions for the value function are available under specific initial policies. %, and in the availability, in certain settings, of explicit expressions for the value function under specific initial policies. 
Random, state-independent dispatching of Poisson arrivals, in particular, decompose an~$\nn$-server setting into~$\nn$ queueing systems fed by independent Poisson processes|these processes can be analyzed separately as the value-function of the multiple-server system separates additively at the server level~\cite{hyytia-peva-2017}. %yielding an additively separable expression for the value function~\cite{hyytia-peva-2017}.
%
%
%\section{Preliminaries}
%\subsection{Value functions of M/G/1 queues}

}%

 %%%%%%%%%%%%%%%
 %%%%%%%%%%%%%%%

 \section{Preliminaries} \label{section:preliminaries}

\subsection{Policy iteration and \acl{FPI}.} \label{section:FPI}

Consider the system depicted in Figure~\ref{figure:dispatcher}, where jobs, arriving according to a Poisson process with rate~$\ar$, are dispatched upon arrival towards one of the~$\nbservers$ servers ($1,\dots,\nbservers$) selected by a (possibly random) dispatching policy $\initialpolicyxy{\bl}{\st}$, where~$\bl=\vect{\bli{1},\dots,\bli{\nbservers}}\in\Realpluszero^\nbservers$ denotes the server backlog vector and~$\st=\vect{\sti{1},\dots,\sti{\nbservers}}\in\Realpluszero^\nbservers$ are the prospective service times of an incoming job at the servers. 
By taking snapshots at intitial time~$\tim=0$ and at the job arrival times ($\tim=1,2,3,\dots$), the continuous-time system reduces to a \ac{MDP}, $(\MDPpolicy{\initialpolicyfunction}{\tim})_{n\in\Natural}$, with state $\MDPpolicy{\initialpolicyfunction}{\tim} = \vect{ \stochbli{\tim}  , \stochstk{\tim} } \in \MDPspace\equiv  \Realpluszero^\nbservers\times\Realpluszero^\nbservers$, where~$\stochstk{\tim}$ is the service time vector of the $\tim$th job and~$\stochbli{\tim}$ is the backlog of the system  at the time of arrival, and with transition probability kernel $\MDPkernelfunction=\vect{\MDPkernelifunction{1},\dots,\MDPkernelifunction{\nbservers}}$ such that, 
%for $\server=1,\dots,\nbservers$,
for any~$\tim$ and every $\vect{\bl,\st}\in\MDPspace$, $\MDPsubspace\subset\MDPspace$,   
\begin{equation} \label{transitionkernel}\notag
 \begin{array}{c} 
\MDPkerneli{\server}{\bl,\st}{\MDPsubspace} 
\defeq
\proba{\MDPpolicy{\initialpolicyfunction}{\tim+1}\in \MDPsubspace | \MDPpolicy{\initialpolicyfunction}{\tim} =
\vect{ \bl,\st},\, 
%\initialpolicy{\tim} 
\initialpolicyxy{\bl}{\st} 
= \server}
\qquad\qquad \\ \hfill 
=
\MDPkerneli{i}{\vect{\bli{1},\dots,\bli{\server}+\st,\dots,\bli{\nbservers}} , 0}{\MDPsubspace} 
.
\end{array}
\end{equation}  
%holds for any~$\tim$ and for every $\MDPsubspace\subset\MDPspace$, $\vect{\bl,\st}\in\MDPspace$. %, where $\initialpolicy{\tim} = \initialpolicyxy{\bl}{\st}  $.
%
%
Assume that the performance of the system is measured by a  cost function $\costfunction=\vect{\costifunction{1},\dots,\costifunction{\nbservers}}$, where 
$\costsx{\server}{\bl,\st}  \equiv \costix{\server}{\bli{\server}}  \, \indicator{\Realplus}{x} $ models a penalty incurred when a job with service time~$\st$ joins server~$\server$, given backlog state $\bl=\vect{\bli{1},\dots,\bli{\nbservers}}$.
We would like to minimize the expected total cost,  defined by
\begin{equation}\notag
%\textstyle
 J_\initialpolicyfunction
 \obsolete{
 = 
 \limsup_{\Tim\to\infty} \frac{1}{\Tim} \sum_{\tim=1}^{\Tim} 
%\condExpectation{  \Costix{\initialpolicy{\tim}}{\MDPpolicy{\initialpolicyfunction}{n,\initialpolicy{\tim}} } }{\MDPpolicy{\initialpolicyfunction}{0} = \bl}
\condExpectation{  
\Costix{
%\initialpolicy{\tim}
\initialpolicys{\MDPpolicy{\initialpolicyfunction}{\tim}} 
}{\MDPpolicy{\initialpolicyfunction}{\tim} }}{\MDPpolicy{\initialpolicyfunction}{0} =\vect{ \bl,\st} }
}
=
\limsup_{N\to\infty} \frac{1}{\Tim} \sum_{\tim=1}^{\Tim} 
%\Expectation{  \Costix{\initialpolicy{\tim}}{\MDPpolicy{\initialpolicyfunction}{n,\initialpolicy{\tim}} } }
\Expectation{  
\Costix{
%\initialpolicy{\tim}
\initialpolicys{\MDPpolicy{\initialpolicyfunction}{\tim}} 
}{\MDPpolicy{\initialpolicyfunction}{\tim} } }
,
\end{equation}
 independently of% the initial state
 ~$\MDPpolicy{\initialpolicyfunction}{0}$.
The optimality equations of the system are % given by
\newcommand{\eqrefOE}{\text{(OE)}}%
\begin{align}
\tag{OEa}
\label{OE1}
\altgenericx{\bl,\st} 
&
\textstyle
= 
\min_\server \left[    
%\costsx{\server}{\bl,\st} 
\costix{\server}{\bli{\server}}
+  \MDPkernelifunction{\server} \altgenericx{\bl,\st}\right] - \varmeancost
,
\\
\tag{OEb}
\label{OE2}
 \osipolicyx{\bl,\st} 
 &
 \textstyle
\in\arg\min_\server \left[     
 %\costsx{\server}{\bl,\st}  
 \costix{\server}{\bli{\server}}
 +  \MDPkernelifunction{\server} \altgenericx{\bl,\st}\right] 
,
\end{align}
where %, for any function~$\altgenericfunction$ on~$\MDPspace$, we write
 $ \MDPkernelifunction{\server} \altgenericx{\bl,\st} 
  =
%\iint
\int_\MDPspace
 \altgenericx{\altbl,\altst}  \MDPkerneli{\server}{\bl,\st}{d\vect{\altbl,\altst}}% =  \int_{\Realpluszero\times\Realpluszero} \genericx{S}  \MDPkerneli{\server}{\bl+\st\, e_i,0}{dS} 
\equiv
 \MDPkernelifunction{\server} \altgenericx{\vect{\bli{1},\dots,\bli{\server}+\sti{\server},\dots,\bli{\nbservers}},0} 
%:=\valuex{\bl+\st\, e_i}
.
$
%
%We know from~\cite{sennott89,arapostathis93,meyn96,meyn97} that if one can 
If one can find  $\varmeancost^*>0$, a  policy~$\osipolicyfunction^*$, and an integrable function~$\altgenericfunction$ % meeting specific conditions, such that%
%\begin{enumerate*}[(i)]
%\item 
such that~$\vect{\altgenericfunction,\varmeancost^*}$ solves~\eqref{OE1} %,
%\item 
and
$\osipolicyfunction^*$  satisfies~\eqref{OE2},
%and
%\item 
%$\altgenericfunction$ satisfies certain integrability conditions,
%$\forall \bl,\osipolicyfunction$ such that $ J_\osipolicyfunction(\bl)<\infty$, $ \frac{1}{\tim}  \MDPkernelifunction{\osipolicyfunction}^n \valuex{\bl} \to 0 $ as $n\to \infty$,
%\end{enumerate*}
then~$\osipolicyfunction^*$ is the {optimal policy} and
$
\textstyle
\varmeancost^*
=
\lim_{\Tim\to\infty} ( 1/ \Tim)
\sum_{\tim=1}^{\verticalshiftby[-5pt]{\scriptstyle{\Tim}}} 
\Expectation{  \Costix{\osipolicyfunction^*(\MDPpolicy{\osipolicyfunction^*}{\tim})}{\MDPpolicy{\osipolicyfunction^*}{\tim} } }
$ 
is the {optimal cost} of the system, \cite{sennott89,arapostathis93,meyn96,meyn97}.
%~\cite{meyn96,sennott89,arapostathis93,meyn97}.
%
The \emph{policy iteration} algorithm for solving~\eqrefOE{} can be described as follows, \cite{howard60}. Given an initial policy~$\osipolicyn{0}$, find, for~$\alttim\geq 0$, a function $\altgenericnfunction{\alttim}$, a mean cost $\varmeancostn{\alttim}$, and  a policy $\osipolicyn{\alttim+1}$ satisfying
\newcommand{\eqrefPI}{\text{(PI)}}%
\begin{align}
\textstyle
\altgenericnx{\alttim}{\bl,\st} &
\textstyle
=   \costsx{  \osipolicynx{\alttim}{\bl,\st}  }{\bl,\st}  +  \MDPkernelifunction{\osipolicynx{\alttim}{\bl,\st}} \altgenericnx{\alttim}{\bl,\st}  - \varmeancostn{\alttim}
\label{PI1}\tag{PIa}
\\
\textstyle
 \osipolicynx{\alttim+1}{\bl,\st} &
 \textstyle
\in\arg\min_\server \left[    { \costix{\server}{\bli{\server}} }  +  
 \valuenx{\alttim}{\bli{1},\dots,\bli{\server}+\sti{\server},\dots,\bli{\nbservers}} 
 %\MDPkernelifunction{\server} \altgenericnx{\alttim}{\bl,\st} 
\right] 
\label{PI2}\tag{PIb}
%\\
\end{align}
where~\eqref{PI1} is the {policy evaluation} step, \eqref{PI2} is the {policy improvement} step 
and
$
\valuenx{\alttim}{\bl}  
:=
% \MDPkernelfunction \altgenericnx{\alttim}{\bl,0}  - \varmeancostn{\alttim}
%=
%\iint
\int_\MDPspace
 \altgenericnx{\alttim}{\altbl,\altst}  \, \MDPkernel{\bl,0}{d\vect{\altbl,\altst}} - \varmeancostn{\alttim}
$ defines the \emph{value function} under policy~$\osipolicyn{\alttim}$. 
% as an expression of the total cost surplus to be expected with respect to the the steady-state costs when the system is initiated at backlog state~$\bl$.
%
Under favorable conditions, $\vect{ \osipolicyn{\alttim}, \varmeancostn{\alttim}}$ eventually converges towards a solution~$ \vect{\osipolicyfunction^*,\varmeancost^*}$.
Solving~\eqref{PI2}, however, is generally difficult.  
% and~\eqref{PI} is only pointwise computable.
%\end{itemize}

The first iteration of~\eqrefPI{} may still be implemented easily if the initial policy~$\osipolicyn{0}\equiv\initialpolicyfunction$ is a random Bernoulli-split between the servers. 
In that case, the multiple-server system decomposes into~$\nbservers$ independent M/G/1 queues with arrival rates $\ari{1},\dots,\ari{\nbservers}$ and transition probability kernels $\MDPkernelarfunction{\ari{1}},\dots,\MDPkernelarfunction{\ari{\nbservers}}$, with $\ari{1}+\dots+\ari{\nbservers}=\ar $.
%with arrival rates $\ari{1},\dots,\ari{\nbservers}$, such that $\ari{1}+\dots+\ari{\nbservers}=\ar $, and the transition probability kernels $\MDPkernelarfunction{\ari{1}},\dots,\MDPkernelarfunction{\ari{\nbservers}}$ characteristic of the M/G/1 queues.
%$$\server=1,\dots,\nbservers$
%\MDPkernelar{\ari{\server}}{\bli{\server}}{S_\server}
%\MDPkerneli{\server}{\bli{\server}}{S_\server}
%=
%\proba{ U_{\server}^{\tim+1} \in S_\server| U_{\server}^{\tim}=\bli{\server}}
%,\hspace{-5cm}\  
%$
% for $\server=1,\dots,\nbservers$. 
% It follows that~\eqref{PI1} decomposes
%
%Stopping the PI algorithm after a single iteration then gives us the \emph{first-policy improvement} approach, which reduces to solving
%
The \emph{first-policy improvement} approach then consists in stopping the PI algorithm after a single iteration, by solving
\noeqref{1SPI}%
\begin{align}
\altgenericix{\server}{\bli{i},\sti{i}} 
&=
%\MDPkernelifunction{\server} 
\MDPkernelarfunction{\ari{\server}} 
\altgenericix{\server}{\bli{i},\sti{i}}  +  \costix{\server}{\bli{\server}}  - \varmeancosti{\server} 
,& (\server=1,\dots,\nbservers)
\label{1SPoisson} \tag{FPIa}
%\end{align}
%%\\
%for~$\altgenericifunction{\server}$ and~$\varmeancosti{\server}$, and then inferring
%\begin{align}
\\
\label{1SPI}\tag{FPIb}
\newpolicyx{\bl,\st} 
&\in 
\arg\min_\server 
%\admissionfoperator{\ari{\server}}{\costifunction{\server}}\vect{\bli{\server},\sti{\server}}
% \admissioniux{\server}{\bli{\server}}{\sti{\server}}
 \admissioniux{\server}{\bl}{\st}
%\defeq  
%\costix{\server}{\bli{\server}}  + \valueix{\server}{\bli{\server}+\sti{\server}} - \valueix{\server}{\bli{\server}}     
%{   \costsx{\server}{\bl}  + \MDPkernelifunction{\server}   \altgenericx{\bl,\st} }   
,&
\end{align}
where
\begin{equation} \label{definitionadmissioncost} \tag{AC}
% \admissionfoperator{\ari{\server}}{\costifunction{\server}}\vect{\bli{\server},\sti{\server}}
 \admissioniux{\server}{\bl}{\st}
=  
\costix{\server}{\bli{\server}}  + \valueix{\server}{\bli{\server}+\sti{\server}} - \valueix{\server}{\bli{\server}}      
\end{equation}
is the {admission cost} at server~$\server$, and $
 \valueix{\server}{\altbl}  
 \defeq
%\altgenericix{\server}{\altbl,0} - \costix{\server}{\bli{\server}}
 %= 
 \MDPkernelarfunction{\ari{\server}}   \altgenericix{\server}{\altbl,0} - \varmeancosti{\server}
$ 
is the {value function} at~$\server$.
% Observe that~\eqref{1SPoisson} is an instance of the Poisson equation~\cite{neveu72,nummelin91}, $\altgenericfunction = P\altgenericfunction + f$, with $\int f(\bl)\, \invariantmeasure{d\bl} = 0$, where~$\invariantmeasurefunction$ denotes the (unique up to a multiplicative constant~\cite{athreya78}) non-trivial measure invariant for the transition kernel~$P$, i.e., such that $P \invariantmeasurefunction = \invariantmeasurefunction$.  
Observe that~\eqref{1SPoisson} is an instance of the Poisson equation $\altgenericfunction = P\altgenericfunction + f$ under $\int f(\bl)\, \invariantmeasure{d\bl} = 0$, 
where~$\invariantmeasurefunction$ denotes the %(unique up to a multiplicative constant) 
non-trivial measure invariant for the transition kernel~$P$ (i.e., $P \invariantmeasurefunction = \invariantmeasurefunction$), \cite{neveu72,nummelin91,athreya78}. In~\eqref{1SPoisson}, $\invariantmeasurefunction$ coincides with the asymptotic probability measure of the waiting times at server~$i$. 
%A first consequence of this observation is that all
All integrable solutions of~\eqref{1SPoisson} with respect to the asymptotic waiting time probability measure are equal up to an additive constant, \cite{glynn96}. Besides, due to the existence of a strong law of large numbers and a central limit theorem for the 
%series 
%$
%S_\server({\tim}) = \sum_{\tim=1}^{\Tim} [ \costix{\server}{\MDP{\tim}} -   \varmeancosti{\server}] 
%$, which behaves like a martingale
costs, \cite{glynn94,glynn96},
$\valueifunction{\server}$ and~$\varmeancosti{\server}$ can be estimated empirically, though at the price of extensive numerical simulations.
Lastly, and preferably, some solutions of~\eqref{1SPoisson} are known to exist in closed form; deriving explicit solutions of this kind is the direction we will explore in this work.

\obsolete{

For a server with Poisson job arrivals with rate~$\ar$ and endowed with a cost function~$\costfunction$, the admission cost  of a job with service time~$\st\in\Realpluszero$ at a server with backlog state~$\bl\in\Realpluszero$ is defined by 
\begin{equation}\label{definitionadmissioncost}
\begin{array}{c}%{rcl}
%\admissionixy{\server}{\bl}{\st}=\valueix{\server}{\bl+\st}-\valueix{\server}{\bl}=\purevalueix{\server}{\bl+\st}-\purevalueix{\server}{\bl} -\frac{\ar\meancosti{\server} }{1-\load} \st
\admissionfoperator{\ar}{\costfunction}\vect{\bl,\st}
=   
%&=
\costx{\bl}+
\valuex{\bl+\st}-\valuex{\bl},
\end{array}
\end{equation}
where~$\valuefunction$ is the value function for~$\costfunction$ under the considered job arrival process.
Consider a $\nbservers$-server system fed with jobs with stochastically independent service times and Poisson arrivals with rate~$\ar$. If the jobs are dispatched to the servers following a  completely random initial policy specified by a random i.i.d. process~$\{\initialpolicy{t}\}$ on $\{1,\dots,\nbservers\}$ such that $\proba{\initialpolicy{t}=i}=\probai{\server}\in[0,1]$ for every arrival time~$t$ and for every $i\in\{1,\dots,\nbservers\}$ ($\sum_{i=1}^\nbservers \probai{\server}=1$), then the arrival process decomposes at the servers  into~$\nbservers$ independent Poisson processes with rates $\ari{1}=\probai{1}\ar,\dots,\ari{\nbservers}=\probai{\nbservers}\ar$. One-step improvement of the inital random policy consists of minimizing  for every new arriving job the global admission cost under the inital random policy, i.e.
\begin{equation}\label{osipolicy}
\osipolicyxy{\bl}{\st} \in \arg \min\nolimits_{i\in\{1,\dots,\nbservers\}} \admissionfoperator{\ari{\server}}{\costifunction{\server}}\vect{\bli{\server},\sti{\server}}
\end{equation}
where~$\bl=\vect{\bli{1},\dots,\bli{\nbservers}}\in\Realpluszero^\nbservers$ is the backlog state of the servers,
the vector $\st=\vect{\sti{1},\dots,\sti{\nbservers}}\in\Realpluszero^\nbservers$ models the prospective service times of the arriving job at the candidate servers, and $\costifunction{1},\dots,\costifunction{\nbservers}$ are given cost functions.

}%

 %%%%%%%%%%%%%%%%%%%
 %%%%%%%%%%%%%%%%%%%

\subsection{Value function of the M/G/1 queue.}

\label{section:valuefunction}

In view of the previous discussion, we consider an individual server modeled by a continuous-state FCFS-M/G/1 queue. The queue is fed with a sequence of jobs with random arrival times modulated by a Poisson point process with rate~$\ar>0$, \cite{gross98,gallager13}.
%We make the assumption of i.i.d. service times for the jobs,  modeling the service time of any job that reaches the queue when busy (i.e., when the queue is processing a previous job) by a unique random variable~$\stochst$ on~$\Realpluszero$. The service times of the jobs reaching the queue when idle (i.e., all the past jobs have been processed) are also i.i.d. and modeled as in~\cite{welch64} by a second random variable~$\stochstzero$ on~$\Realpluszero$, which may differ from~$\stochst$ in distribution, thus accounting for a possible setup delay required by the queue when  waking up from its idle state.
%
The dynamics of the queue is modeled by the equation
\begin{align}
 \label{dynamics}\tag{Q}
 &
 \stochblk{\tim+1} = \sizedprojection{\big}{ \stochblk{\tim} +  \stochstk{\tim} - \stochTk{\tim+1} },
 &
  \tim\geq 0,
\end{align}
where~$\stochstk{\tim}$ denotes the service time of the $\tim$th incoming job, $\stochblk{\tim}$ is the coinciding queue backlog upon arrival, and $\stochTk{\tim+1}$ is the inter-arrival time for~$\stochstk{\tim+1}$.
\begin{notation}
%We let~$\trueWt$ denote a random variable distributed like the 
%For any random variable~$Y$ on~$\Realpluszero$, the cumulative probability distribution of~$Y$  is denoted by $\distriXfunction{Y}:\Realpluszero\mapsto[0,1]$, where $\distriXx{Y}{y}=\proba{Y\leq y}$, and its probability density function is denoted by $\ddistriXfunction{Y}:\Realpluszero\mapsto[0,+\infty]$.
For any real random variable~$Y$, the probability measure associated with~$Y$, its cumulative probability distribution, and its probability density function are respectively denoted by~$\measureXfunction{Y}$, $\distriXfunction{Y}:\Real\mapsto[0,1]$, and $\ddistriXfunction{Y}:\Real\mapsto[0,+\infty]$, with $\proba{Y\leq y}=\measureXx{Y}{(-\infty,y]}=\distriXx{Y}{y}=\int_{-\infty}^{y}\ddistriXx{Y}{\bl}\,d\bl$.
\end{notation}
The service time of every incoming job is assumed as in~\cite{welch64} to be random, conditioned on the activity of the queue at the time of arrival, and independent of the other factors; it is distributed either like the positive random variable~$\stochst$ if on arrival the queue is busy processing a previous job, or like a second positive random variable~$\stochstzero$ if the queue is idle (empty), where~$\stochstzero$ may differ from~$\stochst$ in distribution, thus accounting for a setup delay that the queue might require to  wake up from its idle state.
%We write~$\vect{\stochst,\stochstzero}$ to fully characterize the service time variables at the server.
%
%We make the assumption of i.i.d. service times for the jobs,  modeling the service time of any job that reaches the queue when busy (i.e., when the queue is processing a previous job) by a unique random variable~$\stochst$ on~$\Realpluszero$. The service times of the jobs reaching the queue when idle (i.e., all the past jobs have been processed) are also i.i.d. and modeled as in~\cite{welch64} by a second random variable~$\stochstzero$ on~$\Realpluszero$, which may differ from~$\stochst$ in distribution, thus accounting for a possible setup delay required by the queue when  waking up from its idle state.
%%
%We consider a cost  function~$\costfunction:\bl\in\Realpluszero\mapsto \costx{\bl}\in\CtoR{\Complex}{\Real}$ of the backlog~$\bl$ at the time of  arrival, where the backlog is an expression of the waiting time of the incoming job due to the unfinished work at the queue, so that  the penalty incurred by a job  arriving at state~$\bl$ is $\costx{\bl}$.  
%
It follows that of the transition kernel %~$\MDPkernelfunction$, defined by
$% \begin{array}{rcl}
\MDPkernel{\bl,\st}{\MDPsubspace} 
=
\proba{\vect{\stochblk{\tim+1},\stochstk{\tim+1}}\in \MDPsubspace%\MDPpolicy{\initialpolicyfunction}{\tim+1}
| \vect{\stochblk{\tim},\stochstk{\tim}} =
\vect{ \bl,\st}}
$ of the \ac{MDP} $(\vect{\stochblk{\tim},\stochstk{\tim}})_{n\in\Natural}$ rewrites as
% \begin{equation}\label{fullMDPkernel}
% \nocolsep
% \begin{array}{llll}
% %$
% \MDPkernel{\bl,\st}{\MDPblsubspace\times\MDPstsubspace} 
% &=&
% \MDPkernel{\bl+\st}{\MDPblsubspace} \, \measureXx{\stochst}{\MDPstsubspace} 
% %$
% ,
% &
% \text{if } 
% %$
% 0 \notin \MDPblsubspace
% %$
% , 
% \\
% %$
% \MDPkernel{\bl,\st}{\{0\}\times\MDPstsubspace}   
% &=&
% \MDPkernel{\bl+\st}{\{0\}} \, \measureXx{\stochstzero}{\MDPstsubspace} 
% &
% %$,
% ,
% \end{array}
% \end{equation}
\begin{equation}\label{fullMDPkernel}
%\nocolsep
\MDPkernel{\bl,\st}{\MDPblsubspace\times\MDPstsubspace}
=
\left\lbrace
\begin{array}{llc}
%$
\MDPkernel{\bl+\st}{\{0\}} \, \measureXx{\stochstzero}{\MDPstsubspace} 
&
\text{ if } 
&
\MDPblsubspace = \{0\}
\\
\MDPkernel{\bl+\st}{\MDPblsubspace} \, \measureXx{\stochst}{\MDPstsubspace} 
%$
&
\text{ if } 
&
%$
\MDPblsubspace \not\supset\{0\}
%$
\end{array}
\right\rbrace
,
\end{equation}
where, for all $\bl\geq 0$,
\begin{align}
\label{MDPkernel}
 &
 \MDPkernel{\bl}{[0,\altbl]} =  \exp{-\ar(\bl-\altbl)} \ \  \forall\altbl\in[0,\bl] ,
 \quad
 \MDPkernel{\bl}{\Real\setminus[0,\bl]} = 0  
 .%,
&
%(\bl\in\Realpluszero).
\end{align}

%Letting $\bl\in\Realpluszero$ denote the backlog at the queue,  we introduce 
Consider a cost function $\costfunction:\Realpluszero\mapsto \CtoR{\Complex}{\Real}$ %is considered as a function of , which specifies 
quantifying the (expected) penalty  $\costx{\bl}$ incurred when a job joins the queue at backlog state $\bl\in\Realpluszero$.
%We introduce a cost  function~$\costfunction:\bl\in\Realpluszero\mapsto \costx{\bl}\in\CtoR{\Complex}{\Real}$ of the backlog~$\bl$, so that  the penalty incurred by a job  joining at backlog state~$\bl$ is $\costx{\bl}$.  
%For any two moments in time~$\timei{1},\timei{2}$ with $\timei{1}<\timei{2}$, we let~$\arrivalset{\timei{1}}{\timei{2}}$ denote the random set of the arrival times observed during the time interval~$[\timei{1},\timei{2})$. If~$$
%
%
%We sum up the service time variables of the queue using the notation~$\vect{\stochst,\stochstzero}$. 
The stability of the queue is guaranteed by a server utilization ratio~$\switchtoload{\ar\expectation{\stochst}}{\load=\ar\expectation{\stochst}}$ less than~$1$, and by a finite mean service time at~$\bl=0$, i.e., $\loadzero=\ar\expectation{\stochstzero}<\infty$.
%a finite mean service time at~$\bl=0$\switchtoload{}{ or, equivalently, by $\loadzero<\infty$ where we define~$\loadzero=\ar\expectation{\stochstzero}$}.
% Assumption~\refeq{assumption:stability}
\begin{assumption}[Stability]\label{assumption:stability}
\switchtoload{
$\ar\expectation{\stochst}<1$, $\expectation{\stochstzero}<\infty$%
}{
$\load<1$, $\loadzero<\infty$%
}.
\end{assumption}
\noindent
All in all, the server model considered throughout the paper is:
\begin{model}
The FCFS-M/G/1 queue~\eqref{dynamics} with arrival rate~$\ar$ and 
service times $\vect{\stochst,\stochstzero}$ satisfying Assumption~\ref{assumption:stability},  endowed with a cost function~$\costfunction$. %:\Realpluszero\mapsto \CtoR{\Complex}{\Real}$.
%meeting Assumption~\refeq{assumption:integrablecost}.
\end{model}
\noindent
We complete our model with assumptions on the costs that guarantee existence of the value function. Some notations are first introduced.
%Under Assumption~\refeq{assumption:stability}, the queue fully discharges from any starting state with probability one, and the and discharge times averaged over time are finite.
%
%
%\begin{definition}\label{definition:Wt}
\begin{notation}
Ergodicity %of the queue under Assumption~\ref{assumption:stability} 
implies the existence of a unique asymptotic %(stationary)
probability distribution~$\distriXfunction{\trueWt}$ for the waiting times, where~$\trueWt$ symbolizes a random variable distributed accordingly. %in the queue with service times~$\vect{\stochst,\stochstzero}$, %the empty-state convention~$\stochstzero$, 
A distinction is made between the actual stationary waiting times,  with  service time convention~$\vect{\stochst,\stochstzero}$, and the waiting times that would ensue with the convention~$\vect{\stochst,\stochst}$, modeled by the variable~$\Wt$ with distribution~$\distriXfunction{\Wt}$. 
%By~$\Wt$ we denote a random variable distributed like the stationary waiting times in the queue under the convention~$\vect{\stochst,\stochst}$. %the empty-state convention~$\stochst$.
The Laplace-Stieltjes transforms of~$\stochst$, $\stochstzero$ and~$\Wt$ are denoted 
%by~$\LSTstochsts{\complex}$, $\LSTWts{\complex}$ and~$\LSTtrueWts{\complex}$, respectively.
by~$\LSTstochstfunction$, $\LSTstochstzerofunction$ and~$\LSTWtfunction$, respectively, where $\LSTstochsts{\complex}=\expectation{\exp{-\complex\stochst}}$.
%\end{definition}
%\noindent
\end{notation}
%
%{\color{green}--------------------}
%
\noindent
\begin{assumption}[Cost integrability]\label{assumption:integrablecost}
%The cost function~$\costfunction$ satisfies the conditions
$\modulus{\costfunction}$ is
%\begin{enumerate*}[(a)]
%\item \label{integrablecost}
$\measureXfunction{\Wt}$- and %integrable, 
and
%$\expectation{\modulus{\costx{\bl+\Wt}}}<\infty$
%&\ 
%for all $\bl\in\Realpluszero$,  
%\item \label{trueintegrablecost}
%$\expectation{\costx{\trueWt}}\neq\infty$.
$\measureXfunction{\trueWt}$-integrable.
%\end{enumerate*}
\end{assumption}
\storecompoundcounter{assumption}{assumption:integrablecost}%
\noindent

For any  $\bl\in\Realpluszero$ and any time horizon~$\ti\geq0$, we denote by~$\ValuetxT{\ttt}{\bl}{\ti}$  the (random) total cost incurred over a time interval of the type $[\ttt,\ttt+\ti)$ when the backlog at time~$\ttt$ is~$\bl$. %In the considered setting, the variable~$\ValuetxT{\ttt}{\bl}{\ti}$ is independent of~$\ttt$.
%For any backlog value $\bl\in\Realpluszero$ and time horizon~$\ti\geq0$, we denote by~$\ValuetxT{\ttt}{\bl}{\ti}$  the (random) total  cost incurred in any time interval $[\ttt,\ttt+\ti)$ when the backlog at initial time~$\ttt$ is~$\bl$. In the considered setting, the variable~$\ValuetxT{\ttt}{\bl}{\ti}$ is independent of~$\ttt$.
%
Under Assumption~\ref{assumption:integrablecost}%\ref{trueintegrablecost}
, the quantity~$\ValuetxT{\ttt}{\bl}{\ti}$ averaged over the number of arrivals in the time window tends as~$\ti \to \infty$ to the mean cost per job~$\meancost = \expectation{\costx{\trueWt} }$. 
%\begin{equation}\label{vfexistence}\begin{array}{ll}
%\meancost = \Expectation{\costx{\Wt} } < \infty,
%\end{array}\end{equation}
% where~$\meancost$ is the mean cost per job,
The  \emph{value function}~$\valuefunction:\Realpluszero\mapsto\CtoR{\Complex}{\Real}$ is then defined by% is defined by
%as a function of the backlog~$\bl$ giving
%
\begin{equation}\label{valuefunction} \tag{VF}
\valuex{\bl}=\limBrace{\ti\to\infty}{\expectation{ \ValuetxT{\ttt}{\bl}{\ti} } - \ar\meancost \ti }
,
\quad 
\forall\bl\geq0, 
\end{equation}
\storeuniversalcountertag{valuefunction}{VF}%
 %where $\valuex{\bl}$ 
 as an expression of the state sensitivity of the costs with respect to the steady-state regime. %the  total cost surplus to  be expected from the starting state~$\bl$ relative to the steady-state regime.
 %where%~$\meancost$ is the mean cost per job, %, i.e. it holds that
%\begin{equation}
% \costrate = \lim\nolimits_{\ti\to\infty}   \frac{\ValuetxT{t}{\bl}{\ti}}{\ti}
%\end{equation}
%$ \frac{\ValuetxT{t}{\bl}{\ti}}{\ti} \toas \ar\meancost$ as $\ti\to\infty$, 
%and
%~$\valuefunction$ proves to be independent of~$\ttt$. 
%Useful properties of the value function as defined in~(\refeq{valuefunction}) are provided in Appendix~$\refeq{appendix:vf}$.
%
%\textbf{include rule c equal almost everywhere  in appendix, part of the properties, it follows that c satisfies equation for c+ or c-}
%
%
%
In order to compute~\eqref{valuefunction}, we will regard~$\valuefunction$ as a solution of the following Poisson equation, derived in Appendix~\ref{appendix:LST}.

\begin{proposition}%[Differential equation] 
[Poisson equation]
\label{proposition:poissonequation}
Let 
%the cost function~$\costfunction$ in our server model satisfy
Assumption~\ref{assumption:integrablecost} hold.
%
%\label{proposition:poisson}
%
%
The value function~\eqref{valuefunction} rewrites as $
 \valuex{\bl}  
 =
 \altgenericx{\bl,0} - \costx{\bl}
=
\MDPkernelfunction  
 \altgenericx{\bl,0} - \meancost
%\quad\quad (\server=1,\dots,\nbservers)
$ for some~$\altgenericfunction:\Realpluszero\mapsto\CtoR{\Complex}{\Real}$ solution of the Poisson equation
\begin{equation}
\label{MG1Poisson} \tag{PE}
\altgenericx{\bl,\st} 
=
\MDPkernelfunction
%\MDPkernelarfunction{\ari{\server}} 
\altgenericx{\bl,\st}  +  \costx{\bl}  - \meancost 
,
\end{equation}
\storeuniversalcountertag{MG1Poisson}{PE}%
where % we write
 $ \MDPkernelfunction\altgenericx{\bl,\st} 
  \defeq
%\iint
\int
 \altgenericx{\altbl,\altst}  \MDPkernel{\bl,\st}{d\vect{\altbl,\altst}}
 $.
\end{proposition}
\noindent
A general solution to~\eqref{MG1Poisson} was given in~\cite{glynn94} %with focus set on finite time horizons, 
under the integral form $\altgenericx{\bl,\st}=\int_0^{+\infty}\costx{\altbl}\, \solutionkernel{\bl,\st}{d\altbl} \, d\altbl$, where $\solutionkernelfunction$ defines the solution kernel of the queue. Although closed-form value functions can be inferred from this integral form, it is impractical for a systematic derivation of solutions.
%
%though little is said about the asymptotic form of the solution when the considered horizon grows to infinity, as in~\eqref{valuefunction}.
%
In Section~\ref{section:closedformVF} we take a different approach by considering a transform-domain expression of the solutions of~\ref{MG1Poisson}, obtained by complex analysis of the Poisson equation.

%%%%%%%%%%%%%%
\section{Closed-form value functions.}
\label{section:closedformVF}
In this section we develop the tools that will help us compute value functions.
%\begin{remark}[System interpretation]\label{remark:systeminterpretationbis}
\subsection{Characterization of the value function}
\label{section:complexcharacterization}
Before proceeding, %we give a complex interpretation of our assumptions.
%Recall 
recall the Pollaczek-Khintchine formula for the Laplace-Stieltjes transform of~$\Wt$, \cite{pollaczek30,khintchine32},
which we characterize in Appendix~\ref{appendix:LST}:  %\begin{notation}
%Let~$\dominantpoleX{\Wt}$ denote the  dominant singularity of\,\footnotemark{}%~$\LSTWtfunction$
\begin{equation} \label{PKformula}
% \textstyle
 \LSTWts{\complex} 
 %=
 %\expectation{\exp{-\complex\Wt}} 
=
 \frac{(1-\load )\complex} {\complex- \ar (1- \LSTstochsts{\complex}) }
 .
% \quad\textup{(\emph{Pollaczek-Khintchine})}
\tag{PK}
 \end{equation} 
%useful properties of which are reported in Appendix~\ref{appendix:LST}.
Let~$\dominantpoleX{\Wt}$ denote the  dominant singularity of~$\LSTWtfunction$ which, 
 in view of Proposition~\ref{proposition:analycity}\eqref{analycity:dominantpoles}, %$\dominantpoleX{\Wt}$ 
 is a real negative pole.
%\end{notation}
%\noindent
%\footnotetext{The expression given in~\eqref{PKformula} for the Laplace-Stieltjes transform of~$\Wt$ is the well-known Pollaczek-Khinchin formula, \cite{pollaczek30,khintchine32}. See Appendix~\ref{appendix:LST} for a characterization of~$\LSTWtfunction$.}%
%
In the transform domain, $\measureXfunction{\Wt}$-integrability of~$\modulus{\costfunction}$ reduces to a condition on the relative positions %, in the complex domain, 
of the singularities of~$\LSTWts{-\complex}$ and those of $\laplacetransformfs{\costfunction}{\complex} = \int\nolimits_{0}^{\infty}\exp{-\complex\bl}\, \costx{\bl}  \, d\bl
$, the Laplace transform of~$\costfunction$. Concretely, the \acp{ROC} of~$\LSTWts{-\complex}$ and $\laplacetransformfs{\costfunction}{\complex}$ (two open half-planes with normal vectors pointing in opposite directions) should have nonempty intersection.
This condition (Assumption~\ref{assumption:complexintegrablecost}) is illustrated in Figure~\ref{figure:LSTWt} for the case of constant service times.
%
% 
%
%Figure~\ref{figure:LSTWt}
\begin{figure}
\centering
\begin{tikzpicture}[scale=1.0]
\def\cwint{0.5}
\def\shd{0.7}
\definecolor{cw}{rgb}{\cwint,\cwint,\cwint}
\definecolor{shadecolor}{rgb}{\shd,\shd,\shd}
\def\dotwidth{0.070}%{0.05}
\def\tickwidth{0.05}
\def\gap{-.1}
\def\margin{0.2}
\def\marginsmall{0.1*\margin}
\def\deltaangle{5}
\def\arrowlength{\tauf/4}
\def\arrowangle{70}
\def\Arrowangle{20}
\def\R{2.1}
\def\axislength{2.5}
\def\littlespace{0.02}
\def\ST{1.0}
\def\AR{0.5}
\def\xsc{0.55}%{0.45}
\def\ysc{0.037}%{0.028}
\def\polex{0}
\def\poley{0}
\def\Wmonex{1.25643}
\def\Wmoney{0}
\def\Wzerox{0}
\def\Wzeroy{0}
\def\Wonex{2.789}
\def\Woney{-7.43762}
\def\Wtwox{3.35988}
\def\Wtwoy{-13.8657}
\def\Wthreex{3.72089}
\def\Wthreey{-20.2145}
\def\Wfourx{3.98574}
\def\Wfoury{-26.5361}
\def\Wfivex{4.19505}
\def\Wfivey{-32.8447}
\def\Wsixx{4.36812}
\def\Wsixy{-39.1462}
\def\Wsevenx{4.51566}
\def\Wseveny{-45.4432}
\def\Weightx{4.64424}
\def\Weighty{-51.7372}
\def\Wninex{4.75818}
\def\Wniney{-58.0291}
\def\Wtenx{4.86047}
\def\Wteny{-64.3195}
% computed wwith \ar = 0.5    <------------
%n | -W_n(-0.303265) - 0.5
%1 | 2.789 - 7.43762 i
%2 | 3.35988 - 13.8657 i
%3 | 3.72089 - 20.2145 i
%4 | 3.98574 - 26.5361 i
%5 | 4.19505 - 32.8447 i
%6 | 4.36812 - 39.1462 i
%7 | 4.51566 - 45.4432 i
%8 | 4.64424 - 51.7372 i
%9 | 4.75818 - 58.0291 i
%10 | 4.86047 - 64.3195 i
\def\anchorx{\Wmonex*0.7}
\def\epsil{0.4}
\def\tauf{2.5}
\def\sig{\tauf*.5+\R*.75}
\def\ballsize{.1}
\def\pole{\R*0.44}
\def\npole{9}
\coordinate  (pole)  at  (\polex/\ST*\xsc,\poley/\ST*\ysc);
\coordinate  (Wm1)  at  (\Wmonex/\ST*\xsc,\Wmoney/\ST*\ysc);
\coordinate  (W0)  at  (\Wzerox/\ST*\xsc,\Wzeroy/\ST*\ysc);
\coordinate  (W1)  at (\Wonex/\ST*\xsc,\Woney/\ST*\ysc);
\coordinate  (cW1)  at (\Wonex/\ST*\xsc,-\Woney/\ST*\ysc);
\coordinate  (W2)  at (\Wtwox/\ST*\xsc,\Wtwoy/\ST*\ysc);
\coordinate  (cW2)  at (\Wtwox/\ST*\xsc,-\Wtwoy/\ST*\ysc);
\coordinate  (W3)  at (\Wthreex/\ST*\xsc,\Wthreey/\ST*\ysc);
\coordinate  (cW3)  at (\Wthreex/\ST*\xsc,-\Wthreey/\ST*\ysc);
\coordinate  (W4)  at (\Wfourx/\ST*\xsc,\Wfoury/\ST*\ysc);
\coordinate  (cW4)  at (\Wfourx/\ST*\xsc,-\Wfoury/\ST*\ysc);
\coordinate  (W5)  at (\Wfivex/\ST*\xsc,\Wfivey/\ST*\ysc);
\coordinate  (cW5)  at (\Wfivex/\ST*\xsc,-\Wfivey/\ST*\ysc);
\coordinate  (W6)  at (\Wsixx/\ST*\xsc,\Wsixy /\ST*\ysc);
\coordinate  (cW6)  at (\Wsixx/\ST*\xsc,-\Wsixy /\ST*\ysc);
\coordinate  (W7)  at (\Wsevenx/\ST*\xsc,\Wseveny /\ST*\ysc);
\coordinate  (cW7)  at (\Wsevenx/\ST*\xsc,-\Wseveny /\ST*\ysc);
\coordinate  (W8)  at (\Weightx/\ST*\xsc,\Weighty /\ST*\ysc);
\coordinate  (cW8)  at (\Weightx/\ST*\xsc,-\Weighty /\ST*\ysc);
\coordinate  (W9)  at (\Wninex/\ST*\xsc,\Wniney /\ST*\ysc);
\coordinate  (cW9)  at (\Wninex/\ST*\xsc,-\Wniney /\ST*\ysc);
\coordinate  (W10)  at (\Wtenx/\ST*\xsc,\Wteny /\ST*\ysc);
\coordinate  (cW10)  at (\Wtenx/\ST*\xsc,-\Wteny /\ST*\ysc);
\coordinate (anchor) at (\anchorx/\ST*\xsc,0);
\coordinate (br1) at (\polex/\ST*\xsc+ \littlespace ,-\R-\margin);
\coordinate (br2) at (\Wmonex/\ST*\xsc - \littlespace ,-\R-\margin);
 \draw [white,shading = axis,rectangle, left color=white, right color=shadecolor,shading angle=90, anchor=north, minimum width=\R, minimum height=\R]  (pole) ++ (0,0) ++ (0,-\R-\margin) -- (\anchorx/\ST*\xsc-\axislength,-\R-\margin) -- (\anchorx/\ST*\xsc-\axislength,\axislength) -- (\polex/\ST*\xsc,\axislength) -- cycle;
 \draw [white,shading = axis,rectangle, left color=white, right color=shadecolor,shading angle=270, anchor=north, minimum width=\R, minimum height=\R]  (Wm1) ++ (0,0) ++ (0,-\R-\margin) -- (\anchorx/\ST*\xsc+\axislength,-\R-\margin) -- (\anchorx/\ST*\xsc+\axislength,\axislength) -- (\Wmonex/\ST*\xsc,\axislength) -- cycle;
% real axis
\draw[very thin] (anchor) ++ (-\marginsmall-\axislength,0) -- (anchor) ++(\marginsmall+\axislength,0)  node[below=0]{$\realpart{\complex}$} -- (anchor);
%
% Im axis
\draw[very thin] (pole) node[below left=-0.06]{\scalebox{0.8}{$%\dominantpoleX{\stochst}=
0$}} ++ (0,-\R-\margin) -- (pole) ++(0,\axislength) node[below left=-0.1]{$\imaginarypart{\complex}$} -- (pole);
%
%\node  at (0  +2* \margin,-\epsil -\margin ) {$\contouras{\epsilon}{0}$};
%
%%
\coordinate  (tau)  at  (0,\tauf);
\coordinate  (sig)  at  (0,\sig);
% clockwise
\draw[black,thick]%[fill=black!05!white,thick] 
(anchor) node[below left=-0.06]{\scalebox{0.8}{$\anchor$}} ++ (-\marginsmall,\R) arc (90:270:\R) -- cycle;
\draw[black,->,thin]  (anchor)  ++ (-\marginsmall,0) ++(130-\Arrowangle*0.5:\R+\gap) arc (130-\Arrowangle*0.5:130+\Arrowangle*0.5:\R+\gap) ;
\node[black]  at (\anchorx/\ST*\xsc- \marginsmall+\R*cos 130  +2* \margin*cos 130,\R * sin 130 +\margin*sin 130 ) {$\contoura{\radius}$};
%
% clockwise
\draw[cw,thick, dashed]%[fill=black!05!white,thick] 
(anchor)  ++ (\marginsmall,\R) arc (90:-90:\R) -- cycle;
\draw[cw,->,thin]  (anchor)  ++ (\marginsmall,0) ++(50+\Arrowangle*0.5:\R+\gap) arc (50+\Arrowangle*0.5:50-\Arrowangle*0.5:\R+\gap) ;
\node[cw]  at (\anchorx/\ST*\xsc+ \marginsmall+\R*cos 50  +2* \margin*cos 50,\R * sin 50 +\margin*sin 50 ) {$\contoura{-\radius}$};
%
%\draw[black,thin] (pole) ++ (\anchorx/\ST*\xsc,-\R-\margin/2) -- (\polex/\ST*\xsc + 0.02,-\R-\margin/2) ;
  \draw[decorate,decoration={brace,amplitude=2pt,raise=0pt,mirror},yshift=0pt] (br1) -- (br2) node [midway,yshift=-8pt]{\scalebox{0.8}{$\ROC{\Blaplacetransformf{\extdpurevaluefunction}}$}};
\makedot{(pole)}
\makedot{(Wm1)}
\makewhitedot{(anchor)}
\draw[very thin,dashed] (Wm1) node[below right=-0.06]{\scalebox{0.8}{$\!\!\,-\dominantpoleX{\Wt}$}} ++ (0,-\R-\margin) -- (Wm1) ++(0,\axislength)  -- (Wm1);
\makedot{(W1)}
\makedot{(W2)}
\makedot{(W3)}
\makedot{(W4)}
\makedot{(W5)}
\makedot{(W6)}
\makedot{(W7)}
\makedot{(W8)}
\makedot{(W9)}
%\makedot{(W10)}
\makedot{(cW1)}
\makedot{(cW2)}
\makedot{(cW3)}
\makedot{(cW4)}
\makedot{(cW5)}
\makedot{(cW6)}
\makedot{(cW7)}
\makedot{(cW8)}
\makedot{(cW9)}
\makedot{(cW10)}
\draw (W9) node[left=-0.05]{\scalebox{0.8}{$-\Polef{\LSTWtfunction}$}};
\lengthmark{\anchorx/\ST*\xsc-\marginsmall}{0}{130}{\R}{\epsil*0.2}{$\radius$}
%\lengthmark{\tauf/2}{ - \epsil*sin \deltaangle}{90}{\epsil*sin \deltaangle * 2}{0.3}{\tiny{$\zero{\epsilon}$}}
%\draw[very thin] (C) node[left=0]{$0$} -- (Cy) node[above=0]{$\plusw\plusvnorm{\setofinteresti{\iset}}{\plusvix{\iset}{\llamb}}$};
%\draw  (3,0.5) --  node {hello};
%\draw[black!10!white,thick] (0.5,3) -- (1.5,0) -- (0,0
\end{tikzpicture}
\caption{\label{figure:LSTWt}%
 Convergence of $\Blaplacetransformf{\extdpurevaluefunction}$ for constant service times $\stochst=\st$ and step cost function $\costx{\bl}=  \indicator{[\taufunction,\infty)}{\bl} $, with $\taufunction>0$:
 $\laplacetransformfs{\costfunction}{\complex}= 1 /\complex$ has one pole at~$\complex=0$ with $\ROC{\laplacetransformf{\costfunction}}=\{\complex\in\Complex\setst\realpart\complex>0\}$, while 
 $\LSTWts{-\complex}= {(1-\ar\st )\complex}/[ {\complex+ \ar (1-e^{\complex\st}) }]
$ has an infinity of poles at
 $s=-\ar  [1 + ({1}/{\load})\,\Lambertnx{k}{-\load \exp{-\load } } ]$ for $k\in\Integerzero$, where $\load=\ar\st$ and~$\lambertnfunction{k}$ denotes the $k$th branch of the product logarithm function, with $-\dominantpoleX{\Wt}=-\ar  [1 + ({1}/{\load})\,\Lambertnx{-1}{-\load \exp{-\load } } ]>0$, \cite{corless96}.
 }
\end{figure}%
\storecompoundcounter{figure}{figure:LSTWt}%
%% Assumption~\ref{assumption:complexintegrablecost}
\begin{assumption}%[Cost integrability]
\label{assumption:complexintegrablecost}
The cost function~$\costfunction$ satisfies $-\dominantpoleX{\Wt}\in\ROC{\laplacetransformf{\costfunction}}$. %, where $\laplacetransformfs{\costfunction}{\complex} = \int\nolimits_{0}^{\infty}\exp{-\complex\bl}\, \costx{\bl}  \, d\bl $ denotes the Laplace transform of~$\costfunction$.
\end{assumption}
%
%This raises the question of the existence of a fictitious stochastic process that would be governed by~\eqref{MG1Poisson} over the entire real axis. In the rest of this section, we show that such a process does  exist, and that~\eqref{MG1Poisson} can be solved directly through  complex analysis of the imaginary process presumably involved.
%
%
\noindent
For analysis purposes, we now extend the  nonnegative process~\eqref{dynamics} to negative backlog values by presuming of a (fictitious) stochastic process governed by~\eqref{MG1Poisson} over the entire real axis. We set the scene as follows. 

First, we  let $\costx{\bl}=0$ for $\bl<0$, and  we complete~\eqref{MDPkernel} with 
$ \MDPkernel{\bl}{[\altbl,\bl]} =  1-\exp{-\ar(\bl-\altbl)} $ if $\bl<0$, thus %assuming 
%extending
conjecturing for~\eqref{fullMDPkernel} the $\Realminus$ behaviour %to the negative 
\begin{align}\label{negativedynamics} \tag{Q${}^-$} 
 &
 \stochblk{\tim+1} =  \stochblk{\tim} +  \stochstk{\tim} - \stochTk{\tim+1} 
 %\setst\stochblk{\tim} +  \stochstk{\tim} < 0
 \ \ \ \text{if}\ \ \ \stochblk{\tim} +  \stochstk{\tim} < 0
 ,
 &
 \tim\geq 0
 .
 \end{align}
Observe that the so extended Markov process loses the irreducibility of~\eqref{dynamics}, since the process remains caught in~$\Realpluszero$ once it has occupied a nonnegative state. Otherwise, it is expected to drift towards $\bl=-\infty$, where its chances vanish to ever reach~$\Realpluszero$.
Next, we consider an ancillary, more tractable transition kernel~$\UMDPkernelfunction$ of the type~\eqref{fullMDPkernel} with \emph{uniform} dynamics for the backlogs:
\begin{align}
\label{UMDPkernel}
 &
 \UMDPkernel{\bl}{[\altbl,\bl]} =  1-\exp{-\ar(\bl-\altbl)}  ,
% \quad
% \UMDPkernel{\bl}{(\bl,+\infty)} = 0  
% ,
&
\forall \bl\in\Real
.
\end{align}
%Intuitively, $\UMDPkernelfunction$ models the dynamics of parcels of random lengths~$\stochst$ juxtaposed with rate~$\ar$ on an infinite conveyor belt drifting at speed~$1$ towards~$-\infty$.
%  
The Poisson equation~\eqref{MG1Poisson} then rewrites as the simple form
\begin{equation}
\label{rMG1Poisson} \tag{PE'}
\altgenericx{\bl,\st} 
=
\UMDPkernelfunction \altgenericx{\bl,\st} + \Ucostxy{\bl}{\st}  
, 
\end{equation}  
where $\Ucostxy{\bl}{\st} \defeq \UDeltax{\bl+\st} +  \costx{\bl} - \meancost \, \indicator{[0,+\infty)}{\bl} $, and  $ \UDeltax{\bl}  \defeq (\MDPkernelfunction - \UMDPkernelfunction) \altgenericx{\bl,0}$. %, and~$\stepfunction$ denotes the step function, with the convention~$\stepx{x}=1$ if $x\geq0$ and $\stepx{x}=0$ otherwise. 
Clearly, \eqref{rMG1Poisson} retains the property that its solutions are defined up to a constant. By construction, they  also solve~\eqref{MG1Poisson} on~$\Realpluszero$. The true and virtual parts of these solutions over~$\Realminus$  are identified by Theorem~\ref{theorem:solutionofPoisson}.
%

% Theorem~\ref{theorem:solutionofPoisson}
\begin{theorem}[Extended Poisson equation]\label{theorem:solutionofPoisson}
 Let Assumption~\ref{assumption:complexintegrablecost} hold.
%
%\begin{enumerate}[(i),ref=\roman*,wide,labelwidth=!,labelindent=5pt]
% \item \label{theorem:solutionofPoissoni}
Every solution of~\eqref{rMG1Poisson} has the form $ \altgenericx{\bl,\st} = \Uvaluex{\bl+\st} + \costx{\bl} + \anticausalx{\bl+\st} \, \indicator{(-\infty,-\st)}{\bl}$ for some~$\anticausalfunction:\Real\mapsto\Real$ common to all solutions and for~$\Uvaluefunction:\Real\mapsto\Real$ satisfying  %~$\Uvaluefunction:\Real\mapsto\Real$ defined by 
\begin{align}\label{Ustemsolution}\tag{$\makeU{\text{S}}$}
%\begin{array}{ll}
\Uvaluex{\bl}&=\valuex{0} +\Upurevaluex{\bl} -\frac{\ar\meancost }{1-\load} \bl \, \indicator{[0,+\infty)}{\bl}  + \anticausalx{\bl} \, \indicator{(-\infty,0)}{\bl} , &%\quad 
\forall\bl\in\Real,
\quad
%\end{array}
\end{align}
%where the right derivative of~$\Upurevaluefunction$ has a  two-sided Laplace transform $\Blaplacetransformfs{\rightderivative\Upurevaluefunction}{\complex}= \int\nolimits_{-\infty}^{\infty}\exp{-\complex\bl}\, \rightderivative\Upurevaluex{\bl}  \, d\bl$ given on its nonempty \acl{ROC} by
where the two-sided Laplace transform  of the right derivative of~$\Upurevaluefunction$, $\Blaplacetransformfs{\rightderivative\Upurevaluefunction}{\complex}= \int\nolimits_{-\infty}^{\infty}\exp{-\complex\bl}\, \rightderivative\Upurevaluex{\bl}  \, d\bl$, is given on its nonempty \acl{ROC} by
\begin{equation} \label{bilaterallaplacetransformVF} \tag{C}
\storeuniversalcountertag{bilaterallaplacetransformVF}{C}%
%\textstyle
 \Blaplacetransformfs{\rightderivative\Upurevaluefunction}{\complex}  
=
\frac{\ar}{(1-\load )} \LSTWts{-\complex}\laplacetransformfs{\costfunction}{\complex} % +\Ujump  +\analyticparts{\complex}%\zero{1} 
.
\end{equation}  
\end{theorem} 
\obsolete{
\begin{theorem}[Extended Poisson equation]\label{theorem:solutionofPoisson}
 Let Assumption~\ref{assumption:complexintegrablecost} hold.
%
%\begin{enumerate}[(i),ref=\roman*,wide,labelwidth=!,labelindent=5pt]
% \item \label{theorem:solutionofPoissoni}
Every solution of~\eqref{rMG1Poisson} has the form $ \altgenericx{\bl,\st} = \Uvaluex{\bl+\st} + \costx{\bl} $ for some%~$\Uvaluefunction:\Real\mapsto\Real$ defined by 
\begin{align}\label{Ustemsolution}\tag{$\makeU{\text{S}}$}
%\begin{array}{ll}
\Uvaluex{\bl}&=\Uvaluex{0} +\Upurevaluex{\bl} -\frac{\ar\meancost }{1-\load} \bl \, \indicator{[0,+\infty)}{\bl} , &%\quad 
\forall\bl\in\Real,
%\end{array}
\end{align}
where the right derivative of~$\Upurevaluefunction$ has a  two-sided Laplace transform %of~$\extdpurevaluefunction$,
$\Blaplacetransformfs{\rightderivative\Upurevaluefunction}{\complex}=
\int\nolimits_{-\infty}^{\infty}\exp{-\complex\bl}\, \rightderivative\Upurevaluex{\bl}  \, d\bl$ 
%
%$\Blaplacetransformfs{\rightderivative\Upurevaluefunction}{\complex}$ 
%has nonempty \acl{ROC}, where it is given by
%is 
given on its nonempty \acl{ROC} by
\begin{equation} \label{BlUdpurevaluefunction}
%\textstyle
 \Blaplacetransformfs{\rightderivative\Upurevaluefunction}{\complex}  
=
\frac{\ar}{(1-\load )} \LSTWts{-\complex}\laplacetransformfs{\costfunction}{\complex} +\Ujump  +\analyticparts{\complex}%\zero{1} 
,
\end{equation}  
in which~$\analyticparts{\complex}/\complex$ is analytic 
on~$\{\complex\in\Complex\setst \Realpart{\complex}<- \dominantpoleX{\Wt}\}$ and~$\Ujump$ is common to all solutions, with % and given by
\begin{equation}\label{jump}
\Ujump 
= 
%\frac{\Kquantity}{1-\load}-\frac{\MaclaurinLSTstochsts{2}{\ar}}{(1-\load)^2} \meancost
{\Kquantity}/{(1-\load )}-{\MaclaurinLSTstochsts{2}{\ar}\meancost}/{(1-\load)^2} 
,
\end{equation}
%is common to all solutions, with
where
$
\MaclaurinLSTstochsts{2}{\ar}=1-\ar\expectation{\stochst}+\ar^2\expectation{\stochst^2}/2
$ and~$\Kquantity$ satisfies %denotes the value at~$\ar$ of the second degree Maclaurin polynomial of~$\LSTstochsts{\complex}$.
%$\LSTstochsts{\ar}=\expectation{\exp{-\ar\stochst}}= 1 - \expectation{\stochst} \ar+ \expectation{\stochst^2} \ar^2/2 $
\begin{equation}\label{Kquantity}
\textstyle
\Kquantity
 =
 \bigexpectation{ \altgenericx{0,\stochstzero} 
- \ar  \int_{-\infty}^{0} \altgenericx{\bl,\stochst} \, \exp{\ar\bl} \, d\bl }
.
\end{equation}
%\item \label{theorem:solutionofPoissonii}
%We have $\Upurevaluex{\bl}=\purevaluex{\bl}$ for $\bl\geq 0$, and $\Upurevaluex{\bl}=\extpurevaluex{\bl}-\Ujump$ for $\bl< 0$.
%\end{enumerate}
%
\end{theorem} 
}%
\noindent
Theorem~\ref{theorem:solutionofPoisson} can be shown by transform-domain analysis of the solutions of~\eqref{rMG1Poisson}. The proofs of all the results given in this section are deferred to Appendix~\ref{appendix:LST}.

The function~$\Uvaluefunction$ in~\eqref{Ustemsolution} is an extension of the value function to the negative backlogs, with $\Uvaluex{\bl}\equiv\valuex{\bl}$ if $\bl\geq 0$.
Theorem~\ref{theorem:solutionofPoisson} suggests that the value function~\eqref{valuefunction} characterizes the M/G/1 queue~\eqref{dynamics} as much as the imaginary  process~\eqref{negativedynamics} taking place in the negative backlog values.
What is more, the hidden negative end of the queue seems to hold the key to solving the associated Poisson equation in the transform domain. 
%Notice that the solutions of the extended queue show a jump discontinuity~$\Ujump$ between virtual and actual parts of the queue. %, initially absent from the continuous construction~$\extpurevaluefunction$.
% 

By inverse transformation of~\eqref{Ustemsolution}, we obtain the following results.
\begin{proposition}%[Differential equation] 
[Value function]
\label{proposition:identities}
Let~$\costfunction$ %in our server model 
%in the server model of Section~\ref{section:valuefunction} 
satisfy Assumption~\ref{assumption:integrablecost} and be piecewise continuous.
% % , and define
% % %Endow the server model of Section~\ref{section:valuefunction} with a piecewise continuous cost function~$\costfunction$, and define
% % $\rightcostfunction:\bl\in\Realpluszero\mapsto\rightcostx{\bl}= \lim\nolimits_{\altbl\to\bl^+}\costx{\altbl}$.

\begin{enumerate}[(i),ref=\roman*,wide,labelwidth=!,labelindent=5pt]
\item
 \label{proposition:differentialequation}
% The value function~\eqref{valuefunction} satisfies $\Expectation{\valuex{\bl+\stochst}}\neq\infty$ for all~$\bl\in\Realpluszero$.
%It is
The value function~\eqref{valuefunction} is
continuous, almost everywhere continuously diffe\-ren\-tiable, and semi-diffe\-ren\-tiable with right derivative
\begin{align}\label{deresult}\tag{DE}\storeuniversalcountertag{deresult}{DE}%
%\begin{array}{ll}
 &\rightderivative \valuex{\bl}   
=
 \ar\left(
%\one{\bl-\stochdl}} \cost 
 \rightcostx{\bl}
-  \meancost
+
\Expectation{\valuex{\bl+\stochst}-  \valuex{\bl}  }\right)
,& %\quad 
\forall \bl\in \Realplus,
%\end{array}
\end{align}
where $%\rightcostfunction:\bl\in\Realpluszero\mapsto
\rightcostx{\bl}\defeq \lim\nolimits_{\altbl\to\bl^+}\costx{\altbl}$. %, and $\expectation{\valuex{\bl+\stochst}}\neq\infty$. %, % for all~$\bl\in\Realpluszero$, 
%and we define $\rightcostfunction:\bl\in\Realpluszero\mapsto\rightcostx{\bl}= \lim\nolimits_{\altbl\to\bl^+}\costx{\bl}$.
At~$\bl=0$, one has
\begin{align}%\begin{equation}\begin{array}{l}
\label{boundarycondition}\tag{BCa}\storeuniversalcountertag{boundarycondition}{BCa}%
\valuex{0}
&=%&
{\costx{0}}-\meancost +\expectation{\valuex{\stochstzero}} ,
%\end{array}\end{equation}
%\begin{equation}\begin{array}{l}
\\ 
\label{derivativeatzero} \tag{BCb}
\storeuniversalcountertag{derivativeatzero}{BCb}%
\dvaluex{0}
&=%&
\ar \left(
\rightcostx{0}-\costx{0}
+
\expectation{\valuex{\stochst}-\valuex{\towelch{\stochdelayx{\stochst}}{\stochstzero}}} 
\right)
.
%\end{array}\end{equation}
\end{align}%
%
% \item
% \label{proposition:poisson}
% %
% The value function rewrites as $
%  \valuex{\bl}  
%  =
%  \altgenericx{\bl,0} - \costx{\bl}
% =
% \MDPkernelfunction  
%  \altgenericx{\bl,0} - \meancost
% %\quad\quad (\server=1,\dots,\nbservers)
% $ for some~$\altgenericfunction$ solution of the Poisson equation
% \begin{equation}
% \label{MG1Poisson} \tag{PE}
% %
% \altgenericx{\bl,\st} 
% =
% \MDPkernelfunction
% %\MDPkernelarfunction{\ari{\server}} 
% \altgenericx{\bl,\st}  +  \costx{\bl}  - \meancost 
% ,
% \end{equation}
% where % we write
%  $ \MDPkernelfunction\altgenericx{\bl,\st} 
%   \defeq
% %\iint
% \int
%  \altgenericx{\altbl,\altst}  \MDPkernel{\bl,\st}{d\vect{\altbl,\altst}}
%  $.

\item
\label{proposition:vfdifferentiability}
The value function 
is given by
% $\valuex{\bl}=\valuex{0}-\frac{\ar\meancost\bl}{1-\load}+ \purevaluex{\bl}$, 
\begin{align}\label{stemsolution}\tag{S} \storeuniversalcountertag{stemsolution}{S}%
%\begin{array}{ll}
\valuex{\bl}&=\valuex{0} +\purevaluex{\bl} -\frac{\ar\meancost }{1-\load} \bl , &%\quad 
\forall\bl\in\Realpluszero,
%\end{array}
\end{align}% 
where%~$\meancost=\expectation{\costx{\trueWt}}$, and
~$\purevaluefunction:\Real\mapsto\Real$ is continuous, almost everywhere continuously diffe\-ren\-tiable,  and semi-diffe\-ren\-tiable with right-derivative 
%\begin{subequations}\label{dwfunctionall}%\tag{\acs{wfunction}}
\begin{align}
\label{dwfunction}\tag{\acs{wfunction}}
\storeuniversalcountertag{dwfunction}{\acs{wfunction}}%
&
\rightderivative\purevaluex{\bl}
=
\displaystyle
\frac{\ar}{1-\load}
\Expectation{\costx{\bl+\Wt}}
\obsolete{
\purevaluex{\bl}
=
\frac{\ar}{1-\load}
\int\nolimits_{0}^{\bl}
\Expectation{\costx{\xi+\Wt}}
d\xi
}
,
&%\quad 
%\forall \bl\in\Realpluszero
\forall \bl\in\Real
.%,
%\end{align}
%with~$\Blaplacetransformf{\extdpurevaluefunction}$ given on its nonempty \acl{ROC} by
% \\%\begin{align}  
% \label{bilaterallaplacetransformVF} \tag{\acs{wfunction}b}
% \storeuniversalcountertag{bilaterallaplacetransformVF}{\acs{wfunction}b}%
% %\textstyle 
% &
% %\begin{array}{c}
% %\Blaplacetransformfs{\valuefunction}{\complex}=\frac{\valuex{0}}{\complex} + \frac{\ar}{1-\load} \left( \frac{\Expectation{\exp{\complex\Wt} }}{\complex} \laplacetransformfs{\costfunction}{\complex} - \frac{\meancost}{\complex^2} \right)
% \Blaplacetransformfs{\extdpurevaluefunction}{\complex}
% %\complex \Blaplacetransformfs{\extvaluefunction}{\complex} -{\valuex{0}}
%  = 
% \frac{\ar}{1-\load} \,
% %\left( 
% \LSTWts{-\complex}
% \,
% \laplacetransformfs{\costfunction}{\complex} 
% %- \frac{\meancost}{\complex} 
% %\right)
% ,
% &
% \ROC{\Blaplacetransformf{\extdpurevaluefunction}} \neq \emptyset
% .
%\end{array}
\end{align} 
%\end{subequations}
%where 
%$\rightcostfunction:\bl\in\Realpluszero\mapsto\rightcostx{\bl}= \lim\nolimits_{\altbl\to\bl^+}\costx{\bl}$.
%

\end{enumerate}
\end{proposition}
\storecompoundcounter{proposition}{proposition:identities}%
\noindent
 %A derivation of Proposition~\ref{proposition:identities} is given in Appendix~\ref{appendix:LST}.
 %

Equation~\eqref{deresult}  in Proposition~\ref{proposition:identities}\eqref{proposition:differentialequation} was  for instance used in~\cite{hyytia-peva-2017} to derive the value function of the M/D/1 queue with a step cost function~$\indicatorfunction{[\taufunction,\infty)}$.
However, the expectation of the random jump~$\valuex{\cdot+\stochst}$, makes~(\ref{deresult}) difficult to solve for~$\valuefunction$ in the general case. %|
%From the Poisson equation in~\eqref{proposition:poisson}, we learn that the value function can be derived explicitly within a constant, which appears in~\eqref{stemsolution} as~$\valuex{0}$. A solution to~\eqref{MG1Poisson} is given in~\cite{glynn94} with focus set on finite time horizons, though little is said about the asymptotic form of the solution when the considered horizon grows to infinity, as in~\eqref{valuefunction}.
%
The result reported in~\eqref{proposition:vfdifferentiability} is but the expression taken by the kernel solution of~\cite{glynn94} in the limit case where the invariant measure of the Poisson equation coincides with the stationary measure of the waiting times. % in the queue.
A relation of duality can be observed between~\eqref{stemsolution}, where  the value function follows by cross-correlation of the cost function with the asymptic waiting times, and~\eqref{deresult}, where  the cost function can be recovered by cross-correlation of the value function and the service times. In fact, \eqref{deresult} and~\eqref{stemsolution} are backlog-domain renditions of the same transform-domain solution~\eqref{bilaterallaplacetransformVF}.

A closer look at~\eqref{stemsolution} tells us  that the computation of the value function~$\valuefunction$ reduces to the derivation through~\eqref{dwfunction} of a related function, denoted~$\purevaluefunction$ in this work and referred to as the \acused{wfunction}\emph{`core' value function} or, more concisely, \acli{wfunction}.
Intuitively, $\purevaluex{\bl}$ corresponds to  the expected total cost experienced by the queue from an initial state~$\bl$ until it returns to the empty state~$0$. By construction, $\purevaluex{0}=0$, and the rest of~$\purevaluex{\bl}$ can be obtained by integration from~$0$ of its right-derivative~$\rightderivative\purevaluefunction$, available via~\eqref{bilaterallaplacetransformVF} or~\eqref{dwfunction}. % as an expectation with respect to the asymptotic (delay-free) waiting times. %, where~$\rightcostfunction$  a closed version of the piecewise continuous cost function.
%
%Note that the \acl{wfunction} is common to all queues with identical properties for $\bl>0$, as dictated by~$\ar$, $\stochst$ and~$\rightcostfunction$, with disregard to the parameters~$\stochstzero$ and~$\costx{0}$, which specify their particular behavior at $\bl=0$.
Observe that~$\purevaluefunction$ is fully characterized by~$\ar$, $\stochst$ and~$\rightcostfunction$, independently of the parameters~$\stochstzero$ and~$\costx{0}$, which specify the behavior of the queue at $\bl=0$.

  The rest of the study is  principally concerned with the derivation of the \acl{wfunction}, with disregard to the  other two terms in~\eqref{stemsolution}.  Once~$\purevaluefunction$ is known, the mean cost~$\meancost$ can be inferred from~$\stochstzero$ and~$\costx{0}$ on condition that~$\modulus{\costfunction}$ is $\measureXfunction{\trueWt}$-integrable. %Assumption~\ref{assumption:integrablecost}\ref{trueintegrablecost} holds: %from~$\purevaluefunction$  and   
Combining \eqref{derivativeatzero} %\boundaryconditions{}
with~\eqref{dwfunction} then yields
  %\boundaryconditions{}  yields $ \meancost= ({1-\load})/({\switchtoload{1-\ar \expectation{\stochst-\stochstzero} }{1-\load+\loadzero } }) \big( \frac{1}{\ar}\dpurevaluex{0} - \expectation{ \purevaluex{\stochst} - \purevaluex{\stochstzero}} \big)$ and, in combination with~\eqref{dwfunction}, 
\begin{equation}
%\textstyle%\begin{array}{c}
\label{genericmeancosttwo}% 
\meancost
=
%\frac{1}{\switchtoload{1+\ar\expectation{\stochstzero-\stochst} }{1-\load+\loadzero }} \big\{  (1-\load) (\costx{0}-\rightcostx{0} ) + \expectation{\costx{\Wt}} +  \ar    \expectation{ \int\nolimits_{\stochst}^{\stochstzero} \expectation{\costx{\xi+\Wt}} \, d\xi } \big\}
\left(\frac{1-\load}{\switchtoload{1+\ar\expectation{\stochstzero-\stochst} }{1-\load+\loadzero }} \right) \big\{  \dpurevaluex{0} /\ar +  \costx{0}-\rightcostx{0}  +     \expectation{  \purevaluex{\stochstzero}} -\expectation{ \purevaluex{\stochst}  }  \big\}
%\frac{ (1-\load) (\costx{0}-\rightcostx{0} ) + \expectation{\costx{\Wt}} +  \ar    \expectation{ \int\nolimits_{\stochst}^{\stochstzero} \expectation{\costx{\xi+\Wt}} d\xi } }{ \switchtoload{1+\ar\expectation{\stochstzero-\stochst} }{1-\load+\loadzero }  }
.
%\end{array}
\end{equation}
%$\ar    \expectation{\int\nolimits_{\stochst}^{\stochstzero}\Expectation{\costx{\xi+\Wt}}\, d\xi} =\ar    \expectation{ \frac{1-\load}{\ar}\int\nolimits_{\stochst}^{\stochstzero}\rightderivative\purevaluex{\xi}\, d\xi } = (1-\load)   \expectation{  \purevaluex{\stochstzero} - \purevaluex{\stochst}  } $ 
Note that the \acl{wfunction} and the mean cost~\eqref{genericmeancosttwo} are all we need for \acs{FPI}-dispatching, since the admission cost% defined in
~\eqref{definitionadmissioncost} reduces to
%\begin{align}
\begin{equation}
\label{admissioncost}\tag{AC'}
%\begin{array}{c}%{ll}
%\admissionixy{\server}{\bl}{\st}=\valueix{\server}{\bl+\st}-\valueix{\server}{\bl}=\purevalueix{\server}{\bl+\st}-\purevalueix{\server}{\bl} -\frac{\ar\meancosti{\server} }{1-\load} \st
 \admissioniux{\server}{\bl}{\st}
%\admissionfoperator{\ar}{\costfunction}\vect{\bl,\st}
%&=&\valuex{\bl+\st}-\valuex{\bl}
=
\purevalueix{\server}{\bli{\server}+\sti{\server}}-\purevalueix{\server}{\bli{\server}} - \left(\frac{\ari{\server}\meancosti{\server} }{1-\loadi{\server}}\right) \sti{\server}.
%& \quad \forall \bl,\st\in\Realpluszero,
%\end{array}
\end{equation}
%  induced by~$\stochstzero$ and~$\costx{0}$.
% The value function~$\valuefunction$ of the queueing system can be inferred from~$\purevaluefunction$ using~(\refeq{stemsolution})
%together with~(\refeq{boundarycondition}) and~(\refeq{derivativeatzero}).
%Indeed, we find 
\obsolete{
\begin{equation}\label{meancostperjob}\begin{array}{c}
\meancost= \frac{1-\load}{\switchtoload{1-\ar \expectation{\stochst-\stochstzero} }{1-\load+\loadzero } } \big( \frac{\dpurevaluex{0}}{\ar} - \expectation{ \purevaluex{\stochst} - \purevaluex{\stochstzero}} \big)
,
\end{array}\end{equation}
}%
%where~$\purevaluefunction$  is the \acl{wfunction} for~$\costfunction$, and~$\meancost$ is the corresponding mean cost per job computed by~(\refeq{genericmeancosttwo}).
%

%\end{remark}
%
%

%%%%%%%%%%%%%%

 In Sections~\ref{section:smoothcf}-\ref{section:piecewisedefined}, we exploit these results and derive the causal part of~$\rightderivative\purevaluefunction$  % (the causal part of~$\extdpurevaluex{\bl}$, for~$\bl\geq0$), 
 by inverse  transformation of~\eqref{bilaterallaplacetransformVF}.

%{\color{green}--------------------}
%

\subsection{Basic solutions: analytic cost functions}
\label{section:smoothcf}

%Let~$\dominantpoleX{\Wt}$ denote the dominant pole of~$\LSTWtfunction$ in the complex plane (cf. Appendix~\ref{appendix:LST} and Figure~\ref{figure:LSTWt}). 
%Consider 
%the half-plane $\newSing=\{\complex\in\Complex\setst\realpart{\complex}<-\dominantpoleX{\Wt}\}$, and define 
The analysis of~\eqref{bilaterallaplacetransformVF} is straightforward for 
%the family of cost functions
%
the cost functions belonging to the class
$\Excset\defeq\spanx{\{\costanfunction{\sing}{\nn}\setst\sing\in\Complex,\, \nn\in\Natural\}}$, where $\spanx{S}$ denotes the linear span of a set~$S$, and the function~$\costanfunction{\sing}{\nn}$, defined by
%\begin{align}  
%\label{costanx}
%\begin{array}{ll} 
%&
$\costanx{\sing}{\nn}{\bl}=\bl^\nn\exp{-\sing\bl}
$, %,
%, & %\quad
%(\sing\in\newSing,\ \nn\in\Natural),
%(\sing\in\Complex,\,\nn\in\Natural),
%\end{array}
%\end{align}
is characterized by the {meromorphic} Laplace transform  $\laplacetransformfs{\costanfunction{\sing}{\nn}}{\complex} = \factorial{\nn}/(\complex+\sing)^{\nn+1}$, 
 which is analytic on the complex plane except for a set of  isolated, non-essential singularities, called poles.
%which enjoy only a set of  isolated, non-essential singularities (poles), and are analytic on the rest of the complex plane.
%
%
%Table~\refeq{table:valuefunctions}
\begin{table}
\caption{Explicit \aclp{wfunction} 
for 
%$\costfunction\in\Excset$, ($\sing,\nn\neq 0$).
$\costfunction=\costanfunction{\sing}{\nn}$, ($\sing\in\newSing $, $\nn\in\Natural$). 
%($\sing\in\newSing\setminus\{0\} $, $\nn\in\Naturalpos$). 
%, $\{\specialmatrixnk{}{k}\}$ is given by~(\refeq{table:specialmatrixnk}), and~$\{\poweracoefnk{\sing}{}{k}\}$ by~(\refeq{table:poweracoefnk}). %$\costanx{\sing}{\nn}{\bl}=\bl^{\nn}\exp{-\sing\bl}$
\label{table:valuefunctions}}
%
%\horizontalline
\centerline{
\begin{tabular}{lll}%{|l|l|l|}%{|c|c|c|}
\hline
\vspace{-2mm}
\\
%$\vect{\sing,\nn}$
\beginline{} $\costx{\bl}$
\separator{}
%$\dpurevalueanx{\sing}{\nn}{\bl}$
$\dpurevaluex{\bl}$
\separator{}
%$\purevalueanx{\sing}{\nn}{\bl}$
$\purevaluex{\bl}$
\smallskip
\smallskip
  \\
%\hline
%$\vect{0,0}$
\beginline{} $1$
\separator{}
$\frac{\ar}{1-\load} $
\separator{}
$\frac{\ar}{1-\load}\bl$
\\
%\hline
%$\vect{\sing,0}$
\beginline{} $\exp{-\sing\bl}$
\separator{}
$\frac{\ar\scaley{\LSTWts{\sing}}{}}{1-\load}\scaley{}{\LSTWts{\sing}\,} \exp{-\sing\bl} $
\separator{}
$ \frac{\ar\scaley{\LSTWts{\sing}}{}}{1-\load}\scaley{}{\LSTWts{\sing}} \frac{1-\exp{-\sing\bl}}{\sing} $
\\
%\hline
%$\vect{0,\nn}$
\beginline{} $\bl^{\nn}$
\separator{}
$ \frac{\ar\, \factorial{\nn}}{1-\load}
\sum\nolimits_{k=0}^{\nn} \specialmatrixnk{\nn+1}{\nn-k}\frac{\bl^k}{\factorial{k}} $
\separator{}
%$ \frac{\ar\, \factorial{\nn}}{1-\load}\sum\nolimits_{k=1}^{\nn+1}\specialmatrixnk{\nn+1}{\nn-k+1} \frac{\bl^k}{\factorial{k}} $
$ \frac{\ar\, \factorial{\nn}}{1-\load}
\sum\nolimits_{k=0}^{\nn}\specialmatrixnk{\nn+1}{\nn-k} \frac{\bl^{k+1}}{\factorial{(k+1)}} $
\\
%\hline
%$\vect{\sing,\nn}$   
\beginline{} $\bl^{\nn}\exp{-\sing\bl}$
\separator{}
$ \frac{\ar\factorial{\nn}\scaley{\LSTWts{\sing}}{}}{1-\load}
\sum\nolimits_{k=0}^{\nn}  \poweracoefnk{\sing}{\nn}{\nn-k} \frac{\bl^k \exp{-\sing\bl}}{\factorial{k}} $
\separator{}
$ \frac{\ar\factorial{\nn} \scaley{\LSTWts{\sing}}{}}{1-\load}  \sum\nolimits_{t=0}^{\nn} \left( \sum\nolimits_{k=t}^{\nn}  \frac{\poweracoefnk{\sing}{\nn}{\nn-k}}{\sing^{k+1}} \right) \left( \dirack{t}- \frac{(\sing\bl)^t \exp{-\sing\bl} }{\factorial{t}} \right)  $
\!\!\!\!\! 
\smallskip
\\
%\hline
\end{tabular}
}
\horizontalline
\smallskip
Coefficients:
%$\specialmatrixnk{\nn}{0}=1$ and
\begin{equation}\nocolsep
\label{table:specialmatrixnk}\begin{array}{rll}
\specialmatrixnk{\nn}{0}
&=
1
,
&
\\
\specialmatrixnk{\nn}{k}
&=
%\frac{\ar}{1-\load} \sum\nolimits_{t=0}^{k-1} \powerzeroXcoefnk{k-t+1} \specialmatrixnk{\nn}{t}
[{\ar}/({1-\load})] \sum\nolimits_{t=0}^{k-1} \powerzeroXcoefnk{k-t+1} \specialmatrixnk{\nn}{t}
 , \quad 
% (k\in\Naturalpos) 
(k\geq 1) 
 ,
% \hspace{5cm} \
 \hspace{44.3mm} \
&
%\smallskip
 \\
\powerzeroXcoefnk{k}
&=
{1}/({\factorial{k}})\,\expectation{\stochst^{k}}
,
\quad 
(%\text{$k^{\textup{th}}$ coefficient of the power series of~$\LSTstochsts{-\complex}$ at~$0$}; 
k\geq 0),
& 
  \end{array}
\end{equation}
\storecompoundcounter{equation}{table:specialmatrixnk}%
%where $\powerzeroXcoefnk{k}={1}/({\factorial{k}})\,\expectation{\stochst^{k}}$ denotes the $k^{\textup{th}}$ coefficient of the power series of~$\LSTstochsts{-\complex}$ at~$0$, and
%
%
%For~$k=1,2,3,\dots$: if~$\expectation{\stochst\exp{-\sing\stochst}}, \dots, \expectation{\stochst^{k}\exp{-\sing\stochst}}<\infty$,
\begin{equation}
\label{table:poweracoefnk}
\nocolsep\begin{array}{rcl}
\poweracoefnk{\sing}{\nn}{0}
&=&
\scaley{1}{\LSTWts{\sing}}
,
\\ 
\poweracoefnk{\sing}{\nn}{1} 
&=& %\frac{\sing\ar}{\sing-\ar \left(1-\Expectation{\exp{-\sing\stochst}}\right)}
\scaley{
\frac{\ar\LSTWts{\sing} }{\left( 1-\load \right)\sing^2}
}{
%\frac{\ar }{ 1-\load } \big( \frac{\LSTWts{\sing}}{\sing}\big)^2
[{\ar }/({ 1-\load })] \, \big( {\LSTWts{\sing}}/{\sing}\big)^2
}  
\big( 1-
\LSTstochsts{\sing}%\Expectation{\exp{-\sing\stochst}}
-  \sing
\poweraXcoefnk{\sing}{1}%\Expectation{\stochst\exp{-\sing\stochst}}    
\big)
,
\\
\poweracoefnk{\sing}{\nn}{k} 
&=&
%\frac{\ar\LSTWts{\sing}}{\left(1-\load \right) \sing}  \big[ \frac{1-\ar \poweraXcoefnk{\sing}{1} }{\ar} \poweracoefnk{\sing}{\nn-1}{k-1} - \sum\nolimits_{t=0}^{k-2} \poweraXcoefnk{\sing}{k-t} \poweracoefnk{\sing}{\nn-k+t}{t}  \big] 
[{\ar\LSTWts{\sing}}/{(1-\load ) \sing}]\,  \big[ ({1/{\ar}- \poweraXcoefnk{\sing}{1} })\, \poweracoefnk{\sing}{\nn-1}{k-1}  - \sum\nolimits_{t=0}^{k-2} \poweraXcoefnk{\sing}{k-t} \poweracoefnk{\sing}{\nn-k+t}{t}  \big] 
,
 %\\ &&\hspace{66mm} 
 \quad 
 %(k=2,3,4,\dots) 
 (k\geq 2) 
 , 
% \hspace{18mm} \
 \hspace{1mm} \ 
 %\smallskip
 \\
 \poweraXcoefnk{\sing}{k}
 &=&
 {1}/({\factorial{k}})\,\expectation{\stochst^{k}\exp{-\sing\stochst}},
 \quad
 (%\text{$k$-th coef. of the power series of~$\LSTstochsts{-\complex}$ at~$-\sing$};
 k\geq 0)
 .
\end{array}\end{equation}
\storecompoundcounter{equation}{table:poweracoefnk}%
%where $\poweraXcoefnk{\sing}{k}={1}/({\factorial{k}})\,\expectation{\stochst^{k}\exp{-\sing\stochst}}$ is the $k^{\textup{th}}$ coefficient of the power series of~$\LSTstochsts{-\complex}$ at~$-\sing$.
%\\ 
\obsolete{ 
\horizontalline 
\smallskip
with $\specialmatrixnk{\nn}{0}=1$ and
\begin{equation}\label{table:specialmatrixnk}\begin{array}{ll}
%\specialmatrixnk{\nn}{0}=1
%,
%&
%\\
\specialmatrixnk{\nn}{k}
=
%\frac{\ar}{1-\load} \sum\nolimits_{t=0}^{k-1} \powerzeroXcoefnk{k-t+1} \specialmatrixnk{\nn}{t}
[{\ar}/({1-\load})] \sum\nolimits_{t=0}^{k-1} \powerzeroXcoefnk{k-t+1} \specialmatrixnk{\nn}{t}
 ,& \quad 
% (k\in\Naturalpos) 
(k\geq 1) 
 ,
% \hspace{5cm} \
 \hspace{45mm} \
  \end{array}
\end{equation}
where
$\powerzeroXcoefnk{k}={1}/({\factorial{k}})\,\expectation{\stochst^{k}}$ denotes the $k^{\textup{th}}$ coefficient of the power series of~$\LSTstochsts{-\complex}$ at~$0$, and
%
%
%For~$k=1,2,3,\dots$: if~$\expectation{\stochst\exp{-\sing\stochst}}, \dots, \expectation{\stochst^{k}\exp{-\sing\stochst}}<\infty$,
\begin{equation}
\label{table:poweracoefnk}
\nocolsep\begin{array}{rcl}
\poweracoefnk{\sing}{\nn}{0}
&=&
\scaley{1}{\LSTWts{\sing}}
,
\\
\poweracoefnk{\sing}{\nn}{1} 
&=& %\frac{\sing\ar}{\sing-\ar \left(1-\Expectation{\exp{-\sing\stochst}}\right)}
\scaley{
\frac{\ar\LSTWts{\sing} }{\left( 1-\load \right)\sing^2}
}{
%\frac{\ar }{ 1-\load } \big( \frac{\LSTWts{\sing}}{\sing}\big)^2
[{\ar }/({ 1-\load })] \, \big( {\LSTWts{\sing}}/{\sing}\big)^2
}  
\big( 1-
\LSTstochsts{\sing}%\Expectation{\exp{-\sing\stochst}}
-  \sing
\poweraXcoefnk{\sing}{1}%\Expectation{\stochst\exp{-\sing\stochst}}    
\big)
,
\\
\poweracoefnk{\sing}{\nn}{k} 
&=&
%\frac{\ar\LSTWts{\sing}}{\left(1-\load \right) \sing}  \big[ \frac{1-\ar \poweraXcoefnk{\sing}{1} }{\ar} \poweracoefnk{\sing}{\nn-1}{k-1} - \sum\nolimits_{t=0}^{k-2} \poweraXcoefnk{\sing}{k-t} \poweracoefnk{\sing}{\nn-k+t}{t}  \big] 
[{\ar\LSTWts{\sing}}/{(1-\load ) \sing}]\,  \big[ ({1/{\ar}- \poweraXcoefnk{\sing}{1} })\, \poweracoefnk{\sing}{\nn-1}{k-1}  - \sum\nolimits_{t=0}^{k-2} \poweraXcoefnk{\sing}{k-t} \poweracoefnk{\sing}{\nn-k+t}{t}  \big] 
,
 %\\ &&\hspace{66mm} 
 \quad 
 %(k=2,3,4,\dots) 
 (k\geq 2) 
 , 
% \hspace{18mm} \
 \hspace{1mm} \
\end{array}\end{equation}
where
$\poweraXcoefnk{\sing}{k}={1}/({\factorial{k}})\,\expectation{\stochst^{k}\exp{-\sing\stochst}}$ is the $k^{\textup{th}}$ coefficient of the power series of~$\LSTstochsts{-\complex}$ at~$-\sing$.
%\\ 
}
\horizontalline 
\end{table}
\storecompoundcounter{table}{table:valuefunctions}%
Observe that the condition of existence of the \acl{wfunction}, previously stated in Assumption~\ref{assumption:complexintegrablecost}, reduces for the cost function~$\costanfunction{\sing}{\nn}$ to $\sing\in\newSing$, where we write $\newSing=\{\complex\in\Complex\setst\realpart{\complex}<-\dominantpoleX{\Wt}\}$.

Table~\ref{table:valuefunctions} provides us with the closed-form \aclp{wfunction} for the cost function~$\costanfunction{\sing}{\nn}$, %in~$\Excset$, 
obtained after inversion of~\eqref{bilaterallaplacetransformVF} by integration along a vertical axis in the region of absolute convergence of~$\Blaplacetransformf{\extdpurevaluefunction}$, as we proceed to do now. 
Let $\anchor\in(\sing,-\dominantpoleX{\Wt})$, and consider the contour 
%$\contoura{\radius}=\{\anchor+\i x \setst x\in[-\radius,\radius]\} \cup \arca{\radius}$ 
%\begin{equation}\label{contoura}\begin{array}{c}
$\contoura{\radius}=\{\anchor+\i t \setst t\in[-\radius,\radius]\} \cup \arca{\radius}$,
%\end{array}\end{equation}
where 
$\arca{\radius}=\{\anchor+\radius\exp{\i \alpha }\setst \alpha\in[\frac{\pi}{2},\frac{3\pi}{2}]\}$ is an arc centered in~$\anchor$. 
Since $\lim\nolimits_{\complex\to\infty}\modulus{\LSTWts{\complex}}\leq 1$ (cf. Proposition~\ref{proposition:analycity}\eqref{analycity:nopole}), we find 
%\begin{equation}\label{jordanconditionapplied}\begin{array}{ll}
$
\lim\nolimits_{\radius\to\infty} \LSTWts{-\anchor-\radius\exp{\i \alpha}}\laplacetransformfs{\costanfunction{\sing}{\nn}}{\anchor+\radius\exp{\i \alpha}}= 0
$ 
%,& \forall \alpha\in\left[\frac{\pi}{2},\frac{3\pi}{2}\right],
%\end{array}\end{equation}
%
for $\alpha\in\left[\frac{\pi}{2},\frac{3\pi}{2}\right]$,
and the condition of the third Jordan lemma is satisfied~\cite[\S{}3.1.4, Theorem~1]{mitrinovic84}\cite[\S{}88]{brown14}. % for~$\LSTWts{-(\complex)}\laplacetransformfs{\costanfunction{\sing}{\nn}}{\complex}$.
It follows that integration of 
%$\LSTWts{-\complex}\laplacetransformfs{\costanfunction{\sing}{\nn}}{\complex}$ 
$\Blaplacetransformfs{\extdpurevaluefunction}{\complex}\exp{\complex \bl}$ along the arc~$\arca{\radius}$ vanishes as~$\radius\to\infty$, 
\begin{align}\label{jordanresult}
%\begin{array}{ll}
&\lim_{\radius\to\infty}\int\nolimits_{\arca{\radius}}\LSTWts{-\complex}\laplacetransformfs{\costanfunction{\sing}{\nn}}{\complex}\, \exp{\complex \bl}\, d\complex=0, 
&  \forall \bl\in\Realpluszero,
%\end{array}
\end{align}
and counterclockwise integration of
%$\LSTWts{-\complex}\laplacetransformfs{\costanfunction{\sing}{\nn}}{\complex}$ 
$\Blaplacetransformfs{\extdpurevaluefunction}{\complex}\exp{\complex \bl}$
on the contour~$\contoura{\radius}$ reduces to computing the residue\footnotemark{} at the pole of~$\laplacetransformf{\costanfunction{\sing}{\nn}}$. % the cost function. 
\footnotetext{
Recall that the residue of a meromorphic function~$\genericfunction$ at a pole~$\sing$ of order~$n$ is given by~\cite{brown14}
\begin{equation}\label{residue}\textstyle
 \Residuefx{\genericx\complex}{\complex=\sing} = \frac{1}{\factorial{(n-1)}} \lim_{\complex\to\sing}\frac{d^{n-1}}{d\complex^{n-1}} \big[  (\complex-\sing)^n \genericx{\complex} \big].
\end{equation}
}%
%
%Using the residue theorem, we find for
The residue theorem gives
\begin{equation}\label{bilaterallaplacetransformcostanfunction}
\begin{array}{rcl}
\dpurevalueanx{\sing}{\nn}{\bl}
&\refereq{\eqref{bilaterallaplacetransformVF}}{=} &
\frac{1}{2\pi\i}\lim\nolimits_{t\to\infty}
\int\nolimits_{\anchor-\i t}^{\anchor+\i t}
\big(
\frac{\ar}{1-\load} 
%\left( 
\LSTWts{-\complex}
\frac{\factorial{\nn}}{(\complex+\sing)^{\nn+1}}
\big)  \exp{\complex\bl}
\, d\complex
%\complex \Blaplacetransformfs{\extvaluefunction}{\complex} -{\valuex{0}}
\\
 &\refereq{\eqref{jordanresult}}{=} &
 \frac{1}{2\pi\i}
\ointctrclockwise\nolimits_{\contoura{\radius}}
\big(
\frac{\ar}{1-\load} 
%\big( 
\LSTWts{-\complex}
\frac{\factorial{\nn}}{(\complex+\sing)^{\nn+1}}
\big)  \exp{\complex\bl}
\, d\complex
\\
 &= &
% \frac{1}{2\pi\i} \big[ 2\pi\i 
\bigresiduefx{
\frac{\ar}{1-\load} \LSTWts{-\complex} \frac{\factorial{\nn}}{(\complex+\sing)^{\nn+1}} \exp{\complex\bl}
}{\complex=-\sing}  
%\big]
\\ 
 & \refereq{\eqref{residue}}{=}   &
\frac{\ar}{1-\load} 
\lim\nolimits_{\complex\to-\sing}\frac{1}{\factorial{\nn}} \frac{d^{\nn}}{d\complex^{\nn}}  \big[ \factorial{\nn}\, \LSTWts{-\complex}  \exp{\complex\bl} \big]
\\
 &= &
\frac{\ar}{1-\load} 
 \sum\nolimits_{k=0}^{\nn} \binomialcoef{\nn}{k}  \big( (-1)^{\nn-k}\frac{d^{\nn-k} }{d\complex^{\nn-k}} \LSTWts{\sing} \big)  \bl^k \exp{-\sing\bl}
\\
 &= &
\frac{\factorial{\nn}\ar}{1-\load} 
  \sum\nolimits_{k=0}^{\nn}  
  %\big( \frac{(-1)^{\nn-k}}{\factorial{(\nn-k)}}  \frac{d^{\nn-k} }{d\complex^{\nn-k}} \LSTWts{\sing} \big) 
  \poweracoefnk{\sing}{\nn}{\nn-k}
  \big(\frac{ \bl^k}{\factorial{k}}\big) \exp{-\sing\bl}
,
\end{array}
\end{equation} 
for all 
$\bl\in\Realpluszero$, in which
%the quantity between parentheses 
%\begin{equation}   
% \textstyle
\begin{equation}\label{definition:poweracoefnk}\notag
%\textstyle
\poweracoefnk{\sing}{\nn}{k}
 =
 \frac{(-1)^{k}}{\factorial{k}}\, \frac{d^{k} }{d\complex^{k}} 
 %\LSTWtfunction({\sing})
\LSTWts{\sing}
\end{equation}
\storecompoundcounter{equation}{definition:poweracoefnk}%
%\end{equation}
is the 
$k$th %$(\nn-k)$-th  
coefficient of the Taylor expansion of~$\LSTWts{-\complex}$ at~$\sing$, 
reducing to $\poweracoefnk{0}{\nn}{k}\equiv\specialmatrixnk{}{k}$ if~$\sing=0$.
The coefficients~$\{\specialmatrixnk{}{k}\}$ and~$\{ \poweracoefnk{\sing}{}{k}\}$ will be referred to as the germ of~$\LSTWts{-\complex}$. In~(\ref{table:specialmatrixnk}) and~(\ref{table:poweracoefnk}), they are computed inductively as functions of the coefficients~$\{\powerzeroXcoefnk{k}\}$ and~$\{\poweraXcoefnk{\sing}{k}\}$ of the power series of~$\LSTstochsts{\complex}$. % at the pole of the cost function, 
As such, they are finite by analycity of~$\LSTstochsts{\complex}$ on~$\newSing$ (cf. Proposition~\ref{proposition:analycity}\eqref{analycity:dominantpoles}).
See also Proposition~\ref{proposition:analycity}\eqref{analycity:centered}-\eqref{analycity:excentered} for a derivation of~(\ref{table:specialmatrixnk}) and~(\ref{table:poweracoefnk}), and
Table~\ref{table:moments} for expressions of~$\{\specialmatrixnk{}{k}\}$ specific to standard service time distributions. % (constant, exponential, Erlang).
%
%computed in Proposition~\refeq{proposition:analycity}\eqref{analycity:centered}% in Appendix~\refeq{appendix:LST}
%.
%It is equal to~$\specialmatrixnk{}{\nn-k}$ if~$\sing=0$, and to~$\scaley{\LSTWts{\sing} }{} \poweracoefnk{\sing}{}{\nn-k}  $  if~$\sing\neq0$, with~$\{\specialmatrixnk{}{k}\}$ and~$\{ \poweracoefnk{\sing}{}{k}\}$  respectively given by~(\ref{table:specialmatrixnk}) and~(\ref{table:poweracoefnk}) as functions of the quantities~$\{\poweraXcoefnk{\sing}{k}\}$,
%, which are scaled versions of the moment of the service time distribution, and  the actual 
%which are the coefficients of the power series of~$\LSTstochsts{\complex}$ at the pole of the cost function. 
%
%
The final expressions\footnotemark{} for~$\dpurevalueanfunction{\sing}{\nn}$ and~$\purevalueanfunction{\sing}{\nn}$ are reported in Table~\ref{table:valuefunctions}.

\footnotetext{Alternatively, notice that
$
%\costanfunction{\sing}{\nn}=(-1)^{\nn} \frac{\delta^\nn}{\delta\sing^\nn}[\costanfunction{\sing}{0}]
\costanfunction{\sing}{\nn}=(-1)^{\nn} ({\delta^\nn}/{\delta\sing^\nn})\, \costanfunction{\sing}{0}
$ if~$\sing\in\newSing\setminus\{0\}$. It follows
\showhideproposition{%
 from~\eqref{dwfunction} % in Proposition~\ref{proposition:identities}\eqref{proposition:vfdifferentiability}
 }{%
from~\cite[Proposition~1]{hyytia-vgf-itc-2017} 
  }%
   and
the Leibniz integral rule \hidecalculustheorems{}{(Theorem~\ref{theorem:lir}% in Appendix~\refeq{appendix:ar}
) }%
that, for $\sing\in\newSing\setminus\{0\}$ and $n>0$, %$\nn\in\Naturalpos$,
$%\begin{equation}\label{purevalueanxbydifferentiation}
%\textstyle%\begin{array}{rcll}
\dpurevalueanx{\sing}{\nn}{\bl}
=%&=&
%(-1)^{\nn} \frac{\delta^\nn}{\delta\sing^\nn}[\dpurevalueanx{\sing}{0}{\bl}]
(-1)^{\nn} ({\delta^\nn}/{\delta\sing^\nn}) \dpurevalueanx{\sing}{0}{\bl}
%&\\
=%&=&%\refereq{(\refeq{secondvf})}{=}&
%(-1)^{\nn}\frac{\ar}{1-\load} \frac{\delta^\nn}{\delta\sing^\nn} [\LSTWts{\sing} \exp{-\sing\bl} ]
(-1)^{\nn} [{\ar}/({1-\load})] \, ({\delta^\nn}/{\delta\sing^\nn}) [\LSTWts{\sing} \exp{-\sing\bl} ]
%,&%\forall \bl \in\Realpluszero,
 %\quad (\sing\in\newSing\setminus\{0\},\  \nn\in\Naturalpos),
%\end{array}
$, 
%\end{equation}
and the expressions for~$\dpurevalueanfunction{\sing}{\nn}$  %given in Table~\refeq{table:valuefunctions} 
can be derived by successive differentiations of~$\dpurevalueanfunction{\sing}{0}$.
By continuity arguments, we also find, for $n>0$, %$\nn\in\Naturalpos$, 
$%\begin{equation}\label{purevaluenxbydifferentiation}
%\begin{array}{ll}
\dpurevalueanx{0}{\nn}{\bl}
=
%(-1)^{\nn}\lim\nolimits_{\sing\to 0} \frac{\delta^\nn}{\delta\sing^\nn}[\dpurevalueanx{\sing}{0}{\bl}] %,
(-1)^{\nn}\lim\nolimits_{\sing\to 0} ({\delta^\nn}/{\delta\sing^\nn})[\dpurevalueanx{\sing}{0}{\bl}] %,
%& \quad ( \nn\in\Naturalpos).
%\end{array}
$. %\end{equation}
}%

Since the operation $\costfunction\mapsto\purevaluefunction$ is a linear map, %(Corollary~\ref{vfcorollary}), 
observe that all  cost functions given as linear combinations of~$\costanfunction{\sing}{\nn}$ types are elements of~$\Excset$
%functions in~$\Excset$ 
enjoying explicit value functions. %Equivalently, t
Examples include the trigonometric functions $\cos$ and~$\sin$, which play a part in the developments of Section~\ref{section:trigonometricsum}, or the set of incomplete gamma functions $\{\incompletegamma{\nn+1}{\sing\cdot}\setst\nn\in\Natural,\,\sing\in\Complex\}$, which spans~$\Excset$ completely.
\obsolete{
In particular, $\Excset\equiv\spanx{\{\incompletegamma{\nn+1}{\sing\cdot}\setst\nn\in\Natural,\,\sing\in\Complex\}}$, where%
%The set of these computable cost functions is spanned by the set of the functions~$\incompletegamma{\nn+1}{\sing\bl}$, where $\nn\in\Natural$, $\sing\in\Complex$ and
~$\incompletegammafunction$ denotes the incomplete gamma function.
For $\nn\in\Natural$, $\sing\in\Complex$, and the cost function $\costx{\bl}=\incompletegamma{\nn+1}{\sing \bl} = \factorial{\nn}\, \sum\nolimits_{j=0}^{\nn} [ (\sing\bl)^j/\factorial{j} ] \, e^{-\sing \bl} $, we find, from Table~\ref{table:valuefunctions}, 
\begin{align}
\label{incompletegammadpurevaluex}
%\begin{array}{ll}
&
\textstyle
\dpurevaluex{\bl} 
 =
% \frac{\ar\scaley{\LSTWts{\sing}}{}}{1-\load}\exp{-\sing\bl}  \factorial{\nn}\, \sum\nolimits_{j=0}^{\nn} \sing^j \sum\nolimits_{k=0}^{j}  \poweracoefnk{\sing}{}{j-k} \frac{\bl^k }{\factorial{k}}
%&
%\\
%=
\frac{\ar\scaley{\LSTWts{\sing}}{}\factorial{\nn}}{1-\load}   \sum\nolimits_{k=0}^{\nn}  
[\sum\nolimits_{q=0}^{\nn-k} \sing^{q} \poweracoefnk{\sing}{}{q}] \frac{(\sing\bl)^k }{\factorial{k}} \exp{-\sing\bl} ,
%({\ar\scaley{\LSTWts{\sing}}{}\factorial{\nn}})/({1-\load} ) \,  \sum\nolimits_{k=0}^{\nn}  [\sum\nolimits_{q=0}^{\nn-k} \sing^{q} \poweracoefnk{\sing}{}{q}] \frac{(\sing\bl)^k }{\factorial{k}} \exp{-\sing\bl} ,
&  \forall  \bl\in\Realpluszero . % \nn\in\Natural.
%\end{array}
\end{align}
If $ \sing\neq 0$, the \acl{wfunction} is then given by
\begin{align}\label{incompletegammapurevaluex}
%\begin{array}{ll}
&\textstyle
\purevaluex{\bl} 
=
\frac{\ar\scaley{\LSTWts{\sing}}{}\factorial{\nn}}{\sing(1-\load)}   \sum\nolimits_{k=0}^{\nn}  
[\sum\nolimits_{q=0}^{\nn-k} \sing^{q} \poweracoefnk{\sing}{}{q}] [1-\frac{1}{\factorial{k}}\incompletegamma{\nn+1}{\sing \bl}] ,
& \forall  \bl\in\Realpluszero. %, \sing\in\newSing\setminus\{0\}, \nn\in\Natural.
%\end{array}
\end{align}
}%

\subsection{Piecewise-defined cost functions}
\label{section:piecewisedefined}

Let~$\costifunction{0},\costifunction{1}\in\Excset$, and assume  the cost function is given by $\costfunction = \pwcostifunction{0} 
\, \indicatorfunction{[0,\tau)} + \pwcostifunction{1} 
\, \indicatorfunction{[\tau,\infty)}$, where~$\indicatorfunction{\cdot}$ denotes the indicator function, or, equivalently, 
\begin{equation}\label{piecewisecost}
\begin{array}{ll}
%\costfunction = \pwcostifunction{0}  + \pwdeltafunction \, \indicatorfunction{[\tau,\infty)} 
\costx{\bl} = \pwcostix{0}{\bl}  + \pwdeltax{\bl} \, \indicator{[\taufunction,\infty)}{\bl} 
,
& \quad \forall\bl\in\Realpluszero.
\end{array}
\end{equation} 
where $\pwdeltafunction=\costifunction{1}-\costifunction{0}$.
Since the Laplace transform of~$\fx{\bl}=\bl^\nn\exp{-\sing\bl}\, \indicator{[\taufunction,\infty)}{\bl} %\stepx{\bl-\taufunction}
$ is given on the half-plane $\realpart{\complex}>\realpart{-\sing}$ by
\begin{equation}\label{pwLT}
%\begin{array}{rcl}
\textstyle
\laplacetransformfs{\functionfunction}{\complex} 
=
\int\nolimits_{\taufunction}^{\infty}\bl^{\nn} e^{-(\complex+\sing)\bl}\, d\bl 
%=
%\frac{1}{(\complex+\sing)^{\nn+1}}\int\nolimits_{\taufunction(\complex+\sing)}^{\infty}\altbl^{(\nn+1)-1} e^{-\altbl}\, d\altbl 
%=
%\frac{\incompletegamma{\nn+1}{(\complex+\sing)\taufunction}}{(\complex+\sing)^{\nn+1}}
=
%\factorial{\nn} \sum\nolimits_{q=0}^{\nn} \frac{\taufunction^q}{\factorial{q}(\complex+\sing)^{\nn-q+1}} \, \exp{-(\complex+\sing)\taufunction}
\factorial{\nn} \exp{-\sing\taufunction} \sum\nolimits_{q=0}^{\nn} \frac{\taufunction^q}{\factorial{q}(\complex+\sing)^{\nn-q+1}} \, \exp{-\complex\taufunction}
%\\&=&
%\factorial{\nn}\sum\nolimits_{t=0}^{\nn}\frac{\taufunction^t(\complex+\sing)^{t-\nn-1}}{\factorial{t}}\, \exp{-(\complex+\sing)\taufunction}
,
%\end{array}
\end{equation}
%where~$\incompletegamma{x}{y}=\int\nolimits_{y}^{\infty}t^{x-1} e^{-t}\, dt$ is the incomplete gamma function.
we can find~$\pwshiftedlaplacefunction$ such that $\laplacetransformfs{\costfunction}{\complex} = \laplacetransformfs{\costifunction{0}}{\complex} + \pwshiftedlaplacest{\complex}{\taufunction} e^{-\complex\taufunction}
$
with
$\lim\nolimits_{\complex\to\infty}\pwshiftedlaplacest{\complex}{\taufunction}=0$.
If we place~$\anchor$ in the half-plane~$\Realpart{\complex}>0$ between~$-\dominantpoleX{\Wt}$ and the poles of $\pwcostifunction{0},\pwcostifunction{1}$, \eqref{jordanresult} becomes in the present setting, %
\begin{equation}\label{jordanpwresult}\notag
\begin{array}{ll}
\lim\nolimits_{\radius\to\infty}\int\nolimits_{\arca{\radius}}\LSTWts{-\complex}\laplacetransformfs{\costifunction{0}}{\complex}\, \exp{\complex \bl}\, d\complex=0, 
& \quad \forall \bl\in\Realpluszero,
\end{array}
\end{equation}
for the first term, and
\begin{equation}\label{jordanpiecewisel}\notag
\begin{array}{ll}
\lim\nolimits_{\radius\to\infty}\int\nolimits_{\arca{\radius}}\LSTWts{-\complex}%
\pwshiftedlaplacest{\complex}{\taufunction} 
\,\exp{\complex(\bl-\taufunction)} \, d\complex=0,
& \forall \bl\in(\taufunction,\infty) , % \ (l=1,\dots,\pwn-1),
\\
\lim\nolimits_{\radius\to\infty}\int\nolimits_{\arca{-\radius}}\LSTWts{-\complex}
\pwshiftedlaplacest{\complex}{\taufunction} 
\,\exp{\complex(\bl-\taufunction)} \, d\complex=0,
\quad & \forall \bl\in[0,\taufunction) , % \ (l=1,\dots,\pwn-1),
\end{array}
\end{equation}
for the second term.
Thus, inverse transformation by {counterclockwise integration} along~$\contoura{\radius}$ still applies for all backlog values $\bl>\taufunction $, where 
\begin{equation} \label{solutionFNPlarge}
\textstyle
\rightderivative\purevaluex{\bl}
= \frac{\ar}{1-\load}
\sum_{\pole\in\Polef{\laplacetransformf{\costfunction}}}\bigresiduefx{
 \LSTWts{-\complex} \laplacetransformfs{\costfunction}{\complex} \, e^{\complex\bl} }{\complex=\pole} 
 ,
 \quad
 \forall \bl\in (\taufunction,\infty)
 .
\end{equation}
It is clear from~\eqref{dwfunction} that the derivative~$\rightderivative\purevaluex{\bl}$ for $\bl>\taufunction$ does {not} depend on the values of the cost function on the interval $(0,\taufunction)$. It is therefore equal over $(\taufunction,\infty)$ to the derivative of the \acl{wfunction} for the analytic cost~$\costfunction=\costifunction{1}$, and it can equivalently be derived from~\eqref{bilaterallaplacetransformcostanfunction} (or, alternatively, inferred from Table~\ref{table:valuefunctions}) for the cost function~$\costfunction=\costifunction{1}$.
%where~$\dpurevaluex{\bl}$ is given by~\eqref{bilaterallaplacetransformcostanfunction} (or, alternatively, by  the entries of  Table~\refeq{table:valuefunctions}) for the cost function $\costifunction{1}+\pwdeltafunction\equiv\costifunction{1}$. 

For $\bl<\taufunction $, however, the terms~$\costifunction{0}$ and~$\pwdeltafunction\indicatorfunction{[\taufunction,\infty)}$  in~\eqref{piecewisecost} must be treated separately:    $\costifunction{0}$ by simple inspection of Table~\ref{table:valuefunctions}, and $\pwdeltafunction\indicatorfunction{[\taufunction,\infty)}$ by {clockwise integration} along the contour $\contoura{-\radius}=\{\anchor+\i t \setst t\in[-\radius,\radius]\} \cup \arca{-\radius}$.  
The success of this last operation is conditioned by the singularities of~$\LSTWts{-\complex}$, all contained in the interior of~$\contoura{-\radius}$ as $\radius\to\infty$. In our discussion we consider separately the service time distributions for which~$\LSTWtfunction$ has a finite set of poles~$\Polef{\LSTWtfunction}$ (e.g. exponential or Erlang service time distributions), and those for which~$\LSTWtfunction$ has infinitely many poles (as in discrete service time distributions).

If~$\Polef{\LSTWtfunction}$ is finite,  
the clockwise integral along~$\contoura{-\radius}$ yields~$\cardinality{\Polef{\LSTWtfunction}}$ residues at the poles of~$\LSTWts{-\complex}$, and we find
\begin{equation} \label{solutionFNP}
\textstyle
\begin{array}{c} 
\rightderivative\purevaluex{\bl}
= \frac{\ar}{1-\load}
\sum_{\pole\in\Polef{\laplacetransformf{\costifunction{0}}}}\bigresiduefx{
 \LSTWts{-\complex} \laplacetransformfs{\costifunction{0}}{\complex} \, e^{\complex\bl} }{\complex=\pole} 
 \hspace{2.5cm} \\  \hfill
 -
\frac{\ar}{1-\load}
\sum_{\pole\in\Polef{\LSTWtfunction}}\bigresiduefx{
 \LSTWts{-\complex} \pwshiftedlaplacest{\complex}{\taufunction} \, e^{\complex(\bl-\taufunction)} }{\complex=-\pole} 
 , \  \
 \forall \bl\in (0,\taufunction)
 .
\end{array}
\end{equation}

If otherwise~$\Polef{\LSTWtfunction}$ is infinite, then  
the clockwise integral along~$\contoura{-\radius}$ cannot be computed directly by the residue theorem, which would issue an infinite sum. This difficulty can nevertheless be overcome whenever~$\LSTWtfunction$ rewrites as 
\begin{equation}\label{generaltrick}
%\LSTWts{-\complex}=\iLSTWtxs{\bl}{-\complex}+\fLSTWtxs{\bl}{-\complex}, 
\LSTWtfunction=\iLSTWtxfunction{\bl}+\fLSTWtxfunction{\bl}, 
\quad \forall \bl\in(0,\taufunction),
\end{equation} 
where~$\iLSTWtxfunction{\bl}$ and~$\fLSTWtxfunction{\bl}$ are meromorphic, %analytic on~$\ROC{\Blaplacetransformf{\extdpurevaluefunction}}$,
$\cardinality{\Polef{\fLSTWtxfunction{\bl}}}$ is finite, and
\begin{equation}\notag
\nocolsep
 \begin{array}{c}
%$
\lim\limits_{\radius\to\infty}\int\nolimits_{\arca{-\radius}} 
\fLSTWtxs{\bl}{-\complex}  \pwshiftedlaplacest{\complex}{\taufunction} \,
  \exp{\complex(\bl-\taufunction)}
 \, d\complex
 =
%0
% $,
%,
%\\
%$
\lim\limits_{\radius\to\infty}\int\nolimits_{\arca{\radius}} 
\iLSTWtxs{\bl}{-\complex} \pwshiftedlaplacest{\complex}{\taufunction} 
 \, \exp{\complex(\bl-\taufunction)}
 \, d\complex=0
.%$,
\end{array}
\end{equation}
 \obsolete{
\begin{align}
\label{jordanpiecewisegeneraltrickone}
\begin{array}{c}
\lim\nolimits_{\radius\to\infty}\int\nolimits_{\arca{\radius}} 
\iLSTWtxs{\bl}{-\complex} \pwshiftedlaplacest{\complex}{\taufunction} 
 \, \exp{\complex(\bl-\taufunction)}
 \, d\complex=0, %& \forall \bl\in\Realpluszero ,
\end{array}
\\
\label{jordanpiecewisegeneraltricktwo}
\begin{array}{c}
\lim\nolimits_{\radius\to\infty}\int\nolimits_{\arca{-\radius}} 
\fLSTWtxs{\bl}{-\complex}  \pwshiftedlaplacest{\complex}{\taufunction} 
 \, \exp{\complex(\bl-\taufunction)}
 \, d\complex=0. %& \forall \bl\in\Realpluszero ,
\end{array}
\end{align}
}%
Then, if 
we choose $\anchor \in\ROC{\Blaplacetransformf{\extdpurevaluefunction}}$
%
%$\anchor\in(\dominantpoleX{\stochst},-\dominantpoleX{\Wt})$ in the \acp{ROC} of both~$\iLSTWtxs{\bl}{-\complex}$ and~$\fLSTWtxs{\bl}{-\complex}$ 
and  consider the pole sets 
$\iPolex{\bl} = \{\pole\in \Polef{\iLSTWtxs{\bl}{-\cdot} \, \pwshiftedlaplacest{\cdot}{\taufunction}},\realpart{\pole}<\anchor\}$ 
%$\iPolex{\bl} = \{\pole\in-\Polef{\iLSTWtxfunction                                                                                                                                                                                                                                                           {\bl}},\realpart{\pole}<\anchor\}$ 
%$\iPolex{\bl} = \{\pole\in\Polef{\iLSTWtxfunction                                                                                                                                                                                                                                                           {\bl}},\realpart{-\pole} \notin\ROC{\laplacetransformf{\costfunction}} \}$ 
and 
$\fPolex{\bl}=\{ \pole\in\Polef{\fLSTWtxs{\bl}{-\cdot} \, \pwshiftedlaplacest{\cdot}{\taufunction}},\realpart{\pole}>\anchor \}$,
%$\fPolex{\bl}=\{ \pole\in-\Polef{\fLSTWtxfunction{\bl}},\realpart{\pole}<\anchor \}$,
%$\fPolex{\bl}=\{ \pole\in\Polef{\fLSTWtxfunction{\bl}},\realpart{-\pole}\geq-\dominantpoleX{\Wt} \}$, 
both finite in cardinality, 
we find  
\begin{equation} \label{solutionINP}
\nocolsep  
\begin{array}{rl}
\rightderivative\purevaluex{\bl}
= \hspace{1mm} &
\frac{\ar}{1-\load}
\sum_{\pole\in\Polef{\laplacetransformf{\costifunction{0}}}}\bigresiduefx{
 \LSTWts{-\complex} 
\laplacetransformfs{\costifunction{0}}{\complex} \, e^{\complex\bl}}{\complex=\pole}  
 \hspace{1cm} \\  
 +&
\frac{\ar}{1-\load} \sum_{\pole\in
%\Polef{\pwshiftedlaplacest{\cdot}{\taufunction}}\cup
\iPolex{\bl}} \bigresiduefx{   \iLSTWtxs{\bl}{-\complex} \, \pwshiftedlaplacest{\complex}{\taufunction} \, e^{\complex(\bl-\taufunction)} }{\complex=\pole} 
 \hspace{1cm} \\  
 -&
\frac{\ar}{1-\load}\sum_{\pole\in\fPolex{\bl}}\bigresiduefx{ \fLSTWtxs{\bl}{-\complex} \pwshiftedlaplacest{\complex}{\taufunction} \, e^{\complex(\bl-\taufunction)} }{\complex=\pole} 
 , \quad
 \forall \bl\in (0,\taufunction)
 .
\end{array}
\end{equation}

As we see below, the decomposition proposed in~\eqref{generaltrick} is relevant in particular in the case of discrete service time distributions.

\paragraph{Discrete service time distributions.}\label{section:discretecase}
Consider the M/D/1 queue, where all the jobs have equal service time~$\st$. In this scenario, $\LSTWtfunction$ is given by~\eqref{PKformula} with $\LSTstochsts{\complex}=\exp{-\complex\st}$, 
\obsolete{
\begin{equation}
\label{cstPKformula}
\textstyle%\begin{array}{c}
\LSTWts{\complex} 
=
 \frac{(1-\ar\st )\complex} {\complex- \ar (1-{\exp{-\complex\st}}) },
%\end{array}
\end{equation}
}%
%(\refeq{cstPKformula}) in the appendix, and~$\LSTWtfunction$ 
and rewrites as %for any~$m\in\Naturalpos$ as
\begin{equation}\label{trick}
\begin{array}{ll}
% \LSTWts{-\complex} = [\pwfactors{\complex}]^m \LSTWts{-\complex} +\frac{\complex(1-\load) }{\complex+ \ar} \sum\nolimits_{k=0}^{m-1}   [\pwfactors{\complex}]^k
 \LSTWts{\complex} = [\pwfactors{-\complex}]^m \LSTWts{\complex} +\frac{(1-\ar\st)\complex }{\complex- \ar} \sum\nolimits_{k=0}^{m-1}   [\pwfactors{-\complex}]^k
 , & \quad  \forall m\in\Naturalpos,
\end{array}
\end{equation}
where 
$\pwfactors{\complex}
\defeq 
%[{\ar}/({\complex+\ar})]\LSTstochsts{-\complex}
%=
[{\ar}/({\complex+\ar})]\,\exp{\complex\st}$.
%and $\load=\ar\st$.
%, with $\LSTstochsts{-\complex}=\exp{\complex\st}$.
%$\pwfactors{-\complex}=\frac{\ar}{-\complex+\ar}\LSTstochsts{\complex}$, with $\LSTstochsts{\complex}=\exp{-\complex\st}$.
%$\pwfactors{\complex}=\frac{\ar}{\complex+\ar}\exp{\complex\st}$.
It can be seen that~\eqref{generaltrick} holds % for every $\bl\in(0,\taufunction)$ 
if
\begin{equation}\label{generaltrickdiscrete}\notag
\textstyle
%\iLSTWtxs{\bl}{-\complex} = [\pwfactors{\complex}]^ {\trickx{\bl}}  \LSTWts{-\complex}
\iLSTWtxs{\bl}{\complex} = [\pwfactors{-\complex}]^ {\trickx{\bl}}  \LSTWts{\complex}
,\quad
%\fLSTWtxs{\bl}{-\complex} = \frac{\complex(1-\load) }{\complex+ \ar} \sum\nolimits_{k=0}^{\trickx{\bl}-1}   [\pwfactors{\complex}]^k
\fLSTWtxs{\bl}{\complex} = \frac{\complex(1-\load) }{\complex-\ar} \sum\nolimits_{k=0}^{\trickx{\bl}-1}   [\pwfactors{-\complex}]^k
,
%\end{array}
\end{equation}
with
%where 
%, and the function~$\trickifunction{l}$  reduces to 
$
%\begin{equation}\label{cststtrickix}
%\begin{array}{ll}
 \trickx{\bl}=\Ceil{({\taufunction-\bl})/{\st}}
% ,&\forall\bl\in[0,\taui{l}],\ (l=1,\dots,\pwn-1).
%\end{array}
%\end{equation}
$.

\paragraph{Degenerate cases.}\label{section:degeneratecase}
The decomposition scheme~(\ref{trick}) is not possible for all discrete service time distributions.
%There exist service time distributions with particular transform pole locations for which~(\refeq{jordanpiecewisetrick}) does not hold. 
Consider for instance the geometric service time distribution  
%\begin{equation}\label{gometricst}\begin{array}{ll}
$
\distriXx{\stochst}{\bl}= (\exp{\stgrate}-1) \sum\nolimits_{k=1}^{\infty} \exp{- k\stgrate} \stepx{\bl-k\st}
$ % ,
%&  \quad  \forall
for $
\bl\in\Realpluszero
$, %,
%\end{array}\end{equation}
where $\st>0$ and $\ar<(1-\exp{-\stgrate})/\st$.
We have
$\expectation{ \stochst} =
% (\exp{\stgrate}-1)  (-\st) d/d\stgrate [\sum_{k=1}^{\infty}   \exp{-\stgrate k}] = (\exp{\stgrate}-1)  (-\st) d/d\stgrate [\frac{1}{\exp{\stgrate}-1}]  =   \st \frac{\exp{\stgrate}}{\exp{\stgrate}-1} =   \frac{\st}{1-\exp{-\stgrate}} =
\st/(1-\exp{-\stgrate})$, 
%$\exp{-\stgrate}<1-\ar\st$
%
$\LSTstochsts{\complex}
% =  (\exp{\stgrate}-1) \sum\nolimits_{k=1}^{\infty} \exp{- k\stgrate} \exp{-k\complex\st}=  (\exp{\stgrate}-1) \sum\nolimits_{k=1}^{\infty} \exp{-k(\stgrate +\complex\st)} =  \frac{\exp{\stgrate}-1}{\exp{\stgrate +\complex\st}-1}
= (\exp{\stgrate}-1)/(\exp{\stgrate +\complex\st}-1) $,
and~$\LSTWtfunction$ degenerates into
\begin{equation} \label{LSTWtsinp}
%\begin{array}{c}
\LSTWts{\complex}  
=
\frac{  \fraccst   (\exp{\stgrate +\complex\st}-1) \complex}{(\complex- \ar ) \exp{\stgrate+\complex\st} -\complex+\ar \exp{\stgrate} }
%[{  \fraccst   (\exp{\stgrate +\complex\st}-1) \complex}]/[{(\complex- \ar ) \exp{\stgrate+\complex\st} -\complex+\ar \exp{\stgrate} }]
=
 \fraccst  -\complex(\exp{\stgrate +\complex\st}-1)   \, \degfunctions{\complex} 
,
%\end{array}
\end{equation} 
where $\fraccst= ({1-\exp{-\stgrate}-\ar\st })/({1-\exp{-\stgrate}})$ and $\degfunctions{\complex} = [(\ar-\complex ) \exp{\stgrate+\complex\st} +\complex-\ar \exp{\stgrate} ]^{-1} $. Although~$\degfunctions{\complex}$ decreases like $\magnitude{\radius^{-1}}$ as $\modulus{\complex}\to\infty$ (i.e., not fast enough for counterclockwise integration along~$\contoura{\radius}$), it decomposes as follows:
 \begin{equation}\label{degfunctionstwo}
 \textstyle
%\begin{array}{ll}
\degfunctions{\complex} 
=
% [\degpwfactors{-\complex}]^{m } \degfunctions{\complex}   +  \frac{1}{\ar \exp{\stgrate}-\complex} \sum\nolimits_{k=1}^{m} [\degpwfactors{-\complex}]^{k},
 [\degpwfactors{-\complex}]^{m } \degfunctions{\complex}   +   \sum\nolimits_{k=1}^{m} [\degpwfactors{-\complex}]^{k} /({\ar \exp{\stgrate}-\complex}),
\quad%& 
 (m=1,2,\dots),
%\end{array}
\end{equation}
where %$\fraccst= ({1-\exp{-\stgrate}-\ar\st })/({1-\exp{-\stgrate}})$ and 
$%\begin{equation}\label{pwfactorsdeg}
\textstyle
\degpwfactors{\complex}= 
% \frac{ \complex+\ar \exp{\stgrate}}{\exp{\stgrate}(\complex+ \ar ) } \exp{\complex\st} 
 ({ \complex+\ar \exp{\stgrate}})/{(\complex+ \ar ) } \, \exp{\complex\st-\stgrate} 
% \frac{ \complex+\ar \exp{\stgrate}}{\complex+ \ar  } \exp{\complex\st-\stgrate} 
$ %\end{equation}
is~$\magnitude{ \exp{\realpart{\complex}\st} }$ with just one pole at~$-\ar$.
%
%In view of~(\refeq{pwfactorsdeg}), we ditsribute~(\refeq{LSTWtsinp}), and  use~(\refeq{degfunctionstwo}) twice with parameters~$m+1$ and~$m$. We find, after computations,
By distributing~\eqref{LSTWtsinp} and  using~\eqref{degfunctionstwo} twice with parameters~$m+1$ and~$m$, we find, after computations,
that~\eqref{generaltrick} holds if we set
\obsolete{
\begin{equation}\label{LSTWtsinptwo}
\begin{array}{rcl}
\LSTWts{\complex}  
\obsolete{
&=&
\fraccst  \exp{\stgrate +\complex\st}  -\complex\{ [\degpwfactors{-\complex}]^{m+1 } \degfunctions{-\complex}   +  \frac{1}{-\complex+\ar \exp{\stgrate}} \sum\nolimits_{k=1}^{m+1} [\degpwfactors{-\complex}]^{k} \}
\\&&\hspace{30mm}
 -  \fraccst  -\complex \{ [\degpwfactors{-\complex}]^{m } \degfunctions{-\complex}   +  \frac{1}{-\complex+\ar \exp{\stgrate}} \sum\nolimits_{k=1}^{m} [\degpwfactors{-\complex}]^{k} \}
}
%\\
%=
 %\fraccst  ( \exp{\stgrate +\complex\st} \degpwfactors{-\complex} -  1 ) [\degpwfactors{-\complex}]^{m } \degfunctions{-\complex} + \frac{\fraccst}{-\complex+\ar \exp{\stgrate}} \degpwfactors{-\complex} + \frac{\fraccst}{-\complex+\ar \exp{\stgrate}} \sum\nolimits_{k=1}^{m} ( \exp{\stgrate +\complex\st} \degpwfactors{-\complex} - 1) [\degpwfactors{-\complex}]^{k}
 %\\
%=
 %\fraccst  \left\{ \frac{  \exp{-\complex\st-\stgrate}}{-\complex+ \ar  } + \frac{ \ar (\exp{\stgrate}-1 ) }{(-\complex+ \ar ) (-\complex+\ar \exp{\stgrate})} \sum\nolimits_{k=1}^{m}    [\degpwfactors{-\complex}]^{k} + \ar \frac{ \exp{\stgrate}-1  }{-\complex+ \ar  }   [\degpwfactors{-\complex}]^{m } \degfunctions{-\complex}  \right\}
%\\
&=&
 \frac{  \complex - \fraccst }{\complex-\ar} 
 +  \ar  \fraccst \frac{ (\exp{\stgrate}-1 )  -\complex }{(\complex-\ar) (\complex-\ar \exp{\stgrate})} \sum\nolimits_{k=1}^{m}    \big(\frac{ \complex-\ar \exp{\stgrate}}{\complex-\ar} \big)^{k} \exp{-k(\complex\st+\stgrate)} 
 \\&&\hspace{20mm}
+ \ar  \fraccst  \frac{( \exp{\stgrate}-1  )( \ar \exp{\stgrate}-\complex)^m}{(\ar-\complex)^{m+1}  }  \big(\frac{  \complex}{ (\complex- \ar )   \exp{\stgrate+\complex\st} -\complex+\ar \exp{\stgrate} } \big)\exp{-m(\complex\st+\stgrate)}
,
\end{array}
\end{equation} 
where the first two terms contain a finite number of poles, while the third term has only one candidate pole for counterclockwise integration: $\ar$, with degree $m+1$. 
%The value of~$m$ can therefore be chosen as a function of the backlog~$\bl$ so that integration of the third term along~$\contoura{\radius}$ vanishes for~$\Blaplacetransformfs{\extdpurevaluefunction}{\complex} \exp{\complex\bl}$. In the case of the  cost function $\costx{\bl}=\stepx{\bl-\taufunction}$, it suffices to set $m=\trickix{1}{\bl}=\ceil{({\taufunction-\bl})/{\st}}$ in~(\refeq{LSTWtsinptwo}) and to proceed as in Example~\ref{example:piecewisecst}.
It follows that~\eqref{generaltrick} holds  for $\bl\in(0,\taufunction)$ if we set
}% 
\begin{equation}\label{generaltrickdegenerate}\notag
\begin{array}{rcl}
\iLSTWtxs{\bl}{\complex} 
&=&
%\ar  \fraccst  \frac{( \exp{\stgrate}-1  )( -\complex+\ar \exp{\stgrate})^{\trickix{l}{\bl}-1}}{(-\complex+ \ar)^{\trickix{l}{\bl}+1}  }    \left(\frac{ -\complex+\ar \exp{\stgrate}  }{ (-\complex+ \ar ) \exp{\stgrate+\complex\st} +\complex-\ar \exp{\stgrate} }  + 1 \right)\exp{ \trickix{l}{\bl}(-\complex\st-\stgrate)}
%\ar  \fraccst  \frac{( \exp{\stgrate}-1  )( \ar \exp{\stgrate}-\complex)^{\trickix{l}{\bl}-1}}{(\ar-\complex)^{\trickx{\bl}}  }    \big(\frac{\complex}{ \complex- \ar -(\complex-\ar \exp{\stgrate}) \exp{-(\complex\st+\stgrate)}  }   \big)\exp{-\trickx{\bl}(\complex\st+\stgrate)}
\ar  \fraccst  \frac{\complex( \exp{\stgrate}-1  )( \ar \exp{\stgrate}-\complex)^{\trickix{l}{\bl}-1}}{(\ar-\complex)^{\trickx{\bl}}\, [ \complex- \ar -(\complex-\ar \exp{\stgrate}) \exp{-(\complex\st+\stgrate)} ]  }    \, \exp{-\trickx{\bl}(\complex\st+\stgrate)}
, 
\\ 
\fLSTWtxs{\bl}{\complex} 
&=&
\frac{\complex-\fraccst}{\complex-\ar} 
 +  \ar  \fraccst \frac{ (\exp{\stgrate}-1 )  -\complex}{(\complex-\ar) (\complex-\ar\exp{\stgrate})} \sum\nolimits_{k=1}^{\trickx{\bl}-1}    \big(\frac{\complex-\ar \exp{\stgrate}}{\complex-\ar} \big)^{k} \exp{-k(\complex\st+\stgrate)} 
 ,
\end{array}
\end{equation}
with 
%, and the function~$\trickifunction{l}$  reduces to 
$
%\begin{equation}\label{cststtrickix}
%\begin{array}{ll}
 \trickx{\bl}=\ceil{({\taufunction-\bl})/{\st}}
% ,&\forall\bl\in[0,\taui{l}],\ (l=1,\dots,\pwn-1).
%\end{array}
%\end{equation}
$.

  See Example~\ref{example:piecewisecst} in the appendix for a step-by-step derivation of the \acl{wfunction} in the case of jobs with identical service times.

%%%%%%%%%%%%%%%%%

\section{Value function approximations}
\label{section:approximatevaluefunctions}
 
In the absence of exact expressions for the value functions, the \ac{FPI} step can still be carried out based on value function bounds. Suppose that lower and upper bounds, $\lcostfunction$ and~$\ucostfunction$, are available for~$\costfunction$ with explicitly computable \aclp{wfunction}, denoted by $\lpurevaluefunction$ and~$\upurevaluefunction$, respectively. % (cf. Section~\ref{section:closedformVF}). 
%compute lower and upper bounds on their respective mean costs,  and infer bounds on the value function for any backlog value.
%Then,  bounds  on~$\purevaluefunction$ can then be inferred from  the developments of Section~\ref{section:piecewisedefined}.
%
Using the interval arithmetic notation\footnotemark{}, we write
 $\costfunction\in\intervx{\costfunction}\equiv[\lcostfunction,\ucostfunction]$ and, %in view of Corollary~\ref{vfcorollary}, 
 by linearity  of the map $\costfunction\mapsto\purevaluefunction$,  %(Corollary~\ref{vfcorollary}),
 we find in $\intervx{\purevaluefunction}\equiv[\lpurevaluefunction,\upurevaluefunction]$ a bounding interval for the \acl{wfunction}, while~\eqref{genericmeancosttwo} provides the bounds $\intervx{\meancost}\equiv[\lmeancost,\umeancost]$  for the mean cost~$\meancost$.
 
In the $\nbservers$-server system of Figure~\ref{figure:dispatcher} with arrival rates $\ari{1},\dots,\ari{\nbservers}$ and cost functions %~$\costifunction{1},\dots, \costifunction{\nbservers}$  
bounded by~${\intervx{\costifunction{1}}},\dots, {\intervx{\costifunction{\nbservers}}}$, the admission cost~\eqref{admissioncost} inherits the bounds
%we can compute interval bounds for the admission cost~\eqref{admissioncost},
\begin{equation}\label{admissioncostbounds}\tag{[AC]}
%\textstyle%\begin{array}{ll}
\obsolete{
\admissionfoperator{\ar}{\intervx{\costfunction}}\vect{\bl,\st}
=
\intervx{\purevaluefunction}\vect{\bl+\st}
-
\intervx{\purevaluefunction}\vect{\bl}
-\left(\frac{\ar\st}{1-\load} \right)
\intervx{\meancost}
,
& \quad
\forall \vect{\bl,\st} \in \Realpluszero \times \Realpluszero,
}
\intervx{\admissionicoperator{\server}{\intervx{\costifunction{\server}}}}\vect{\bl,\st}
=  
\intervx{\purevalueifunction{\server}}\vect{\bli{\server}+\sti{\server}}-\intervx{\purevalueifunction{\server}}\vect{\bli{\server}} -\left(\frac{\ari{\server}\intervx{\meancosti{\server}}}{1-\loadi{\server}}\right) \sti{\server} 
,
%
%\end{array}
\end{equation}
where ${\intervx{\purevalueifunction{1}}},\dots, {\intervx{\purevalueifunction{\nbservers}}}$ and ${\intervx{\meancosti{1}}},\dots, {\intervx{\meancosti{\nbservers}}}$ are the corresponding interval bounds for the \acl{wfunction} and mean costs.
The \ac{FPI} decision at state~$\vect{\bl,\st}$  can be made in favor of a server $i\in\{1,\dots,\nbservers\}$ iff
\begin{equation}\label{decision}\tag{D}
\begin{array}{ll}
%\admissionfoperator{\ari{i}}{\intervx{\costifunction{i}}}\vect{\bli{i},\st}\leq\admissionfoperator{\ari{j}}{\intervx{\costifunction{j}}}\vect{\bli{j},\st} 
\intervx{\admissionicoperator{\server}{\intervx{\costifunction{\server}}}}\vect{\bl,\st}
\leq
\intervx{\admissionicoperator{\altserver}{\intervx{\costifunction{\altserver}}}}\vect{\bl,\st}
,& \quad
\forall 
%j\in\{1,\dots,i-1,i+1,\dots,\nbservers\}
\altserver\neq\server
.
\end{array}\end{equation}
If otherwise no server satisfies~\eqref{decision}, the precision of the interval bounds for the cost functions   must be improved  until a decision can be made.
\footnotetext{%
% Notation~\ref{notation:intervalarithmetic}
\begin{inlinenotationwithtitle}{Interval arithmetic}{notation:intervalarithmetic}  
We use~$\intervx{x}\equiv[x_1,x_2]$ to represent an interval on~$\Real$. We write~$\intervx{x}\in\intervx{\Real}$ where $\intervx{\Real}=\{[x_1,x_2]\setst x_1\leq x_2;\, x_1,x_2\in\Real\}$, $a\in\intervx{x} $ iff $x_1\leq a \leq x_2$,  $\modulus{\intervx{x}}=x_2-x_1$, and  $-\intervx{x}
%=
%\{-a\setst a\in\intervx{x}\}
=
[-x_2,-x_1]
$. For~$\intervx{x},\intervx{y}\in\intervx{\Real}$ we have  $\intervx{x}+\intervx{y}
=
%\{a+b\setst a \in\intervx{x},\, b\in\intervx{y}\}
%=
[x_1+y_1,x_2+y_2]
$,  $\intervx{x}<\intervx{y}$ iff 
%$a<b$ for every $\vect{a,b}\in\intervx{x}\times\intervx{y}$, 
$x_2<y_1$,
and $\intervx{x}\leq\intervx{y}$, $\intervx{x}>\intervx{y}$ and $\intervx{x}\geq\intervx{y}$ are defined similarly.
\end{inlinenotationwithtitle}
}%

In the rest of this section we discuss various cost approximation schemes.
% : % for the cost functions: 
% Taylor series %in Section%s~\ref{section:continuouslydifferentiable} and
% (Section~\ref{section:entirecostfunctions}) %, trigonometric approximations of periodic functions (\ref{section:periodiccostfunctions}), 
% and uniform approximations by polynomials %in~
% (\ref{section:continuouscostfunctions}).

%From the perspective of extending the body of computable value functions, this part of the study considers 
%In Sections~\ref{section:continuouslydifferentiable} and~\ref{section:entirecostfunctions} we consider cost functions that can be expressed as Taylor series and seek to derive  value functions in the form of power series, thus elaborating on an approach suggested  in~\cite{hyytia-peva-2014}.

\subsection{Analytic cost functions and Taylor series} \label{section:entirecostfunctions}
Due to the availability of explicit value functions for the type $\costx{\bl}=\bl^nn$, Taylor/Maclaurin series have been cited as natural candidates for the approximation of analytic cost functions, \cite{hyytia-peva-2014}.  
Let~$\costfunction$ be  an infinitely smooth real function on~$\Realpluszero$ with $k$-th derivative~$\dncostfunction{k}$. For~$\nn\in\Natural$, consider an interval~$\intervx{\remaindercostnfunction{\nn}}$ such that $\costfunction \in \polycostnfunction{\nn} + \intervx{\remaindercostnfunction{\nn}} $, where
$
%\begin{equation}%\label{taylorpolynomialcost}
%\begin{array}{ll}
\polycostnx{\nn}{\bl} = \sum\nolimits_{k=0}^{\nn} %\frac{d^{k}\costx{0}}{d\bl^{k}} 
\, \dncostx{k}{0}
{\bl^k}/{\factorial{k}}
%,&\quad
%\forall\bl\in\Realpluszero,
%\end{array}
%\end{equation}
$
is the Taylor polynomial of order~$n$. % for~$\costfunction$. 
If~$\polypurevaluenfunction{\nn}$ denotes the \acl{wfunction} associated with~$\polycostnfunction{\nn}$, and~$\intervx{\complexboundpurevaluenfunction{\nn}}$ is a bounding interval covering the \aclp{wfunction} for all cost functions comprised in~$\intervx{\remaindercostnfunction{\nn}}$, then using Table~\ref{table:valuefunctions} we find 
\begin{equation} \label{polypurevaluenx}\notag
\textstyle
\polypurevaluenx{\nn}{\bl}
= 
\frac{\ar}{1-\load} \sum\nolimits_{k=0}^{\nn} \{ \sum\nolimits_{t=0}^{\nn-k} \specialmatrixnk{t+1}{t} \, \dncostx{k+t}{0} \} \frac{ \bl^{k+1}}{\factorial{(k+1)}} , 
\end{equation}
and, by linearity, $\purevaluefunction\in \intervx{\purevaluenfunction{\nn}}$, where $\intervx{\purevaluenfunction{\nn}} =
\polypurevaluenfunction{\nn}+ \intervx{\complexboundpurevaluenfunction{\nn}}$.

If~$\costfunction$ is analytic, then $\intervx{\remaindercostnfunction{\nn}}$ vanishes pointwise near $\bl=0$ as $\nn\to\infty$, and our hopes are that the remainder $ \intervx{\complexboundpurevaluenfunction{\nn}}$ will become small as well, with~$\intervx{\purevaluenfunction{\nn}}$ converging towards~$\purevaluefunction$ in some sense.
The next result, however, claims that a cost function~
$\costfunction$ given as a convergent Taylor series only yields a convergent sequence of \aclp{wfunction} if~$\costfunction$ is entire (i.e., its Taylor series converges everywhere)  with order of growth\footnotemark{} less than the exponential type~$\modulus{\dominantpoleX{\Wt}}$, 
\footnotetext{Recall that the \emph{order of growth}  of an entire function~$\genericfunction$~\cite{levin96}, defined by $\eorder = \limsup_{\radius\to\infty}\ln\ln{\normi{\genericfunction}{\infty,\radius}}/\ln \radius$, where $\normi{\genericfunction}{\infty,\radius} = \sup_{\complex}\{\modulus{\genericx{\complex}}\setst\modulus{\complex}<\radius\}$, is the infimum of all~$m$ such that $\genericx{\complex}=\magnitude{\exponential(\modulus{\complex}^m)}$, while the \emph{type} of~$\genericfunction$ is defined by $\etype = \limsup_{\radius\to\infty}\ln{\normi{\genericfunction}{\infty,\radius}}/ \radius^\eorder$. If $\eorder=1$, then~$\genericfunction$ is said to be of exponential type~$\etype$.}%
whereas any function~$\costfunction$ falling outside this restrictive category is expected to produce a divergent sequence for~$\purevaluefunction$. A proof of  Theorem~\ref{theorem:convergence} is given in  Appendix~\ref{appendix:complement}.
%
% 
%
% Theorem~\ref{theorem:convergence}
\begin{theorem}[Taylor series for~$\purevaluefunction$]\label{theorem:convergence}
Let~$\costfunction$
%in the server model of Section~\ref{section:valuefunction} 
satisfy Assumption~\ref{assumption:complexintegrablecost} and be %analytic.
%Let~$\dominantpoleX{\Wt}$ denote the dominant pole of~$\LSTWtfunction$. 
%Further assume that~$\costfunction$ is 
entire with order of growth~$\eorder$ and type~$\etype$, so that
%\begin{equation}
\begin{align}
 \label{entirecost}
%\begin{array}{ll}
%\textstyle
&
\costx{\bl} = \sum\nolimits_{\nn=0}^{\infty} 
%\frac{d^{\nn}\costx{\altbl}}{d\bl^{\nn}}
% \dncostx{\nn}{\altbl} \frac{(\bl-\altbl)^\nn}{\factorial{\nn}}
 [\dncostx{\nn}{\altbl}/{\factorial{\nn}}]\, {(\bl-\altbl)^\nn}
,  &%\quad
\forall \bl,\altbl \in\Realpluszero.
%\end{equation}
\end{align}
For $k\in \Natural$, let $\valuederivativek{k}=\lim_{\nn\to\infty} \valuederivativenk{\nn}{k}$, where
\begin{align}
 %\begin{equation}
\label{valuederivativek}
%\textstyle
&
\valuederivativenk{\nn}{k}
=[{\ar}/({1-\load})]\, 
 \sum\nolimits_{q=0}^{\nn-k} \specialmatrixnk{q+1}{q} \, \dncostx{k+q}{0}
 ,&%\quad
 (\nn\in\Natural),
\obsolete{
\valuederivativek{\nn}=\frac{\ar}{1-\load} \sum\nolimits_{q=0}^{\infty} \specialmatrixnk{}{q}
%\frac{d^{\nn+q}\costx{0}}{d\bl^{\nn+q}}
\, \dncostx{\nn+q}{0}
,\quad (\nn\in\Natural),
}
\end{align}
%\end{equation}
   and define the functions~$\sumsumfunction$ and~$\sumfunction$   as
   \begin{subequations}
\begin{align}
 & 
\label{sumsumx}
%\begin{array}{rcl}
\sumsumx{\bl}
 =%&= &
%
% \int\nolimits_{0}^{\bl} \sum\nolimits_{k=0}^{\infty} \valuederivativek{k}  \frac{ \xi^{k}}{\factorial{k}} \, d\xi ,
 \int\nolimits_{0}^{\bl} \sum\nolimits_{k=0}^{\infty} \frac{\valuederivativek{k}}{\factorial{k}}\, { \xi^{k}} \, d\xi ,
%\end{array}
&\\
\label{sumx}&
%\begin{array}{rcl}
\sumx{\bl}
 =%&=& 
%=
\frac{\ar}{1-\load}   \int\nolimits_{0}^{\bl}  \sum\nolimits_{q=0}^{\infty}\specialmatrixnk{}{q} 
\, \dncostx{q}{\xi}
\, d\xi  .
%\end{array}
%\hspace{22mm}
&
\end{align}
   \end{subequations}
\begin{enumerate}[(i)]
\item \label{convergence:yes}
If either~$\eorder<1$ or~$\eorder=1$ and~$\etype<\modulus{\dominantpoleX{\Wt}}$, %^{-1}$,
 then the coefficients~$\valuederivativek{\nn}$ are finite for all~$\nn$, \eqref{sumsumx} and~\eqref{sumx} converge on~$\Realpluszero$, and~$\sumsumfunction=\sumfunction=\purevaluefunction$.
\item \label{convergence:no}
If either~$\eorder>1$ or~$\eorder=1$ and~$\etype>\modulus{\dominantpoleX{\Wt}}$, %^{-1}$,
 then~\eqref{valuederivativek}  diverges for all~$k$. % as $\nn\to\infty$.
 %~(\refeq{sumsumx})  diverges for~$\bl\in\Realpluszero$.
%
\end{enumerate} 
\end{theorem}
Equation~\eqref{sumsumx} %in Theorem~\ref{theorem:convergence}
 is the Taylor series (in convergence conditions) of~$\purevaluefunction$ at $\bl=0$. The coefficients of the series are given by~$\{\valuederivativek{k}\}$,  the sequence of the successive derivatives of~$\dpurevaluex{\bl}$ at~$0$, obtained by cross-correlation between the sequence~$\{\dncostx{k}{0}\}$ of the  derivatives of~$\costfunction$ at $\bl=0$ and~$\{\specialmatrixnk{}{k}\}$, the germ of~$\LSTWtfunction$ at the origin, given in~\eqref{table:specialmatrixnk}. Equation~\eqref{sumsumx} may be understood as an extension of~\eqref{valuederivativek} to  $\bl\geq 0$, in the sense that $\dpurevaluex{\bl}$ is computed directly by cross-correlation of the cost derivatives at~$\bl$ with the the germ of~$\LSTWtfunction$ at~$0$.

%
%\begin{figure}
%\includegraphics[scale=1]{analyticvaluefuctionsfigexp250.pdf}
%\end{figure}
%
\begin{figure}
 \centering
%~
\hspace{-7mm}
\subfigure[$\sing=0.5\, (\strate-\ar)$]{
        \centering
%        \makeVFboundsplot{\setpathfigs%
%ts-oneminusexpmax-exp10-250-25}{50}
\includegraphics[scale=1.0]{\setpathfigs%
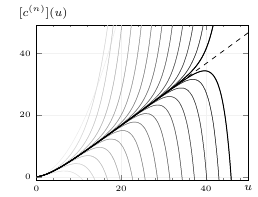}%analyticvaluefuctionsfigexp250.pdf} 
        } 
%\\ 
\hspace{-7mm}
%~
    \subfigure[$\sing=0.8\, (\strate-\ar)$]{
        \centering
%        \makeVFboundsplot{\setpathfigs%
%ts-oneminusexpmax-exp10-400-25}{50}
\includegraphics[scale=1.0]{\setpathfigs%
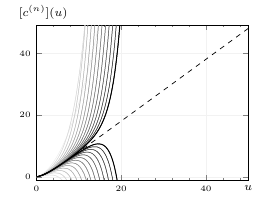}%analyticvaluefuctionsfigexp400.pdf} 
        } 
\hspace{-7mm}
%~
    \subfigure[ $\sing=1.0\, (\strate-\ar)$]{
        \centering
%        \makeVFboundsplot{\setpathfigs%
%ts-oneminusexpmax-exp10-500-25}{50}
\includegraphics[scale=1.0]{\setpathfigs%
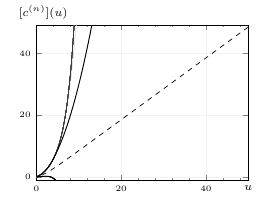}%analyticvaluefuctionsfigexp500.pdf}
        }
\hspace{-3mm}
\caption{(Example~\ref{example:exponentialtwo}) Taylor series for the \ac{wfunction} with exponential service times 
%($\distriXx{\stochst}{x}= 1-\exp{-\strate x}$ with~$\strate=2\ar$) 
($\stochst\sim\exponentialdistri{\strate}$, with~$\strate=2\ar$), 
and cost function~$\costx{\bl}=1-\exp{-\sing \bl}$ of exponential type $\modulus{\sing}$ : $\purevaluefunction$ %\Mainsolution{} 
(dashed line) %with  constant service time~$\st=1$ 
and~$
\intervx{\purevaluenfunction{\nn}}
%{\purevaluenfunction{\nn}}
$ for $\nn=1,\dots,25$.
%
%$\eorder<1$ or~$\eorder=1$ and~$\etype<\modulus{\dominantpoleX{\Wt}}$
The series only converges if $\modulus{\sing}<\modulus{\dominantpoleX{\Wt}}=\strate-\ar$.
\label{figure:expexample}}
\end{figure}

The message of Theorem~\ref{theorem:convergence} is illustrated by Figure~\ref{figure:expexample}, which exposes through an elementary problem the hazards of processing cost functions as Taylor series. See Example~\ref{example:exponentialtwo} in the appendix for computational details.
%

%\begin{remark}[System interpretation]\label{remark:systeminterpretation}
\paragraph{System interpretation.}
%In convergence conditions, the sequence~$\{\valuederivativek{k}\}$ gives the successive derivatives of~$\dpurevaluex{\bl}$ at~$0$. 
%
For $\nn\in\Natural$, let $\costderivativevectorn{\nn}=\vect{\costderivativevectornk{\nn}{0},%\costderivativevectornk{\nn}{1},
\dots, \costderivativevectornk{\nn}{\nn-1}   }$, with~$\costderivativevectornk{\nn}{k} ={\dncostx{k}{0}}$, %($k=0,1,\dots,\nn-1$) 
and  
$\valuederivativevectorn{\nn}=\vect{\valuederivativenk{\nn}{0},%\valuederivativenk{\nn}{1},
\dots,\valuederivativenk{\nn}{\nn-1}}$,
where~$\valuederivativenk{\nn}{k}$ is defined by~\eqref{valuederivativek}. %, then~$\costderivativevectorn{\infty}$ and~$\valuederivativevectorn{\infty}$ are germs at~$0$ of% the cost function and of the \mainsolution{}, $\costfunction$ and~$\dpurevaluefunction$,  respectively. 
%, whereas these functions rewrite, for $\bl\geq 0$, as~$\costx{\bl}=\sum\nolimits_{k=0}^{\infty} \costderivativevectornk{\nn}{k}\frac{\bl^k}{\factorial{k}}$ and~$\purevaluex{\bl}=\sum\nolimits_{k=0}^{\infty} \valuederivativevectornk{\nn}{k}\frac{\bl^{k+1}}{\factorial{(k+1)}}$.
%
%Now, consider the Toeplitz, upper triangular matrices
Using the matrix inversion lemma, one shows that the Toeplitz, upper triangular matrices
\begin{equation}\label{specialmatricesn}\notag \textstyle
\uppertriangularn{\nn}=
\scalebox{1.0}{$
\left(
\begin{array}{ccccc}
\powerzeroXcoefnk{1} & \powerzeroXcoefnk{2} &\powerzeroXcoefnk{3}&\cdots & \powerzeroXcoefnk{\nn}%&  \powerzeroXcoefnk{\nn}
\\ 
0 & \powerzeroXcoefnk{1} &\powerzeroXcoefnk{2} & \hdots & \powerzeroXcoefnk{\nn-1}%& \powerzeroXcoefnk{\nn-1}
\\
0 & 0& \powerzeroXcoefnk{1} &\hdots&\powerzeroXcoefnk{\nn-2}%& \powerzeroXcoefnk{\nn-2}
\\
\vdots&\vdots&\vdots&\ddots&\vdots%&
\\
0 & 0& 0&\hdots&\powerzeroXcoefnk{1}%&\powerzeroXcoefnk{1} 
%\\
%0 & 0& 0&\cdots&0&\powerzeroXcoefnk{0}
\end{array}
\right)
$}
,
\qquad
\specialmatrixn{\nn}=\frac{\ar}{1-\load}
\scalebox{1.0}{$
\left(
\begin{array}{ccccc}
\specialmatrixnk{\nn+1}{0} & \specialmatrixnk{\nn+1}{1} &\specialmatrixnk{\nn+1}{2}&\cdots & \specialmatrixnk{\nn+1}{\nn-1}%&  \specialmatrixnk{\nn+1}{\nn}
\\
0 & \specialmatrixnk{\nn+1}{0} &\specialmatrixnk{\nn+1}{1} & \hdots & \specialmatrixnk{\nn+1}{\nn-2}%& \specialmatrixnk{\nn+1}{\nn-1}
\\
0 & 0& \specialmatrixnk{\nn+1}{0} &\hdots&\specialmatrixnk{\nn+1}{\nn-3}%& \specialmatrixnk{\nn+1}{\nn-2}
\\
\vdots&\vdots&\vdots&\ddots&\vdots%&
\\
0 & 0& 0&\hdots&\specialmatrixnk{\nn+1}{0}%&\specialmatrixnk{\nn+1}{1} 
%\\
%0 & 0& 0&\cdots&0&\specialmatrixnk{\nn+1}{0}
\end{array}
\right)
$}
.
\end{equation}
%where $\powerzeroXcoefnk{k}=\expectation{\stochst^{k}}/{\factorial{k}}$, and~$\{\powerzeroXcoefnk{k}\}$ and~$\{\specialmatrixnk{}{k}\}$ form the respective germs of~$\LSTstochstfunction$ and~$\LSTWts{-\complex}$ (cf. Proposition~\ref{proposition:analycity}\eqref{analycity:centered}).
%Using~\eqref{table:specialmatrixnk} and the matrix inversion lemma, it is straightforward to show that $%\specialmatrixn{\nn}=({\identityn{\nn}}/{\ar}-\uppertriangularn{\nn})^{-1}
satisfy $
\specialmatrixn{\nn}({\identityn{\nn}}/{\ar}-\uppertriangularn{\nn}) = \identityn{\nn}
$ for all~$\nn$, 
where~$\identityn{\nn}$ denotes the identity matrix. % of dimension~$\nn$.
Besides, \eqref{valuederivativek} rewrites in matrix form as
%\begin{equation}\label{filter}
%\begin{array}{ll}
$
\valuederivativevectorn{\nn}=\specialmatrixn{\nn} 
\costderivativevectorn{\nn}
%& \quad (\nn=0,1,2,\dots),
%\end{array}
%\end{equation}
$. 
As $\nn\to\infty$, $\{\valuederivativevectornk{\nn}{t}\}$ %_{t\in\Natural}$ can be regarded in~(\refeq{valuederivativek}) as 
becomes the output (at nonnegative times) of the cross-correlation of~$\{\costderivativevectornk{\nn}{t}\}$  with the sequence defined by $\impulsex{t}= {\ar}/({1-\load})\specialmatrixnk{}{t} $  for $t\geq 0$, %where~$\{\specialmatrixnk{}{t}\}$  is the germ of~$\LSTWts{-\complex}$ at the origin, 
thus giving us  an interpretation for analytic functions 
\showhideproposition{%
of   Proposition~\ref{proposition:identities}\eqref{proposition:vfdifferentiability}, where~$\dpurevaluefunction$ was obtained by cross-correlation of~$\costx{\bl}$ with $\ar/(1-\load) \ddistriXx{\Wt}{\bl}$. % (cf. system interpretation in Section~\ref{section:deriving}). 
}{%
of~\cite[Proposition~1]{hyytia-vgf-itc-2017}, where~$\dpurevaluefunction$ is obtained by cross-correlation of~$\costx{\bl}$ with $\ar/(1-\load) \ddistriXx{\Wt}{\bl}$. 
}%
%
%and $\impulsex{t}=0$ if~$t>0$.
%convolution of the sequence~$\{\costvectornk{\nn}{t}\}_{t\in\Natural}$  with the anticausal digital filter  characterized by the impulse response $\impulsefunction : t\in\Integer\mapsto\Real$ such that $\impulsex{t}={\ar}/({1-\load})\specialmatrixnk{}{-t} $ if $t\leq 0$ and $\impulsex{t}=0$ if~$t>0$.
\obsolete{
In the time-domain, we get, for $\nn\in\Natural$,
\begin{align}
\label{forwardconvolution}\begin{array}{ll}
\valuederivativek{k} = \frac{\ar}{1-\load} \sum\nolimits_{t=k}^{\nn-1} \specialmatrixnk{t+1}{t-k} \costderivativevectornk{\nn}{t},
\end{array}
& k=0,1,\dots,\nn-1,
\\
\label{backwardconvolution}\begin{array}{ll}
\costderivativevectornk{\nn}{k} =  \frac{\valuederivativek{k}}{\ar} - \sum\nolimits_{t=k}^{\nn-1}  \frac{\expectation{\stochst^{t+1-k}}}{\factorial{(t+1-k)}} \valuederivativek{t},
\end{array}
& k=0,1,\dots,\nn-1,
\end{align}
}%
Similarly, (\ref{sumx}) expresses~$\dpurevaluex{\bl}$ as the cross-correlation of the sequence of derivatives of~$\costfunction$ at~$\bl$ with the sequence $\impulsex{t}$.

From our observation follows that the Z-transforms of the sequences satisfy
\begin{equation}\label{convolution}
%\begin{array}{c}
\Ztransformfz{\valuederivativevectorn{\infty}}{\zcomplex}
=
\Ztransformfz{\impulsefunction}{{1}/{\zcomplex}}\,
\Ztransformfz{\costderivativevectorn{\infty}}{\zcomplex}
,
%\end{array}
\end{equation}
where 
$
\Ztransformfz{\genericsequencefunction}{\zcomplex}
=
\sum\nolimits_{k=0}^{\infty} \genericsequence{k} \zcomplex^{-k}$ denotes the Z-transform of a sequence~$\genericsequence{t}$.
\obsolete{
$\Ztransformfz{\costderivativevectorn{\infty}}{\zcomplex}
=
\sum\nolimits_{k=0}^{\infty} \costderivativevectornk{\infty}{k} \zcomplex^{-k}$,
$\Ztransformfz{\impulsefunction}{\zcomplex}
=
{\ar}/({1-\load})\sum\nolimits_{k=0}^{\infty} \specialmatrixnk{}{k} \zcomplex^{-k}$, 
$\Ztransformfz{\valuederivativevectorn{\infty}}{\zcomplex} 
=
\sum\nolimits_{k=0}^{\infty} \valuederivativevectornk{\infty}{k} \zcomplex^{-k}$.
}%
The vector~$\valuederivativevectorn{\infty}$ can be recovered\footnotemark{}  from~\eqref{convolution} by inverse Z-transform provided that the regions of convergence of~$\Ztransformfz{\costderivativevectorn{\infty}}{\zcomplex}$ and~$\Ztransformfz{\impulsefunction}{{1}/{\zcomplex}}$ intersect on a non-empty circular band|%when~$\infty\to\infty$, 
this condition is to be linked to those of Theorem~\ref{theorem:convergence}\ref{convergence:yes}.
\obsolete{
Without known expression for~$\Ztransformf{\impulsefunction}$, the coefficients~$\{\valuederivativek{\nn}\}$ can be derived directly from~\eqref{valuederivativek}.
}%
\obsolete{
Conversely, $
\Ztransformfz{\costderivativevectorn{\infty}}{\zcomplex}
=
(\Ztransformfz{\impulsefunction}{{1}/{\zcomplex}})^{-1}
\Ztransformfz{\valuederivativevectorn{\infty}}{\zcomplex}
$, where $(\Ztransformfz{\impulsefunction}{\zcomplex})^{-1}$
is the Z-transform of $\dirack{t}/\ar - \expectation{\stochst^{t+1}}/\factorial{(t+1)} $. 
}

\footnotetext{%
Conversely, inverting~$\specialmatrixn{\nn}$ yields $\costderivativevectorn{\nn}=({\identityn{\nn}}/{\ar}-\uppertriangularn{\nn})
\valuederivativevectorn{\nn}$ and, as $\nn\to\infty$, 
\begin{equation}
 \label{costderivativek}
 \textstyle
\costderivativevectornk{\infty}{k} =   {\valuederivativek{k}}/{\ar} - \sum\nolimits_{t=0}^\infty%{\nn-k-1}  
[{\valuederivativek{t+k}}/{\factorial{(t+1)}}]\, \expectation{\stochst^{t+1}}
,
\end{equation}
which provides us with a converse for Theorem~\ref{theorem:convergence}, where the source cost function of a given  \acl{wfunction} with germ~$\{\valuederivativevectornk{\nn}{k}\}$ can be recovered from~$\purevaluefunction$ through~\eqref{costderivativek},
on the condition that~$\purevaluefunction$ grows slower than the exponential type~$\modulus{\dominantpoleX{\stochst}}$|where~$\dominantpoleX{\stochst}$  denotes the dominant pole of~$\stochst$|, in which case $({\identityn{\nn}}/{\ar}-\uppertriangularn{\nn})\,\valuederivativevectorn{\nn}$  converges as~$\nn\to\infty$.
%
%that~$({\identityn{\nn}}/{\ar}-\uppertriangularn{\nn})\,\valuederivativevectorn{\nn}$  converges as~$\nn\to\infty$.
%Similarly, this condition is that~$\purevaluefunction$ should grow slower than the exponential type~$\modulus{\dominantpoleX{\stochst}}$,
%either $\eorder<1$, or~$\eorder=1$ and~$\etype<\modulus{\dominantpoleX{\stochst}}$, where~$\dominantpoleX{\stochst}$  denotes the dominant pole of~$\stochst$.
%^{-1}$,
Similarly, $
\Ztransformfz{\costderivativevectorn{\infty}}{\zcomplex}
=
(\Ztransformfz{\impulsefunction}{{1}/{\zcomplex}})^{-1}
\Ztransformfz{\valuederivativevectorn{\infty}}{\zcomplex}
$, where $(\Ztransformfz{\impulsefunction}{\zcomplex})^{-1}$
is the Z-transform of $\dirack{t}/\ar - \expectation{\stochst^{t+1}}/\factorial{(t+1)} $. 
}%

The conclusions of Theorem~\ref{theorem:convergence} lead us to consider, in the rest of Section~\ref{section:approximatevaluefunctions}, approximations schemes no longer on~$\Realpluszero$, where the growth  of the functions~$\{\bl^\nn\}$ as $\bl\to\infty$ causes divergence, but on finite intervals where the series converge safely.

%%%%%%%%%%%%%%%%%
\obsolete{
\subsection{Periodic cost functions} 
\label{section:periodiccostfunctions}

First suppose that the cost function has a periodic component designed, for instance, in correlation with the hour of arrival. Let
$\costfunction= %\costx{0} \diract{\bl} +
 \periodcostfunction\, \indicatorfunction{(0,\infty)} $,
%$\costfunction= \costx{0} \diracfunction + \periodcostfunction \,  \indicatorfunction{(0,\infty)}$ 
where
$\periodcostfunction:\Real\mapsto\Real$ is periodic with period~\moveperiod{$\period$}{$2\period$} and continuous.
%We approximate the cost function by trigonometric sums, which are known by 
The {Weierstrass approximation theorem} (see, e.g., \cite[Weierstrass first theorem]{korovkin60}, \cite[Theorem~1.1]{rivlin69}, \cite[Theorem~2.7]{koralov07}) claims that the periodic, continuous function~$\periodcostfunction$  can be approximated by a trigonometric sum with arbitrary precision with respect to the uniform norm $\norm{\functionfunction}=\sup\nolimits_{\bl\in\moveperiod{[-\period/2,\period/2]}{[-\period,\period]}}\modulus{\fx{\bl}}
$.
This implies
that for any~$\epsilon>0$ one can find $\order<\infty$ and a trigonometric sum~$\tsnfunction{\order}$ such that $\errn{\order} = \norm{\periodcostx{\bl}-\tsnx{\order}{\bl}}<\epsilon$.
%,
\obsolete{
where the uniform norm of a function $\functionfunction:\Real\mapsto\Real$ with period~$\period$ is defined by
\begin{equation}\label{uniformnormperiodic}
\norm{\functionfunction}=\sup\nolimits_{\bl\in\moveperiod{[-\period/2,\period/2]}{[-\period,\period]}}{\modulus{\fx{\bl}}}.
\end{equation} 
}%
 It then holds that $\costfunction\in\intervx{\costfunction}$, where we set
\begin{align}\label{approximatedcostfunctionperiodic}
%\begin{array}{r}
&
\intervxy{\costfunction}{\bl} =
\tsnx{\order}{\bl}+ 
[-\errn{\order},\errn{\order}] 
,
&%\hspace{5mm}
\forall \bl\in\Realpluszero.
%\end{array}
\end{align}
Such a trigonometric sum %approximating~$\periodcostfunction$ 
is given the partial Fourier series
\begin{equation}\label{partialfourierseries}
%\textstyle
\tsnx{\order}{\bl}=\sum\nolimits_{k=-\order}^{\order} \fouriercoefk{k}\, \exp{\i\frac{ \moveperiod{2}{} k \pi \bl }{\period}} ,
\end{equation}
where the Fourier coefficients $ \fouriercoefk{k}$ satisfy
\begin{align}\label{fouriercoefficients} 
%\begin{array}{ll}
&
\fouriercoefk{k} = \moveperiod{ \frac{1}{\period} \int\nolimits_{-\frac{\period}{2}}^{\frac{\period}{2}} \periodcostx{\bl} \exp{-\i\frac{ 2 k \pi \bl }{\period}} }{ \frac{1}{2\period} \int\nolimits_{-{\period}}^{{\period}} \periodcostx{\bl} \exp{-\i\frac{  k \pi \bl }{\period}} } \, d\bl ,
&  (k\in\Integer).
%\end{array}
\end{align}
Since $\periodcostfunction$ is real, $\fouriercoefk{0}$ is  real and  $\fouriercoefk{-k}$ and~$\fouriercoefk{k}$ are complex conjugate for all~$k$, so that~$\costnfunction{\order}$ has~$2\order+1$ real parameters.
\obsolete{
The Fourier series~\eqref{partialfourierseries}   converges towards~$\periodcostfunction$ with rate 
$ \magnitude{\errn{\order}} = \log(\order)\,  \continuitymodulusSfx{\Real}{\periodcostfunction}{1/n}$~\cite[\S{}21]{jackson41}, where
the \emph{modulus of continuity} of a function~$f:\Real\to\Real$ on an interval~$I\subset\Real$ is defined by
%$\continuitymodulusSfx{S}{f}{\delta}=\sup\{\modulus{f(x_1)-f(x_2)}: \modulus{x_1-x_2}\leq\delta,\ x_1,x_2  \in S \}$
\begin{equation}\label{continuitymodulus}\continuitymodulusSfx{I}{f}{\delta}=\sup\limits_{\substack{ x_1,x_2  \in I,\\ \modulus{x_1-x_2}\leq\delta}} \modulus{f(x_1)-f(x_2)}.
\end{equation}
If~$\periodcostfunction$ satisfies the $\hoeldernexponent$-H\"oldern condition $\modulus{\periodcostx{\bli{1}}-\periodcostx{\bli{2}}}\leq \hoeldernconstant  \modulus{\bli{1}-\bli{2}}^\hoeldernexponent  \ \forall \bli{1},\bli{2} \in \Real$ ($0<\hoeldernexponent \leq 1$), then  $\continuitymodulusSfx{\Real}{\periodcostfunction}{\delta}\leq  \hoeldernconstant \delta^\hoeldernexponent$, and% the   trigonometric sum
~\eqref{partialfourierseries}  converges uniformly. %In particular, when $\hoeldernexponent=1$, $\periodcostfunction$ is Lipschitz continuous.
}%
The Fourier series~\eqref{partialfourierseries}   converges towards~$\periodcostfunction$ with rate 
$ \magnitude{\errn{\order}} = \log(\order)\,  \continuitymodulusSfx{\period}{\periodcostfunction}{\period/(n\pi)}$, where
$
%\begin{equation}\label{continuitymodulus}
\continuitymodulusSfx{\period}{\periodcostfunction}{\delta}
=
\sup \{ \modulus{\periodcostx{\bli{1}}-\periodcostx{\bli{2}}}
 :  \bli{1},\bli{2}  \in \Real, \modulus{\bli{1}-\bli{2}}\leq\delta\} 
%\sup\nolimits_{\bli{1},\bli{2}  \in \Real, \modulus{\bli{1}-\bli{2}}\leq\delta} \modulus{\periodcostx{\bli{1}}-\periodcostx{\bli{2}}}
%
%\sup\limits_{\substack{\bli{1},\bli{2}  \in \Real,\\ \modulus{\bli{1}-\bli{2}}\leq\delta}} \modulus{\periodcostx{\bli{1}}-\periodcostx{\bli{2}}}
%.
%\end{equation}
$
defines the \emph{modulus of continuity} of the periodic function, \cite[\S{}21]{jackson41}. % scaled back to a $2\pi$-period.
In particular, if for some $\hoeldernexponent\in(0,1]$ the function $\periodcostfunction$ satisfies the $\hoeldernexponent$-H\"oldern condition $\modulus{\periodcostx{\bli{1}}-\periodcostx{\bli{2}}}\leq \hoeldernconstant  \modulus{\bli{1}-\bli{2}}^\hoeldernexponent  
$ for all $\bli{1},\bli{2} \in \Real$,
%($0<\hoeldernexponent \leq 1$), 
then  
%$\continuitymodulusSfx{\period}{\periodcostfunction}{\delta}\leq  \hoeldernconstant (\delta\pi/\period)^\hoeldernexponent$, 
$\continuitymodulusSfx{\period}{\periodcostfunction}{\delta}\leq  \hoeldernconstant \delta^\hoeldernexponent$, 
and% the   trigonometric sum
~\eqref{partialfourierseries}  converges uniformly towards~$\periodcostfunction$. %In particular, when $\hoeldernexponent=1$, $\periodcostfunction$ is Lipschitz continuous.
%
%
%for uniform convergence of apprximations by trigonometric polynomials one needs the function to be periodic and continuous between periods (see convergence of fourier series KEVIN STEPHEN STOTTER CUDDY Corollary 2.15. )
%
%=> $\omega$ function defined across the border.
 
A faster convergence rate can be obtained by slightly modifying the Fourier coefficients in~\eqref{partialfourierseries}. For this, consider the trigonometric sum 
\begin{equation}
\label{improvedpartialfourierseries}
%\textstyle
\itsnx{\order}{\bl}=\sum\nolimits_{k=-\order}^{\order} \itscoefnk{\order}{\modulus{k}} \, \fouriercoefk{k} \, \exp{\i\frac{  \moveperiod{2}{} k \pi \bl }{\period}} ,
\end{equation}
where $\itscoefnk{\order}{0},\dots,\itscoefnk{\order}{\order}\in\Real$. The choice of parameters proposed in~\cite[\S{}3]{korovkin60},
\obsolete{
\begin{equation}\label{itscoefnk}\begin{array}{ll}
\itscoefnk{\order}{0} =1,
&\\
\itscoefnk{\order}{1} = \cosx{\frac{\pi}{\order+2}},
&\\
\itscoefnk{\order}{k} = \frac{\sum\nolimits_{q=0}^{\order-k}\Sinx{\frac{q+1}{\order+2}\pi} \Sinx{\frac{q+k+1}{\order+2}\pi} }{ \sum\nolimits_{q=0}^{\order} \Sinsquaredx{\frac{q+1}{\order+2}\pi}   } ,
&\quad (k=2,\dots,\order),
\end{array}
\end{equation}
}%
\begin{equation}\label{itscoefnk}\textstyle
\itscoefnk{\order}{0} {\,=\,} 1,
\
\itscoefnk{\order}{1} {\,=\,} \cosx{\frac{\pi}{\order+2}},
\
\itscoefnk{\order}{k} {\,=\,} \frac{\sum\nolimits_{q=0}^{\order-k}\Sinx{\frac{q+1}{\order+2}\pi} \Sinx{\frac{q+k+1}{\order+2}\pi} }{ \sum\nolimits_{q=0}^{\order} \Sinsquaredx{\frac{q+1}{\order+2}\pi}   } 
\textup{ for } k=2,\dots,\order,
\end{equation}
grants~\eqref{improvedpartialfourierseries} the convergence rate
\begin{equation}\label{jacksontrigonometric}
\begin{array}{l}
\errn{\order} \leq 6\,\ContinuitymodulusSfx{\period}{\periodcostfunction}{\frac{\tau}{\order\pi}},
\end{array}
\end{equation}
(see \cite[\emph{first Jackson Theorem}]{korovkin60}, or~\cite[Theorem~1.3]{rivlin69}).
%
%We refer to~\cite[\S{}1.1]{rivlin69}  and references therein for an interpretation of~(\refeq{jacksontrigonometric}) as a near-best convergence rate for trigonometric sums.
%

Inequality~\eqref{jacksontrigonometric} provides us with interval bounds for the cost function, from which we infer bounds for the value function.
\begin{proposition}[Periodic cost]\label{proposition:periodic}
Consider the server model of Section~\ref{section:valuefunction} with  cost function~$\costfunction= %\costx{0} \diract{\bl} +
 \periodcostfunction\,  \indicatorfunction{(0,\infty)} $ 
meeting Assumption~\ref{assumption:integrablecost}, where $\periodcostfunction:\Real\mapsto\Real$ is continuous and $ \moveperiod{}{2}\period$-periodic. % and modulus of continuity~$\continuitymodulusSfx{\Real}{\periodcostfunction}{\cdot}$.
The \acl{wfunction} satisfies $\purevaluefunction\in\intervx{\purevaluefunction}$, where
\begin{equation}\begin{array}{ll}
\intervxy{\purevaluefunction}{\bl} =
[{\ar}/({1-\load})]\, \left\{\left({\fouriercoefk{0}} + 6\,\ContinuitymodulusSfx{\period}{\periodcostfunction}{{\tau}/({\order\pi})} [-1,1] \right) \bl \right. 
&
\\ \hspace{10mm}
\left.
+ [{ \moveperiod{}{2}\period}/({  k \pi})]\, \sum\nolimits_{k=1}^{\order} \itscoefnk{\order}{{k}} \,  \imaginarypart{ \fouriercoefk{-k}
\LSTWts{\i{  \moveperiod{2}{} k \pi}/{\period}} (1-\exp{-\i{ \moveperiod{2}{} k \pi\bl}/{\period}})} \right\}
,
&\quad \forall\bl\in\Realpluszero,
%\\
%=
%\frac{\ar}{1-\load} \left( \frac{\fouriercoefk{0}}{2} \bl + \frac{\period}{ 2 k \pi} \sum\nolimits_{k=1}^{\order} \itscoefnk{\order}{\modulus{k}} \, 2 \Imaginarypart{ \fouriercoefk{-k} \LSTWTS{\i\frac{ 2 k \pi}{\period}} (1-\Cosx{\frac{ 2 k \pi\bl}{\period}}+\i  \Sinx{\frac{ 2 k \pi\bl}{\period}}) } \right)
\end{array}
\end{equation}
with $\itscoefnk{\order}{{k}},\dots,\itscoefnk{\order}{{k}}$  given by~\eqref{itscoefnk}, and $\fouriercoefk{-\order},\dots,\fouriercoefk{0}$ by~\eqref{fouriercoefficients}.
%
%$\errn{\order} \leq 6\,\continuitymodulusSfx{\Real}{\periodcostfunction}{1/\order}$.
\end{proposition}
\begin{proof}
The result follows by computation of the \acl{wfunction} of~\eqref{approximatedcostfunctionperiodic} with the modified trigonometric sum~\eqref{improvedpartialfourierseries} and an uniform error bound~$\errn{\order}$ satisfying~\eqref{jacksontrigonometric}. We used $\LSTWts{0}=1$, and  the results of Table~\ref{table:valuefunctions}  for $\costx{\bl}=1$, and for $\costx{\bl}=\exp{-\sing\bl}$ with $\sing=-\i\, ({  \moveperiod{2}{} k \pi}/{\period})$ and $k=-\order,\dots,-1,1,\dots,\order$.
\qed\end{proof}

%improved trigonometric sum of~\cite{korovkin60,rivlin69} improved on convergence rate $\log(n)\,  \omega(1/n)$ of~\cite[\S{}21]{jackson41}

}%%%%%%%%%%%%%%%%%

\subsection{Continuous cost functions} 
\label{section:continuouscostfunctions}

Assume now that the cost function~$\costfunction$ is continuous\footnote{Piecewise continuous functions can be treated similarly by partitioning~$\Realplus$ into as many intervals as required by their discontinuities.}, and partition the backlog axis into an interval~$(0,\taufunction)$ %($\taufunction>0$) 
where~$\costfunction$ is approximated precisely (in virtue of the Weierstrass approximation theorem) with respect to the uniform norm 
%\begin{equation}\label{uniformnorminterval}
$
\norm{\functionfunction}=\sup\nolimits_{\bl\in[0,\taufunction]}{\modulus{\fx{\bl}}}
$
%\end{equation} 
%
%In the case of a continuous (possibly, piecewise continuous) cost function~$\costfunction:\Realpluszero\mapsto\Real $  we suggest to partition the backlog time line  into two or more intervals  and to derive interval bounds for the cost function on each interval.
%For simplicity we consider only the following two-interval setup: an interval~$(0,\taufunction)$ ($\taufunction>0$) of the frequent backlog values where the cost function is continuous and approximated precisely with respect to the uniform norm 
\obsolete{
\begin{equation}\label{uniformnorminterval}
\norm{\functionfunction}=\sup\nolimits_{\bl\in[0,\taufunction]}{\modulus{\fx{\bl}}}
\end{equation} 
}%
by a finite sum~$\polycostnfunction{\order}$ %a polynomial 
of degree~$\order$, 
%~\cite[Weierstrass first theorem]{korovkin60}, \cite[Theorem~1.1]{rivlin69}, \cite[Theorem~2.7]{koralov07}|,
\obsolete{ 
\begin{equation}\label{polynomial}
\begin{array}{ll}
\polycostnx{\order}{\bl}=  \sum\nolimits_{k=0}^{\order} 
\polynomialcoef{\order}{k}
 \bl^k,
 &\forall \bl\in[0,\taufunction]
 ,
 \end{array}
\end{equation}
}%
 and its complement~$(\taufunction,\infty)$, %, containing more occasional backlog values,
 where unrefined bounds in~$\Excset$ are chosen for~$\costfunction$. Bounds for the \acl{wfunction} can be inferred from the developments of Section~\ref{section:piecewisedefined}. 
  \begin{notation}[$\makepurevffunction$]\label{notation:Makepurevf}
%  Given $\taufunction>0$, let $\intervx{\polycostfunction}$ and $\intervx{\excfunction}$ define two intervals of real functions in~$\Excset$, and define   $ \intervx{\costfunction} = \intervx{\polycostfunction} \, \indicatorfunction{(0,\taufunction)} +  \intervx{\excfunction} \, \indicatorfunction{(\taufunction,\infty)} $. We denote by $\intervMakepurevf{\intervx{\polycostfunction}}{\intervx{\excfunction}}$ the resulting interval for the \aclp{wfunction}  of all  $  \costfunction \in \intervx{\costfunction} $.
  Given $\taufunction>0$ and $\polycostfunction,\excfunction\in\Excset$, we denote by $\makepurevf{\polycostfunction}{\excfunction}$ the  \ac{wfunction}  relative to the cost function $ \costfunction = \polycostfunction \, \indicatorfunction{(0,\taufunction)} +  \excfunction \, \indicatorfunction{(\taufunction,\infty)} $.
 \end{notation}
\begin{proposition}[Continuous cost]\label{proposition:polynomial}
Consider the server model of Section~\ref{section:valuefunction} with   a cost function~$\costfunction$ 
meeting Assumption~\ref{assumption:complexintegrablecost}, continuous on a nonempty interval~$(0,\taufunction)$, %($\taufunction>0$), 
and such that $\costfunction\in\intervx{\costfunction}$, where 
\begin{equation}\label{approximatedcostfunction}
\begin{array}{r}
\intervx{\costfunction} =
%\costx{0}\, \diracfunction +
 \{ \polycostnfunction{\order}+
[-\errn{\order},\errn{\order}] \} \, \indicatorfunction{(0,\taufunction)}
+  \intervx{\excfunction} \, \indicatorfunction{(\taufunction,\infty)}
 ,%\hspace{5mm}
%\\ \forall \bl\in\Realpluszero.
\end{array} 
\end{equation}
%$\polycostnfunction{\order}$ 
in which
$ \polycostnfunction{\order}$ is a finite sum of degree~$\order\in\Natural$,
$\errn{\order}\geq 0$,
$\intervx{\excfunction}\equiv[\lexcfunction,\uexcfunction]$, %:\Realpluszero\mapsto\intervx{\Real}$, with  $\lexcfunction,\uexcfunction\in\Excset$.
and $\polycostnfunction{\order},\lexcfunction,\uexcfunction$ are real elements of~$\Excset$.
The \acl{wfunction} %of~$\costfunction$  
satisfies 
\begin{equation}\notag\begin{array}{l}
\purevaluefunction \in \left[ \makepurevf{\polycostnfunction{\order}- \errn{\order}}{\lexcfunction} ,  \makepurevf{\polycostnfunction{\order}+ \errn{\order}}{\uexcfunction} 
\right]
 ,
\end{array}
\end{equation}
where $\makepurevffunction$ (cf. Notation~\ref{notation:Makepurevf})  is computed as in Section~\ref{section:piecewisedefined}.
%
%$\errn{\order} \leq 6\,\continuitymodulusSfx{\Real}{\periodcostfunction}{1/\order}$.
\end{proposition}
\obsolete{
 \begin{notation}[$\makepurevffunction$]\label{notation:Makepurevf}
We denote by $\Makepurevf{\pwkjin{0} ,\dots,\pwkjin{\nn}}{\excfunction}$ the \acl{wfunction} for 
%the cost function defined %(almost everywhere)   
%for $\bl\in\Realplus$
%by $
\begin{equation}\textstyle
 \costx{\bl} =   \sum\nolimits_{j=0}^{\nn} \pwkjin{j}  \bl^{j} \, \indicator{(0,\taufunction)}{\bl}
  + \excx{\bl}  \, \indicator{(\taufunction,\infty)}{\bl}
%$,
,
\end{equation}
where $\taufunction>0$, $\nn\in\Natural$, $\pwkjin{1},\dots,\pwkjin{\nn}\in\Real$, and $\excfunction\in\Excset$. % is of type~\eqref{costanx}.
% %  
 \end{notation}
\begin{proposition}[Continuous cost]\label{proposition:polynomial}
Consider the server model of Section~\ref{section:valuefunction} with   a cost function~$\costfunction$ 
meeting Assumption~\ref{assumption:complexintegrablecost}, continuous on a nonempty interval~$(0,\taufunction)$, %($\taufunction>0$), 
and such that $\costfunction\in\intervx{\costfunction}$, where 
\begin{equation}\label{approximatedcostfunction}
\begin{array}{r}
\intervx{\costfunction} =
%\costx{0}\, \diracfunction +
 \{ \polycostnfunction{\order}+
[-\errn{\order},\errn{\order}] \} \, \indicatorfunction{(0,\taufunction)}
+  \intervx{\excfunction} \, \indicatorfunction{(\taufunction,\infty)}
 ,%\hspace{5mm}
%\\ \forall \bl\in\Realpluszero.
\end{array} 
\end{equation}
%$\polycostnfunction{\order}$ 
$ \polycostnx{\order}{\bl}=  \sum\nolimits_{k=0}^{\order} 
\polynomialcoef{\order}{k}
 \bl^k$ is a polynomial of degree~$\order\in\Natural$,
 \obsolete{
 given by
\begin{equation}\label{polynomial}
\begin{array}{ll}
\polycostnx{\order}{\bl}=  \sum\nolimits_{k=0}^{\order} 
\polynomialcoef{\order}{k}
 \bl^k,
 &\quad \forall \bl\in[0,\taufunction]
 ,
 \end{array}
\end{equation}
}%
%the uniform error bound~$\errn{\order}$ is a positive constant, 
$\errn{\order}\geq 0$,
and~$\intervx{\excfunction}\equiv[\lexcfunction,\uexcfunction]:\Realpluszero\mapsto\intervx{\Real}$, with  $\lexcfunction,\uexcfunction\in\Excset$.
%bounding function~$\lexcfunction,\uexcfunction$ of the type~\eqref{costanx}.
The \acl{wfunction} %of~$\costfunction$  
satisfies %$\purevaluefunction\in\intervx{\purevaluefunction}$, where
%
%$$ \Intervx{\purevaluefunction} \in \intervxy{\mathcal{W}}{\translatedpolcoscoef{\order}{0}   + 6\, \continuitymodulussIfx{\taufunction}{[0,\taufunction]}{\costfunction}{\frac{\taufunction}{2\order}} [-1,1],\translatedpolcoscoef{\order}{1},\dots,\translatedpolcoscoef{\order}{\order};\intervx{\excfunction}}   $$
\begin{equation}\begin{array}{l}
\purevaluefunction \in \left[ \makepurevf{\polynomialcoef{\order}{0}   - \errn{\order},\polynomialcoef{\order}{1},\dots,\polynomialcoef{\order}{\order}}{\lexcfunction} ,  \makepurevf{\polynomialcoef{\order}{0}   + \errn{\order},\polynomialcoef{\order}{1},\dots,\polynomialcoef{\order}{\order}}{\uexcfunction} 
\right]
 ,
\end{array}
\end{equation}
\obsolete{
\begin{equation}\begin{array}{ll}
\intervxy{\purevaluefunction}{\bl} =
\frac{\ar}{1-\load}
\left\{
\left(   \translatedpolcoscoef{\order}{0}   + 
6\, \continuitymodulussIfx{\taufunction}{[0,\taufunction]}{\costfunction}{\frac{\taufunction}{2\order}} [-1,1]
 \right) \bl
\vphantom{\frac{1^1}{1}} \right. 
&
\\ \hspace{35mm}
\left.
+        \sum\nolimits_{k=1}^{\order} \factorial{k} \, \translatedpolcoscoef{\order}{k} 
  \sum\nolimits_{t=0}^{k}\specialmatrixnk{k+1}{k-t} \frac{\bl^{t+1}}{\factorial{(t+1)}}
\right\} ,
& \forall\bl\in\Realpluszero,
%\\
%=
%\frac{\ar}{1-\load} \left( \frac{\fouriercoefk{0}}{2} \bl + \frac{\period}{ 2 k \pi} \sum\nolimits_{k=1}^{\order} \itscoefnk{\order}{\modulus{k}} \, 2 \Imaginarypart{ \fouriercoefk{-k} \LSTWTS{\i\frac{ 2 k \pi}{\period}} (1-\Cosx{\frac{ 2 k \pi\bl}{\period}}+\i  \Sinx{\frac{ 2 k \pi\bl}{\period}}) } \right)
\end{array}
\end{equation}
}%
where $\makepurevffunction$ (cf. Notation~\ref{notation:Makepurevf})  is computed as in Section~\ref{section:piecewisedefined}.
%
%$\errn{\order} \leq 6\,\continuitymodulusSfx{\Real}{\periodcostfunction}{1/\order}$.
\end{proposition}
}%
 \obsolete{
 Let 
\begin{equation}\label{uniformnormtwo}
\begin{array}{l}
\errn{\order}=\norm{\costfunction-\polycostnfunction{\order}} = \sup\nolimits_{\bl\in[0,\period]}{\modulus{\costx{\bl}-\polycostnx{\order}{\bl}}},
\end{array}
\end{equation}
and denote by~$\intervx{\excfunction}\equiv[\lexcfunction,\uexcfunction]\in\Realpluszero\mapsto\intervx{\Real}$ the bounding interval on~$(\taufunction,\infty)$, so that $\costx{\bl}\in[\costnx{\order}{\bl}-\errn{\order},\costnx{\order}{\bl}+\errn{\order}]$ for $\bl\in(0,\taufunction)$ and $\costx{\bl}\in\intervx{\excx{\bl}}$ for $\bl\in(\taufunction,\infty)$. The cost function thus satisfies $\costfunction\in\intervx{\costfunction}$, where we define
\begin{equation}\label{approximatedcostfunction}
\begin{array}{r}
\intervx{\costfunction} =
%\costx{0}\, \diracfunction +
 \left( \costnfunction{\order}+
[-\errn{\order},\errn{\order}] \right) \, \indicatorfunction{(0,\taufunction)}
+  \intervx{\excfunction} \, \indicatorfunction{(\taufunction,\infty)}
. %,%\hspace{5mm}
%\\ \forall \bl\in\Realpluszero.
\end{array}
\end{equation}
%where~$\diracfunction$ is the Dirac delta function and~$\indicatorfunction{}$ is the indicator function.
%
}%
%

%, and in particular~(\refeq{cstexamplebvf}) in Example~\refeq{example:piecewisecst}. 
%

\noindent
The \ac{FPI} step can be implemented based on the interval bounds~\eqref{admissioncostbounds} for~$\intervx{\costfunction}$, in place of the actual admission cost~\eqref{admissioncost}, % with respect the true cost function~$\costfunction$, 
provided that the parameters~$\taufunction$ and~$\nn$ chosen for the servers allow for it. Otherwise, the parameter values should be refined (by increasing~$\taufunction$ and~$\nn$) until decision~\eqref{decision} can be made.  

A pseudocode for the resulting procedure
is given in Algorithm~\ref{algorithm:dispatching}, where the cost function~$\costifunction{\server}$ of each server~$\server$ is supplied with a continuum of bounding interval functions~$\intervx{\excifunction{\server}}$ such that, for every~$\taufunction>0$, $\costix{\server}{\bli{\server}}\in\intervxy{\excit{\server}{\taufunction}}{\bli{\server}}$ if $\bli{\server}>\taufunction$.
 Algorithm~\ref{algorithm:dispatching} infers the \ac{FPI} decision~$ \newpolicyx{\bl,\st} $ at any state~$\vect{\bl,\st}$ by gradually decreasing the error tolerance~$\epsilon_t$ of the admission cost bounds at each server, computed by~\eqref{approximatedcostfunction}. To guarantee the error margin~$\epsilon_t$ at a server~$\server$, the parameter~$\taui{\server}$ is first taken large enough for the  approximation error in the $\bl>\taui{\server}$ window to be less than~${\epsilon_t}/{2}$ (line~\ref{line:one}), then the sum~$\polycostnfunction{\order}$ %in~\eqref{approximatedcostfunction} 
 is given enough terms for the approximation error in the $0<\bl<\taui{\server}$ window to be less than~${\epsilon_t}/{2}$ (line~\ref{line:two}), so that the overall precision~$\epsilon_t$ is secured for the bounds~$\intervx{\polycostifunction{\server}}$ (line~\ref{line:three}). All servers with exceeding admission costs  will be ignored (line~\ref{line:four}) for the rest of the procedure, which resumes with a smaller margin $\epsilon_{t+1}$.
\begin{algorithmenv}[t]
%\hrulefill
    \caption{\ac{FPI} with interval value functions\protect\footnotemark
    \label{algorithm:dispatching}}
\horizontalline
%\vspace{-2mm}\\
\renewcommand{\makenindex}[2]{{#1}^{\scalebox{0.8}{\tiny$(#2)$}}}%
\centering %\removelatexerror
\begin{algorithm}[H]%[t]
\def\inlineIf{\textbf{If }}
\def\inlineThen{\textbf{ then }}
\def\inlineWith{\textbf{ with }}
\def\inlineHoldsFor{\textbf{ holds for }}
\def\inlineST{\textbf{ such that }}
\def\And{\textup{\textbf{ and }}}
\SetKwFor{While}{While}{do}{fintq}%
\SetKwFor{For}{For}{do}{fintq}%
    \SetKwInOut{Input}{Input}
    \SetKwInOut{Output}{Output}
\SetKwInput{KwInit}{Initialization}
\KwData{$\{\vect{\ari{1},\costifunction{1},\intervx{\excifunction{1}}},\dots,\vect{\ari{\nbservers},\costifunction{\nbservers},\intervx{\excifunction{\nbservers}}}\}$, $\{\epsilon_t\} $ with $\epsilon_t\downarrow 0$%_{t=0}^{t_\textup{max}}
}
\smallskip
%\texttt{function}$(\bl,\st)$\;
 \Input{%$\vect{\bli{1},\dots,\bli{\nbservers}}$, 
 $\vect{\bl,\st}\in\Realpluszero^\nbservers\times\Realpluszero^\nbservers$% with $\epsilon_t\downarrow 0$
 }
 \Output{$\decisionset\subset\{1,\dots,\nbservers\}$}
 \smallskip 
 \KwInit{$t\leftarrow0$,  $\decisionset \leftarrow \{1,\dots,\nbservers\}$,  $\intervx{\polycostifunction{\server}} \leftarrow (-\infty,\infty) $ for all $ i\in\decisionset$}
\nonl\vspace{-4pt}\rule{0.95\linewidth}{.2pt} \; \vspace{-9pt}
\nonl%
\While{ $t\leq t_\textup{max}$\And
$\modulus{\decisionset}>1$}{
 \nonl \For{$i\in\decisionset$}{
   $\taui{\server}\leftarrow \arg \inf\nolimits_{\taufunction\geq 0} \big\{   
   %\modulus{\admissionfoperator{\ari{\server}}{  \intervx{\excifunction{\server}}  \indicatorfunction{(\taufunction,\infty) }}\vect{\bli{\server},\sti{\server}}}  
   %\bigmodulus{\intervx{\admissionicostoperator{\server}{\intervx{\excit{\server}{\taufunction}}  \indicatorfunction{(\taufunction,\infty) }}}\vect{\bl,\st}}
\bigmodulus{\intervadmissionicostux{\server}{\intervx{\excit{\server}{\taufunction}} \, \indicatorfunction{(\taufunction,\infty) }}{\bl}{\st}} 
   \leq {\epsilon_t}/{2} \big\} $ 
   %\hfill (set $ \taui{\server} $)
   \;        \label{line:one} 
   $\orderi{\server} \leftarrow \arg\min\nolimits_{\order} \big\{ 
   %\errn{\order}\modulus{\admissionfoperator{\ari{\server}}{ \indicatorfunction{(0,\taufunction)} }\vect{\bli{\server},\sti{\server}}}
      %\bigmodulus{\admissionicostoperator{\server}{ \indicatorfunction{(\taufunction,\infty) }}\vect{\bl,\st}}
      \bigmodulus{\admissionicostux{\server}{ \indicatorfunction{(0,\taui{\server}) }}{\bl}{\st}}
    \leq {\epsilon_t}/({4\errn{\order}}) \big\}$ 
%\hfill (set $ \orderi{\server} $)
\;     \label{line:two}
$\intervx{\polycostifunction{\server}}
   %\refereq{(\refeq{approximatedcostfunction})}{\leftarrow}
   \leftarrow
%\costx{0}\, \diracfunction +
 \big\{ \polycostinfunction{\server}{\orderi{\server}}+
[-\errn{\orderi{\server}},\errn{\orderi{\server}}] \big\} \, \indicatorfunction{(0,\taui{\server})}
+  \intervx{\excit{\server}{\taui{\server}}} \, \indicatorfunction{(\taui{\server},\infty)} $\;\label{line:three}
}
 \nonl \For{$i\in\decisionset$}{
%$$
  \inlineIf{$\exists j\in \decisionset {\setminus}\{i\}$}\inlineST{$
  %\intervx{\admissionicostoperator{\server}{\intervx{\polycostifunction{\server}}}}\vect{\bl,\st} > \intervx{\admissionicostoperator{\altserver}{\intervx{\polycostifunction{\altserver}}}}\vect{\bl,\st}
\intervadmissionicostux{\server}{\intervx{\polycostifunction{\server}}}{\bl}{\st} > \intervadmissionicostux{\altserver}{\intervx{\polycostifunction{\altserver}}}{\bl}{\st}
$
%\admissionfoperator{\ari{\server}}{\intervx{\costifunction{\server}}}\vect{\bli{\server},\sti{\server}}>\admissionfoperator{\ari{\altserver}}{\intervx{\costifunction{\altserver}}}\vect{\bli{\altserver},\sti{\altserver}} $
  }\inlineThen{$\decisionset \leftarrow \decisionset {\setminus}\{i\} $\;  \label{line:four}
  }
  }
}
%    \caption{\ac{FPI} with interval value functions\label{algorithm:dispatching}}
%\nonl\rule{0.96\linewidth}{.4pt} 
\end{algorithm}
\horizontalline
\end{algorithmenv}
\footnotetext{In Algorithm~\ref{algorithm:dispatching}, the first argument of $\admissionicostux{\server}{\genericfunction}{\cdot}{\cdot}$ (or $\intervadmissionicostux{\server}{\intervx\genericfunction}{\cdot}{\cdot}$) indicates the cost function~$\genericfunction$ (resp. the interval function~$\intervx\genericfunction$) for which the admission cost at server~$\server$ is computed.}

\obsolete{
Recalling~\eqref{admissioncost}, we find a bounding interval function~$\admissionfoperator{\ar}{\intervx{\costfunction}}$ for the admission cost of a job with service time $\st\in\Realpluszero$ for the considered queueing sytem at state $\bl\in\Realpluszero$, defined by
\begin{equation}\label{admissioncostbounds}
\begin{array}{ll}
\admissionfoperator{\ar}{\intervx{\costfunction}}\vect{\bl,\st}
=
%[\lpurevaluex{\bl+\st} ,\upurevaluex{\bl+\st} ]
\intervx{\purevaluefunction}\vect{\bl+\st}
-
%[\lpurevaluex{\bl} ,\upurevaluex{\bl} ]
\intervx{\purevaluefunction}\vect{\bl}
%[\lpurevaluex{\bl+\st}-\upurevaluex{\bl} ,\upurevaluex{\bl+\st}-\lpurevaluex{\bl} ]
-\frac{\ar\st}{1-\load} 
%[\lmeancost ,\umeancost ]
\intervx{\meancost}
,
&
\forall \vect{\bl,\st} \in \Realpluszero \times \Realpluszero,
\end{array}
\end{equation}
where~$\intervx{\meancost}$ are the bounds
on the mean cost per job computed by~\eqref{genericmeancosttwo} for the bounds~\eqref{approximatedcostfunction}.
In the presence of~$\nbservers$ queueing systems with cost functions~$\costifunction{1},\dots, \costifunction{\nbservers}$ respectively approximated as above by the interval cost functions~${\intervx{\costifunction{1}}},\dots, {\intervx{\costifunction{\nbservers}}}$, %  and respective admission cost interval functions~$\admissionfoperator{\intervx{\costifunction{1}}},\dots, \admissionfoperator{\intervx{\costifunction{\nbservers}}}$, 
the first-step dispatching decision can be made at state~$\vect{\bli{1},\dots,\bli{\nbservers}}$ iff there is an $i\in\{1,\dots,\nbservers\}$ such that $\admissionfoperator{\ar}{\intervx{\costifunction{i}}}\vect{\bli{i},\st}\leq\admissionfoperator{\ar}{\intervx{\costifunction{j}}}\vect{\bli{j},\st}$ for every $j\in\{1,\dots,i-1,i+1,\dots,\nbservers\}$.
Otherwise, the approximation of the cost function should be improved by increasing~$\taufunction$ and/or~$\order$ until the condition is met.
}%
% 
%

%  In Sections~\ref{section:bersteinpolynomials} and~\ref{section:nearoptimalpolynomials} we discuss the methods for deriving the finite sum~$\intervx{\polycostnfunction{\nn}} $ when the cost function~$\costfunction$ is continuous and bounded on any interval~$[0,\taufunction]$. A consequence of our assumption is that~$\costfunction$ is uniformly bounded on~$[0,\taufunction]$, where its modulus of continuity,
% $ \continuitymodulussIfx{\taufunction}{[0,\taufunction]}{\costfunction}{\delta} 
% =
% \sup\{\modulus{\costx{\bli{1}}-\costx{\bli{2}}} : \bli{1},\bli{2}  \in [0,\taufunction],\, \modulus{\bli{1}-\bli{2}}\leq \delta\}$, 
% vanishes with~$\delta$.

 In Sections~\ref{section:bersteinpolynomials} and~\ref{section:nearoptimalpolynomials} we discuss the methods for deriving the finite sum~$\intervx{\polycostnfunction{\nn}} $ when the cost function~$\costfunction$ is continuous on any support~$[0,\taufunction]$.
\subsubsection{Bernstein polynomials}
\label{section:bersteinpolynomials}

%Let~$\costfunction$ %:\Realpluszero\mapsto\Real$ 
%be continuous and bounded on the interval~$[0,\taufunction]$. 
The function~$\costfunction$ can be approximated on~$[0,\taufunction]$ by the  \emph{Bernstein polynomial}~\cite{bernstein12}
\begin{align}\label{bernsteinnx}
%\textstyle
&
\bernsteinnx{\order}{\bl}
 = 
  \sum\nolimits_{\lind=0}^{\order}   \Bigbinomialcoef{\order}{\lind}  \, \Big(\frac{\bl}{\taufunction}\Big)^\lind \Big(1-\frac{\bl}{\taufunction}\Big)^{\order-\lind} \Costx{{\lind\taufunction}/{\order}}
%   \sum\nolimits_{\lind=0}^{\order} \Bigbinomialcoef{\order}{\lind} \, ({\bl}/{\taufunction})^\lind \, (1-{\bl}/{\taufunction})^{\order-\lind}\, \costx{{\lind\taufunction}/{\order}}
 ,
  & \forall \bl\in[0,\taufunction].
  \quad
\end{align}
Notice that~\eqref{bernsteinnx} rewrites as $\bernsteinnx{\order}{\bl}
 = 
  \expectation{ \costx{{K\taufunction}/{\order}} }$, where the random variable~$K\sim \textup{B}(\order,\bl/\taufunction)$
  is distributed according to   the binomial distribution with~$\order$ trials and success probability~$\bl/\taufunction$.
  %, which follows  the binomial distribution with parameters~$\order$ and~$\bl/\taufunction$, can be seen as the total number of successes of~$\order$ independent Bernouilli trials with probability of success~$\bl/\taufunction$. 
  The quantity~$K/\order$ has mean~$\bl/\taufunction$ and variance~$({\bl}/{\taufunction})(1-{\bl}/{\taufunction})/\order\leq 1/(4\order)$, which vanishes uniformly on~$[0,\taufunction]$. It follows
  %by uniform continuity of~$\costfunction$ on~$[0,\taufunction]$ 
  from continuity arguments that~$\expectation{ \costx{{K\taufunction}/{\order}} }$ converges uniformly towards~$\costx{\bl}$ on that interval, \cite[proof of Theorem~2.7]{koralov07}.
So does~\eqref{bernsteinnx}, with rate  
  %The uniform convergence rate for~(\refeq{bernsteinnx}), 
\begin{equation}\label{bernsteinrate}
%\textstyle
%\begin{array}{l}
%\errn{\order} 
\norm{\costfunction-\bernsteinnfunction{\order}}
\leq 
 ({3}/{2})\, \bigcontinuitymodulussIfx{\taufunction}{[0,\taufunction]}{\costfunction}{{\taufunction}/{\sqrt{\order}}} ,
%\end{array}
\end{equation}
\cite[Theorem~1.2]{rivlin69},  where
%
%$
\begin{equation}\label{continuitymodulusbis}\notag
\textstyle
\continuitymodulussIfx{\taufunction}{[0,\taufunction]}{\costfunction}{\delta} 
=
\sup \{ \modulus{\costx{\bli{1}}-\costx{\bli{2}}}
 :  \bli{1},\bli{2}  \in [0,\taufunction], \modulus{\bli{1}-\bli{2}}\leq\delta\} 
%\sup\nolimits_{\bli{1},\bli{2}  \in \Real, \modulus{\bli{1}-\bli{2}}\leq\delta} \modulus{\periodcostx{\bli{1}}-\periodcostx{\bli{2}}}
%
%\sup\limits_{\substack{\bli{1},\bli{2}  \in \Real,\\ \modulus{\bli{1}-\bli{2}}\leq\delta}} \modulus{\periodcostx{\bli{1}}-\periodcostx{\bli{2}}}
%.
\end{equation}
%$
defines the \emph{modulus of continuity} of~$\costfunction$ on the interval~$[0,\taufunction]$, \cite[\S{}21]{jackson41}.
%
%is derived in~\cite[Theorem~1.2]{rivlin69}, where
  %
  \obsolete{
\begin{equation}\label{continuitymodulusbis}
\textstyle
\continuitymodulussIfx{\taufunction}{[0,\taufunction]}{\costfunction}{\delta} 
=
\sup\limits_{\substack{\bli{1},\bli{2}  \in [0,\taufunction],\\ \modulus{\bli{1}-\bli{2}}\leq \delta}} \modulus{\costx{\bli{1}}-\costx{\bli{2}}}
\end{array}
\end{equation}
defines the modulus of continuity of~$\costfunction$ on the interval~$[0,\taufunction]$. Note that~$\costfunction$ is uniformly continuous on the interval iff $\lim\nolimits_{\delta\downarrow 0}\continuitymodulussIfx{\taufunction}{[0,\taufunction]}{\costfunction}{\delta} =0 $.
  }%
To conform with~\eqref{approximatedcostfunction}, we rewrite~\eqref{bernsteinnx} as\footnotemark
\begin{align}\label{bernsteinnxexpanded}
%\textstyle
&
\bernsteinnx{\order}{\bl}
= 
\sum\nolimits_{\kind=0}^{\order} 
 \bernsteincoef{\order}{\kind}
 \bl^{\kind} 
  ,
  &
  \forall \bl\in[0,\taufunction],
\end{align}
where
$%\begin{equation}\label{bernsteincoef}
\textstyle
\bernsteincoef{\order}{\kind} 
=
(-\taufunction)^{-\kind} \binomialcoef{\order}{\kind} \sum\nolimits_{\lind=0}^{\kind} 
%\binomialcoef{\order}{\lind}  \binomialcoef{\order-\lind}{\kind-\lind}   
 \binomialcoef{\kind}{\lind}   
(-1)^{\lind}  \costx{{\lind\taufunction}/{\order}} 
 % ,
 % \quad (\kind=0\dots,\order)
  %.
$, for $\kind=0\dots,\order$. %\end{equation}
\footnotetext{
The coefficients~$\bernsteincoef{\order}{0},\dots,\bernsteincoef{\order}{\order}$ can be computed recursively. Indeed, one show by induction that $\bernsteincoef{\order}{\kind}=({1}/{\factorial{\kind}})\, \bernsteinderivx{\order}{\kind}{0}$, where, for $\kind=1,\dots,\order$,  
\begin{equation}\label{bernsteinderivx}
%\nocolsep
%\textstyle%\begin{array}{ll}
\bernsteinderivx{\order}{0}{\lind}=\costx{{\lind\taufunction}/{\order}},
%&(\lind=0,\dots,\order),
%\\ 
\quad
\bernsteinderivx{\order}{\kind}{\lind}=[({\order-\kind+1})/{\taufunction} ] \, \big[\bernsteinderivx{\order}{\kind-1}{\lind+1}-\bernsteinderivx{\order}{\kind-1}{\lind} \big] , \quad  (\lind=0,\dots,\order-\kind
%,\ \kind=1,\dots,\order
).
 %\end{array}
\end{equation}
\noeqref{bernsteinderivx}%\storeequationcounter{bernsteinderivx}%
\storecompoundcounter{equation}{bernsteinderivx}%
}%

\obsolete{
Observe that the coefficients~$\bernsteincoef{\order}{0},\dots,\bernsteincoef{\order}{\order}$ satisfy~$\bernsteincoef{\order}{k}=\bernsteinderivx{\order}{k}{0}/\factorial{k}$, where we compute, by induction\footnotemark{},  
\begin{equation}\label{bernsteinderivx}
\begin{array}{ll}
\bernsteinderivx{\order}{0}{t}=\Costx{\frac{t\taufunction}{\order}},
&(t=0,\dots,\order),
\\
\bernsteinderivx{\order}{k}{t}=\frac{(\order-k+1) \left[\bernsteinderivx{\order}{k-1}{t+1}-\bernsteinderivx{\order}{k-1}{t} \right] }{\taufunction} ,\quad&  (t=0,\dots,\order-k,\ k=1,\dots,\order).
 \end{array}
\end{equation}

\footnotetext{
Indeed, $\bernsteincoef{\order}{0}=\bernsteinderivx{\order}{0}{0}/\factorial{0}$ is immediate. Suppose  now that
\begin{equation}\label{bernsteinderivxextended}
\begin{array}{ll}
\bernsteinderivx{\order}{k}{l} = \frac{\factorial{k}}{
(-\taufunction)^{k} } \binomialcoef{\order}{k} \sum\nolimits_{t=0}^{k} 
%\binomialcoef{\order}{t}  \binomialcoef{\order-t}{k-t}   
  \binomialcoef{k}{t}   
(-1)^{t}  \Costx{\frac{t\taufunction}{\order}} 
  ,
 & (l=0,\dots,\order-k),
 \end{array}
\end{equation}
holds for $k=1,\dots,q-1$, where $1\leq q \leq \order-1$. Then, for $l=0,\dots,\order-q$,
\begin{equation}\label{inductionbernstein}
\begin{array}{l}
\bernsteinderivx{\order}{q}{l}
\refereq{\eqref{bernsteinderivx}}{=}
\frac{(\order-q+1)}{\taufunction} \left[  \bernsteinderivx{\order}{q-1}{l+1}-\bernsteinderivx{\order}{q-1}{l} \right] 
\\
\refereq{\eqref{bernsteincoef}}{=}
\frac{\order-q+1}{\taufunction} \left(\frac{\factorial{(q-1)}}{(-\taufunction)^{q-1}} \right) \binomialcoef{\ \order}{q-1}\sum\nolimits_{t=0}^{q-1} 
%\binomialcoef{\order}{t}  \binomialcoef{\order-t}{q-1-t}   
  \binomialcoef{q-1}{\ t} 
(-1)^{t} \left[  \Costx{\frac{(t+l+1)\taufunction}{\order}} -   \Costx{\frac{(t+l)\taufunction}{\order}} \right]
\\
=
\left( \frac{\factorial{q}}{ (-\taufunction)^{q}} \right) \frac{ (\order-q+1) }{q}   \binomialcoef{\ \order}{q-1} \left[  \sum\nolimits_{t=1}^{q}%\binomialcoef{\order}{t-1}  \binomialcoef{\order+1-t}{q-t} 
 \binomialcoef{q-1}{t-1} 
(-1)^{t} \Costx{\frac{(t+l)\taufunction}{\order}} +  \sum\nolimits_{t=0}^{q-1} 
%\binomialcoef{\order}{t}  \binomialcoef{\order-t}{q-1-t}   
\binomialcoef{q-1}{\ t} 
(-1)^{t} \Costx{\frac{(t+l)\taufunction}{\order}} \right]  
\\
=
\left( \frac{\factorial{q}}{ (-\taufunction)^{q}} \right) \binomialcoef{\order}{q} \left\{   \Costx{\frac{l\taufunction}{\order}}  +  (-1)^{q} \Costx{\frac{(q+l)\taufunction}{\order}} %\right. 
%\\
%\hspace{20mm} \left.
+ \sum\nolimits_{t=1}^{q-1}\left[
%\binomialcoef{\order}{t-1}  \binomialcoef{\order+1-t}{q-t} 
  \binomialcoef{q-1}{t-1} 
+ 
%\binomialcoef{\order}{t}  \binomialcoef{\order-t}{q-1-t}   
 \binomialcoef{q-1}{\ t} 
 \right] 
 (-1)^{t} \Costx{\frac{(t+l)\taufunction}{\order}} \right\}    
\\
\obsolete{
=
 \factorial{q} (\order-q+1) (-\taufunction)^{-q} \left[   \binomialcoef{\order}{q-1} /q   \Costx{\frac{l\taufunction}{\order}} + \sum\nolimits_{t=1}^{q-1} [ \binomialcoef{\order}{t-1}  \binomialcoef{\order+1-t}{q-t}   + \binomialcoef{\order}{t}  \binomialcoef{\order-t}{q-1-t}  ]/q  (-1)^{t}\Costx{\frac{(t+l)\taufunction}{\order}} + \binomialcoef{\order}{q-1}    /q (-1)^{q} \Costx{\frac{(q+l)\taufunction}{\order}}  \right]
\\
=
 \factorial{q} (\order-q+1) (-\taufunction)^{-q}   \left[   \binomialcoef{\order}{q} /(\order-q+1)   \Costx{\frac{l\taufunction}{\order}} + \sum\nolimits_{t=1}^{q-1}  \binomialcoef{\order}{t}  \binomialcoef{\order-t}{q-t}  /(\order+(q-1))  (-1)^{t}\Costx{\frac{(t+l)\taufunction}{\order}} + \binomialcoef{\order}{q}    /(\order-q+1) (-1)^{q} \Costx{\frac{(q+l)\taufunction}{\order}}  \right]
\\
=
 \factorial{q} (-\taufunction)^{-q} \left[ \binomialcoef{\order}{0}    \binomialcoef{\order-0}{q-0}   \Costx{\frac{(l+0)\taufunction}{\order}} + \sum\nolimits_{t=1}^{q-1}  \binomialcoef{\order}{t}  \binomialcoef{\order-t}{q-t}   (-1)^{t}\Costx{\frac{(t+l)\taufunction}{\order}} + \binomialcoef{\order}{q}  \binomialcoef{\order-q}{q-q}    (-1)^{q} \Costx{\frac{(q+l)\taufunction}{\order}}  \right]
 \\
 }%
=
\frac{\factorial{q}}{(-\taufunction)^{q}}  \binomialcoef{\order}{q}  \sum\nolimits_{t=0}^{q} 
% \binomialcoef{\order}{t}  \binomialcoef{\order-t}{q-t}   
  \binomialcoef{q}{t}   
 (-1)^{t}\Costx{\frac{(t+l)\taufunction}{\order}} 
 \end{array}
\end{equation}
and~\eqref{bernsteinderivxextended} holds for~$k=q$. By induction, \eqref{bernsteinderivxextended} is true for $0\leq k \leq \order$. By setting~$l=0$ in~\eqref{bernsteinderivxextended}, we infer from~\eqref{bernsteincoef} that 
 $\bernsteincoef{\order}{k}=\bernsteinderivx{\order}{k}{0}/\factorial{k}$ for~$0\leq k \leq \order$.
}%
}%

From~\eqref{bernsteinrate} and~\eqref{bernsteinnxexpanded}, we infer bounds for the value function.
\begin{corollary}[Bernstein polynomials]
% Proposition~\ref{proposition:polynomial} holds for $\polycostnfunction{\order}\equiv\bernsteinnfunction{\order}$, where~$\bernsteinnfunction{\order}$ is defined by~\eqref{bernsteinnxexpanded}, with the uniform error bound~$\errn{\order}= {3}/{2}\,\continuitymodulussIfx{\taufunction}{[0,\taufunction]}{\costfunction}{{\taufunction}/{\sqrt{\order}}}$. 
Proposition~\ref{proposition:polynomial} holds for $\polycostnfunction{\order}\equiv\bernsteinnfunction{\order}$ defined by~\eqref{bernsteinnxexpanded} with the uniform error bound~$\errn{\order}= {3}/{2}\,\continuitymodulussIfx{\taufunction}{[0,\taufunction]}{\costfunction}{{\taufunction}/{\sqrt{\order}}}$.   
%,  where~$\bernsteinnfunction{\order}$ is defined by~\eqref{bernsteinnxexpanded}. %,  with   $\bernsteincoef{\order}{0},\dots,\bernsteincoef{\order}{\order} $  given by~\eqref{bernsteincoef}.
\end{corollary}

\subsubsection{Approximation by trigonometric sums}
\label{section:trigonometricsum}

A better convergence rate for~$\polycostnfunction{\order}$ can be obtained using trigonometric sums; we refer to~\cite[\S{}1.1]{rivlin69} for details on this topic. 
Consider the continuous, $2\taufunction$-periodic function $\mirrorcostfunction:\Real\mapsto\Real$  defined on~$[-\taufunction,\taufunction]$ by $\mirrorcostx{\bl}=\costx{\modulus{\bl}} $.
The {Weierstrass approximation theorem} (see, e.g., \cite[Weierstrass first theorem]{korovkin60}, \cite[Theorem~1.1]{rivlin69}, \cite[Theorem~2.7]{koralov07}) claims that% the periodic, continuous function
~$\mirrorcostfunction$  can be approximated by a trigonometric sum with arbitrary precision with respect to the uniform norm $\norm{\functionfunction}=\sup\nolimits_{\bl\in\moveperiod{[-\period/2,\period/2]}{[-\period,\period]}}\modulus{\fx{\bl}}
$.
This implies
that for any~$\epsilon>0$ one can find $\order<\infty$ and a trigonometric sum~$\tsnfunction{\order}$ such that $\errn{\order} = \norm{\mirrorcostx{\bl}-\tsnx{\order}{\bl}}<\epsilon$.
%, 
%
It then follows that $
\costfunction\in\intervx{\costfunction}
=
\tsnfunction{\order}+ 
[-\errn{\order},\errn{\order}] 
$.
\obsolete{ 
It then holds that $\costfunction\in\intervx{\costfunction}$, where we set
\begin{align}\label{approximatedcostfunctionperiodic}
%\begin{array}{r}
&
\intervxy{\costfunction}{\bl} =
\tsnx{\order}{\bl}+ 
[-\errn{\order},\errn{\order}] 
,
&%\hspace{5mm}
\forall \bl\in\Realpluszero.
%\end{array}
\end{align}
}%
Such a trigonometric sum %approximating~$\periodcostfunction$ 
is given by the partial Fourier series, which for the real, even function~$\mirrorcostfunction$ reduces to
\begin{equation}\label{mirrorpartialfourierseries}
%\textstyle
\tsnx{\order}{\bl}=\mirrorfouriercoefk{0} + 2 \sum\nolimits_{k=1}^{\order} \mirrorfouriercoefk{k}\, \cosx{{ \moveperiod{2}{} k \pi \bl }/{\period}} ,
\end{equation}
where 
\begin{align}\label{mirrorfouriercoefficients} 
%\begin{array}{ll}
&
\mirrorfouriercoefk{k} = \moveperiod{ \frac{1}{\period} \int\nolimits_{-\frac{\period}{2}}^{\frac{\period}{2}} \mirrorcostx{\bl} \exp{-\i\frac{ 2 k \pi \bl }{\period}} }{ \frac{1}{\period} \int\nolimits_{0}^{{\period}} \costx{{\bl}} \cosx{{  k \pi \bl }/{\period}} } \, d\bl ,
&  (k\in\Natural).
%\end{array}
\end{align}
 are the Fourier coefficients.
\obsolete{
\begin{equation}\label{partialfourierseries}
%\textstyle
\tsnx{\order}{\bl}=\sum\nolimits_{k=-\order}^{\order} \mirrorfouriercoefk{k}\, \exp{\i\frac{ \moveperiod{2}{} k \pi \bl }{\period}} ,
\end{equation}
where 
$
%\begin{align}\label{mirrorfouriercoefficients} 
%\begin{array}{ll}
%&
\mirrorfouriercoefk{k} = \moveperiod{ \frac{1}{\period} \int\nolimits_{-\frac{\period}{2}}^{\frac{\period}{2}} \mirrorcostx{\bl} \exp{-\i\frac{ 2 k \pi \bl }{\period}} }{ ({1}/{2\period}) \int\nolimits_{-{\period}}^{{\period}} \periodcostx{\bl} \exp{-\i\frac{  k \pi \bl }{\period}} } \, d\bl ,
%&  (k\in\Integer).
%\end{array}
%\end{align}
$ are the the Fourier coefficients.
Since $\periodcostfunction$ is real, $\mirrorfouriercoefk{0}$ is  real and  $\mirrorfouriercoefk{-k}$ and~$\mirrorfouriercoefk{k}$ are complex conjugate for all~$k$, so that~$\costnfunction{\order}$ has~$2\order+1$ real parameters.
}%
With the {modulus of continuity} of% the periodic function
~$\mirrorcostfunction$ defined by
%
%$
\begin{equation}\label{continuitymodulus}
\continuitymodulusSfx{\period}{\mirrorcostfunction}{\delta}
=
\sup \{ \modulus{\mirrorcostx{\bli{1}}-\mirrorcostx{\bli{2}}}
 :  \bli{1},\bli{2}  \in \Real, \modulus{\bli{1}-\bli{2}}\leq\delta\} ,
%\sup\nolimits_{\bli{1},\bli{2}  \in \Real, \modulus{\bli{1}-\bli{2}}\leq\delta} \modulus{\periodcostx{\bli{1}}-\periodcostx{\bli{2}}}
%
%\sup\limits_{\substack{\bli{1},\bli{2}  \in \Real,\\ \modulus{\bli{1}-\bli{2}}\leq\delta}} \modulus{\periodcostx{\bli{1}}-\periodcostx{\bli{2}}}
%.
\end{equation}
%$
the Fourier series~\eqref{mirrorpartialfourierseries}   converges towards the periodic function~$\mirrorcostfunction$ with rate 
$ 
\magnitude{\errn{\order}} 
= 
\log(\order)\,  \continuitymodulusSfx{\period}{\mirrorcostfunction}{\period/(n\pi)}
$, 
\cite[\S{}21]{jackson41}. % scaled back to a $2\pi$-period.
%
%
%for uniform convergence of apprximations by trigonometric polynomials one needs the function to be periodic and continuous between periods (see convergence of fourier series KEVIN STEPHEN STOTTER CUDDY Corollary 2.15. )
%
%=> $\omega$ function defined across the border.
%
Faster convergence can be obtained by slightly modifying the Fourier coefficients in~\eqref{mirrorpartialfourierseries}. For this, consider % the sum %trigonometric sum 
\begin{equation} 
\label{improvedmirrorpartialfourierseries}
%\textstyle
\mirroritsnx{\order}{\bl}= \itscoefnk{\order}{0} \, \mirrorfouriercoefk{0} + 2 \sum\nolimits_{k=0}^{\order} \itscoefnk{\order}{k} \, \mirrorfouriercoefk{k} \, \cosx{{  \moveperiod{2}{} k \pi \bl }/{\period}} ,
\end{equation}
where $\itscoefnk{\order}{0},\dots,\itscoefnk{\order}{\order}\in\Real$. The choice of parameters proposed in~\cite[\S{}3]{korovkin60},
\obsolete{
\begin{equation}\label{itscoefnk}\begin{array}{ll}
\itscoefnk{\order}{0} =1,
&\\
\itscoefnk{\order}{1} = \cosx{\frac{\pi}{\order+2}},
&\\
\itscoefnk{\order}{k} = \frac{\sum\nolimits_{q=0}^{\order-k}\Sinx{\frac{q+1}{\order+2}\pi} \Sinx{\frac{q+k+1}{\order+2}\pi} }{ \sum\nolimits_{q=0}^{\order} \Sinsquaredx{\frac{q+1}{\order+2}\pi}   } ,
&\quad (k=2,\dots,\order),
\end{array}
\end{equation}
}%
\begin{equation}\label{itscoefnk}\textstyle
\itscoefnk{\order}{0} {\,=\,} 1,
\
\itscoefnk{\order}{1} {\,=\,} \cosx{\frac{\pi}{\order+2}},
\
\itscoefnk{\order}{k} {\,=\,} \frac{\sum\nolimits_{q=0}^{\order-k}\Sinx{\frac{q+1}{\order+2}\pi} \Sinx{\frac{q+k+1}{\order+2}\pi} }{ \sum\nolimits_{q=0}^{\order} \Sinsquaredx{\frac{q+1}{\order+2}\pi}   } 
\textup{ for } k=2,\dots,\order,
\end{equation}
lends~\eqref{improvedmirrorpartialfourierseries} the convergence rate
\begin{equation}\label{mirrorjacksontrigonometric}
\begin{array}{l}
\errn{\order} \leq 6\,\ContinuitymodulusSfx{\period}{\mirrorcostfunction}{\frac{\tau}{\pi\order}},
\end{array}
\end{equation}  
(see \cite[\emph{first Jackson Theorem}]{korovkin60}, or~\cite[Theorem~1.3]{rivlin69}).
Since $\mirroritsnfunction{\order}\in\Excset$ and by construction $ \continuitymodulusSfx{\period}{\mirrorcostfunction}{\cdot} \equiv \continuitymodulussIfx{\taufunction}{[0,\taufunction]}{\costfunction}{\cdot} $, $\mirroritsnfunction{\order}$ is a candidate finite sum for  Proposition~\ref{proposition:polynomial} and \eqref{mirrorjacksontrigonometric} gives us bounds for the value function.
 
\begin{corollary}[Trigonometric sums]\label{corollary:fourier}
% Proposition~\ref{proposition:polynomial} holds for~$\polycostnfunction{\order}\equiv\ptsnfunction{\order}$, where~$\ptsnfunction{\order}$ is defined in~\eqref{improvedpartialfourierseriester}, with the uniform error bound~$\errn{\order} = 6\, \continuitymodulussIfx{\taufunction}{[0,\taufunction]}{\costfunction}{{\taufunction}/({\pi\order})}$.
Proposition~\ref{proposition:polynomial} holds for $\polycostnfunction{\order}\equiv\bernsteinnfunction{\order}$ defined by~\eqref{bernsteinnxexpanded} with the uniform error bound~$\errn{\order}= {3}/{2}\,\continuitymodulussIfx{\taufunction}{[0,\taufunction]}{\costfunction}{{\taufunction}/{\sqrt{\order}}}$.   
\end{corollary}
\noindent
In particular, if for some $\hoeldernexponent\in(0,1]$ the cost function satisfies the $\hoeldernexponent$-H\"oldern condition $\modulus{\costx{\bli{1}}-\costx{\bli{2}}}\leq \hoeldernconstant  \modulus{\bli{1}-\bli{2}}^\hoeldernexponent  
$ for all $\bli{1},\bli{2} \in [0,\taufunction]$,
%($0<\hoeldernexponent \leq 1$), 
then  
%$\continuitymodulusSfx{\period}{\mirrorcostfunction}{\delta}\leq  \hoeldernconstant (\delta\pi/\period)^\hoeldernexponent$, 
$\continuitymodulussIfx{\taufunction}{[0,\taufunction]}{\costfunction}{\delta} \leq  \hoeldernconstant \delta^\hoeldernexponent$, 
and% the   trigonometric sum
~\eqref{mirrorpartialfourierseries}  converges uniformly towards~$\costfunction$ on~$[0,\taufunction]$ with $\errn{\order} = \magnitude{({\taufunction}/\order)^\hoeldernexponent}$. %In particular, when $\hoeldernexponent=1$, $\mirrorcostfunction$ is Lipschitz continuous.
If~$\costfunction$ is Lipschitz continuous on~$[0,\taufunction]$ with modulus~$\Lipschitz$, then $\errn{\order} < 2 \Lipschitz {\taufunction}/\order$.

\subsubsection{Near-optimal polynomial approximation}
\label{section:nearoptimalpolynomials}

\obsolete{%
A better convergence rate for~$\polycostnfunction{\order}$ can be obtained using trigonometric sums, which we first briefly discuss,  referring to~\cite[\S{}1.1]{rivlin69} for details on the topic.

Let $\thetacostfunction:\Real\mapsto\Real$ be continuous, periodic with period~$2\pi$, and defined  by $\thetacostx{\theta}=\costx{\bltheta{\theta}} $ where $\bltheta{\theta} =  ({\taufunction}/{2})\, (1 + \cosx{\theta})$.
The {Weierstrass approximation theorem} (see, e.g., \cite[Weierstrass first theorem]{korovkin60}, \cite[Theorem~1.1]{rivlin69}, \cite[Theorem~2.7]{koralov07}) claims that% the periodic, continuous function
~$\periodcostfunction$  can be approximated by a trigonometric sum with arbitrary precision with respect to the uniform norm $\norm{\functionfunction}=\sup\nolimits_{\bl\in\moveperiod{[-\period/2,\period/2]}{[-\period,\period]}}\modulus{\fx{\bl}}
$.
This implies
that for any~$\epsilon>0$ one can find $\order<\infty$ and a trigonometric sum~$\tsnfunction{\order}$ such that $\errn{\order} = \norm{\periodcostx{\bl}-\tsnx{\order}{\bl}}<\epsilon$.
%, 
%
It then holds that $
\costfunction\in\intervx{\costfunction}
=
\tsnfunction{\order}+ 
[-\errn{\order},\errn{\order}] 
$.
\obsolete{ 
It then holds that $\costfunction\in\intervx{\costfunction}$, where we set
\begin{align}\label{approximatedcostfunctionperiodic}
%\begin{array}{r}
&
\intervxy{\costfunction}{\bl} =
\tsnx{\order}{\bl}+ 
[-\errn{\order},\errn{\order}] 
,
&%\hspace{5mm}
\forall \bl\in\Realpluszero.
%\end{array}
\end{align}
}%
Such a trigonometric sum %approximating~$\periodcostfunction$ 
is given the partial Fourier series
\begin{equation}\label{partialfourierseries}
%\textstyle
\tsnx{\order}{\bl}=\sum\nolimits_{k=-\order}^{\order} \fouriercoefk{k}\, \exp{\i\frac{ \moveperiod{2}{} k \pi \bl }{\period}} ,
\end{equation}
where the Fourier coefficients $ \fouriercoefk{k}$ satisfy
\begin{align}\label{fouriercoefficients} 
%\begin{array}{ll}
&
\fouriercoefk{k} = \moveperiod{ \frac{1}{\period} \int\nolimits_{-\frac{\period}{2}}^{\frac{\period}{2}} \periodcostx{\bl} \exp{-\i\frac{ 2 k \pi \bl }{\period}} }{ \frac{1}{2\period} \int\nolimits_{-{\period}}^{{\period}} \periodcostx{\bl} \exp{-\i\frac{  k \pi \bl }{\period}} } \, d\bl ,
&  (k\in\Integer).
%\end{array}
\end{align}
Since $\periodcostfunction$ is real, $\fouriercoefk{0}$ is  real and  $\fouriercoefk{-k}$ and~$\fouriercoefk{k}$ are complex conjugate for all~$k$, so that~$\costnfunction{\order}$ has~$2\order+1$ real parameters.
The Fourier series~\eqref{partialfourierseries}   converges towards~$\periodcostfunction$ with rate 
$ \magnitude{\errn{\order}} = \log(\order)\,  \continuitymodulusSfx{\period}{\periodcostfunction}{\period/(n\pi)}$, where
%
%$
\begin{equation}\label{continuitymodulus}
\continuitymodulusSfx{\period}{\periodcostfunction}{\delta}
=
\sup \{ \modulus{\periodcostx{\bli{1}}-\periodcostx{\bli{2}}}
 :  \bli{1},\bli{2}  \in \Real, \modulus{\bli{1}-\bli{2}}\leq\delta\} 
%\sup\nolimits_{\bli{1},\bli{2}  \in \Real, \modulus{\bli{1}-\bli{2}}\leq\delta} \modulus{\periodcostx{\bli{1}}-\periodcostx{\bli{2}}}
%
%\sup\limits_{\substack{\bli{1},\bli{2}  \in \Real,\\ \modulus{\bli{1}-\bli{2}}\leq\delta}} \modulus{\periodcostx{\bli{1}}-\periodcostx{\bli{2}}}
%.
\end{equation}
%$
is the {modulus of continuity} of the periodic function, \cite[\S{}21]{jackson41}. % scaled back to a $2\pi$-period.
In particular, if for some $\hoeldernexponent\in(0,1]$ the function $\periodcostfunction$ satisfies the $\hoeldernexponent$-H\"oldern condition $\modulus{\periodcostx{\bli{1}}-\periodcostx{\bli{2}}}\leq \hoeldernconstant  \modulus{\bli{1}-\bli{2}}^\hoeldernexponent  
$ for all $\bli{1},\bli{2} \in \Real$,
%($0<\hoeldernexponent \leq 1$), 
then  
%$\continuitymodulusSfx{\period}{\periodcostfunction}{\delta}\leq  \hoeldernconstant (\delta\pi/\period)^\hoeldernexponent$, 
$\continuitymodulusSfx{\period}{\periodcostfunction}{\delta}\leq  \hoeldernconstant \delta^\hoeldernexponent$, 
and% the   trigonometric sum
~\eqref{partialfourierseries}  converges uniformly towards~$\periodcostfunction$. %In particular, when $\hoeldernexponent=1$, $\periodcostfunction$ is Lipschitz continuous.
%
%
%for uniform convergence of apprximations by trigonometric polynomials one needs the function to be periodic and continuous between periods (see convergence of fourier series KEVIN STEPHEN STOTTER CUDDY Corollary 2.15. )
%
%=> $\omega$ function defined across the border.
%
Faster convergence can be obtained by slightly modifying the Fourier coefficients in~\eqref{partialfourierseries}. For this, consider the sum %trigonometric sum 
\begin{equation}
\label{improvedpartialfourierseries}
%\textstyle
\itsnx{\order}{\bl}=\sum\nolimits_{k=-\order}^{\order} \itscoefnk{\order}{\modulus{k}} \, \fouriercoefk{k} \, \exp{\i\frac{  \moveperiod{2}{} k \pi \bl }{\period}} ,
\end{equation}
where $\itscoefnk{\order}{0},\dots,\itscoefnk{\order}{\order}\in\Real$. The choice of parameters proposed in~\cite[\S{}3]{korovkin60},
\obsolete{
\begin{equation}\label{itscoefnk}\begin{array}{ll}
\itscoefnk{\order}{0} =1,
&\\
\itscoefnk{\order}{1} = \cosx{\frac{\pi}{\order+2}},
&\\
\itscoefnk{\order}{k} = \frac{\sum\nolimits_{q=0}^{\order-k}\Sinx{\frac{q+1}{\order+2}\pi} \Sinx{\frac{q+k+1}{\order+2}\pi} }{ \sum\nolimits_{q=0}^{\order} \Sinsquaredx{\frac{q+1}{\order+2}\pi}   } ,
&\quad (k=2,\dots,\order),
\end{array}
\end{equation}
}%
\begin{equation}\label{itscoefnk}\textstyle
\itscoefnk{\order}{0} {\,=\,} 1,
\
\itscoefnk{\order}{1} {\,=\,} \cosx{\frac{\pi}{\order+2}},
\
\itscoefnk{\order}{k} {\,=\,} \frac{\sum\nolimits_{q=0}^{\order-k}\Sinx{\frac{q+1}{\order+2}\pi} \Sinx{\frac{q+k+1}{\order+2}\pi} }{ \sum\nolimits_{q=0}^{\order} \Sinsquaredx{\frac{q+1}{\order+2}\pi}   } 
\textup{ for } k=2,\dots,\order,
\end{equation}
lends~\eqref{improvedpartialfourierseries} the convergence rate
\begin{equation}\label{jacksontrigonometric}
\begin{array}{l}
\errn{\order} \leq 6\,\ContinuitymodulusSfx{\period}{\periodcostfunction}{\frac{\tau}{\order\pi}},
\end{array}
\end{equation}
(see \cite[\emph{first Jackson Theorem}]{korovkin60}, or~\cite[Theorem~1.3]{rivlin69}).
}%

Alternatively, the convergence rate of Corollary~\ref{corollary:trigonometric} can be obtained using polynomials. % Recall the continuous function~$\costfunction$ introduced in the beginning of Section~\ref{section:continuouscostfunctions}, and s
Set~$\thetacostfunction:\Real\mapsto\Real$ to be the continuous, $2\pi$-periodic function defined  by $\thetacostx{\theta}=\costx{\bltheta{\theta}} $, where $\bltheta{\theta} =  ({\taufunction}/{2})\, (1 + \cosx{\theta})$.
%
%
%A better convergence rate for~$\polycostnfunction{\order}$ can be obtained by extending the results  of Section~\ref{section:periodiccostfunctions}  as in~\cite[\S{}1.1]{rivlin69}. Let $\thetacostfunction:\Real\mapsto\Real$ be continuous, periodic with period~$2\pi$, and defined  by $\thetacostx{\theta}=\costx{\bltheta{\theta}} $ where $\bltheta{\theta} =  ({\taufunction}/{2})\, (1 + \cosx{\theta})$.
%
It follows  from the $({\taufunction}/{2})$-Lipschitz continuity of~$\blthetafunction$ and the definition~\eqref{continuitymodulus} of the modulus of continuity, that
$%\begin{equation}\label{continuitymoduluster}
%\textstyle%\begin{array}{l}
\continuitymodulusSfx{\pi}{\thetacostfunction}{\delta}
\leq
\continuitymodulussIfx{\taufunction}{[0,\taufunction]}{\costfunction}{{\tau\delta}/{2}} 
%,\qquad \forall\delta>0.
%\end{equation}
$ for $\delta>0$.
Proceeding as in~\eqref{improvedmirrorpartialfourierseries}, we consider the trigonometric sum for~$\thetacostfunction$ given by the modified Fourier series
%Since~$\thetacostfunction$ is even,  its Fourier coefficients~$\{\fouriercoefk{k} \}$ are real, and  the  trigonometric sum~\eqref{improvedpartialfourierseries} for~$\thetacostfunction$ reduces to
\begin{align}\label{improvedpartialfourierseriesbis}
%\begin{array}{ll}
&
\itsnx{\order}{\theta}= 
%\itscoefnk{\order}{0}  \fouriercoefk{0} + 2 \sum\nolimits_{k=1}^{\order} \itscoefnk{\order}{k}  {\fouriercoefk{k}} \cosx{  k \theta} ,
\sum\nolimits_{k=0}^{\order} \itscoefnk{\order}{k}  \realfouriercoefk{k} \cosx{  k \theta} ,
&
\forall \theta\in\Real,
%\end{array}
\end{align}
where 
$  \realfouriercoefk{0} = 
%\Realpart{ \fouriercoefk{0}}=
\fouriercoefk{0}$,  $\realfouriercoefk{k} 
%= 2 \realpart{\fouriercoefk{k} } 
= 2 \fouriercoefk{k}$ if $k\geq 1$, and
%
%
%
%
%
%From the conclusions of Section~\ref{section:periodiccostfunctions} 
%and we 
$
%\begin{align}\label{fouriercoefficients} 
%\begin{array}{ll}
%&
\fouriercoefk{k} = \moveperiod{ \frac{2}{\pi} \int\nolimits_{0}^{\frac{\pi}{2}} \periodcostx{\theta} \cosx{ 2 k  \theta } }{ ({1}/{\pi}) \int\nolimits_{0}^{{\pi}} \periodcostx{\theta} \cosx{k\theta}} \, d\theta 
%&  (k\in\Natural).
%\end{array}
%\end{align}
 $ for $k\in\Natural$.
Defining~$\itscoefnk{\order}{k}$ as in~\eqref{itscoefnk},
yields the uniform 
convergence rate  %for~$\itsnfunction{\order}$
%The uniform convergence rate for~(\refeq{improvedpartialfourierseriesbis}) is given by~(\refeq{jacksontrigonometric}) as 
\begin{equation}\label{jacksontrigonometricbis}
%\textstyle
%\begin{array}{l}
%\errn{\order}
\norm{\costfunction-\itsnfunction{\order}}
=
%\sup\nolimits_{\bl\in[0,\period]}{\modulus{\costx{\bl}-\itsnx{\order}{\bl}}}
%\leq
\refereq{\eqref{mirrorjacksontrigonometric}}{\leq}
6\,\ContinuitymodulusSfx{\pi}{\thetacostfunction}{  1 / \order} 
%\refereq{(\refeq{continuitymoduluster})}{\leq} 
\leq
6\, \ContinuitymodulussIfx{\taufunction}{[0,\taufunction]}{\costfunction}{{\taufunction}/({2\order})} .
%\end{array}
\end{equation} 
 \obsolete{
 $  \cosx{  k \theta} = \Realpart{\cosx{  k \theta}+ \i\sinx{  k \theta}} = \Realpart{(\cosx{   \theta}+ \i\sinx{   \theta})^k} = \Realpart{\sum\nolimits_{t=0}^{k}  \binomialcoef{k}{t} \i^{t}\cosnx{k-t}{   \theta} \sinnx{t}{   \theta}} = \sum\nolimits_{t=0}^{\floor{k/2}}  \binomialcoef{k}{2t}   (-1)^{t}\cosnx{k-2t}{   \theta} \sinnx{2t}{   \theta} =  \sum\nolimits_{t=0}^{\floor{k/2}}   \binomialcoef{k}{2t}   (-1)^{t}\cosnx{k-2t}{   \theta} (1-\cosnx{2}{   \theta}  )^t =  \sum\nolimits_{t=0}^{\floor{k/2}}   \binomialcoef{k}{2t}    (-1)^{t}  \cosnx{k-2t}{   \theta}   \sum\nolimits_{l=0}^{t}  \binomialcoef{t}{l} (-1)^{t-l} \cosnx{2(t-l)}{   \theta} =  \sum\nolimits_{t=0}^{\floor{k/2}}   \binomialcoef{k}{2t}         \sum\nolimits_{l=0}^{t}  \binomialcoef{t}{l} (-1)^{l} \cosnx{k-2l}{   \theta} $
 % q=k-2l
}%
It remains to rewrite~\eqref{improvedpartialfourierseriesbis} as a polynomial in~$\bl$ by returning to the backlog domain. For this, we develop $ \cosx{  k \theta}  = \Realpart{(\cosx{   \theta}+ \i\sinx{   \theta})^k}$ and find  $\cosx{  k \theta} =\polcosnx{k}{\cosx{\theta}}$, where the polynomial~$ \polcosnx{ k}{x} $, characterized by its~$k$ real roots located in~$(-1,1)$, is defined by
\obsolete{
$\polcosnx{ k}{x} =  \sum\nolimits_{q=0}^{{k}/{2}}   \coscoef{k}{q} 
 x^{2q} $ if~$k$ is even and
$\polcosnx{ k}{x} = \sum\nolimits_{q=0}^{({k-1})/{2}}  \coscoef{k}{q}  x^{2q+1} $ if~$k$ is odd.
}%
\obsolete{
 \begin{equation}\label{polycostnx}
 \cosx{  k \theta} = \left|
 \begin{array}{ll}
 \sum\nolimits_{q=0}^{{k}/{2}} 
 % \left[ (-1)^{\frac{k}{2}-q}  \sum\nolimits_{t=\frac{k}{2}-q}^{\frac{k}{2}} \binomialcoef{k}{2t} \binomialcoef{\ \ t}{\frac{k}{2}-q} \right]
 \coscoef{k}{q} 
 \cosnx{2q}{\theta}           
 ,
 & \textup{if~$k$ is even,}
\\
 \sum\nolimits_{q=0}^{({k-1})/{2}} 
 %\left[ (-1)^{\frac{k-1}{2}-q}  \sum\nolimits_{t=\frac{k-1}{2}-q}^{\frac{k-1}{2}} \binomialcoef{k}{2t} \binomialcoef{\ \ t}{\frac{k-1}{2}-q} \right]
 \coscoef{k}{q}
  \cosnx{2q+1}{\theta}           
,
 & \textup{if~$k$ is odd,}
 \end{array}\right.
\end{equation}
}%
\begin{equation}\label{polycostnx}
 \polcosnx{ k}{x} = \left\{
 \begin{array}{ll} 
 \sum\nolimits_{q=0}^{{k}/{2}} 
  \coscoef{k}{q}  
 x^{2q}           
 ,
 &\quad \textup{if~$k$ is even}
\\
 \sum\nolimits_{q=0}^{({k-1})/{2}} 
 \coscoef{k}{q}
  x^{2q+1}           
,
 & \quad\textup{if~$k$ is odd}
 \end{array}\right\}
 ,
\end{equation}
\storecompoundcounter{equation}{polycostnx}%
where~$\coscoef{0}{0}=1$, and
\begin{equation}\label{coscoef}\notag
\textstyle
\coscoef{k}{q}
 =
 % (-1)^{\floor{\frac{k}{2}}-q}  \sum\nolimits_{t=\floor{\frac{k}{2}}-q}^{\floor{\frac{k}{2}}} \binomialcoef{k}{2t} \binomialcoef{\ \ \, t}{\floor{\frac{k}{2}}-q},
  (-1)^{\floor{\frac{k}{2}}-q}  \sum\nolimits_{t=0}^{q} \Binomialcoef{\ \ \ \ k}{2(\floor{\frac{k}{2}}-t)} \Binomialcoef{\floor{\frac{k}{2}}-t}{\floor{\frac{k}{2}}-q},
  \quad(q=0\dots,\Floor{\frac{k}{2}},\, k\in\Naturalpos ).  
\end{equation}
%
%The polynomial~$ \polcosnx{ k}{x} $ is characterized by its~$k$ real roots located in~$(-1,1)$.
%$\coscoef{k}{0} = (-1)^{\floor{\frac{k}{2}}}   \binomialcoef{k}{2(\floor{\frac{k}{2}})} $, 
Since $\thetacostx{\theta}=\costx{\taufunction (1 + \cosx{\theta})/2} $, a polynomial approximation of~$\costfunction$ on~$[0,\taufunction]$ is obtained by setting $\cosx{  k \theta} =\polcosnx{k}{2\bl/\taufunction-1}$  in~\eqref{improvedpartialfourierseriesbis}, and we find, after straightforward computations,
\begin{align}\label{improvedpartialfourierseriester}
%\begin{array}{ll}
&
\ptsnx{\order}{\bl}
\obsolete{
=
\sum\nolimits_{k=0}^{\order} \itscoefnk{\order}{k}  \realfouriercoefk{k} \, \polcosnx{k}{\frac{2\bl}{\taufunction}-1} 
}
% = 
%  \sum\nolimits_{t=0}^{\order} \polcoscoef{\order}{t}\left(\frac{2\bl}{\taufunction}-1\right)^t
  =
  \sum\nolimits_{k=0}^{\order} 
%\left[\frac{ \sum\nolimits_{t=k}^{\order} \binomialcoef{t}{k}  (-1)^{t}\polcoscoef{\order}{t}}{\left(\frac{-\taufunction}{2}\right)^k}\right]
\translatedpolcoscoef{\order}{k}
 \bl^k
  ,
  & \forall \bl\in[0,\taufunction],
%\end{array}
\end{align}
%with the coefficients~$\translatedpolcoscoef{\order}{k}$ reducing to
where we define
\begin{align}\label{translatedpolcoscoef}\notag
%\textstyle
&\translatedpolcoscoef{\order}{k} = 
%\left(\frac{2}{\taufunction}\right)^k \sum\nolimits_{t=0}^{\order-k} \binomialcoef{t+k}{\ k}  (-1)^{t}\polcoscoef{\order}{t+k}
({2}/{\taufunction})^k \, \sum\nolimits_{t=0}^{\order-k} \Bigbinomialcoef{t+k}{\ k} \, (-1)^{t} \, \polcoscoef{\order}{t+k}
  ,&
  (k=0\dots,\order).
\end{align}
%for~$k=0\dots,\order$, with
%where we define
and
%\begin{equation}\label{polcoscoef}
$%\begin{array}{ll}
\polcoscoef{\order}{t} = 
%\sum\nolimits_{k\in\funnyset{\order}{t}}  \itscoefnk{\order}{k}  \fouriercoefk{k}   \coscoef{k}{\floor{t/2}} ,
\sum\nolimits_{k\in\funnyset{\order}{t}}  \itscoefnk{\order}{k} \, \realfouriercoefk{k} \,   \coscoef{k}{\floor{t/2}} 
$ for $t=0,\dots,\order$, %,
% & \quad(t=0,\dots,\order),
% \end{array}
%\end{equation}
% $ \polcoscoef{\order}{t} = \sum\nolimits_{k\in\funnyset{\order}{t}}  \itscoefnk{\order}{k}  \fouriercoefk{k}   \coscoef{k}{\floor{t/2}} $
%\end{array}
%\end{equation}
in which %we write 
$\funnyset{\order}{t}  = \{ t,t+2,t+4,\dots,\order \}$ if~$\order-t$ is even, and $\funnyset{\order}{t}  = \{ t,t+2,t+4,\dots,\order-1 \}$ otherwise ($0\leq t\leq\order$). 
As for the Fourier coefficients of~$\thetacostfunction$, they reduce to
%$  \realfouriercoefk{0} = \Realpart{ \fouriercoefk{0}}= \fouriercoefk{0}$, and $\realfouriercoefk{k} = 2 \realpart{\fouriercoefk{k} } = 2 \fouriercoefk{k}$ for $k\geq 1$, with %$
%$\fouriercoefk{0},\dots,\fouriercoefk{\order}$ satisfying %, for~$ k=0,\dots,\order$,
\obsolete{
\begin{equation}\label{realfouriercoefficients} \begin{array}{ll}
%\realfouriercoefk{k} 
\realpart{\fouriercoefk{k} }\refereq{\eqref{fouriercoefficients}}{=}
%\frac{1}{2\pi} \int\nolimits_{-{\pi}}^{{\pi}} \periodcostx{\theta} \cosx{ k  \theta}  \, d\theta,
\frac{1}{\pi} \int\nolimits_{0}^{{\pi}} \periodcostx{\theta} \cosx{ k  \theta}\, d\theta
,
&\quad ( k\in\Naturalnn).
\end{array}
\end{equation}
The Fourier coefficients~\eqref{realfouriercoefficients} reduce, on the other hand, to
}%
\begin{equation}
\label{realfouriercoefficientstwo} 
%\begin{array}{l}
%\realfouriercoefk{k}
%\realpart{\fouriercoefk{k}}
\fouriercoefk{k}
=%\refereq{\eqref{fouriercoefficients}}{=}
\obsolete{
%\frac{1}{2\pi} \int\nolimits_{-{\pi}}^{{\pi}} \periodcostx{\theta} \cosx{ k  \theta}  \, d\theta,
\frac{1}{\pi} \int\nolimits_{0}^{{\pi}} \periodcostx{\theta} \cosx{ k  \theta}\, d\theta
 =
}
%\refereq{(\refeq{realfouriercoefficients})}{=}
\frac{1}{\pi} \int\nolimits_{0}^{{\pi}} 
%\bigcostx{\frac{\taufunction (1 + \cosx{\theta})}{2}}  
\costx{\bltheta{\theta}}  
\cosx{ k  \theta}\, d\theta
%\quad\\\hfill
\refereq{\eqref{polycostnx}}{=}
%- \frac{1}{\pi} \int\nolimits_{\taufunction}^{0} \Costx{\arCosx{\frac{2\bl}{\taufunction}-1}}    \frac{\cosx{ k  \theta}}{\sqrt{\bl(\taufunction-\bl)}} \, d\bl
 \frac{1}{\pi} \int\nolimits_{0}^{\taufunction}  \costx{\bl}    \frac{\polcosnx{k}{\frac{2\bl}{\taufunction}-1}}{\sqrt{\bl(\taufunction-\bl)}} \, d\bl
,
%\end{array}
\end{equation}
\storecompoundcounter{equation}{realfouriercoefficientstwo}%
where we have used the change of variable  
%$\bl=\taufunction (1 + \cosx{\theta})/2$.  
$\bl=\bltheta{\theta}$. 
For many cost functions, the  coefficients~$\{\fouriercoefk{k} \}$ can be derived exactly. 
See Lemma~\ref{lemma:fouriercoefsforquotientofpolynomials} for expressions of these coefficients  in the case when $\costfunction$ is given as a quotient of polynomials.
%Expressions for these coefficients  are given in Lemma~\ref{lemma:fouriercoefsforquotientofpolynomials}  in the case when $\costfunction$ is given as a quotient of polynomials.
%
%

From~\eqref{jacksontrigonometricbis}, we infer the following bounds for the value function.

%Because the uniform convergence rate~\eqref{jacksontrigonometricbis} holds for~\eqref{improvedpartialfourierseriester}, we infer the following bounds for the value function.
%  
%   Corollary~\ref{corollary:trigonometric}
\begin{corollary}[Near-optimal polynomials]\label{corollary:trigonometric}
% Proposition~\ref{proposition:polynomial} holds for~$\polycostnfunction{\order}\equiv\ptsnfunction{\order}$, where~$\ptsnfunction{\order}$ is defined  by~\eqref{improvedpartialfourierseriester}, with the uniform error bound~$\errn{\order} = 6\, \continuitymodulussIfx{\taufunction}{[0,\taufunction]}{\costfunction}{{\taufunction}/({2\order})}$.
Proposition~\ref{proposition:polynomial} holds for~$\polycostnfunction{\order}\equiv\ptsnfunction{\order}$ defined  by~\eqref{improvedpartialfourierseriester} with  uniform error bound~$\errn{\order} = 6\, \continuitymodulussIfx{\taufunction}{[0,\taufunction]}{\costfunction}{{\taufunction}/({2\order})}$.
\end{corollary}

\noindent
Without further assumptions on~$\costfunction$, the convergence rate~$\magnitude{\continuitymodulussIfx{\taufunction}{[0,\taufunction]}{\costfunction}{{\taufunction}/({2\order})}}$ guaranteed by~\eqref{improvedpartialfourierseriester} is non-improvable. 
The performance of~$\mirroritsnfunction{\order}$ %in~\eqref{improvedmirrorpartialfourierseries} 
and~$\ptsnfunction{\order}$ %in~\eqref{improvedpartialfourierseriester} 
in Corollaries~\ref{corollary:fourier} and~\ref{corollary:trigonometric} are then really close,  and the choice of either approach (Section~\ref{section:trigonometricsum} or~\ref{section:nearoptimalpolynomials}),  mostly dependent on the computability of the Fourier coefficients~\eqref{mirrorfouriercoefficients} or~\eqref{realfouriercoefficientstwo}, respectively,
%the trigonometric sum~\eqref{improvedmirrorpartialfourierseries} and the polynomial~\eqref{improvedpartialfourierseriester} 
is left to the appreciation of the reader.
 The second approach nevertheless prevails in the event  the cost function %is $k$~times differentiable on~$[0,\taufunction]$, 
has a $k$th derivative~$\dncostfunction{k}$ on~$[0,\taufunction]$. Then, the convergence rate in Corollary~\ref{corollary:trigonometric} can be lowered to~$\magnitude{\order^{-k}\continuitymodulussIfx{\taufunction}{[0,\taufunction]}{\dncostfunction{k}}{{\taufunction}/[{2(\order-k)}]}}$ by using the derivatives as the targets of approximation, \cite[Theorem~1.5]{rivlin69}. %, with the possible risk of increased computational complexity for  the coefficients of the polynomial. 
This distinguishing property of approach~\ref{section:nearoptimalpolynomials} stems from the fact that~$\thetacostx{\theta}$ retains the smoothness of the cost function, whereas~$\mirrorcostx{\bl}$ shows irregularities at $\bl=(2k+1)\taufunction$.
We refer to~\cite[\S{}1.1]{rivlin69}  and references therein for further considerations on the optimality  of~\eqref{jacksontrigonometricbis} as a convergence rate for polynomial approximations.
%In the particular case $x^{n+1}$ by $P_n(x)$ then Chebychev, allows for reduction of order of polynomials.

\paragraph{Case study: quotient cost function.}%\label{example:quotientexample}
Let
$%\begin{equation}\label{polynomialexample}\begin{array}{ll}
\costx{\bl}={\bl^2}/({\sing^2+\bl^2}) %, \quad & \forall \bl\in\Realpluszero,
$, %\end{array}\end{equation}
where $\sing>0$ is a positive parameter.
The Fourier coefficients~\eqref{realfouriercoefficientstwo} for~$\costfunction$ are given by  Lemma~\ref{lemma:fouriercoefsforquotientofpolynomials} %in  Appendix~\ref{appendix:complement} 
with $\polydegreek{k}\equiv k$. After computation of the residues at the complex conjugate poles~$\i\sing$ and~$-\i\sing$,  \eqref{solutionrealfouriercoefk} reduces to
\begin{equation} \label{solutionrealfouriercoefkexample}\notag
\begin{array}{ll}
%\realfouriercoefk{k}
%\realpart{\fouriercoefk{k}}
\fouriercoefk{k}
=
\sqrt{\pi} \sum\nolimits_{q=0}^{k}     \frac{ \laurentcoef{-q}\,  (-\taufunction)^{q} }{ \factorial{q} \Gammax{\frac{1}{2}-q}} 
- 
 \frac{\sqrt{\sing}}{\sqrt[4]{\sing^2+\taufunction^2}} \sum\nolimits_{q=0}^{\floor{\frac{k}{2}}}  
 \evenoddcoefksingtauq{k}{\sing}{\taufunction}{q} \, \coscoef{k}{q},
 & \quad (k\in\Naturalnn),
\end{array}
\end{equation}
 where~$\{\laurentcoef{-q}\}_{q=0}^{k}$ are  the first~$k+1$  coefficients 
 (i.e., those associated with nonnegative powers) 
 of the Laurent series at~$+\infty$ of~$\costx{\bl} \, \polcosnx{k}{{2\bl}/{\taufunction}-1} $, 
 % i.e., the coefficients associated with nonnegative powers in~(\refeq{laurentallpolycostnx}), 
 equal in this example  to
\begin{equation}\label{laurentcoefexample}\notag
\laurentcoef{-q}
\refereq{\eqref{laurentcoef}}{=}
\left\{
 \begin{array}{ll}
 (\frac{-2}{\taufunction})^{q}  \sum\nolimits_{l=\ceil{\frac{q}{2}}}^{\frac{k}{2}} 
 \big[
 \sum\nolimits_{t= 0 }^{l-\ceil{\frac{q}{2}}} 
  \binomialcoef{\, 2l}{q+2t}  
\big(\frac{-4\sing^{2}}{\taufunction^2}\big)^{t}   
 \big]
 \coscoef{k}{l}    
 ,
 &\textup{if~$k$ even }
\\
 - (\frac{-2}{\taufunction})^{q}  
  \sum\nolimits_{l=\ceil{\frac{q-1}{2}}}^{\frac{k-1}{2}} 
 \big[
 \sum\nolimits_{t= 0}^{l-\ceil{\frac{q-1}{2}}} 
   \binomialcoef{2l+1}{q+2t}   
\big(\frac{-4\sing^{2}}{\taufunction^2}\big)^{t} 
   \big]
    \coscoef{k}{l}
, \quad
 &\textup{if $k$ odd}
 \end{array}\right\}
 \end{equation}
 \obsolete{
 \begin{equation}\label{laurentcoefexample}
\laurentcoef{-q}
\refereq{(\eqref{laurentcoef})}{=}
\left|
 \begin{array}{ll}
 (\frac{-2}{\taufunction})^{q}  \sum\nolimits_{j=0}^{
 %\frac{k}{2}-\ceil{\frac{q}{2}}
 \floor{\frac{k-q}{2}}
 } 
 \left[
 \sum\nolimits_{t= 0 }^{j} 
  \binomialcoef{2\ceil{\frac{q}{2}}+2j}{\ \, q+2t}  
\left(\frac{-4\sing^{2}}{\taufunction^2}\right)^{t}   
 \right]
 \Coscoef{k}{j+\ceil{\frac{q}{2}}}    
 ,
 &\textup{if~$k$ even, }
\\
 - (\frac{-2}{\taufunction})^{q}  
  \sum\nolimits_{j=0}^{
%  \frac{k-1}{2}-\ceil{\frac{q-1}{2}}
\floor{\frac{k-q}{2}}
  } 
 \left[
 \sum\nolimits_{t= 0}^{j} 
   \binomialcoef{2\ceil{\frac{q-1}{2}}+2j+1}{\ \ \ \, q+2t}   
\left(\frac{-4\sing^{2}}{\taufunction^2}\right)^{t} 
   \right]
    \Coscoef{k}{j+\ceil{\frac{q-1}{2}}}
, \quad
 &\textup{if $k$ odd},
 \end{array}\right. 
 \end{equation}
 }%
%are  the  coefficients of the Laurent series at~$+\infty$ of the analytic continuation of~$\allpolycostnfunction{k}$, 
where we used~$\{{\evenoddcoefksingtauq{k}{\sing}{\taufunction}{q} }\}_{q=0}^{\floor{\frac{k}{2}}}$, defined for~$k$ even and $q=0,\dots,\frac{k}{2}$ by
\begin{equation}\notag
% \begin{array}{rl}
\textstyle
\evenoddcoefksingtauq{k}{\sing}{\taufunction}{q} 
 =
 \bigcosx{\frac{\anglesingtau{\sing}{\tau}}{2}} \sum\nolimits_{t=0}^{q} \binomialcoef{2q}{2t}
 \big(\frac{-4\sing^2}{\taufunction^2}\big)^{t}    
% \hspace{20mm}
% &
% \\
 % +  \bigsinx{\frac{\anglesingtau{\sing}{\tau}}{2}} \left(\frac{-2\sing}{\taufunction}\right) \sum\nolimits_{t=0}^{q-1} \binomialcoef{2q}{2t+1}\big(\frac{-4\sing^2}{\taufunction^2}\big)^{t}
- \frac{2\sing}{\taufunction} \bigsinx{\frac{\anglesingtau{\sing}{\tau}}{2}}  \sum\nolimits_{t=0}^{q-1} \binomialcoef{2q}{2t+1}\big(\frac{-4\sing^2}{\taufunction^2}\big)^{t}
%,\quad &
%(k \text{ even},\ q=0,\dots,\frac{k}{2})
,
% \end{array}
\end{equation}
 and for~$k$ odd and $q=0,\dots,\frac{k-1}{2}$ by
\begin{equation}\notag
\textstyle
 %\begin{array}{rl}
 \evenoddcoefksingtauq{k}{\sing}{\taufunction}{q}
 {=}{-}\bigcosx{\!\frac{\anglesingtau{\sing}{\tau}}{2}\!} \!\sum\nolimits_{t=0}^{q}\! \binomialcoef{{2q{+}1}}{\ \, 2t}
 \big(\!\frac{{-}4\sing^2}{\taufunction^2}\!\big)^{t} 
 %\hspace{10mm}
 %&
 %\\
 %- \bigsinx{\frac{\anglesingtau{\sing}{\tau}}{2}} \big(\frac{-2\sing}{\taufunction}\big)  \sum\nolimits_{t=0}^{q} \binomialcoef{2q+1}{2t+1}\big(\frac{-4\sing^2}{\taufunction^2}\big)^{t}     
  {+}{\frac{2\sing}{\taufunction}} \bigsinx{\!\frac{\anglesingtau{\sing}{\tau}}{2}\!}  \! \sum\nolimits_{t=0}^{q}\! \binomialcoef{2q{+}1}{2t{+}1}\big(\!\frac{-4\sing^2}{\taufunction^2}\!\big)^{t}
  ,
%,\quad &
%(k \text{ odd}, \ q=0,\dots,\frac{k-1}{2}),
% \end{array}
\end{equation}
in which $ \anglesingtau{\sing}{\tau} = 
%{\pi}/{2} - \atanx{{\sing}/{\taufunction}}  
 \atanx{\frac{\taufunction}{\sing}}  
  $. % is the acute angle between the line segments connecting the points~$\i$ and~$0$ and the points~$\i$ and~$\taufunction$, respectively.
 %
%
%%%%%%%
\begin{figure}[!t]
\centering
%\begin{subfigure}%{3cm}%{.5\columnwidth}
\subfigure[Cost function intervals.\label{figure:costfinterval}]{
        \centering
        \includegraphics{\setpathfigs%
        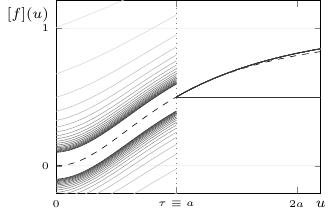}%standalone-korovkin-graphs-fcost}
}% 
%\caption{Cost function intervals.}\label{figure:costfinterval}
%\end{subfigure}
%
\subfigure[Value function intervals (exponentially distributed service times; $\strate=2$, $\ar=1$).\label{figure:vfinterval}]{
        \centering
        \includegraphics{\setpathfigs%
        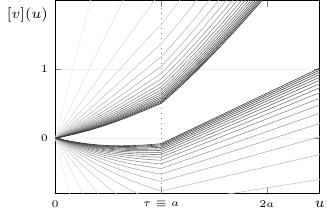}%standalone-korovkin-graphs-vf}
}%
\caption{%
%
%Example~\ref{example:quotientexample}: 
%Case study: 
Intervals for $\costx{\bl}={\bl^2}/({\sing^2+\bl^2})$ as per  Corollary~\ref{corollary:trigonometric} (%$\taufunction\equiv\sing$, 
$\order=1,\dots,20$).
} \label{figure:interval}
\end{figure}%
%
%% CONTINUITY MODULUS
%
\obsolete{
$\continuitymodulussIfx{\taufunction}{[0,\taufunction]}{\costfunction}{\delta} :$

$ d/d\bl  [ \frac{(\bl+\delta)^2}{\sing^2+(\bl+\delta)^2}  \frac{\bl^2}{\sing^2+\bl^2}   ]  
=
2\frac{(\bl+\delta)[\sing^2+(\bl+\delta)^2]-(\bl+\delta)^2(\bl+\delta)}{[\sing^2+(\bl+\delta)^2]^2} - 2\frac{\bl[\sing^2+\bl^2]-\bl^2 \bl}{[\sing^2+\bl^2]^2}
=
2\frac{\sing^2(\bl+\delta)]}{[\sing^2+(\bl+\delta)^2]^2} - 2\frac{\sing^2\bl}{[\sing^2+\bl^2]^2}
$

$
=
\frac{2}{[\sing^2+\bl^2]^2[\sing^2+(\bl+\delta)^2]^2}
\left(
\sing^2(\bl+\delta)] [\sing^2+\bl^2]^2- \sing^2\bl [\sing^2+(\bl+\delta)^2]^2
\right)
$

$ \sing^2(\bl+\delta) [\sing^2+\bl^2]^2- \sing^2\bl [\sing^2+(\bl+\delta)^2]^2
=
\sing^2\delta [\sing^2+\bl^2]^2 + \sing^2\bl \{ [\sing^2+\bl^2]^2 - [\sing^2+\bl^2 + (2\delta\bl+\delta^2)]^2\}
=
\sing^2\delta [\sing^2+\bl^2]^2 - \sing^2\bl [ 2\sing^2+2\bl^2 + 2\delta\bl+\delta^2 ] [2\delta\bl+\delta^2]  
=
- \sing^2\bl  [2\delta\bl+\delta^2 ]^2 + 
[\sing^4\delta - 3\sing^2 \delta\bl^2  - 2\sing^2 \delta^2\bl  ]
[\sing^2+\bl^2]
=
 - \sing^2  \delta^4 \bl - 4 \sing^2 \delta^3\bl^2  - 4\sing^2 \delta^2\bl^3   
+ \sing^6\delta  -3 \sing^4\delta  \bl^2 - 2 \sing^4\delta^2\bl  
+ \sing^4\delta\bl^2  -3 \sing^2\delta  \bl^4 - 2 \sing^2\delta^2 \bl^3
$

$
=
- 3 \sing^2\delta [
\bl^4
+2 \delta \bl^3  
+2 (\sing^2+ 2 \delta^2)/3    \bl^2  
 + ( \delta^3 + 2 \sing^2\delta)/3  \bl 
 - \sing^4 /3 
 ]
$

Consider : $ ( (\bl+\delta/2)^2)^2 + \alpha (\bl+\delta/2)^2 + \beta $

$= 
\bl^4+2\delta\bl^3+(\alpha+3\delta^2/2)\bl^2+ (\alpha\delta+\delta^3/2)\bl + \beta +\alpha\delta^2/4+\delta^4/16$

Take $(\alpha+3\delta^2/2)=2 (\sing^2+ 2 \delta^2)/3 $
or $\alpha=(4 \sing^2-\delta^2)/6= 2(\sing^2-(\delta/2)^2)/3$

Then $\alpha\delta+\delta^3/2=(4 \sing^2\delta-\delta^3)/6+3\delta^3/6=(\delta^32\sing^2\delta)/3$ OK!

Lastly $\beta +\alpha\delta^2/4+\delta^4/16= - \sing^4 /3 $
or
$\beta =- [ (\delta/2)^4 + 2\sing^2(\delta/2)^2 +  \sing^4 ]/3 = - [\sing^2+(\delta/2)^2]^2/3 $

Compute the delta:
$\Delta = \alpha^2 - 4 \beta =  4(\sing^2-(\delta/2)^2)^2/9 + 12 [\sing^2+(\delta/2)^2]^2/9 =   [4 \sing^4 - 8 \sing^2 (\delta/2)^2 + 4 (\delta/2)^4   + 12 \sing^4   + 24\sing^2 (\delta/2)^2 + 12(\delta/2)^4 ] /9
=
  [16 \sing^4 +16 \sing^2 (\delta/2)^2 + 16 (\delta/2)^4  ] /9
  =
  16 [ \sing^4 + \sing^2 (\delta/2)^2 +  (\delta/2)^4  ] /9
$

$\Delta = \alpha^2 - 4 \beta =\frac{4 ( \sing^2-\delta^2/4)^2 + 12 (\delta^2/4 + \sing^2)^2}{9} = \frac{ (4 \sing^2-\delta^2)^2 + 3 (\delta^2 + 4 \sing^2)^2}{36} $
$
=  \frac{ 64 \sing^4 +16 \sing^2\delta^2+4\delta^4 }{36}
= \frac{ 16 \sing^4 +4 \sing^2\delta^2+\delta^4 }{9}
$ >0

There is a negative solution and a larger solution

$\frac{-\alpha+\sqrt{\Delta}}{2}=\frac{-\left( \sing^2-\left(\frac{\delta}{2}\right)^2\right)+2\sqrt{  \sing^4 + \sing^2\left(\frac{\delta}{2}\right)^2+\left(\frac{\delta}{2}\right)^4 }}{3}$

We find
$(\bl+\delta/2)^2 = =\frac{-\left( \sing^2-\left(\frac{\delta}{2}\right)^2\right)+2\sqrt{  \sing^4 + \sing^2\left(\frac{\delta}{2}\right)^2+\left(\frac{\delta}{2}\right)^4 }}{3}$
or
$(\bl+\delta/2)^2  =\frac{-\left( \sing^2-\left(\frac{\delta}{2}\right)^2\right)+2\sqrt{  \sing^4 + \sing^2\left(\frac{\delta}{2}\right)^2+\left(\frac{\delta}{2}\right)^4 }}{3}$
or
\begin{equation}\label{inflexion}\begin{array}{c}
\inflexionpoint{\delta} = 
\sqrt{\frac{\left(\frac{\delta}{2}\right)^2 - \sing^2+2\sqrt{  \sing^4 + \sing^2\left(\frac{\delta}{2}\right)^2+\left(\frac{\delta}{2}\right)^4 }}{3} } -\frac{\delta}{2}
.
\end{array}\end{equation}
}%
%
%
%
%By locating the maximum of $\costx{\bl+\delta/2} - \costx{\bl+\delta/2}$,  one shows that t
The continuity modulus of~$\costfunction$ on~$[\delta/2,\taufunction-\delta/2]$ is given, for $\delta \in[0,{\taufunction}/{2}]$, by
$\continuitymodulussIfx{\taufunction}{[0,\taufunction]}{\costfunction}{\delta} = \costx{\inflexionpoint{\delta}+\delta/2} - \costx{\inflexionpoint{\delta}-\delta/2} $,
where
\begin{equation}\label{inflexion}\notag\textstyle
\inflexionpoint{\delta} = 
\min\Big\{
\sqrt{\frac{(\frac{\delta}{2})^2 - \sing^2+2\sqrt{  \sing^4 + \sing^2(\frac{\delta}{2})^2+(\frac{\delta}{2})^4 }}{3} } 
,\taufunction-\frac{\delta}{2}
\Big\}
\end{equation}
\obsolete{
$\left(\frac{\delta}{2}\right)^2 - \sing^2+2\sqrt{  \sing^4 + \sing^2\left(\frac{\delta}{2}\right)^2+\left(\frac{\delta}{2}\right)^4 } < 3 \left(\frac{\delta}{2}\right)^2$

$ \sing^4 + \sing^2\left(\frac{\delta}{2}\right)^2+\left(\frac{\delta}{2}\right)^4  < \left( \left(\frac{\delta}{2}\right)^2 + \sing^2/2 \right)^2$

$ \sing^4 + \sing^2\left(\frac{\delta}{2}\right)^2+\left(\frac{\delta}{2}\right)^4  <  \left(\frac{\delta}{2}\right)^4 + \left(\frac{\delta}{2}\right)^2 \sing^2 + \sing^4  /4$

$ \sing^4  <  \sing^4  /4$ impossible. So square root $>\delta/2$.
}%
satisfies ${\delta}/{2}< \inflexionpoint{\delta} \leq \taufunction -{\delta}/{2}$. The cost function~$\costfunction$ is approximated by~\eqref{approximatedcostfunction}, with~$\polycostnfunction{\order}\equiv\ptsnfunction{\order}$ given by~\eqref{improvedpartialfourierseriester}
and~$\intervx{\excfunction}$ set to
\begin{equation}\label{excfunctionpolynomialexample}\notag
\begin{array}{ll}
\intervxy{\excfunction}{\bl}=\bigintervx{\costx{\taufunction},-\left(1-\costx{\taufunction}\right) \exponential\big\{{-\big(\frac{\dcostx{\taufunction}}{1-\costx{\taufunction}}\big) (\bl-\taufunction)}\big\}},
&\quad \bl\geq\taufunction,
\end{array}\end{equation}
in which $\dcostx{\taufunction}=2\sing^2\taufunction(\sing^2+\taufunction^2)^{-2}$.
The intervals produced by Corollary~\ref{corollary:trigonometric} for~$\costfunction$, and for its value function in the presence of jobs with exponentially-distributed service times are displayed in  Figure~\ref{figure:interval}, for fixed~$\taufunction$ and~$\order=1,\dots,20$. %The cost function intervals combine the above developments with~(\refeq{approximatedcostfunction}) and  the near-optimal uniform approximation~(\refeq{improvedpartialfourierseriester}), while the 
The value function intervals shown in Figure~\ref{figure:vfinterval} followed from%in accordance with  Corollary~\ref{corollary:trigonometric} 
%using the near-optimal uniform approximation
~(\ref{improvedpartialfourierseriester}) %, the implementation of which exploited 
%from
and  the developments of  Examples~\ref{example:piecewiseanalyticexp} and~\ref{example:piecewiseexp}.
The interval gaps can be arbitrarily reduced by increasing both~$\order$ and~$\taufunction$, as in Algorithm~\ref{algorithm:dispatching}.

Consider a system of two parallel servers~$1$ and~$2$, with server~$1$  twice faster than~$2$. Feed the system a sequence of jobs with arrival rate~$\ar=3/2$ and service times exponentially distributed with parameters  $ \vect{\strate_{1},\strate_{2}}=\vect{2,1}$. Assume that the workload is initially balanced between the two servers, i.e.,  $\vect{\ari{1},\ari{2}}=\vect{1,1/2}$, and let $\costx{\bl}={\bl^2}/({\sing^2+\bl^2})$. % with~$\sing=1$, 
\begin{figure}[!t]
\centering
\subfigure[$\osipolicyxy{\bl}{\st}$ for $\st=\vect{1,2}$.\label{figure:policyservers}]{
        \centering
        \includegraphics{\setpathfigs%
        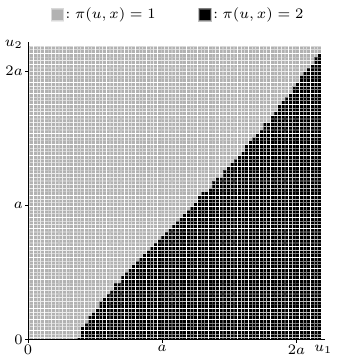}%approximatevaluefunction-korovkin-server-fig}
\obsolete{
\begin{tikzpicture}[scale=1]
\def\sizex{5}
\def\sizey{5}
\def\minx{0}
\def\dx{0.1}
\def\miny{0}
\def\dy{0.1}
\def\xstep{0.2}
\def\ystep{0.2}
\def\greymin{0.1}
\def\greymax{1.0}
\def\colone{0.7}
\def\coltwo{0.0}
\def\colframe{1.0}
\definecolor{serverone}{gray}{\colone};
\definecolor{servertwo}{gray}{\coltwo};
\definecolor{frame}{gray}{\colframe};
\def\edge{0.4}
\tikzfading[name=fade left, left color=transparent!0,right color=transparent!100]
\tikzfading[name=fade right, right color=transparent!0,left color=transparent!100]
\tikzfading[name=fade bottom, bottom color=transparent!0,top color=transparent!100]
\tikzfading[name=fade top, top  color=transparent!0,bottom color=transparent!100]
\makesquare{0}{0}{\dx/2}{\dy/2}{\edge}{serverone}{frame}
\makesquare{\dx}{0}{\dx/2}{\dy/2}{\edge}{serverone}{frame}
\makesquare{\dy}{\dy}{\dx/2}{\dy/2}{\edge}{serverone}{frame}
\makesquare{0}{\dy}{\dx/2}{\dy/2}{\edge}{servertwo}{frame}
\makesquare{2*\dx}{\dy}{\dx/2}{\dy/2}{\edge}{servertwo}{frame}
\end{tikzpicture}
}%
}%
\subfigure[Minimum order~$\order^\star \vect{\bl,\st} $ needed in~(\ref{improvedpartialfourierseriester}) for decision, $\st=\vect{1,2}$.\label{figure:policyn}]{
        \centering
        \includegraphics{\setpathfigs%
        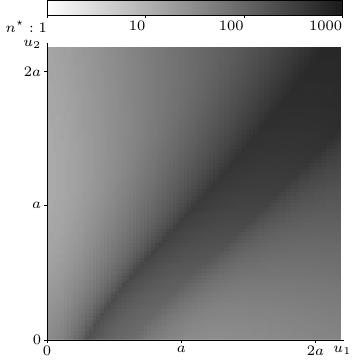}%approximatevaluefunction-korovkin-n-fig}
\obsolete{
\begin{tikzpicture}[scale=1]
\def\sizex{5}
\def\sizey{5}
\def\minx{0}
\def\dx{1}
\def\miny{0}
\def\dy{1}
\def\xstep{0.2}
\def\ystep{0.2}
\def\greymin{0.1}
\def\greymax{1.0}
\def\colone{0.45}
\def\coltwo{0.85}
\input{approximatevaluefunction-korovkin-n}
\end{tikzpicture}
}%
}%
\caption{%
One-step policy improvement for a two-server system~$\vect{1,2}$ with arrival rates~$ \vect{\ari{1},\ari{2}}=\vect{1,1/2}$, exponentially distributed service times with parameters $ \vect{\strate_{1},\strate_{2}}=\vect{2,1}$, and cost function $\costx{\bl}={\bl^2}/({\sing^2+\bl^2})$. %, where~$\sing=1$ 
%(cf. case study). %Example~\ref{example:quotientexample}).
%
} \label{figure:policy}
\end{figure}%
Figure~\ref{figure:policyservers} depicts, for a particular job with service times $\vect{\sti{1},\sti{2}}=\vect{1,2}$ and for various backlog~$\bl=\vect{\bli{1},\bli{2}}$, the \ac{FPI} policy~$\osipolicyxy{\bl}{\st}$ issued by  Algorithm~\ref{algorithm:dispatching}. 
  %The value function intervals  were computed in accordance with  Corollary~\ref{corollary:trigonometric} using the near-optimal uniform approximation~(\ref{improvedpartialfourierseriester}), the implementation of which exploited the developments of Examples~\ref{example:piecewiseexp} and~\ref{example:piecewiseanalyticexp}.
%
The quantity~$\order^\star\vect{\bl,\st}$ displayed in Figure~\ref{figure:policyn} is the minimum order~$\order$ required in~(\ref{improvedpartialfourierseriester})  for  dispatching at~$\vect{\bl,\st}$. This quantity was estimated by reporting the minimum order that allowed for dispatching for a coarse grid of values of the parameter~$\taufunction$. It can be seen that~$\order^\star\vect{\bl,\st}$  grows with the distance to the origin~$\bl=\vect{0,0}$, and increases abruptly near the frontiers of the dispatching policy~$\osipolicyfunction$.  The relatively high orders rendered by  Figure~\ref{figure:policyn}  are due to the  conservativeness of the uniform error bound~$\errn{\order}\equiv 6\, \continuitymodulussIfx{\taufunction}{[0,\taufunction]}{\costfunction}{{\taufunction}/({2\order})}$ for this particular choice of the cost function (cf. Figure~\ref{figure:costfinterval}). In practice, more accurate estimates of the error bound would contribute to reducing the estimation orders. More generally, building the function approximations from the $k$ first derivatives of~$\costfunction$, as previously suggested, will significantly accelerate convergence.
\obsolete{
for $0<\modulus{\complex}<\sing$, that
\begin{equation}\label{partlaurentone}\begin{array}{rcl}
\frac{1}{(\sing\complex)^2+1}
&=&
\left(\frac{1}{1-[-(\sing\complex)^2]}\right)
 \\
 &=&
\sum\nolimits_{t=0}^{\infty} [-(\sing\complex)^2]^t 
\\
 &=&
\sum\nolimits_{t=0}^{\infty} (-1)^t (\sing\complex)^{2t}   .
 \end{array}
\end{equation}
Hence,
\begin{equation}\begin{array}{rcl}
\lim\nolimits_{\complex\to 0^+}  \frac{d^{t}}{d\complex^t} \left[ \Costx{\frac{1}{\complex}} \right]
&=&
 \left|
 \begin{array}{ll}
  \factorial{t}\,  (-1)^{\frac{t}{2}} \sing^{t}   ,\quad&\text{if }t\text{ is even},
\\
 0, &\text{otherwise.} 
 \end{array}\right.
 \end{array}
\end{equation}
Besides,
\begin{equation}\label{partlaurenttwo}\begin{array}{l}
\lim\nolimits_{\complex\to 0^+}  \frac{d^{t}}{d\complex^t}  \left[\complex^{k}  \Polcosnx{k}{\frac{2}{\taufunction\complex}-1}\right]
\\ \quad \quad
\refereq{\eqref{polycostnx}}{=}
\obsolete{
\lim\nolimits_{\complex\to 0^+}  \frac{d^{t}}{d\complex^t} 
\left|
 \begin{array}{ll}
 \sum\nolimits_{l=0}^{\frac{k}{2}} 
  \coscoef{k}{l} 
 \sum\nolimits_{j=0}^{2l} \binomialcoef{2l}{\, j} (\frac{2}{\taufunction})^j (-1)^{2l-j}  \complex^{k-j}         
 ,
 &\quad \textup{if~$k$ is even,}
\\
 \sum\nolimits_{l=0}^{\frac{k-1}{2}} 
 \coscoef{k}{l} \sum\nolimits_{j=0}^{2l+1} \binomialcoef{2l+1}{\ j}
  (\frac{2}{\taufunction})^j (-1)^{2l+1-j}   \complex^{k-j}        
,
 & \quad\textup{if~$k$ is odd,}
 \end{array}\right.
}%
\left|
 \begin{array}{ll}
  \factorial{t}\, (\frac{-2}{\taufunction})^{k-t}   \sum\nolimits_{l=\ceil{\frac{k-t}{2}}}^{\frac{k}{2}} 
 \binomialcoef{\, 2l}{k-t}  \, \coscoef{k}{l}    
 ,
 &\textup{if~$k$ even, $t \leq k$, }
\\
 - \factorial{t} \, (\frac{-2}{\taufunction})^{k-t} \sum\nolimits_{l=\ceil{\frac{k-1-t}{2}}}^{\frac{k-1}{2}} 
   \binomialcoef{2l+1}{\, k-t}    \, \coscoef{k}{l}
, \quad
 &\textup{if $k$ odd, $t \leq k$,}
 \\ 0,&\text{otherwise}.
 \end{array}\right. 
 \end{array}
\end{equation}
The coefficients~\eqref{laurentcoef} follow from~\eqref{partlaurentone}, \eqref{partlaurenttwo} and Leibniz's product rule.

for $0\leq q \leq k$:
\obsolete{
\begin{equation}\label{laurentcoefexample}\begin{array}{rcl} 
\laurentcoef{-q}
 &= &
\frac{1}{\factorial{(k-q)}} \lim\nolimits_{\complex\to 0}
\frac{d^{k-q}}{d\complex^{k-q}} \left[\complex^{k} \Costx{k}{\frac{1}{\complex}}  \Polcosnx{k}{\frac{2}{\taufunction\complex}-1} \right]
\\
&=&
\left|
 \begin{array}{ll}
 \sum\nolimits_{t=
 %\ceil{\frac{-q}{2}}
 0
 }^{\frac{k}{2}-\ceil{\frac{q}{2}}} 
(-1)^t \sing^{2t} (\frac{-2}{\taufunction})^{q+2t}   \sum\nolimits_{l=t+\ceil{\frac{q}{2}}}^{\frac{k}{2}} 
 \binomialcoef{\, 2l}{q+2t}  \, \coscoef{k}{l}    
 ,
 &\textup{if~$k$ even, }
\\
 -  \sum\nolimits_{t=
 %\ceil{\frac{-q}{2}}
 0
 }^{\frac{k-1}{2}-\ceil{\frac{q-1}{2}}} 
(-1)^t \sing^{2t} (\frac{-2}{\taufunction})^{q+2t} \sum\nolimits_{l=t+\ceil{\frac{q-1}{2}}}^{\frac{k-1}{2}} 
   \binomialcoef{2l+1}{q+2t}    \, \coscoef{k}{l}
, \quad
 &\textup{if $k$ odd}.
 \end{array}\right. 
\\
&=&
%  \ceil{\frac{q}{2}}} <= l <= {\frac{k}{2} 
%  0 <= t <= l - \ceil{\frac{q}{2}} , \frac{k}{2} - \ceil{\frac{q}{2}}
%
% \ceil{\frac{q-1}{2}}} <= l <= {\frac{k-1}{2} 
%  0 <= t <= l - \ceil{\frac{q-1}{2}} , \frac{k-1}{2} - \ceil{\frac{q-1}{2}}
%
%
\left|
 \begin{array}{ll}
 (\frac{-2}{\taufunction})^{q}  \sum\nolimits_{l=\ceil{\frac{q}{2}}}^{\frac{k}{2}} 
 \left[
 \sum\nolimits_{t= 0 }^{l-\ceil{\frac{q}{2}}} 
  \binomialcoef{\, 2l}{q+2t}  
(\frac{-4\sing^{2}}{\taufunction^2})^{t}   
 \right]
 \coscoef{k}{l}    
 ,
 &\textup{if~$k$ even, }
\\
 - (\frac{-2}{\taufunction})^{q}  
  \sum\nolimits_{l=\ceil{\frac{q-1}{2}}}^{\frac{k-1}{2}} 
 \left[
 \sum\nolimits_{t= 0}^{l-\ceil{\frac{q-1}{2}}} 
   \binomialcoef{2l+1}{q+2t}   
(\frac{-4\sing^{2}}{\taufunction^2})^{t} 
   \right]
    \coscoef{k}{l}
, \quad
 &\textup{if $k$ odd}
 \end{array}\right. 
\obsolete{
\\
&=&
\left|
 \begin{array}{ll}
 (\frac{-2}{\taufunction})^{q}   
 \sum\nolimits_{j=0}^{\floor{\frac{k-q}{2}}} 
 \left[
 \sum\nolimits_{t= 0 }^{j} \binomialcoef{2j+2\ceil{\frac{q}{2}}}{q+2t}  
(\frac{-4}{\taufunction^2})^{t}   
 \right]
 \coscoef{k}{j+\ceil{\frac{q}{2}}}    
 ,
 &\textup{if~$k$ even, }
\\
 -  (\frac{-2}{\taufunction})^{q}   
 \sum\nolimits_{j=0}^{\floor{\frac{k-q}{2}}} 
 \left[
 \sum\nolimits_{t= 0}^{j} 
   \binomialcoef{2j+2\ceil{\frac{q-1}{2}}+1}{q+2t}   
(\frac{-4}{\taufunction^2})^{t} 
   \right]
    \coscoef{k}{j+\ceil{\frac{q-1}{2}}}
, \quad
 &\textup{if $k$ odd}.
 \end{array}\right. 
 }
\end{array}\end{equation}
}%
\begin{equation}\label{laurentcoefexample}
\laurentcoef{-q}
=
\left|
 \begin{array}{ll}
 (\frac{-2}{\taufunction})^{q}  \sum\nolimits_{l=\ceil{\frac{q}{2}}}^{\frac{k}{2}} 
 \left[
 \sum\nolimits_{t= 0 }^{l-\ceil{\frac{q}{2}}} 
  \binomialcoef{\, 2l}{q+2t}  
(\frac{-4\sing^{2}}{\taufunction^2})^{t}   
 \right]
 \coscoef{k}{l}    
 ,
 &\textup{if~$k$ even, }
\\
 - (\frac{-2}{\taufunction})^{q}  
  \sum\nolimits_{l=\ceil{\frac{q-1}{2}}}^{\frac{k-1}{2}} 
 \left[
 \sum\nolimits_{t= 0}^{l-\ceil{\frac{q-1}{2}}} 
   \binomialcoef{2l+1}{q+2t}   
(\frac{-4\sing^{2}}{\taufunction^2})^{t} 
   \right]
    \coscoef{k}{l}
, \quad
 &\textup{if $k$ odd}
 \end{array}\right. 
 \end{equation}

The Fourier coefficients~\eqref{realfouriercoefficientstwo} for this cost function are given, for~$=0,1,\dots$, by
\begin{equation} \label{solutionrealfouriercoefkexample}\begin{array}{c}
\realpart{\fouriercoefk{k}}
%\realfouriercoefk{k}
=
\sqrt{\pi} \sum\nolimits_{q=0}^{\polydegree}     \frac{ \laurentcoef{-q}\,  (-\taufunction)^{q} }{ \factorial{q} \Gammax{\frac{1}{2}-q}} 
- 
\sum\nolimits_{\sing=\i\sing,-\i\sing} \Residuefx{
\allpolycostnx{k}{\complex}  \,  \complexexponent{\complex}{-\frac{1}{2}}{-\pi} \complexexponent{(\complex-\taufunction)}{-\frac{1}{2}}{-\pi}
}{\complex=\sing},
\end{array}
\end{equation}

 $\i-\taufunction= \sqrt{1+\taufunction^2} \, \exp{i\left(\pi-\Atanx{\frac{1}{\taufunction}}\right)}$

 $(\i-\taufunction)^{\alpha}=\exp{ \alpha \ln \left[ \sqrt{1+\taufunction^2} \, \exp{i\left(\pi-\Atanx{\frac{1}{\taufunction}}\right)} \right] }   = \exp{ \alpha \left(  \ln (\sqrt{1+\taufunction^2} +i  \left(\pi-\Atanx{\frac{1}{\taufunction}}\right)  \right) } = (1+\taufunction^2)^{\frac{\alpha}{2}}  \exp{ i \alpha \left(\pi-\Atanx{\frac{1}{\taufunction}}\right)  }    =  (1+\taufunction^2)^{\frac{\alpha}{2}} \left[ \Cosx{\alpha \left(\pi-\Atanx{\frac{1}{\taufunction}}\right) } + \i \Sinx{\alpha \left(\pi-\Atanx{\frac{1}{\taufunction}}\right) } \right]   $

 $(\i-\taufunction)^{-\frac{1}{2}}
% =  (1+\taufunction^2)^{-\frac{1}{4}} \left[ \Cosx{-\frac{\pi-\Atanx{\frac{1}{\taufunction}}}{2}  } + \i \Sinx{-\frac{\pi-\Atanx{\frac{1}{\taufunction}}}{2} } \right]  
=  (1+\taufunction^2)^{-\frac{1}{4}} \exp{- \i \frac{\pi-\Atanx{\frac{1}{\taufunction}}}{2}  }
 %=  (1+\taufunction^2)^{-\frac{1}{4}} \left[ \Sinx{\frac{\Atanx{\frac{1}{\taufunction}}}{2}  } - \i \Cosx{\frac{\Atanx{\frac{1}{\taufunction}}}{2} } \right] 
  $
 
 $\i^{-\frac{1}{2}}=     \frac{\sqrt{2}}{2} (1- \i) =  \Cosx{-\frac \pi 4} + \i \Sinx{-\frac \pi 4} = \exp{-\i \frac \pi 4}$

 $\i^{-\frac{1}{2}} (\i-\taufunction)^{-\frac{1}{2}}
 =   (1+\taufunction^2)^{-\frac{1}{4}} \exp{ \i \left( \frac{\Atanx{\frac{1}{\taufunction}}}{2} - \frac{3\pi}{4} \right) }
  =    (1+\taufunction^2)^{-\frac{1}{4}} \left[ \Cosx{\frac{\Atanx{\frac{1}{\taufunction}}}{2} - \frac{3\pi}{4}  } + \i \Sinx{\frac{\Atanx{\frac{1}{\taufunction}}}{2} - \frac{3\pi}{4}  } \right]    $

$(\i\sing)^{-\frac{1}{2}} (\i\sing-\taufunction)^{-\frac{1}{2}}
 =
 \frac{1}{\sing}  \i^{-\frac{1}{2}} \left(\i-\frac{\taufunction}{\sing}\right)^{-\frac{1}{2}}
 =  \frac{1}{\sing}  (1+\left(\frac{\taufunction}{\sing}\right)^2)^{-\frac{1}{4}} \exp{ \i \left( \frac{\Atanx{\frac{\sing}{\taufunction}}}{2} - \frac{3\pi}{4} \right) }
=  \frac{\left(1+\left(\frac{\taufunction}{\sing}\right)^2\right)^{-\frac{1}{4}}}{\sing }   \exp{ \i \left( \frac{\Atanx{\frac{\sing}{\taufunction}}}{2} - \frac{3\pi}{4} \right) }
  $

 \begin{equation}\begin{array}{c}
 \Residuefx{
\allpolycostnx{k}{\complex}  \,  \complexexponent{\complex}{-\frac{1}{2}}{-\pi} \complexexponent{(\complex-\taufunction)}{-\frac{1}{2}}{-\pi}
}{\complex=\pm\i\sing}
\\=
\lim\nolimits_{\complex\to\pm \i\sing} (\complex\mp\i\sing) \left( \frac{\complex^2}{(\complex+\i\sing)(\complex-\i\sing)} \right)  \Polcosnx{k}{\frac{2\complex}{\taufunction}-1} \,  \complexexponent{\complex}{-\frac{1}{2}}{-\pi} \complexexponent{(\complex-\taufunction)}{-\frac{1}{2}}{-\pi}
\\=
 \left( \frac{\pm\i\sing}{2} \right)  \Polcosnx{k}{\frac{\pm2\i\sing}{\taufunction}-1} \,   \frac{\left(1+\left(\frac{\taufunction}{\sing}\right)^2\right)^{-\frac{1}{4}}}{\sing }   \exp{ \pm \i \left( \frac{\Atanx{\frac{\sing}{\taufunction}}}{2} - \frac{3\pi}{4} \right) }
\\=
 \frac{ \sqrt{\sing} \, \Polcosnx{k}{\frac{\pm2\i\sing}{\taufunction}-1}}{2 \left(\sing^2+\taufunction^2\right)^{\frac{1}{4}}}  \,  \exp{\pm \i \left( \frac{\Atanx{\frac{\sing}{\taufunction}}}{2} - \frac{\pi}{4} \right) }
\end{array}\end{equation}

 \begin{equation}\label{cosxkthetacomplex}\begin{array}{l}
 \Polcosnx{k}{\frac{\pm2\i\sing}{\taufunction}-1} 
 \\
 \obsolete{
 = \left|
 \begin{array}{ll}
 \sum\nolimits_{q=0}^{\frac{k}{2}} 
 \coscoef{k}{q} 
 \left(\frac{\pm2\i\sing}{\taufunction}-1\right)^{2q}           
 ,
 &\quad \textup{if~$k$ is even,}
\\
 \sum\nolimits_{q=0}^{\frac{k-1}{2}} 
 \coscoef{k}{q}
  \left(\frac{\pm2\i\sing}{\taufunction}-1\right)^{2q+1}           
,
 & \quad\textup{if~$k$ is odd,}
 \end{array}\right.
 \\
 }
  = \left|
 \begin{array}{r}
 \sum\nolimits_{q=0}^{\frac{k}{2}}  
\left[ \sum\nolimits_{t=0}^{q} \binomialcoef{2q}{2t}
 \left(\frac{-4\sing^2}{\taufunction^2}\right)^{t}     \mp \i \left(\frac{2\sing}{\taufunction}\right) \sum\nolimits_{t=0}^{q-1} \binomialcoef{2q}{2t+1}
  \left(\frac{-4\sing^2}{\taufunction^2}\right)^{t}    \right]  \coscoef{k}{q}      
 ,\quad  \quad  \quad  
 \\ \textup{if~$k$ is even,} 
\\
 \sum\nolimits_{q=0}^{\frac{k-1}{2}} 
\left[-\sum\nolimits_{t=0}^{q} \binomialcoef{2q+1}{2t}
 \left(\frac{-4\sing^2}{\taufunction^2}\right)^{t} \pm \i \left(\frac{2\sing}{\taufunction}\right)  \sum\nolimits_{t=0}^{q} \binomialcoef{2q+1}{2t+1}
 \left(\frac{-4\sing^2}{\taufunction^2}\right)^{t}           
 \right]   \coscoef{k}{q}
, \quad 
\\ \textup{if~$k$ is odd,}
 \end{array}\right.
 \end{array}
\end{equation}

\begin{equation}\begin{array}{c}
{\frac{2}{\sqrt{\sing}} \, \left(\sing^2+\taufunction^2\right)^{\frac{1}{4}}} \, 2\Realpart{\Residuefx{
\allpolycostnx{k}{\complex}  \,  \complexexponent{\complex}{-\frac{1}{2}}{-\pi} \complexexponent{(\complex-\taufunction)}{-\frac{1}{2}}{-\pi}
}{\complex=\i} } 
\\
=
\left|
 \begin{array}{r}
 2\sum\nolimits_{q=0}^{\frac{k}{2}}  
\left[ \Cosx{ \anglesingtau{\sing}{\tau}} \sum\nolimits_{t=0}^{q} \binomialcoef{2q}{2t}
 \left(\frac{-4\sing^2}{\taufunction^2}\right)^{t}     + \left(\frac{2\sing}{\taufunction}\right) \Sinx{ \anglesingtau{\sing}{\tau}} \sum\nolimits_{t=0}^{q-1} \binomialcoef{2q}{2t+1}
  \left(\frac{-4\sing^2}{\taufunction^2}\right)^{t}    \right]  \coscoef{k}{q}      
 ,\quad  \quad  \quad  
 \\ \textup{if~$k$ is even,} 
\\
 -2 \sum\nolimits_{q=0}^{\frac{k-1}{2}} 
\left[\Cosx{ \anglesingtau{\sing}{\tau}} \sum\nolimits_{t=0}^{q} \binomialcoef{2q+1}{2t}
 \left(\frac{-4\sing^2}{\taufunction^2}\right)^{t} + \Sinx{ \anglesingtau{\sing}{\tau}} \left(\frac{2\sing}{\taufunction}\right)  \sum\nolimits_{t=0}^{q} \binomialcoef{2q+1}{2t+1}
 \left(\frac{-4\sing^2}{\taufunction^2}\right)^{t}           
 \right]   \coscoef{k}{q}
, \quad 
\\ \textup{if~$k$ is odd,}
 \end{array}\right.
 \\
=
\left|
 \begin{array}{r}
 2\sum\nolimits_{q=0}^{\frac{k}{2}}  
 \evencoefsingtauq{\sing}{\taufunction}{q} \, \coscoef{k}{q}      
 ,\quad \textup{if~$k$ is even,} 
\\
 2 \sum\nolimits_{q=0}^{\frac{k-1}{2}} 
 \oddcoefsingtauq{\sing}{\taufunction}{q}  \, \coscoef{k}{q}    
, \quad \textup{if~$k$ is odd,}
 \end{array}\right.
 \\
=
 2\sum\nolimits_{q=0}^{\floor{\frac{k}{2}}}  
 \evenoddcoefksingtauq{k}{\sing}{\taufunction}{q} \, \coscoef{k}{q}      
 \end{array}
\end{equation}
where we define
\begin{equation}
 \begin{array}{rl}
 \evenoddcoefksingtauq{k}{\sing}{\taufunction}{q} 
 =
 \Cosx{ \anglesingtau{\sing}{\tau}} \sum\nolimits_{t=0}^{q} \binomialcoef{2q}{2t}
 \left(\frac{-4}{\taufunction^2}\right)^{t}    
 \hspace{20mm}
 &
 \\
  + \left(\frac{2}{\taufunction}\right) \Sinx{ \anglesingtau{\sing}{\tau}} \sum\nolimits_{t=0}^{q-1} \binomialcoef{2q}{2t+1}
  \left(\frac{-4}{\taufunction^2}\right)^{t}
,\quad &
(k \text{ even},\ q=0,\dots,\frac{k}{2}),
 \end{array}
\end{equation}
 and
\begin{equation}
 \begin{array}{rl}
 \evenoddcoefksingtauq{k}{\sing}{\taufunction}{q}
 =
 - \Cosx{ \anglesingtau{\sing}{\tau}} \sum\nolimits_{t=0}^{q} \binomialcoef{2q+1}{2t}
 \left(\frac{-4}{\taufunction^2}\right)^{t} 
 \hspace{10mm}
 &
 \\
  - \Sinx{ \anglesingtau{\sing}{\tau}} \left(\frac{2}{\taufunction}\right)  \sum\nolimits_{t=0}^{q} \binomialcoef{2q+1}{2t+1}
 \left(\frac{-4}{\taufunction^2}\right)^{t}     
,\quad &
(k \text{ odd}, \ q=0,\dots,\frac{k-1}{2}),
 \end{array}
\end{equation}
where $ \anglesingtau{\sing}{\tau} = [ \atanx{{\sing}/{\taufunction}}]/2 - {\pi}/{4} =1/2 [ \atanx{{\sing}/{\taufunction}}- {\pi}/{2} ] =-1/2 [ \atanx{{\taufunction}/{\sing}} ]$ for positive $\taufunction,\sing$.
 
}%
%%%%%%%%%%%%%%%%%%%%%
 %
 %
% \remarkmarker\end{example}

\obsolete{
----------

 take trigonometric approx of $g(\theta)$ : $a_0  /2+...+\cosx{n \theta}=P_n(\cosx{\theta})  =P_n(2\bl/\taufunction-1) = P^\prime_n(\bl)$

 show that the scaled sup error of $\costfunction$ on $[0,\taufunction] $ is more than the modulus of continuity of $g$, which is periodic and even, thus the modulus of continuity is realised by two points of $[0,\pi]$

 do the same the other way round starting from the modulus of c. of $g$

 \begin{equation}
 \sup\limits_{\substack{\bli{1},\bli{2}  \in [0,\taufunction],\\ \modulus{\bli{1}-\bli{2}}\leq 2\delta/\taufunction }} \modulus{\periodcostx{\bli{1}}-\periodcostx{\bli{2}}} 
 =
 \modulus{\periodcostx{\bl^\star}-\periodcostx{\bl^\star+\delta^\star}} 
\end{equation}

take $\theta^\star$ and  $\theta^\star+\xi^\star$ such that $\bl^\star=\taufunction/2 \times (1 + \cosx{\theta^\star})$ and $\bl^\star+\delta^\star = \taufunction/2 \times (1 + \cosx{\theta^\star+\xi^\star}) $, both points in $[0,\pi]$. Then $\modulus{2\delta^\star/\taufunction} = \modulus{\cosx{\theta^\star+\xi^\star}-\cosx{\theta^\star}} = 2 \modulus{\sinx{\xi^\star/2} \sinx{\theta^\star+\xi^\star/2} } \leq \modulus{ \xi^\star }$

 \begin{equation}
 \sup\limits_{\substack{\bli{1},\bli{2}  \in [0,\taufunction],\\ \modulus{\bli{1}-\bli{2}}\leq 2\delta/\taufunction }} \modulus{\periodcostx{\bli{1}}-\periodcostx{\bli{2}}} 
 =
 \modulus{\periodcostx{\bl^\star}-\periodcostx{\bl^\star+\delta^\star}} 
 =
 \modulus{ g(\theta^\star) - g(\theta^\star+\xi^\star) }
\end{equation}

Consider the periodic function with period~$2\pi$ obtained by mirroring~$\costx{\pi\bl/\taufunction}\, \indicatorfunction{(0,\pi\bl/\taufunction)}$ %of the interval $[0,\taufunction]$, mirror it 
around~$0$, and replicating the $[-\pi,\pi]$ section

 it's periodic and continuous, take trigonometric approximation: consider variable $\altbl$ such that  $\bl = \taufunction \cosx{2\pi\altbl/T}$, and function $g$ such that $g(\altbl) = \costx{\taufunction\cosx{2\pi\altbl/T}}= \costx{\bl}$ on  $[0,\pi]$ and $g(\altbl)=g(-\altbl)$ on $[-\pi,0]$, then $g$ is continuous and  periodic on $[-\pi,\pi]$ with $T=2\pi$. Take trigonometric approximation of $g(\altbl)$ in the form  $a_0/2+...+\cosx{n \altbl}=P_n(\cosx{\altbl})=P\prime_n(\taufunction\cosx{\altbl})$, then we have an approximation of $ \costx{\taufunction\cosx{2\pi\altbl/T}}$ in the polynomial $P\prime_n(\bl)$.
 }

%%%%%%%%%%%%%%%%%%%%%

%%%%%%%%%%%%%%%%%%%%%

%%%%%%%%%%

\section{Discussion}\label{section:discussion}

Integral transformations of the Poisson equation $\altgenericfunction = P\altgenericfunction + \genericfunction$ have the quality of simplifying the analysis, as they provide a principled framework for the systematic derivation of solutions.
Although it is known that the candidate functions for closed-form solutions form a dense set where any~$\genericfunction$ can be approximated with arbitrary precision, one  should be cautious that a   convergent series for~$\genericfunction$ does not always produce a convergent series for~$\altgenericfunction$; 
Taylor series of~$\genericfunction$, in particular,  are subject to tail effects and most likely to diverge after $\measureXfunction{\Wt}$-integration with respect to the stationary probability measure of the waiting times. 

In the context of \acl{FPI}, such tail effects can be avoided  by considering approximations of~$\genericfunction$ on finite supports---preferably trigonometric sums, which for Lipschitz-continuous~$\genericfunction$ achieve the convergence rate~$\magnitude{\taufunction/\order}$ in the number~$\order$ of approximation terms, improvable to $\magnitude{{\taufunction}/[(\order-k)\order^{k}]}$ if~$\genericfunction$ is $k$-times continuously differentiable---, while using tractable bounds for the larger backlog values. 
The availability of closed forms for bounding intervals of this type with a diversity of service time distribution models gives the green light to a systematized implementation of the \ac{FPI} step.

We believe that the techniques developed in this study, combined with well-chosen supervised learning methods,  make it possible, in large multiple-server systems, to devise efficient online algorithms for learning \ac{FPI} policies gradually, as the incoming jobs are  dispatched and the (possibly high-dimensional) state space is visited. The design and assessement of \ac{FPI} dispatching policies in such systems is left to future work.

\obsolete{%
In the M/G/1 queue, value functions derived from Taylor series expansions of the cost function are only expected to converge for a narrow range of cost functions, namely, entire functions with asymptotic rates of growth less than the exponential type~$\modulus{\dominantpoleX{\Wt}}$, corresponding to the asymptotic decreasing rate of the limiting distribution of the waiting times in the queue. % at which the limiting distribution of the waiting times in the queue decreases.
% 
%In Example~\ref{example:piecewiseanalyticexp} we saw that this restriction could be bypassed for any analytic cost function by confining the Taylor series expansions to one interval of finite length (or to a finite union of such intervals). 
%Since however the maximum interval length is bounded by the radius of convergence of the Taylor series of the cost function, the possibility is excluded of devising unequivoqual dispatching policies with arbitrary precision.
%All the same, devising in this fashion unequivoqual dispatching policies with arbitrary precision is in general impossible because the maximum interval length is bounded by the radius of convergence of the Taylor series of the cost function.
Tail effects of the waiting time distribution can be avoided by instead considering function approximations on finite intervals~$[0,\taufunction]$, while using tractable bounds for the large-backlog section of the cost function. 
The availability of closed-form expressions to bounding intervals of this type %for all cost functions with explicit Fourier coefficients~(\refeq{realfouriercoefficientstwo}) and 
for a diversity of service time distribution models gives the green light to a systematic implementation of \acf{FPI}  in the context of multiple-server dispatching.

The \ac{FPI}  policy in a multiple-server system can be inferred from such approximation models for the value functions of the individual servers on condition that the function approximimation algorithms converge fast enough on their intervals. 
%
%Uniform approximations of the cost function inside intervals are a better strategy for the problem of value function approximation by polynomials for the reason the estimation error they induce on the value function is bounded above and below by closed-form expressions which vanish for every backlog value as the  order of the approximating polynomial is increased.
%grows linearly with the backlog at the queue.
%
Bernstein polynomials offer a flexible and easily implementable solution, though their relatively slow convergence rate, $%\magnitude{\continuitymodulussIfx{\taufunction}{[0,\taufunction]}{\costfunction}{{\taufunction}/{\sqrt{\order}}}}
\magnitude{\taufunction/{\sqrt{\order}}}
$ in the number~$\order$ of approximation terms under Lipschitz-continuous cost, is prohibitive in practice. %, where numerical issues are likely to arise.
Trigonometric sums on the other hand achieve the non-improvable convergence rate~$
%\magnitude{\continuitymodulussIfx{\taufunction}{[0,\taufunction]}{\costfunction}{{\taufunction}/({2\order})}}
\magnitude{\taufunction/\order}
$, which can be further lowered to $\magnitude{{\taufunction}/[(\order-k)\order^{k}]}$ in the case of $k$-times continuously differentiable cost by a certain class of polynomials derived from trigonometric sums.
%Interval bounds for the value function are available in closed-form with many types of cost functions|those with explicit Fourier coefficients~(\refeq{realfouriercoefficientstwo})|for a diversity of service time distribution models, allowing for the systematic implementation of the first policy improvement step in the context multi-server dispatching.
%%Future work will assess the precision  and efficiency of the produced approximation bounds, and will address the relevance of the suggested approach for one-step policy improvement in the context multi-server dispatching.
%

We believe that the techniques developed in this study, combined with well-chosen supervised learning methods, will make it possible, in large systems with many servers, to devise online algorithms for learning \ac{FPI} policies gradually, as the incoming jobs are  dispatched and the (possibly high-dimensional) state space is visited. The actual design, implementation and assessement of \ac{FPI} dispatching policies in such systems is left to future work.
}%
%%%%%%%%%%%

\obsolete{
\section*{Acknowledgements}
This study was supported in part by the Academy of Finland within the 
FQ4BD project (grant no.\ 296206).
}%

 %%%%%%%%%%%%%%%
 %%%%%%%%%%%%%%%

\bibliographystyle{spmpsci}      % mathematics and physical sciences

\bibliography{\setpath%
QUESTA20_bibliography}

\begin{thebibliography}{10}
\providecommand{\url}[1]{{#1}}
\providecommand{\urlprefix}{URL }
\expandafter\ifx\csname urlstyle\endcsname\relax
  \providecommand{\doi}[1]{DOI~\discretionary{}{}{}#1}\else
  \providecommand{\doi}{DOI~\discretionary{}{}{}\begingroup
  \urlstyle{rm}\Url}\fi

\bibitem{Aalto96NTS}
Aalto, S., Virtamo, J.: Basic packet routing problem.
\newblock 13th Nordic Teletraffic Seminar pp. 85--97 (1996)

\bibitem{arapostathis93}
Arapostathis, A., Borkar, V., Fernández-Gaucherand, E., Ghosh, M., Marcus, S.:
  Discrete-time controlled markov processes with average cost criterion: A
  survey.
\newblock SIAM J. Control Optim. \textbf{31}(2), 282--344 (1993).
\newblock \doi{10.1137/0331018}

\bibitem{athreya78}
Athreya, K.B., Ney, P.: A new approach to the limit theory of recurrent markov
  chains.
\newblock Trans. Amer. Math. Soc. \textbf{245}, 493--501 (1978)

\bibitem{bernstein12}
Bernstein, S.N.: {D{\'e}monstration du Th{\'e}or{\`e}me de Weierstrass
  fond{\'e}e sur le calcul des Probabilit{\'e}s}.
\newblock Comm. Soc. Math. Kharkov \textbf{13}(1), 1--2 (1912)

\bibitem{bertsekas07}
Bertsekas, D.P.: Dynamic Programming and Optimal Control, vol.~II, 3rd edn.
\newblock Athena Scientific (2007)

\bibitem{bhulai06}
Bhulai, S.: On the value function of the {M/Cox(r)/1} queue.
\newblock J. Appl. Probab. \textbf{43} (2006).
\newblock \doi{10.1239/jap/1152413728}

\bibitem{bhulai03}
Bhulai, S., Spieksma, F.M.: {O}n the uniqueness of solutions to the {P}oisson
  equations for average cost {M}arkov chains with unbounded cost functions.
\newblock Math. Methods Oper. Res. \textbf{58}(2), 221--236 (2003)

\bibitem{brown14}
Brown, J.: Complex Variables and Applications, 9th edn.
\newblock McGraw-Hill Education, New York, NY (2014)

\bibitem{corless96}
Corless, R.M., Gonnet, G.H., Hare, D.E.G., Jeffrey, D.J., Knuth, D.E.: On the
  {Lambert} {W} function.
\newblock Adv. Comput. Math. \textbf{5}(1), 329--359 (1996).
\newblock \doi{10.1007/BF02124750}

\bibitem{turck12}
De~Turck, K., De~Clercq, S., Wittevrongel, S., Bruneel, H., Fiems, D.:
  {Transform-Domain Solutions of Poisson's Equation with Applications to the
  Asymptotic Variance}.
\newblock In: K.~Al-Begain, D.~Fiems, J.M. Vincent (eds.) Analytical and
  Stochastic Modeling Techniques and Applications, pp. 227--239. Springer
  Berlin Heidelberg, Berlin, Heidelberg (2012)

\bibitem{gallager13}
Gallager, R.G.: Stochastic processes : theory for applications.
\newblock Cambridge University Press, Cambridge (2013)

\bibitem{glynn94}
Glynn, P.W.: Poisson's equation for the recurrent {M/G/1} queue.
\newblock Adv. Appl. Probab. \textbf{26}(4), 1044–1062 (1994).
\newblock \doi{10.2307/1427904}

\bibitem{glynn96}
Glynn, P.W., Meyn, S.P.: A {L}iapounov bound for solutions of the {P}oisson
  equation.
\newblock Ann. Probab. \textbf{24}(2), 916--931 (1996).
\newblock \doi{10.1214/aop/1039639370}

\bibitem{gross98}
Gross, D., Harris, C.M.: Fundamentals of queueing theory.
\newblock J. Wiley \& sons, New York, Chichester, Weinheim (1998)

\bibitem{howard60}
Howard, R.A.: Dynamic Programming and Markov Processes.
\newblock MIT Press, Cambridge, MA (1960)

\bibitem{hyytia12OnTheVF}
Hyyti{\"a}, E., Aalto, S., Penttinen, A., Virtamo, J.: On the value function of
  the {M/G/1} {FCFS} and {LCFS} queues.
\newblock J. Appl. Probab. \textbf{49}(4), 1052--1071 (2012)

\bibitem{hyytia-ejor-2012}
Hyyti{\"a}, E., Penttinen, A., Aalto, S.: Size- and state-aware dispatching
  problem with queue-specific job sizes.
\newblock European J. Oper. Res. \textbf{217}(2), 357--370 (2012)

\bibitem{hyytia-peva-2014}
Hyyti{\"a}, E., Righter, R., Aalto, S.: Task assignment in a heterogeneous
  server farm with switching delays and general energy-aware cost structure.
\newblock Perform. Eval. \textbf{75--76}(0), 17--35 (2014)

\bibitem{hyytia-peva-2017}
Hyyti{\"a}, E., Righter, R., Bilenne, O., Wu, X.: Dispatching fixed-sized jobs
  with multiple deadlines to parallel heterogeneous servers.
\newblock Perform. Eval. \textbf{114}(Supplement C), 32 -- 44 (2017).
\newblock \doi{10.1016/j.peva.2017.04.003}

\bibitem{hyytia-peva-2020}
Hyyti{\"a}, E., Righter, R., Virtamo, J., Viitasaari, L.: On value functions
  for {FCFS} queues with batch arrivals and general cost structures.
\newblock Perform. Eval. \textbf{138}, 102083 (2020).
\newblock \doi{10.1016/j.peva.2020.102083}

\bibitem{hyytia14TaskAssignment}
Hyytiä, E., Righter, R., Aalto, S.: Task assignment in a heterogeneous server
  farm with switching delays and general energy-aware cost structure.
\newblock Perform. Eval. \textbf{75-76}, 17 -- 35 (2014).
\newblock \doi{10.1016/j.peva.2014.01.002}

\bibitem{hyytia11MM1PS}
Hyytiä, E., Virtamo, J., Aalto, S., Penttinen, A.: M/m/1-ps queue and
  size-aware task assignment.
\newblock Perform. Eval. \textbf{68}(11), 1136 -- 1148 (2011).
\newblock \doi{10.1016/j.peva.2011.07.011}.
\newblock Special Issue: Performance 2011

\bibitem{sennott89}
I.~Sennott, L.: Average cost optimal stationary policies in infinite state
  markov decision processes with unbounded costs.
\newblock Oper. Res. \textbf{37}, 626--633 (1989).
\newblock \doi{10.1287/opre.37.4.626}

\bibitem{jackson41}
Jackson, D.: {F}ourier {S}eries and {O}rthogonal {P}olynomials.
\newblock Dover Books on Mathematics. Dover Publications (1941)

\bibitem{khintchine32}
Khintchine, A.Y.: Mathematical theory of a stationary queue.
\newblock Mat. Sb. \textbf{39}(4), 73--84 (1932)

\bibitem{koralov07}
Koralov, L., Sinai, Y.: Theory of Probability and Random Processes.
\newblock Universitext. Springer Berlin Heidelberg (2007)

\bibitem{korovkin60}
Korovkin, P.: Linear Operators and Approximation Theory.
\newblock International monographs on advanced mathematics \& physics.
  Hindustan Pub. Corp. (1960)

\bibitem{krishnan87}
Krishnan, K.R.: Joining the right queue: A markov decision-rule.
\newblock In: Proc. IEEE Conf. Decis. Control, vol.~26, pp. 1863--1868 (1987).
\newblock \doi{10.1109/CDC.1987.272835}

\bibitem{levin96}
Levin, B.: Lectures on Entire Functions.
\newblock Amer. Math. Soc., Providence, RI (1996)

\bibitem{meyn96}
Meyn, S.P.: Convergence of the policy iteration algorithm with applications to
  queueing networks and their fluid models.
\newblock In: Proc. {IEEE} Conf. Decis. Control, vol.~1, pp. 366--371 vol.1
  (1996).
\newblock \doi{10.1109/CDC.1996.574337}

\bibitem{meyn97}
Meyn, S.P.: The policy iteration algorithm for average reward markov decision
  processes with general state space.
\newblock IEEE Trans. Automat. Control \textbf{42}(12), 1663--1680 (1997).
\newblock \doi{10.1109/9.650016}

\bibitem{mitrinovic84}
Mitrinovi{\'c}, D., Ke{\v{c}}ki{\'c}, J.: The Cauchy Method of Residues: Theory
  and Applications.
\newblock Math. Appl. D. Reidel Publishing Company, Dordrecht, Holland (1984)

\bibitem{neveu72}
Neveu, J.: Potentiel markovien r\'ecurrent des cha\^{i}nes de harris.
\newblock Ann. Inst. Fourier \textbf{22}(2), 85--130 (1972).
\newblock \doi{10.5802/aif.414}

\bibitem{nummelin91}
Nummelin, E.: On the poisson equation in the potential theory of a single
  kernel.
\newblock Math. Scand. \textbf{68}, 59--82 (1991).
\newblock \doi{10.7146/math.scand.a-12346}

\bibitem{ott92}
Ott, T.J., Krishnan, K.R.: Separable routing: A scheme for state-dependent
  routing of circuit switched telephone traffic.
\newblock Ann. Oper. Res. \textbf{35}(1-4), 43--68 (1992).
\newblock \doi{10.1007/BF02023090}

\bibitem{pollaczek30}
Pollaczek, F.: {{\"U}ber eine Aufgabe der Wahrscheinlichkeitstheorie. I}.
\newblock Math. Z. \textbf{32}(1), 64--100 (1930).
\newblock \doi{10.1007/BF01194620}

\bibitem{rivlin69}
Rivlin, T.J.: An Introduction to the Approximation of Functions.
\newblock Dover Publications, New York (1969)

\bibitem{sassen97}
Sassen, S., Tijms, H., Nobel, R.: A heuristic rule for routing customers to
  parallel servers.
\newblock Stat. Neerl. \textbf{51}, 107 -- 121 (2001).
\newblock \doi{10.1111/1467-9574.00040}

\bibitem{tijms03}
Tijms, H.: A First Course in Stochastic Models.
\newblock Wiley (2003)

\bibitem{welch64}
Welch, P.D.: On a generalized {M/G/1} queuing process in which the first
  customer of each busy period receives exceptional service.
\newblock Oper. Res. \textbf{12}(5), 736--752 (1964).
\newblock \doi{10.1287/opre.12.5.736}

\bibitem{wijngaard79}
Wijngaard, J.: Decomposition for dynamic programming in production and
  inventory control.
\newblock Eng. Process. Econ. \textbf{4}(2), 385 -- 388 (1979).
\newblock \doi{10.1016/0377-841X(79)90051-2}

\end{thebibliography}
%\bibliographystyle{plain} 

    %%%%%%%%%%%%
    %%%%%%%%%%%%
    
    \counterwithin{table}{section}

\appendix 

\numberwithin{theorem}{section}
\numberwithin{proposition}{section}
\numberwithin{lemma}{section}
\numberwithin{corollary}{section}
\numberwithin{result}{section}
\numberwithin{definition}{section}
\numberwithin{example}{section}

\section{Characterization of the value function}

%\section{On the Laplace-Stieltjes transform of~$\Wt$. }
\label{appendix:LST}

Before showing Propositions~\ref{proposition:poissonequation}-\ref{proposition:identities}, and Theorem~\ref{theorem:solutionofPoisson}, we  characterize~$\LSTWtfunction$ in the complex plane. 
%  
% Proposition~\refeq{proposition:analycity}\eqref{analycity:dominantpoles}
% Proposition~\refeq{proposition:analycity}\eqref{analycity:nopole}
% Proposition~\refeq{proposition:analycity}\eqref{analycity:isolatedpole}
% Proposition~\refeq{proposition:analycity}\eqref{analycity:centered}
% Proposition~\refeq{proposition:analycity}\eqref{analycity:excentered}
\begin{proposition}[Analycity of~$\LSTWtfunction$ and pole location] \label{proposition:analycity} 
Under Assumption~\ref{assumption:stability}:

\begin{enumerate}[(i),ref=\roman*,wide,labelwidth=!,labelindent=5pt]
\item \label{analycity:dominantpoles}
%
%If~$\LSTstochstfunction$ has at least one singularity, and~$\dominantpoleX{\stochst}$ and~$\dominantpoleX{\Wt}$ respectively denote the dominant singularities of~$\LSTstochstfunction$ and~$\LSTWtfunction$, then~$\dominantpoleX{\stochst},\dominantpoleX{\Wt}\in\Realminuszero$ with~$\dominantpoleX{\stochst}<\dominantpoleX{\Wt}$.
The dominant singularity~$\dominantpoleX{\Wt}$ of~$\LSTWtfunction$  (i.e.,  that with largest real value)
%$\LSTWtfunction$ has no singularity with nonnegative real value. Its dominant singularity~$\dominantpoleX{\Wt}$, i.e that with largest real value, %of~$\LSTWtfunction$
 is a pole with degree~$1$ lying on the negative real axis~$\Realminus$. The dominant singularity~$\dominantpoleX{\stochst}$ of~$\LSTstochstfunction$ is real, negative (possibly infinite) and satisfies~$\dominantpoleX{\stochst}<\dominantpoleX{\Wt}$. $\LSTstochstfunction$ is analytic on~$\{\complex\in\Complex\setst \Realpart{\complex}> \dominantpoleX{\Wt}\}$.
\item \label{analycity:nopole}
$\LSTWtfunction$ is analytic on~$\{\complex\in\Complexzero\setst \Realpart{\complex}> \dominantpoleX{\Wt}\}$, %i.e,
where
$\lim\nolimits_{\complex\to\infty}\modulus{\LSTWts{\complex}}\leq 1$.

\item \label{analycity:isolatedpole}
%
%\item \label{analycity:analycity}
%$\LSTWtfunction$ has no singularity with nonnegative real value. In particular, it is analytic on the half-plane~$\realpart{\complex}> 0$. 
%
One can find~$\epsilon>0$ such that $\LSTWtfunction$ is analytic on~$\{\complex\in\Complexzero\setminus\{\dominantpoleX{\Wt}\}\setst \realpart{\complex}>\dominantpoleX{\Wt}-\epsilon\}$. 

\item \label{analycity:centered}
%Let~$\cut\in\Realplus$. The series~$\{\expectation{\Wt^k}/ \factorial{k}\}$ is asymptotically geometric with asymptotic rate~$1/\cut$ if and only if~$-\cut$ is the singularity of~$\LSTWtfunction$ with largest real value. In that case, $\LSTWtfunction$ is analytic at~$0$, where it rewrites as the series
$\LSTWtfunction$ is analytic in a neighborhood of~$0$, where it rewrites as the series
\noeqref{PKformulaseriescentered}%
\begin{equation} 
\label{PKformulaseriescentered}
\begin{array}{ll}
\LSTWts{\complex}
=
\sum\nolimits_{k=0}^\infty
\specialmatrixnk{}{k}\, (-\complex)^k
,
& \forall\complex\in\{\altcomplex\in\Complexzero:\modulus{\altcomplex-\sing} < \modulus{\dominantpoleX{\Wt}} \},
\end{array}
\end{equation}
\storecompoundcounter{equation}{PKformulaseriescentered}%
in which the coefficients~$\{\specialmatrixnk{}{k}\}$ are given by~\eqref{table:specialmatrixnk}  in Table~\ref{table:valuefunctions}, % in Table~\refeq{table:valuefunctions}, 
and satisfy~$\specialmatrixnk{}{k}=\expectation{\Wt^k}/ \factorial{k}$, for $k\in\Natural$.
The series~$\{\specialmatrixnk{}{k}\}$ is asymptotically geometric with asymptotic rate~$\modulus{\dominantpoleX{\Wt}}^{-1}$.

\item \label{analycity:excentered} At any point~$\sing\in\Complexzero$ where~$\LSTWtfunction$ is analytic, $\LSTWtfunction$ rewrites as the series
\noeqref{PKformulaseriesexcentered}%
\begin{equation}
\label{PKformulaseriesexcentered}
\begin{array}{ll}
\LSTWts{\complex}
= 
\scaley{\LSTWts{\sing}}{}
 \sum\nolimits_{k=0}^{\infty}  \poweracoefnk{\sing}{}{k}   (\sing-\complex)^{k} ,
 &
 \forall\complex\in\{\altcomplex\in\Complex:\modulus{\altcomplex-\sing} < \radiussing{\sing} \} ,
%\textup{if }\modulus{\complex-\sing} < \inf_{\singularity\in\Singf{\LSTWtfunction}} \modulus{\singularity-\sing} \} ,
 \end{array}
\end{equation}
\storecompoundcounter{equation}{PKformulaseriesexcentered}%
where~$\radiussing{\sing}$ % = \inf_{\singularity\in\Singf{\LSTWtfunction}} \modulus{\singularity-\sing}$ is the distance 
denotes the distance from~$\sing$ to the closest singularity of~$\LSTWtfunction$. The coefficients~$\{\poweracoefnk{\sing}{}{k}\}$ are given by~\eqref{table:poweracoefnk} in Table~\ref{table:valuefunctions}.
\end{enumerate}
\end{proposition}
\storecompoundcounter{proposition}{proposition:analycity}%
%

%\section{Value function identities} 
%\label{appendix:vf}

Next, we derive the identities of Section~\ref{section:complexcharacterization} for the value function. 
\obsolete{
The result first reported in Lemma~\ref{lemma:linear} is well known. It is a consequence of the law of large numbers.

\newcommand{\eqreflinearresultsT}{\noeqref{linearresults}\text{(\ref{linearresults}a)}}%
\newcommand{\eqreflinearresultsN}{\noeqref{linearresults}\text{(\ref{linearresults}b)}}%
\begin{lemma}\label{lemma:linear}
Let $0\leq \bli{2} \leq \bli{1}$ and
%and let the cost function~$\costfunction$ be equal to a constant~$\fixedcost$ almost everywhere on~$(\bli{1},\bli{2})$. 
suppose that the M/G/1 queue, initially at backlog state~$\bli{1}$, reaches state~$\bli{2}$ for the first time after a random period of time~$\stochT$, in which~$\stochN$ jobs have arrived. Then, 
\begin{equation}\label{linearresults}\begin{array}{rcl}
(a)\quad
\Expectation{\stochT}
=%&=&
%\frac{\bli{1}-\bli{2}}{1-\load}
 \frac{1}{1-\load}\, ({\bli{1}-\bli{2}})
,
\qquad
&%\\
(b) \quad
\Expectation{\stochN}
=%&=&
\frac{\ar}{1-\load} \, (\bli{1}-\bli{2})
%,&%\\
%\valuex{\bli{1}}
%-%&=&
%\valuex{\bli{2}}
%= 
%\frac{\ar(\fixedcost-\meancost)(\bli{1}-\bli{2})}{1-\load} %-\ar \meancost\Expectation{\stochT}
.
\end{array}\end{equation} 
\storecompoundcounter{equation}{linearresults}%
\end{lemma}
\storecompoundcounter{lemma}{lemma:linear}%
\obsolete{%
\begin{proof}
Lemma~\ref{lemma:linear} is a consequence of the law of large numbers. Consider~$n$ realizations of the setting, and denote by~$\stochTk{1},\dots,\stochTk{n}$ the random values observed for the variable~$\stochT$, by~$\stochNk{1},\dots,\stochNk{n}$ those observed for the variable~$\stochN$, and by~$\{\stochstkl{k}{l}\}_{l=1}^{\stochNk{k}},\dots,\{\stochstkl{k}{l}\}_{l=1}^{\stochNk{n}}$. Since the rate of the Poisson process is equal to the density of arrivals per unit of time,  
\begin{equation}\label{poissondensity}%\begin{array}{c}
\textstyle
\ar
=
\lim\limits_{n\to\infty}\frac{\sum\nolimits_{k=1}^{n}\stochNk{k}}{\sum\nolimits_{k=1}^{n}\stochTk{k}}
=
\Big(\lim\limits_{n\to\infty}\frac{\sum\nolimits_{k=1}^{n}\stochNk{k}}{n}\Big) \Big(\lim\limits_{n\to\infty} \frac{n}{\sum\nolimits_{k=1}^{n}\stochTk{k}}\Big)
=
\frac{\Expectation{\stochN}}{\Expectation{\stochT}} .
%\end{array}
\end{equation}
Similarly, 
\begin{equation}\label{stsumdensity}\textstyle%\begin{array}{c}%{rcl}
\lim\limits_{n\to\infty}\frac{1}{n}\sum\nolimits_{k=1}^{n}\sum\nolimits_{l=1}^{\stochNk{k}}\stochstkl{k}{l}
=%&=&
\Big(\lim\limits_{n\to\infty}\frac{\sum\nolimits_{k=1}^{n}\sum\nolimits_{l=1}^{\stochNk{k}}\stochstkl{k}{l}}{\sum\nolimits_{k=1}^{n}\stochNk{k}}\Big) \Big(\lim\limits_{n\to\infty} \frac{\sum\nolimits_{k=1}^{n}\stochNk{k}}{n}\Big)
=%\\&=&
\Expectation{\stochst}\Expectation{\stochN}.
%\end{array}
\end{equation}
By definition of the variables, we also have
%\begin{equation}\label{earlylinearidentity}\begin{array}{ll}
$
\bli{2}= \bli{1} +\sum\nolimits_{l=1}^{\stochNk{k}}\stochstkl{k}{l}-\stochTk{k}
%,
%&
% \quad(k=1,\dots,n),
%\end{array}\end{equation}
$
for $k=1,\dots,n$, and it follows that
\begin{equation}\label{linearidentity}\begin{array}{l}
\Expectation{\stochT}
=%&=&
\lim\nolimits_{n\to\infty}\frac{1}{n}\sum\nolimits_{k=1}^{n}\stochTk{k} 
%\\\refereq{(\refeq{earlylinearidentity})}
=%{=}& 
\bli{1}-\bli{2}+\lim\nolimits_{n\to\infty}\frac{1}{n}\sum\nolimits_{k=1}^{n}\sum\nolimits_{l=1}^{\stochNk{k}}\stochstkl{k}{l} 
%\\&\refereq{(\refeq{stsumdensity})}{=}&
\hspace{20mm}\\\hfill
\refereq{\eqref{stsumdensity}}{=} \bli{1}-\bli{2} + \Expectation{\stochst}\Expectation{\stochN}
%\\&\refereq{(\refeq{poissondensity})}{=}&
\refereq{\eqref{poissondensity}}{=}
 \bli{1}-\bli{2} + \load\Expectation{\stochT}
\end{array}
\end{equation}
which yields~(\ref{linearresults}a), and~(\ref{linearresults}b) follows from~\eqref{poissondensity}.
%Since the penalties are constant, their total cost over the time period that moves the system from~$\bli{1}$ to~$\bli{2}$ is equal to $\ar\fixedcost\stochN$, so that~(\refeq{valuefunction}) gives $\valuex{\bli{1}}= \ar \fixedcost \Expectation{\stochN} - \ar\meancost \Expectation{\stochT} +  \valuex{\bli{2}}$, which reduces to the third result and completes the proof.
\end{proof}
}%
 
}%

\begin{proof}[Proposition~\ref{proposition:poissonequation}]
Start the queue at state~$\bl$. The quantity~$ \ValuetxT{\ttt}{\bl}{\ti} $ appearing in~\eqref{valuefunction} rewrites, for any~$\stochT\geq0$ and for~$\ti$ large enough, as $ \ValuetxT{\ttt}{\bl}{\ti} = \ValuetxT{\ttt}{\bl}{\stochT} +  \ValuetxT{\ttt+\stochT}{\stochblt{\stochT}}{\ti-\stochT} $, where~$\stochblt{\stochT}$ denotes the backlog observed after time~$\stochT$.
It follows from the Markov property of the system and from the the definition~\ref{valuefunction} of the value function that 
\begin{equation}\label{Markovproperty} 
 \valuex{\bl} 
 =
 \expectation{ \ValuetxT{\ttt}{\bl}{\stochT}  - \ar\meancost \stochT  } + \expectation{ \valuex{\stochblt{\stochT}} }
.
\end{equation} 
%We now turn to show the claims, starting with~\eqref{proposition:vfdifferentiability}, from which useful properties of~$\valuefunction$ (existence, continuity) can be inferred. 

\renewcommand{\event}{Y}%
%
%\eqref{proposition:poisson} 
Now, consider the function 
\begin{equation}\label{MDPvaluefunction}
\textstyle
\altgenericx{\bl,\st}
=
\limbracenolimits{\Tim\to\infty}{ \sum_{\tim=1}^{\Tim} 
\Expectation{  
\costx{\stochblk{\tim} } -\meancost} \setst \vect{\stochblk{1},\stochstk{1}}=\vect{\bl,\st}
} + \meancost
,
\quad
\forall \bl,\st\in\Realpluszero,
\end{equation}
which can be verified to satisfy Equation~\eqref{MG1Poisson} by application of the Markov property to the \ac{MDP}.
The function~$\altgenericfunction$ defined by~\eqref{MDPvaluefunction} can be seen as a discrete-time counterpart of the value function~\eqref{valuefunction}, which 
follows  from~\eqref{MDPvaluefunction} by using the convention $\vect{\stochblk{0},\stochstk{0}}=\vect{\bl,0}$ and setting $\valuex{\bl}\equiv\MDPkernelfunction\altgenericx{\bl,0} - \meancost$ or, equivalently, from~\eqref{Markovproperty} by defining~$\stochT$ as the arrival time of the first job so that $\expectation{ \ValuetxT{\ttt}{\bl}{\stochT}  - \ar\meancost \stochT  } = -\meancost$ and $\expectation{ \valuex{\stochblt{\stochT}} } \equiv \MDPkernelfunction   \altgenericx{\bl,0} $.
\qed\end{proof}

\newcommand{\ignorealphafactor}[1]{}%
\begin{proof}[Theorem~\ref{theorem:solutionofPoisson}]
%\eqref{theorem:solutionofPoissoni}
A simple calculation reveals that $\UDeltax{\bl}=0$ if $\bl<0$, and  
$%\begin{equation}
\label{UDeltax}
\textstyle
\UDeltax{\bl}
=
(\MDPkernelfunction - \UMDPkernelfunction) \altgenericx{\bl,0}
% =
%  \int_{0}^{+\infty} \altgenericx{\altbl,\altst} \exp{-\ar\bl} \, \measureXx{\stochstzero}{d\altst}
% - \ar \int_{0}^{+\infty} \measureXx{\stochst}{d\altst} \int_{-\infty}^{0} \altgenericx{\altbl,\altst} \exp{-\ar(\bl-\altbl)} \, d\altbl
% =
%   \exp{-\ar\bl} \big[ \int_{0}^{+\infty} \altgenericx{\altbl,\altst} \, \measureXx{\stochstzero}{d\altst}
% - \ar \int_{0}^{+\infty} \measureXx{\stochst}{d\altst} \int_{-\infty}^{0} \altgenericx{\altbl,\altst} \exp{\ar\altbl} \, d\altbl \big]
= \Kquantity  \exp{-\ar\bl} 
$, %,
%\end{equation}
for $\bl\geq0$, with~$\Kquantity$ satisfying by%~\eqref{Kquantity}. %, where we define 
\begin{equation}\label{Kquantity}
\textstyle
\Kquantity
 =
 \bigexpectation{ \altgenericx{0,\stochstzero} 
- \ar  \int_{-\infty}^{0} \altgenericx{\bl,\stochst} \, \exp{\ar\bl} \, d\bl }
.
\end{equation}
We characterize the extended value function
$
\Uvaluefunction:\bl\in\Real\mapsto\Uvaluex{\bl}=\altgenericx{\bl,0} - \costx{\bl} \ignorealphafactor{+ \alphafactor\bl   \,\stepx{\bl} }  %+ \meancost   \,\stepx{\bl} + \frac{\ar\meancost }{1-\load} \bl \,\stepx{\bl}   
$
associated with some~$\altgenericfunction$ solution of~\eqref{rMG1Poisson}.
Note that, by construction, $\Uvaluefunction$ coincides with the value function on~$\Realpluszero$, i.e., $\Uvaluex{\bl} \equiv \valuex{\bl}$ if $\bl\geq 0$.
Once~$\Uvaluefunction$ is known, it will be possible to recover~$\altgenericfunction$  using %the identity
\begin{equation} \label{altgenericxidentity}
 \altgenericx{\bl,\st} 
 = 
 \altgenericx{\bl+\st,0} - \costx{\bl+\st} + \costx{\bl}  
 =
 \Uvaluex{\bl+\st} + \costx{\bl} \ignorealphafactor{- \alphafactor(\bl+\st)   \,\stepx{\bl+\st} } %- \meancost\, \stepx{\bl}
 .
\end{equation}
Consider $\complex$  %\in\ROC{\Blaplacetransformf{\Uvaluefunction}}$ 
in the region of absolute convergence of~$\Blaplacetransformf\Uvaluefunction$, where the orders of integration in our developments may be permuted. The two-sided Laplace transform of~$\Uvaluefunction$ is given by
\begin{equation}\notag
 \nocolsep
 \begin{array}{rcl}
 \Blaplacetransformfs{\Uvaluefunction}{\complex} 
 &=&
 \int_{-\infty}^{+\infty} \Uvaluex{\bl}\, \exp{-\complex\bl} \, d\bl
\\
&=&
\int_{-\infty}^{+\infty} [ \altgenericx{\bl,0} - \costx{\bl} \ignorealphafactor{+   \alphafactor\bl   \,\stepx{\bl}}%+\meancost\, \stepx{\bl}
] \, \exp{-\complex\bl} \, d\bl
\\
&\refereq{\eqref{rMG1Poisson}}{=}&
\int_{-\infty}^{+\infty} [ \UMDPkernelfunction \altgenericx{\bl,0} + \UDeltax{\bl} -\meancost \,\indicator{[0,+\infty)}{\bl}  ]\, \exp{-\complex\bl} \, d\bl \ignorealphafactor{ + \frac{\alphafactor}{\complex^2} }
\\  
&=& 
\int_{-\infty}^{+\infty} \exp{-\complex\bl} \, d\bl \int \altgenericx{\altbl,\st} \UMDPkernel{\bl,0}{d\vect{\altbl,\st}}  + \int_{0}^{+\infty} \Kquantity  \exp{-(\complex+\ar)\bl} \, d\bl - \frac{\meancost}{\complex} \ignorealphafactor{ + \frac{\alphafactor}{\complex^2} }
\\
&\refereq{\eqref{UMDPkernel}}{=}&
 \ar \bigexpectation{ \int_{-\infty}^{+\infty} \exp{-\complex\bl} \, d\bl \,  \int_{-\infty}^{\bl}  \altgenericx{\altbl,\stochst} \exp{-\ar(\bl-\altbl)}  \,d\altbl     }  - \frac{\meancost}{\complex}  + \frac{\Kquantity}{\complex+\ar} \ignorealphafactor{+ \frac{\alphafactor}{\complex^2} }
\\
&=&
\ar  \bigexpectation{ \int_{-\infty}^{+\infty}   \altgenericx{\altbl,\stochst}\, \exp{\ar\altbl} \,d\altbl \int_{\altbl}^{+\infty} \exp{-(\complex+\ar)\bl}     \, d\bl      } - \frac{\meancost}{\complex} + \frac{\Kquantity}{\complex+\ar}   \ignorealphafactor{+ \frac{\alphafactor}{\complex^2} }
\\
&\refereq{\eqref{altgenericxidentity}}{=}&
 \frac{\ar}{\complex+\ar}  \big\{ \bigexpectation{ \int_{-\infty}^{+\infty}    \Uvaluex{\altbl+\stochst}\, \exp{-\complex\altbl} \,d\altbl  \ignorealphafactor {-  \int_{-\stochst}^{+\infty}    \alphafactor(\altbl+\stochst) \,    \exp{-\complex\altbl} \,d\altbl  }  }    +  \int_{0}^{+\infty}  \costx{\altbl}  \,    \exp{-\complex\altbl} \,d\altbl   \big\}  - \frac{\meancost}{\complex} + \frac{\Kquantity}{\complex+\ar} \ignorealphafactor{+ \frac{\alphafactor}{\complex^2} }
\\
% &=&
% \frac{\ar}{\complex+\ar}  \big\{\LSTstochsts{-\complex} \,    \Blaplacetransformfs{\Uvaluefunction}{\complex}    +  \laplacetransformfs{\costfunction}{\complex}   \big\} + \frac{\Kquantity}{\complex+\ar} - \frac{\meancost}{\complex} 
% \ignorealphafactor{ + \frac{\alphafactor}{\complex^2} - \frac{\alphafactor\lambda\LSTstochsts{-\complex}}{\complex^2(\complex+\ar)} }
% %+\frac{\alphafactor(\complex+\ar)-\alphafactor\ar\LSTstochsts{-\complex}}{\complex^2(\complex+\ar)} 
% %+\frac{\alphafactor[\complex-\ar(\LSTstochsts{-\complex}-1)]}{\complex^2(\complex+\ar)} 
% %+\frac{\alphafactor(1-\load)\nearones{\complex}}{\complex(\complex+\ar)} 
% %+\frac{\alphafactor(1-\load)\nearones{\complex}}{\ar\complex} -\frac{\alphafactor(1-\load)\nearones{\complex}}{\ar(\complex+\ar)} 
% \\
&=& 
\frac{\ar}{\complex+\ar}  \big\{\LSTstochsts{-\complex} \,    \Blaplacetransformfs{\Uvaluefunction}{\complex}    +  \laplacetransformfs{\costfunction}{\complex}   - \frac{ \ignorealphafactor{\ar}\meancost \ignorealphafactor{- \alphafactor(1-\load)\nearones{\complex} }}{\ignorealphafactor{\ar}\complex}+ \frac{\Kquantity-\meancost}{\ar}   \big\}
.
\end{array}
\end{equation}
\ignorealphafactor{where $\nearones{\complex}=[1-(\ar/\complex)(\LSTstochsts{-\complex}-1)]/(1-\load)=1+\magnitude{\complex}$.}
Solving the above equation for~$\Blaplacetransformfs{\Uvaluefunction}{\complex}$ and using $  \LSTWts{-\complex} =  1+{\ar\expectation{\stochst^2}}/[{2(1-\load)}]\complex+\zero{\complex}$ yields, after computations,
\begin{equation} \label{BlUpurevaluefunction}
\nocolsep
\begin{array}{rcl}
% (\complex+\ar) \Blaplacetransformfs{\Uvaluefunction}{\complex}  =  \ar\LSTstochsts{-\complex} \,    \Blaplacetransformfs{\Uvaluefunction}{\complex}    -  \ar\laplacetransformfs{\costfunction}{\complex} -\frac{\ar\meancost}{\complex} + \Kquantity  
 \Blaplacetransformfs{\Uvaluefunction}{\complex}  
 &\refereq{\eqref{PKformula}}{=}&
 \frac{\ar\LSTWts{-\complex} }{(1-\load )\complex}\left[\laplacetransformfs{\costfunction}{\complex} -\frac{\ignorealphafactor{\ar}\meancost \ignorealphafactor{-\alphafactor(1-\load)\nearones{\complex}}}{\ignorealphafactor{\ar}\complex} + \frac{\Kquantity-\meancost}{\ar}  \right]
 \\
%  &=&
%  \frac{\ar}{(1-\load )\complex}\LSTWts{-\complex}\laplacetransformfs{\costfunction}{\complex}  
% +
%  \frac{\ar[1+\frac{\ar\expectation{\stochst^2}}{2(1-\load)}\complex+\zero{\complex}] }{(1-\load )\complex}\left[ -\frac{\ignorealphafactor{\ar}\meancost \ignorealphafactor{-\alphafactor(1-\load)\nearones{\complex}}}{\ignorealphafactor{\ar}\complex} + \frac{\Kquantity-\meancost}{\ar}  \right]
%  \\
%  &=&
%  \frac{\ar}{(1-\load )}\left\{\frac{1}{s} \LSTWts{-\complex}\laplacetransformfs{\costfunction}{\complex}  - \frac{\meancost}{\complex^2} + \big[\frac{\Kquantity-\meancost}{\ar} - \frac{\ar\expectation{\stochst^2}\meancost}{2(1-\load)}\big] \frac{1}{\complex} +\zero{\frac{1}{\complex}} \right\}
%  \\
%  &=&
%  \frac{\ar}{(1-\load )}\left\{\frac{1}{s} \LSTWts{-\complex}\laplacetransformfs{\costfunction}{\complex}  - \frac{\meancost}{\complex^2} + \big[\frac{(1-\load)\Kquantity-\MaclaurinLSTstochsts{2}{\ar}\meancost}{\ar(1-\load)}\big] \frac{1}{\complex} +\zero{\frac{1}{\complex}} \right\}
%  \\
%  &=&
%  \frac{\ar}{(1-\load )}\left\{\frac{1}{s} \LSTWts{-\complex}\laplacetransformfs{\costfunction}{\complex}  - \frac{\meancost}{\complex^2} + \big[\frac{\Kquantity}{\ar}-\frac{\MaclaurinLSTstochsts{2}{\ar}\meancost}{\ar(1-\load)}\big] \frac{1}{\complex} +\zero{\frac{1}{\complex}} \right\}
% \\
 &=&
  \frac{\ar}{(1-\load )\complex} \LSTWts{-\complex}\laplacetransformfs{\costfunction}{\complex}  -  \frac{\ar\meancost}{(1-\load )\complex^2} +  \frac{\Ujump}{\complex} +
  \frac{\analyticparts{\complex}}{\complex}
  %\zero{\frac{1}{\complex}} 
  ,
\end{array}
\end{equation} 
where $\analyticparts{\complex}$ has no singularities on~$\{\complex\in\Complex\setst \Realpart{\complex}<- \dominantpoleX{\Wt}\}$, and~$\Ujump$ is given by 
\begin{equation}\label{jump}
\Ujump 
= 
%\frac{\Kquantity}{1-\load}-\frac{\MaclaurinLSTstochsts{2}{\ar}}{(1-\load)^2} \meancost
{\Kquantity}/{(1-\load )}-{\MaclaurinLSTstochsts{2}{\ar}\meancost}/{(1-\load)^2} 
,
\end{equation}
%is common to all solutions, with
with $
\MaclaurinLSTstochsts{2}{\ar}\defeq 1-\ar\expectation{\stochst}+\ar^2\expectation{\stochst^2}/2
$.
Since~$\Uvaluefunction$ is expected to be asymptotically flat for $\bl\to-\infty$, the $- {\ar\meancost}/[{(1-\load )\complex^2}] $ term in~\eqref{BlUpurevaluefunction} is necessarily due to a term $ -{\ar\meancost }/({1-\load})\, \bl$ on~$\Realpluszero$ in the backlog domain. 
By inverse transformation of~\eqref{BlUpurevaluefunction}, we find %that~$\Uvaluefunction$ has the form
\begin{align}\label{earlyUstemsolution}\tag{$\makeU{\text{S}}$}
\textstyle
%\begin{array}{ll}
\Uvaluex{\bl}
&\textstyle=
\Uvaluex{0} +\Upurevaluex{\bl} -\frac{\ar\meancost }{1-\load} \bl \, \indicator{[0,+\infty)}{\bl}  + \anticausalx{\bl} \, \indicator{(-\infty,0)}{\bl}
, &%\quad 
\forall\bl\in\Real,
%\end{array}
\end{align}
where~$\anticausalfunction$ satisfies 
%$\laplacetransformfs{[\anticausalx{-\cdot}-\Ujump]}{\complex} =   -{\analyticparts{-\complex}}/{\complex}$.
$\laplacetransformfs{\anticausalx{-\cdot}}{\complex} =   -[{\analyticparts{-\complex}+\Ujump}]/{\complex}$.
The general form for~$\altgenericfunction$ follows from~\eqref{altgenericxidentity}, \eqref{earlyUstemsolution} and $\Uvaluex{\bl} \equiv \valuex{\bl}$ on~$\Realpluszero$. 
 The non-empty \ac{ROC} of~$\Blaplacetransformf{\rightderivative\Uvaluefunction}$ is the consequence of Assumption~\ref{assumption:complexintegrablecost}.

It remains to show that the function~$\anticausalfunction$ is identical for all solutions or, equivalently, that the quantity~$\Kquantity$  in~\eqref{Kquantity} is the same for all~$\altgenericfunction$.
% 
% ------------
% 
% 
% $ \altgenericx{\bl,\st} = \Uvaluex{\bl+\st} + \costx{\bl} + \anticausalx{\bl+\st} \, \indicator{(-\infty,-\st)}{\bl}$
% 
% \begin{align}\label{Ustemsolution}\tag{$\makeU{\text{S}}$}
% %\begin{array}{ll}
% \Uvaluex{\bl}&=\valuex{0} +\purevaluex{\bl} -\frac{\ar\meancost }{1-\load} \bl \, \indicator{[0,+\infty)}{\bl}  + \anticausalx{\bl} \, \indicator{(-\infty,0)}{\bl} , &%\quad 
% \forall\bl\in\Real,
% %\end{array}
% \end{align}
% 
% 
% \begin{equation} \label{bilaterallaplacetransformVF} \tag{\acs{wfunction}b}
% \storeuniversalcountertag{bilaterallaplacetransformVF}{\acs{wfunction}b}%
% %\textstyle
%  \Blaplacetransformfs{\rightderivative\Upurevaluefunction}{\complex}  
% =
% \frac{\ar}{(1-\load )} \LSTWts{-\complex}\laplacetransformfs{\costfunction}{\complex} +\Ujump  +\analyticparts{\complex}%\zero{1} 
% ,
% \end{equation}  
% 
% ---------
% 
%It can be seen that~$\Blaplacetransformf{\rightderivative\Uvaluefunction}$ in~\eqref{BlUdpurevaluefunction} has non-empty \ac{ROC} under  Assumption~\ref{assumption:complexintegrablecost}.
%
%Since inverse transformation of~$\Blaplacetransformf{\rightderivative\Uvaluefunction}$ is unaffected by the term $\Ujump +\analyticparts{\complex}$, we infer from~\eqref{bilaterallaplacetransformVF} that the derivative~$\rightderivative\Upurevaluefunction$ coincides with the derivative of the extended \acl{wfunction}~$\rightderivative\extpurevaluefunction$.
% 
 \newcommand{\harmonic}{\alpha}%
 \newcommand{\basisjump}{\bar\Ujump}%
 \newcommand{\probapos}{p^+}%
 \newcommand{\probaneg}{p^-}%
 %Statement~\eqref{theorem:solutionofPoissonii}  will then follow from~\eqref{theorem:solutionofPoissoni} and from~\eqref{BlUpurevaluefunction} on condition 
 %It remains to show that the quantity~$\Kquantity$ given in~\eqref{Kquantity} is identical for all solutions. 
To see this, consider a solution~$\altgenericifunction{1}$ of~\eqref{rMG1Poisson} with associated value function~$\Uvalueifunction{1}$ and jump~$\Ujumpi{1}$ at $\bl=0$. The value function for every other solution~$\altgenericifunction{2}$ rewrites as $\Uvalueifunction{2}=\harmonic + (\Uvalueifunction{1} - \Ujumpi{2} + \Ujumpi{1} ) \indicatorfunction{\Realminus}+ \Uvalueifunction{1} \indicatorfunction{\Realpluszero}$, where $\Ujumpi{1}$, $\Ujumpi{2}$ and~$\harmonic$ are constants. We show that $\Ujumpi{1}=\Ujumpi{2}$. If we successively compute the expression~\eqref{jump} for~$\Ujumpi{1}$ and~$\Ujumpi{2}$, using~\eqref{Kquantity}, \eqref{altgenericxidentity} and the extension of~\eqref{MG1Poisson}, we get, after simplifications, $
 \Ujumpi{2} = 
 %\Ujumpi{1}  - (  \Ujumpi{1} - \Ujumpi{2} ) \expectation{       \int_{-\infty}^{0}   \indicator{\Realminus}{\altbl+\stochst}  \, \ar\exp{\ar\altbl} \, d\altbl}/{(1-\load )}
 %=
 \Ujumpi{1}  - (  \Ujumpi{1} - \Ujumpi{2} ) \LSTstochsts{\ar}    /{(1-\load )}
 $. Exploiting twice the strict convexity of~$\exp{-x}$, we find $
\LSTstochsts{\ar} = \expectation{ \exp{-\lambda\stochst} } >  \exp{-\lambda\expectation{\stochst} }  > 1 - \lambda\expectation{\stochst} = 1-\load
$. Hence, $\LSTstochsts{\ar}    /{(1-\load )} \neq 1$ and, consequently, $\Ujumpi{1}=\Ujumpi{2}$.
\qed\end{proof}

\begin{proof}[Proposition~\ref{proposition:identities}]
\eqref{proposition:differentialequation}
Consider~$\complex$  %\in\ROC{\Blaplacetransformf{\Uvaluefunction}}$ 
in the region of absolute convergence of~$\Blaplacetransformf\Uvaluefunction$. 
Since $\Blaplacetransformfs{\rightderivative\Upurevaluefunction}{\complex}=\complex\Blaplacetransformfs{\Upurevaluefunction}{\complex}$,   \eqref{bilaterallaplacetransformVF} rewrites as
%  $\laplacetransformfs{\anticausalx{-\cdot}}{-\complex} =   [{\analyticparts{\complex}+\Ujump}]/{\complex}$
\begin{equation} \label{dualBlUdpurevaluefunction}\notag
%\textstyle
[\complex+ \ar (1- \LSTstochsts{-\complex}) ] \, \Blaplacetransformfs{\Upurevaluefunction}{\complex}
\refereq{\eqref{PKformula}}{=}
\ar\laplacetransformfs{\costfunction}{\complex} + \complex \laplacetransformfs{\anticausalx{-\cdot}}{-\complex} 
%\left[\Ujump  +\analyticparts{\complex}\right] 
\left[1+ \ar/\complex (1- \LSTstochsts{-\complex}) \right]%\zero{1} 
,
\end{equation}
while  transformation of~\eqref{Ustemsolution} gives
$%\begin{equation}
%\begin{array}{ll}
\Blaplacetransformfs{\Upurevaluefunction}{\complex}  = \Blaplacetransformfs{\Uvaluefunction}{\complex}
 + {\ar\meancost }/[{(1-\load)\complex^2}]  
%\end{array}
$. %\end{equation}
Besides,
\begin{equation}\notag
\nocolsep
\begin{array}{l}
\LSTstochsts{-\complex} \Blaplacetransformfs{\Uvaluefunction}{\complex} 
=
\expectation{\exp{\complex\stochst}} \int_{-\infty}^{+\infty} \Uvaluex{\bl} \exp{-\complex\bl} d\bl   
=
\expectation{ \int_{-\infty}^{+\infty} \Uvaluex{\bl} \exp{-\complex(\bl-\stochst)} d\bl }   
\\\qquad\qquad\qquad\qquad
=
\expectation{ \int_{-\infty}^{+\infty} \Uvaluex{\altbl+\stochst} \exp{-\complex\altbl} d\altbl }   
=
\int_{-\infty}^{+\infty} \expectation{ \Uvaluex{\altbl+\stochst} } \exp{-\complex\altbl} d\altbl    
=
\Blaplacetransformfs{\expectation{\Uvaluex{\cdot+\stochst}}}{\complex}.
\end{array}
\end{equation}
Combining the above with $\Blaplacetransformfs{\rightderivative\Uvaluefunction}{\complex}=\complex\Blaplacetransformfs{\Uvaluefunction}{\complex}$, we get, after computations,
% \begin{equation} \label{dualBlUdvaluefunction}
% %\textstyle
% \Blaplacetransformfs{\rightderivative\Uvaluefunction}{\complex}   + \ar (\Blaplacetransformfs{\Uvaluefunction}{\complex} - \Blaplacetransformfs{\expectation{\Uvaluex{\cdot+\stochst}}}{\complex} )  %= \ar\laplacetransformfs{\costfunction}{\complex} - \left[{\ar\meancost }/[{(1-\load)\complex}]   -\Ujump  -\analyticparts{\complex}\right]  \left[1+ \ar/\complex (1- \LSTstochsts{-\complex}) \right]%\zero{1} 
% = \ar\laplacetransformfs{\costfunction}{\complex} - \left[{\ar\meancost }/[{(1-\load)\complex}]   -\Ujump  -\analyticparts{\complex}\right]  \left[1-\load +\zero{1} \right]%\zero{1} 
% ,
% \end{equation}
\begin{equation} \label{dualBlUdvaluefunction}
%\textstyle
\Blaplacetransformfs{\rightderivative\Uvaluefunction}{\complex}      
%= \ar\left[\laplacetransformfs{\costfunction}{\complex} + \Blaplacetransformfs{\expectation{\Uvaluex{\cdot+\stochst}}}{\complex}   - \Blaplacetransformfs{\Uvaluefunction}{\complex}\right] - \left[{\ar\meancost }/\complex   - (1-\load)\Ujump  - (1-\load) \analyticparts{\complex}\right]  \left[1+\zero{1} \right]%\zero{1} 
%= \ar\left[\laplacetransformfs{\costfunction}{\complex} -\meancost/\complex + \Blaplacetransformfs{\expectation{\Uvaluex{\cdot+\stochst}}}{\complex}   - \Blaplacetransformfs{\Uvaluefunction}{\complex}\right] + \left[    (1-\load)\Ujump  + (1-\load) \analyticparts{\complex}\right] + \left[{\ar\meancost }/\complex   - (1-\load)\Ujump  - (1-\load) \analyticparts{\complex}\right] \zero{1} %\zero{1} 
= \ar\left[\laplacetransformfs{\costfunction}{\complex} -\meancost/\complex + \Blaplacetransformfs{\expectation{\Uvaluex{\cdot+\stochst}}}{\complex}   - \Blaplacetransformfs{\Uvaluefunction}{\complex}\right] + \altanalyticparts{\complex} %\left[    (1-\load)\Ujump  + (1-\load) \analyticparts{\complex}\right] + \left[{\ar\meancost }/\complex   - (1-\load)\Ujump  - (1-\load) \analyticparts{\complex}\right] \zero{1} %\zero{1} 
,
\end{equation}
where~$\altanalyticparts{\complex}$ shows no singularity on~$\{\complex\in\Complex\setst \Realpart{\complex}<- \dominantpoleX{\Wt}\}$, and we have used Proposition~\refeq{proposition:analycity}\eqref{analycity:dominantpoles} and $ 1+ \ar/\complex (1- \LSTstochsts{-\complex})
%= 1+ \ar/\complex (1- \expectation{\exp{\complex\stochst}})= 1+ \ar/\complex (1- \expectation{1+\complex\stochst+\zero{\complex}}) =  1- \ar\expectation{\stochst} +\zero{1}
= 1-\load +\zero{1} $.
Inverse Laplace transformation of~\eqref{dualBlUdvaluefunction} then gives, at every $\bl\geq0$ where~$\Uvaluefunction$ is differentiable,
\begin{align}\label{Uderesult}%\tag{DE}
%\begin{array}{ll}
\dUvaluex{\bl}   
=
 \ar\left(
%\one{\bl-\stochdl}} \cost 
 \costx{\bl}
-  \meancost
+
\Expectation{\Uvaluex{\bl+\stochst}-  \Uvaluex{\bl}  }\right)
,
%, \forall \bl\in \Realpluszero,
%\end{array}
\end{align}
which holds for almost every~$\bl>0$ by piecewise continuity of~$\costfunction$. 
Since by construction $\valuex{\bl}\equiv\Uvaluex{\bl}$ for $\bl\geq 0$, we find~\eqref{deresult}.
From Theorem~\ref{theorem:solutionofPoisson}, we have
\begin{equation} \label{Uvaluexzero}
 \textstyle
\Uvaluex{0} =  \altgenericx{0,0} - \costx{0}  \refereq{\eqref{rMG1Poisson}}{=} 
\MDPkernelfunction\altgenericx{0,0}    - \meancost 
\refereq{\eqref{fullMDPkernel}}{=}
\expectation{\Uvaluex{\stochstzero}} + \costx{0} - \meancost 
,
\end{equation}
which yields~\eqref{boundarycondition}.
Finally, we find~\eqref{derivativeatzero} by taking the limit of~\eqref{Uderesult} as $\bl\to0^+$,
\begin{align}\notag%\tag{DE}
\begin{array}{l}
\dUvaluex{0}   
=%\refereq{\eqref{Uderesult}}{=}
 \ar\left(
%\one{\bl-\stochdl}} \cost 
 \rightcostx{0}
-  \meancost
+
\Expectation{\Uvaluex{\stochst}-  \Uvaluex{0}  }\right)
%\qquad \qquad\qquad \qquad\qquad \qquad \\ \hfill
\refereq{\eqref{Uvaluexzero}}{=}
 \ar\left( \rightcostx{0}  + \Expectation{\Uvaluex{\stochst}}-  \expectation{\Uvaluex{\stochstzero}}  - \costx{0}  \right)
 .
\end{array}
 \end{align}
 %which is~\eqref{derivativeatzero}.

% $  \Uvaluex{0} =  \altgenericx{0,0} - \costx{0}$
% 
% From either~\eqref{MG1Poisson} or~\eqref{rMG1Poisson}
% $
% %
% \altgenericx{0,0} 
% =
% \MDPkernelfunction
% %\MDPkernelarfunction{\ari{\server}} 
% \altgenericx{0,0}  +  \costx{0}  - \meancost 
% $
% 
% $ \MDPkernelfunction\altgenericx{0,0} 
%   =
% %\iint
% \int
%  \altgenericx{\altbl,\altst}  \MDPkernel{0,0}{d\vect{\altbl,\altst}}
%  =
% \int
%  ( \Uvaluex{\altbl+\altst} + \costx{\altbl})   \MDPkernel{0,0}{d\vect{\altbl,\altst}}
% \refereq{\eqref{fullMDPkernel}}{=}
% \expectation{\Uvaluex{\stochstzero}} + \costx{0}
% $

\eqref{proposition:vfdifferentiability}
Equation~\eqref{stemsolution} follows directly from~\eqref{Ustemsolution} and the fact that $\valuex{\bl}\equiv\Uvaluex{\bl}$ for $\bl\geq 0$.
It remains to compute~$\rightderivative\Upurevaluefunction$.
From Theorem~\ref{theorem:solutionofPoisson}, we get
%
% In view of~\eqref{Ustemsolution} and~\eqref{stemsolution}, we can derive~$\purevaluex{\bl}$ on~$\Realpluszero$ by inverse Laplace transform of~\eqref{bilaterallaplacetransformVF}, which rewrites on  its \ac{ROC} as
\begin{equation} \label{altBlUdpurevaluefunction}
\nocolsep
\begin{array}{l}
 \Blaplacetransformfs{\rightderivative\Upurevaluefunction}{\complex} 
 %-\Ujump  -\analyticparts{\complex}
\refereq{\eqref{bilaterallaplacetransformVF}}{=}
\frac{\ar}{(1-\load )} \LSTWts{-\complex}\laplacetransformfs{\costfunction}{\complex} 
=
\frac{\ar}{(1-\load )} \expectation{\exp{\complex\Wt}}\int_{-\infty}^{+\infty}\costx{\bl}\exp{-\complex\bl}\, d\bl 
\qquad\qquad\\\hfill
=
\frac{\ar}{(1-\load )} \expectation{\int_{-\infty}^{+\infty}\costx{\bl}\exp{-\complex(\bl-\Wt)}\, d\bl } 
=
\frac{\ar}{(1-\load )} \int_{-\infty}^{+\infty}\expectation{\costx{\altbl+\Wt}}\exp{-\complex(\altbl)}\, d\altbl  
\\\hfill
=
\frac{\ar}{(1-\load )} \Blaplacetransformfs{\expectation{\costx{\cdot+\Wt}}}{\complex} 
,
\end{array}
\end{equation}  
where we have used %the convention  
$\costx{\bl}=0$ if $\bl<0$.
%Since~$\analyticparts{\complex}/\complex$ is analytic  on~$\{\complex\in\Complex\setst \Realpart{\complex}<- \dominantpoleX{\Wt}\}$, inverse transformation of~$\Blaplacetransformf{\rightderivative\Uvaluefunction}$ is unaffected by the term $\Ujump +\analyticparts{\complex}$, and we find that~\eqref{stemsolution} holds for $\bl\geq 0$. The $\bl<0$ part of~$\valuefunction$ is arbitrary, and we chose to define it  by continuation of~\eqref{stemsolution} to the negative backlog values. The Laplace transform~\eqref{bilaterallaplacetransformVF} follows immediately.
%
Equation~\eqref{bilaterallaplacetransformVF} follows by inversion of~\eqref{altBlUdpurevaluefunction}.
\qed\end{proof}

\obsolete{
Corollary~\ref{vfcorollary} follows directly from Proposition~\ref{proposition:identities}\eqref{proposition:vfdifferentiability}.
%  Corollary~\ref{vfcorollary}
\begin{corollary}\label{vfcorollary} %Consider an M/G/1 queue with arrival rate~$\ar$ and service times $\vect{\stochst,\stochstzero}$ meeting Assumption~\refeq{assumption:stability}. %. 
In the server model of Section~\ref{section:valuefunction}, 
let~
%Let
$\costifunction{1},\costifunction{2}:\Realpluszero\mapsto\CtoR{\Complex}{\Real}$ be two cost functions satisfying Assumption~\ref{assumption:complexintegrablecost}, and denote the corresponding value functions, mean costs per job, and \aclp{wfunction} %[cf. ~(\refeq{stemsolution})] 
by~$\valueifunction{1}$, $\meancosti{1}$, $\purevalueifunction{1}$ and~$\valueifunction{2}$, $\meancosti{2}$, $\purevalueifunction{2}$, respectively. 
\begin{enumerate}[(a),ref=\roman*,wide,labelwidth=!,labelindent=5pt]
\item \label{lemma:equality} If $\costix{1}{\bl}=\costix{2}{\bl}$ almost everywhere on~$\Realpluszero$ with respect to the probability measure~$\measureXfunction{\Wt}$ (\lq{}$\measureXfunction{\Wt}$-a.e.\rq{}), then $\meancosti{1}=\meancosti{2}$,   $\purevalueifunction{1}=\purevalueifunction{2}$, and~$\valueifunction{1}=\valueifunction{2}$. 
\item \label{lemma:additivity}
If $\valuefunction$, $\meancost$, $\purevaluefunction$ are the value function, the mean cost %per job 
and the \acl{wfunction} associated with the cost function~$\costifunction{1}+\costifunction{2}$, then $\meancost=\meancosti{1} + \meancosti{2}$%
, $ \purevaluefunction=\purevalueifunction{1} + \purevalueifunction{2}$, and
$\valuefunction-\valuex{0}=\valueifunction{1}-\valueix{1}{0} + \valueifunction{2}-\valueix{2}{0} 
$.
\obsolete{
and 
\begin{equation}\label{equationsum}\begin{array}{rcll}
%\meancost&=&\meancosti{1} + \meancosti{2},&
%\\
\purevaluex{\bl}&=& \purevalueix{1}{\bl} + \purevalueix{2}{\bl}
,& \quad\forall\bl\in\Realpluszero,
\\
\valuex{\bl}-\valuex{0}&= &\valueix{1}{\bl}-\valueix{1}{0} + \valueix{2}{\bl}-\valueix{2}{0} 
,& \quad\forall\bl\in\Realpluszero.
\end{array}
\end{equation}
}%
\item \label{lemma:comparison}
If \CtoR{%
$\costifunction{1},\costifunction{2}:\Realpluszero\mapsto\Real$ and
}{%
}%
$\costix{1}{\bl}\leq\costix{2}{\bl}$ %for every $\bl\in\Realpluszero$, 
$\measureXfunction{\Wt}$-a.e. on~$\Realpluszero$,
then $\meancosti{1} \leq \meancosti{2}$ and 
$\purevalueifunction{1}\leq \purevalueifunction{2}$.
\obsolete{
\begin{equation}\label{equationcomparison}\begin{array}{ll}
\purevalueix{1}{\bl} \leq \purevalueix{2}{\bl}
,& \quad\forall\bl\in\Realpluszero.
\end{array}
\end{equation}
}%
\item \label{lemma:comparisoncomplex}
If \CtoR{%
$\costifunction{2}:\Realpluszero\mapsto\Realpluszero$ and
}{%
$\costifunction{2}\geq 0$ and
}
$\modulus{\costix{1}{\bl}}\leq\costix{2}{\bl}$ $\measureXfunction{\Wt}$-a.e. on~$\Realpluszero$,
%for every $\bl\in\Realpluszero$, 
then $\modulus{\meancosti{1} } \leq \meancosti{2}$ and
$
\modulus{\purevalueifunction{1} }\leq \purevalueifunction{2}
$.
\obsolete{
\begin{equation}\label{equationcomparisoncomplex}\begin{array}{ll}
\Modulus{\purevalueix{1}{\bl} }\leq \purevalueix{2}{\bl}
%\valueix{1}{\bl}-\valueix{1}{0}\leq \valueix{2}{\bl}-\valueix{2}{0} -  \frac{\ar(\meancosti{1}-\meancosti{2})}{1-\load} \bl
,&\quad \forall\bl\in\Realpluszero.
\end{array}
\end{equation}
}%
\end{enumerate}
\end{corollary}
 }%
 
 \obsolete{
\begin{proof}[Proposition~\ref{proposition:bilateraltransformdWF}]
In the region of absolute convergence of~$\Blaplacetransformf{\extdpurevaluefunction}$, we find, %\hidecalculustheorems{, using Fubini's theorem,}{}
\begin{equation}\label{bilaterallaplace}
\begin{array}{rcl}
\Blaplacetransformfs{\extdpurevaluefunction}{\complex}
&\refereq{\eqref{dwfunction}}{=}&
\frac{\ar}{1-\load}%\complex}
\int\nolimits_{-\infty}^{+\infty}
\exp{-\complex\bl} 
\, \expectation{
\costx{\bl+\Wt}\, \stepx{\bl+\Wt}}
\, d\bl
\\&\hidecalculustheorems{=}{\refereq{\eqref{fubini}}{=}}&
\frac{\ar}{1-\load}%\complex}
\expectation{
\int\nolimits_{-\infty}^{+\infty}
\exp{-\complex\bl}\,  
\costx{\bl+\Wt}\, \stepx{\bl+\Wt}
\, d\bl
}
\\&=&
\frac{\ar}{1-\load}%\complex}
\expectation{
\int\nolimits_{-\infty}^{+\infty}
\exp{-\complex(\altbl-\Wt)} 
 \costx{\altbl}\, \stepx{\altbl}
\, d\altbl
}
\\&=&
\frac{\ar}{1-\load}%\complex} 
 \expectation{\exp{\complex\Wt} }
%\LSTWts{-\complex}
\int\nolimits_{0}^{+\infty}
\exp{-\complex\altbl}
\costx{\altbl}
\, d\altbl
\obsolete{
\\&=&
\frac{\ar}{(1-\load)}%\complex} 
% \Expectation{\exp{\complex\Wt} }
\LSTWts{-\complex}
\laplacetransformfs{\costfunction}{\complex}
}
,
\end{array}
\end{equation}
where the Laplace-Stieltjes transform of~$\Wt$,
$\LSTWts{\complex} 
 =
 \expectation{\exp{-\complex\Wt}}
 $ is  given by the Pollaczek-Khinchin formula~\eqref{PKformula}, \cite{pollaczek30,khintchine32,gross98}. 
 %A possible way to derive~\eqref{PKformula} is to first infer the probability distribution function of~$\Wt$ from Proposition~\ref{proposition:identities}\eqref{proposition:vfdifferentiability} by using the cost function~$\costx{\altbl}=\stepx{\bl-\altbl}$, with the convention~$\stepx{x}=1$ if $x\geq0$ and $\stepx{x}=0$ otherwise. 
Then, the nonemptiness of the \ac{ROC} of~$\Blaplacetransformf{\extdpurevaluefunction}$ is a direct consequence of Assumption~\ref{assumption:complexintegrablecost}.
\qed\end{proof} 
}%

%%%%%%%%%%

\obsolete{
\begin{proof}[Equation~\eqref{BlDEdpurevaluefunction}]
We characterize a solution
$
\DEvaluefunction:\bl\in\Real\mapsto\DEvaluex{\bl}=\altgenericx{\bl,0} - \costx{\bl} \ignorealphafactor{+ \alphafactor\bl   \,\stepx{\bl} }  %+ \meancost   \,\stepx{\bl} + \frac{\ar\meancost }{1-\load} \bl \,\stepx{\bl}   
$
of an extension of~\eqref{deresult} to negative backlogs.
In the region of absolute convergence of~$\Blaplacetransformf\DEvaluefunction$,
\begin{equation}   
 \nocolsep
 \begin{array}{rcl}
 \complex \Blaplacetransformfs{\DEvaluefunction}{\complex}
& \refereq{\eqref{deresult}}{=}&
  \int_{-\infty}^{+\infty} \ar  \big\{ \rightcostx{\bl}-  \meancost \, \indicator{\Realpluszero}{\bl} +\Expectation{\DEvaluex{\bl+\stochst}-  \DEvaluex{\bl}  }\big\} \, \exp{-\complex\bl} \, d\bl
\\
&=&
  \ar  \big\{ \laplacetransformfs{\costfunction}{\complex} - \frac{\meancost}{\complex} + \int_{-\infty}^{+\infty}  \Expectation{\DEvaluex{\bl+\stochst}-  \DEvaluex{\bl}  } \, \exp{-\complex\bl} \, d\bl \big\}
\\
&=&
  \ar  \big\{ \laplacetransformfs{\costfunction}{\complex} - \frac{\meancost}{\complex} +  \bigexpectation{\int_{-\infty}^{+\infty} [\DEvaluex{\bl+\stochst}-  \DEvaluex{\bl}  ] \, \exp{-\complex\bl} \, d\bl} \big\}
\\
&=&
  \ar  \big\{ \laplacetransformfs{\costfunction}{\complex} - \frac{\meancost}{\complex} +  (\expectation{\exp{\complex\stochst}}-  1  )  \Blaplacetransformfs{\DEvaluefunction}{\complex} \big\}
,
\end{array}
\end{equation}
where $\expectation{\exp{\complex\stochst}}=\LSTstochsts{-\complex}$. We find~\eqref{BlDEdpurevaluefunction} by solving the above for~$\Blaplacetransformfs{\DEvaluefunction}{\complex}$, and using~\eqref{PKformula} and $  \LSTWts{-\complex} =  1+{\ar\expectation{\stochst^2}}/[{2(1-\load)}]\complex+\zero{\complex}$.
%  $$
%  % (\complex +\ar[1 - \LSTstochsts{-\complex}] )
%  \frac{(1-\load )\complex}{\LSTWts{-\complex}}
%   \Blaplacetransformfs{\DEvaluefunction}{\complex}
%  =
%    \ar  \big\{ \laplacetransformfs{\costfunction}{\complex} - \frac{\meancost}{\complex} \big\}
%  $$
% $$
%  \Blaplacetransformfs{\DEvaluefunction}{\complex}
% =
% \frac{\ar}{(1-\load )\complex} \LSTWts{-\complex}
%      \laplacetransformfs{\costfunction}{\complex} - \left(\frac{\ar}{(1-\load )}\right)  \frac{\meancost}{\complex^2}  - \left(\frac{\ar^2\expectation{\stochst^2} \meancost}{2(1-\load)^2}\right) \frac{1}{\complex} +\zero{1/\complex}  
% $$
% $$
%  \LSTWts{-\complex} 
% =
%  \frac{(1-\load )\complex} {\complex+ \ar (1- \LSTstochsts{-\complex}) }
% $$
\qed\end{proof}
}%

%%%%%%%%%%
%%%%%%%%%%
%%%%%%%%%%

%Appendix~\ref{appendix:rates}
\section{Moments of the asymptotic waiting times and rates of growth}
\label{appendix:rates}
In Table~\ref{table:moments} we derive the coefficient sequence~$\{\specialmatrixnk{}{k}\}$ for standard service time distributions (constant, exponential, Erlang), and study its asymptotic growth. 
The moments of~$\Wt$ and their growth rates can be inferred from those of~$\{\specialmatrixnk{}{k}\}$ using the identity $\expectation{\Wt^k}=\factorial{k}\specialmatrixnk{}{k}$.
%
%
% where~$\agr=1$ if~$\stochst$ is identically distributed, and~$\agr\in(1,\infty]$ otherwise.
%

% Table~\ref{table:moments}
\begin{table}
 \caption{Germ of~$\LSTWts{-\complex}$ at $\complex=0$ for constant (M/D/1), exponentially-distributed (M/M/1) and Erlang-distributed (M/E\textsubscript{$\shape$}/1) service times.
\label{table:moments}}
\horizontalline
\smallskip
\textbf{M/D/1 (constant):} $\stochst=\st$, with~$\st>0$; $\expectation{\stochst^{k}}=\st^k$; $\LSTWts{\complex} 
=
 \frac{(1-\ar\st )\complex} {\complex- \ar (1-{\exp{-\complex\st}}) }$; $\Polef{\LSTWtfunction}=\big\{\ar  [1 + \frac{1}{\ar\st}\Lambertnx{k}{-\ar\st \exp{-\ar\st } } ] \setst k\in\Integerzero\big\}$; $-\dominantpoleX{\Wt}=-\ar  [1 + \frac{1}{\ar\st}\Lambertnx{-1}{-\ar\st \exp{-\ar\st } } ]$;
 %$\specialmatrixnk{\nn}{0}=1$ and%
%
\begin{equation}\label{cstspecialmatrixnk}\begin{array}{ll}
%\specialmatrixnk{\nn}{0}=1,
%&
%\\
%
\obsolete{
\specialmatrixnk{\nn}{1}
=
\phi \frac{\st^{2}}{\factorial{2}} 
&
\\
\specialmatrixnk{\nn}{2}
=
\phi \frac{\st^{3}}{\factorial{3}} + \phi^2 (\frac{\st^{2}}{\factorial{2}})^2 
&
\\
\specialmatrixnk{\nn}{3}
=
\phi ( \frac{\st^{4}}{\factorial{4}} + \frac{\st^{3}}{\factorial{3}}\phi \frac{\st^{2}}{\factorial{2}} + \frac{\st^{2}}{\factorial{2}}(\phi \frac{\st^{3}}{\factorial{3}} + \phi^2 (\frac{\st^{2}}{\factorial{2}})^2 ) )
=
\phi  \frac{\st^{4}}{\factorial{4}} +2 \phi^2 \frac{\st^{2}}{\factorial{2}} \frac{\st^{3}}{\factorial{3}} + \phi^3 (\frac{\st^{2}}{\factorial{2}})^3
&
\\
\specialmatrixnk{\nn}{4}
= 
\phi  \frac{\st^{5}}{\factorial{5}}     +  \phi^2 \left(2\frac{\st^{2}}{\factorial{2}}  \frac{\st^{4}}{\factorial{4}} + ( \frac{\st^{3}}{\factorial{3}})^2 \right)   +3 \phi^3( \frac{\st^{2}}{\factorial{2}})^2 \frac{\st^{3}}{\factorial{3}} + \phi^4  (\frac{\st^{2}}{\factorial{2}})^4
&
\\
\specialmatrixnk{\nn}{5}
= 
\phi (   \frac{\st^{6}}{\factorial{6}} + \frac{\st^{5}}{\factorial{5}}( \phi \frac{\st^{2}}{\factorial{2}}  ) + ( \phi \frac{\st^{3}}{\factorial{3}} + \phi^2 (\frac{\st^{2}}{\factorial{2}})^2  ) + \frac{\st^{3}}{\factorial{3}} ( \phi  \frac{\st^{4}}{\factorial{4}} +2 \phi^2 \frac{\st^{2}}{\factorial{2}} \frac{\st^{3}}{\factorial{3}} + \phi^3 (\frac{\st^{2}}{\factorial{2}})^3 ) +  \frac{\st^{2}}{\factorial{2}} ( \phi  \frac{\st^{5}}{\factorial{5}}     +  \phi^2 \left(2\frac{\st^{2}}{\factorial{2}}  \frac{\st^{4}}{\factorial{4}} + ( \frac{\st^{3}}{\factorial{3}})^2 \right)   
+3 \phi^3( \frac{\st^{2}}{\factorial{2}})^2 \frac{\st^{3}}{\factorial{3}} 
+ \phi^4  (\frac{\st^{2}}{\factorial{2}})^4 )  )
&
\\
= 
  \phi  \frac{\st^{6}}{\factorial{6}} 
+ 2 \phi^2 \left(
 \frac{\st^{2}}{\factorial{2}}  \frac{\st^{5}}{\factorial{5}} 
+   \frac{\st^{3}}{\factorial{3}} \frac{\st^{4}}{\factorial{4}}
\right)
+3 \phi^3 \left( 
 (\frac{\st^{2}}{\factorial{2}})^2 \frac{\st^{4}}{\factorial{4}}  
+ \frac{\st^{2}}{\factorial{2}} (\frac{\st^{3}}{\factorial{3}})^2
\right)
+ 4 \phi^4  (\frac{\st^{2}}{\factorial{2}})^3  \frac{\st^{3}}{\factorial{3}}    
+ \phi^5  (\frac{\st^{2}}{\factorial{2}})^5 
&
\\
}%
\specialmatrixnk{\nn}{k}
= 
 \big[ \sum\nolimits_{t=1}^{k}\big(\frac{\ar\st}{1-\ar\st}\big)^t  
 \normscenarios{t}{k+t}%\frac{\scenarios{t}{k+t}}{\factorial{(k+t)}} 
 \big]\, \st^{k}
, &\quad (k\geq 1),
 \end{array}
\end{equation}
\noeqref{cstspecialmatrixnk}% 
\storecompoundcounter{equation}{cstspecialmatrixnk}%
where~$\lambertnfunction{n}$ denotes the $n$th branch of the  product logarithm function, \cite{corless96}, and
\begin{subequations}
\begin{align}
\label{scenariostrivial}\textstyle
\normscenarios{1}{n}%\frac{\scenarios{1}{n}}{\factorial{n}}
&\textstyle
=\frac{1}{\factorial{n}},
%\hspace{26mm}
&&
(n \geq 2) ,
\\ 
\label{scenariosidentity}
\textstyle
\normscenarios{m+1}{n}
&\textstyle
=\sum\nolimits_{p=2m}^{n-2} \frac{\normscenarios{m}{p}}{\factorial{(n-p)}} 
%\frac{\scenarios{m+1}{n}}{\factorial{n}}=\sum\nolimits_{p=2m}^{n-2} \frac{1}{\factorial{(n-p)}} \big(\frac{\scenarios{m}{p}}{\factorial{p}}\big)
,
&& \textstyle
%n=2m+2,2m+3,\dots, \ m=1,2,\dots\, .
(m=1,\dots,\floor{\frac{n-2}{2}}, n\geq 2), 
\\ \label{scenariossimplify}
\textstyle
%\frac{\scenarios{m}{n+1}}{\factorial{(n+1)}}=\frac{m}{n+1}\left(\frac{\scenarios{m}{n}}{\factorial{n}}+\frac{\scenarios{m-1}{n-1}}{\factorial{(n-1)}}\right),
\normscenarios{m}{n+1}
&\textstyle
=\frac{m}{n+1}[\normscenarios{m}{n}+\normscenarios{m-1}{n-1}],
\hspace{-5mm}
&& \textstyle
(m=2,\dots,\floor{\frac{n}{2}},\, n\geq 4).
\end{align}
\end{subequations}
\noeqref{scenariostrivial}%
\noeqref{scenariosidentity}%
\noeqref{scenariossimplify}%
\storesubequationtriple{scenariostrivial}{scenariosidentity}{scenariossimplify}
%\storecompoundcounter{equation}{scenariostrivial}%
%
\horizontalline
\smallskip
\textbf{M/M/1 (exponential):}
%
%$\proba{\stochst\leq x}= 1-\exp{-\strate x}$ 
$\ddistriXx{\stochst}{\st}=
\strate\exp{-\strate \st} $ with rate
 $\strate>\ar$; $\expectation{\stochst^{k}}=\factorial{k}\, \strate^{-k}$; $
\LSTWts{\complex} 
=
 \frac{(\strate-\ar)(\complex+\strate)} {\strate(\complex+\strate-\ar)}
 $;
 $\Polef{\LSTWtfunction}=\{\ar-\strate\}$; $-\dominantpoleX{\Wt}=\strate-\ar$;
 %$\specialmatrixnk{\nn}{0}=1$ and%
\begin{equation}\label{expspecialmatrixnk}\begin{array}{ll}
%\specialmatrixnk{\nn}{0}=1,
%&
%\\
%
%
%
\obsolete{
\specialmatrixnk{\nn}{k}
=
\frac{\ar\strate}{\strate-\ar} \sum\nolimits_{t=0}^{k-1} \frac{1}{\strate^{k-t+1}}\specialmatrixnk{\nn}{t}
, & k=1,\dots,\nn,
\\
\specialmatrixnk{\nn}{1}
=
\frac{\ar\strate}{\strate-\ar}  \frac{1}{\strate^{2}}
&
\\
\specialmatrixnk{\nn}{2}
=
\frac{\ar\strate}{\strate-\ar} ( \frac{1}{\strate^{3}} + \frac{1}{\strate^{2}} \frac{\ar\strate}{\strate-\ar}  \frac{1}{\strate^{2}})
=
\frac{\ar\strate}{\strate-\ar}  \frac{1}{\strate^{3}} + \left(\frac{\ar\strate}{\strate-\ar} \right)^2 \frac{1}{\strate^{4}}
&
\\
\specialmatrixnk{\nn}{3}
=
\frac{\ar\strate}{\strate-\ar} ( \frac{1}{\strate^{4}} + \frac{1}{\strate^{3}} \frac{\ar\strate}{\strate-\ar}  \frac{1}{\strate^{2}} + \frac{1}{\strate^{2}} \frac{\ar\strate}{\strate-\ar} ( \frac{1}{\strate^{3}} + \frac{1}{\strate^{2}} \frac{\ar\strate}{\strate-\ar}  \frac{1}{\strate^{2}}))
=
\frac{\ar\strate}{\strate-\ar}  \frac{1}{\strate^{4}} + 2 \left( \frac{\ar\strate}{\strate-\ar} \right)^2 \frac{1}{\strate^{5}}  +  \left(\frac{\ar\strate}{\strate-\ar} \right)^3 \frac{1}{\strate^{6}}
&
\\
\specialmatrixnk{\nn}{4}
=
\frac{\ar\strate}{\strate-\ar}  (\frac{1}{\strate^{5}} + \frac{1}{\strate^{4}} \frac{\ar\strate}{\strate-\ar}  \frac{1}{\strate^{2}} + \frac{1}{\strate^{3}} \frac{\ar\strate}{\strate-\ar} ( \frac{1}{\strate^{3}} + \frac{1}{\strate^{2}} \frac{\ar\strate}{\strate-\ar}  \frac{1}{\strate^{2}}) + \frac{1}{\strate^{2}} \frac{\ar\strate}{\strate-\ar} ( \frac{1}{\strate^{4}} + \frac{1}{\strate^{3}} \frac{\ar\strate}{\strate-\ar}  \frac{1}{\strate^{2}} + \frac{1}{\strate^{2}} \frac{\ar\strate}{\strate-\ar} ( \frac{1}{\strate^{3}} + \frac{1}{\strate^{2}} \frac{\ar\strate}{\strate-\ar}  \frac{1}{\strate^{2}})) )
&
\\
=
\frac{\ar\strate}{\strate-\ar} \frac{1}{\strate^{5}} 
+ 
3 \left(\frac{\ar\strate}{\strate-\ar}\right)^2  \frac{1}{\strate^{6}} 
 + 
 3 \left(\frac{\ar\strate}{\strate-\ar}\right)^3  \frac{1}{\strate^{7}} 
 + 
  \left(\frac{\ar\strate}{\strate-\ar}\right)^4 \frac{1}{\strate^{8}}  
&
\\
\specialmatrixnk{\nn}{5}
=
\frac{\ar\strate}{\strate-\ar}\frac{1}{\strate^{6}}
+
4\left(\frac{\ar\strate}{\strate-\ar}\right)^2 \frac{1}{\strate^{7}}  
 + 
6 \left(\frac{\ar\strate}{\strate-\ar} \right)^3 \frac{1}{\strate^{8}}
 +  
 4 \left(\frac{\ar\strate}{\strate-\ar} \right)^4 \frac{1}{\strate^{9}}
 + 
 \left(\frac{\ar\strate}{\strate-\ar}\right)^5 \frac{1}{\strate^{10}}  
&
\\
}
\specialmatrixnk{\nn}{k}
%=
%\frac{\ar\strate}{\strate-\ar} \frac{1}{\strate^{k+1}}   \sum\nolimits_{t=0}^{k-1}\left(\begin{array}{c} k-1\\ t\end{array}\right)\left(\frac{\ar}{\strate-\ar} \right)^t 1^{k-1-t}
%=
%\frac{\ar}{\strate^{k}(\strate-\ar)}   \left(\frac{\ar}{\strate-\ar} + 1\right)^{k-1} 
=
\frac{\ar}{\strate(\strate-\ar)^k}  
, &\quad 
( k\geq 1 ).
  \end{array}
\end{equation}
\storecompoundcounter{equation}{expspecialmatrixnk}%
\horizontalline
\smallskip
\textbf{%Erlang-distributed service times (
M/E\textsubscript{$\shape$}/1 (Erlang):%
%)
}
$
\ddistriXx{\stochst}{\st}=
\frac{\strate}{\factorial{(\shape-1)}} (\strate \st)^{\shape-1} \exp{-\strate \st}
%\frac{(\strate \st)^{\shape-1}}{\factorial{(\shape-1)}} \strate \exp{-\strate \st}
$, with shape 
$\shape\geq 1$ and rate 
$\strate>\shape \ar$; $\expectation{\stochst^{k}}=\frac{\factorial{(k+\shape-1)}}{\factorial{(\shape-1)}}  \,\strate^{-k}$; 
$
\LSTWts{\complex} 
%=
% \frac{(1-\ar\st )(\complex+\strate)^\shape \complex} {(\complex- \ar) (\complex+\strate)^\shape +\ar\strate^\shape}
 =
 \frac{(1-\ar\st )(\complex+\strate)^\shape } {\strate^{\shape} \big[ (\frac{\complex}{\strate}+1)^\shape - \frac{\ar}{\strate}  \sum\nolimits_{k=0}^{\shape-1} \binomialcoef{\ \shape}{k+1} (\frac{\complex}{\strate})^k \big]  }
 $; $\cardinality{\Polef{\LSTWtfunction}}=\shape$;
%$\specialmatrixnk{\nn}{0}=1$ and %
\begin{equation}\label{erlspecialmatrixnk}\begin{array}{ll}
%\specialmatrixnk{\nn}{0}=1,
%&
%\\
%
\specialmatrixnk{\nn}{k}
= 
 \frac{1}{\strate^k}  \sum\nolimits_{t=1}^{k} \big(\frac{\ar}{\strate-\shape\ar}\big)^t \scenariosshape{\shape}{t}{k+t}  
, & \quad (k\geq 1),
 \end{array}
\end{equation}
\noeqref{erlspecialmatrixnk}%
\storecompoundcounter{equation}{erlspecialmatrixnk}%
where 
\begin{subequations}
 \begin{align}
\label{scenariosshapetrivial}
\textstyle
\scenariosshape{\shape}{1}{n} 
&\textstyle 
=
\binomialcoef{n+\shape-1}{ \ \ \shape-1},
%$
%\end{array}
&&\textstyle 
(n\geq 2),
\\ 
\label{scenariosshapeidentity}
\textstyle 
\scenariosshape{\shape}{m+1}{n} 
&\textstyle 
=
\sum\nolimits_{p=2m}^{n-2} \binomialcoef{\shape+n-p-1}{\ \ \ \shape-1}\,\scenariosshape{\shape}{m}{p}
,
\hspace{-7mm}
&&\textstyle 
(m=1,\dots,\floor{\frac{n-2}{2}},
  \ n\geq 2).
\end{align}
\end{subequations}
\noeqref{scenariosshapetrivial}%
\noeqref{scenariosshapeidentity}%
\storesubequationdouble{scenariosshapetrivial}{scenariosshapeidentity}%
\horizontalline 
\end{table}
\storecompoundcounter{table}{table:moments}%
 
\obsolete{   
 \begin{proof}[Derivation of Table~\ref{table:moments}]
  The expression for~$\{\specialmatrixnk{\nn}{k}\}$ were derived
  either by combination of~\eqref{PKformula} and~\eqref{definition:poweracoefnk}, or by computation of the moments~$\expectation{\stochst^{k}}$ and development of \eqref{table:specialmatrixnk}.
  
  {(M/D/1)} 
  Before showing~\eqref{cstspecialmatrixnk}, observe from~\eqref{scenariostrivial}-\eqref{scenariosidentity} that the quantity~$\factorial{n}\,\normscenarios{m}{n}$ is in fact the number of possible scenarios that may occur  when placing~$n$ distinct objects  (unordered) into~$m$ numbered urns so that each urn contains at least two objects ($n\geq 2m$).
  \noindent
  For~$k=1$, \eqref{cstspecialmatrixnk} follows directly from~\eqref{table:specialmatrixnk}.
If we suppose that~\eqref{cstspecialmatrixnk} holds for~$k=1,\dots,p$, then, using $\expectation{\stochst^{k}}=\st^k$,
\begin{equation}\label{proofonstantservicetimes}\nocolsep\begin{array}{rcl}
\specialmatrixnk{\nn}{p+1}
&\refereq{\eqref{table:specialmatrixnk}}{=} &
\frac{\ar}{1-\ar\st}
\big[\frac{\st^{p+2}}{\factorial{(p+2)}}  +  \sum\nolimits_{t=1}^{p} \frac{\st^{p-t+2}}{\factorial{(p-t+2)}}  \specialmatrixnk{\nn}{t} \big]
\\
&\refereq{\eqref{cstspecialmatrixnk}}{=} &
\frac{\ar}{1-\ar\st}
\big[\frac{\st^{p+2}}{\factorial{(p+2)}}  +  \sum\nolimits_{t=1}^{p} \frac{\st^{p-t+2}}{\factorial{(p-t+2)}}  \sum\nolimits_{q=1}^{t}(\frac{\ar}{1-\ar\st})^q
%\frac{\scenarios{q}{t+q}}{\factorial{(t+q)}}
\,\normscenarios{q}{t+q}\,
\st^{t+q} \big]
\\
&=&
% t : 1 -> p, q : 1 -> t 
% => t>=q 
% q : 1 -> p, t : q -> p
\frac{\ar}{1-\ar\st}
\frac{\st^{p+2}}{\factorial{(p+2)}}
+
  \sum\nolimits_{q=1}^{p}
(\frac{\ar}{1-\ar\st})^{q+1}
%\big(
\sum\nolimits_{t=q}^{p} 
%\frac{\scenarios{q}{t+q}}{\factorial{(t+q)}}
\normscenarios{q}{t+q}
\frac{\st^{p+q+2}}{\factorial{(p-t+2)}}
%\big)
\\
&\refereq{\eqref{scenariostrivial}}{=}&
\frac{\ar}{1-\ar\st}
%\frac{\scenarios{1}{p+2}}{\factorial{(p+2)}}
\normscenarios{1}{p+2}
\,\st^{p+2}
%\\&&\hfill
+
  \sum\nolimits_{q=2}^{p+1}
(\frac{\ar}{1-\ar\st})^{q}
\big[
%\sum\nolimits_{t=1}^{q-1} 
\sum\nolimits_{l=q-1}^{p} 
%\frac{\factorial{(p+1+q)}}{\factorial{(p+1+q-l)}\factorial{l}}
%\binomialcoef{p+1+q}{q-1+l} \scenarios{q-1}{q-1+l}
\frac{\normscenarios{q-1}{q-1+l}}{\factorial{(p-l+2)}}%\frac{\scenarios{q-1}{q-1+l}}{\factorial{(q-1+l)}}
 \big]
\st^{p+1+q} 
\\
&\refereq{\eqref{scenariosidentity}}{=} & 
 \sum\nolimits_{q=1}^{p+1}
\big(\frac{\ar}{1-\ar\st}\big)^{q}
%\frac{ \scenarios{q}{p+1+q}}{\factorial{(p+1+q)}}
\normscenarios{q}{p+1+q}
\, \st^{p+1+q}
 \end{array}
\end{equation}
and~\eqref{cstspecialmatrixnk} holds for all~$k$ by induction.
We obtain in~\eqref{scenariostrivial}-\eqref{scenariosidentity} a recusrsive procedure for computing the values of~$\{\normscenarios{m}{n}\}$. Identity~\eqref{scenariossimplify}, which follows from~\eqref{scenariostrivial}-\eqref{scenariosidentity}, tends to simplify and accelerate the process in practice.
\noindent
As for the poles of~$\LSTWtfunction$ (cf. Figure~\ref{figure:LSTWt}), they are the (nonzero) solutions of $\complex- \ar (1-{\exp{-\complex\st}})=0$ or, equivalently,  $(\complex-\ar)\st \exp{(\complex-\ar)\st} =-\ar\st\exp{-\ar\st}$, which is an instance of the equation $ \zcomplex e^\zcomplex = \sing $, the solutions of which are given by $\Lambertnx{n}{\sing}$, where~$\lambertnfunction{n}$ is the $n$th branch of the  product logarithm function~$\lambertnfunction{n}$~\cite{corless96}.

%$\ar+\frac{1}{\st}\Lambertnx{n}{-\ar\st\exp{-\ar\st}} $

(M/M/1)
$\LSTWtfunction$ follows from~\eqref{PKformula} and $\LSTstochsts{\complex}={\strate}/({\complex+\strate})$. 
After successive derivations of~$\LSTWts{-\complex} $, which rewrites as
$
\LSTWts{-\complex} 
%=
% \frac{(\strate-\ar)(\complex-\strate)} {\strate(\complex-\strate+\ar)}
 =
% \frac{\strate-\ar}{\strate}[1+ \frac{\ar}{\strate-\ar-\complex}] 
 [({\strate-\ar})/{\strate}]\,[1+ {\ar}/({\strate-\ar-\complex})] 
 $, we compute~\eqref{definition:poweracoefnk} and find $\poweracoefnk{\sing}{.}{k} = [({\strate-\ar})/{\strate}]\,[\dirack{k}+{\ar}{(\strate-\ar-\sing)^{-k-1}}]$, which reduces to~\eqref{expspecialmatrixnk} when~$\sing=0$.

 (M/E\textsubscript{$\shape$}/1)
 In~\eqref{scenariosshapetrivial}-\eqref{scenariosshapeidentity}, $\scenariosshape{\shape}{m}{n}$ can be understood, for $\shape\geq1$ and $n\geq 2m$, as the number of possible outcomes when~$m$ ordered collections of~$n_1,n_2,\dots,n_m$ unordered objects are  respectively picked out of~$m$ urns~$1,2,\dots,m$ so that~$n_1+\dots+n_m=n$ and $n_p\geq 2$ for $p=1,\dots,m$, where each urn~$p$ is initially assumed to contain~$\shape-1$ distinct objets plus~$n_p$ objects randomly drawn (without repetition) from a common set of~$n$ additional distinct objects ($p=1,\dots,m$).
 We now turn to show~\eqref{erlspecialmatrixnk}.
 \noindent
  It is easy to verify that~\eqref{erlspecialmatrixnk} is true for $k=1$.  Suppose now that~\eqref{erlspecialmatrixnk} holds for $k=1,\dots,p$, then
\begin{equation}\label{prooferlservicetimes}\begin{array}{rcl}
\specialmatrixnk{\nn}{p+1}
&\refereq{\eqref{table:specialmatrixnk}}{=} &
\frac{\ar\strate}{\strate-\shape\ar}
\big[ \binomialcoef{p+\shape+1}{\ \shape-1}\strate^{-(p+2)}+   \sum\nolimits_{t=1}^{p} \binomialcoef{p-t+\shape+1}{\ \ \ \shape-1}\strate^{-(p-t+2)}\specialmatrixnk{}{t} \big]
\\
&\refereq{\eqref{erlspecialmatrixnk}}{=} &
\frac{\ar}{\strate^{p+1}(\strate-\shape\ar)}
\big[ \binomialcoef{p+\shape+1}{\ \shape-1}+   \sum\nolimits_{t=1}^{p} \binomialcoef{p-t+\shape+1}{\ \ \ \shape-1} \strate^{t} \big(  
  \sum\nolimits_{l=1}^{t} (\frac{\ar}{\strate-\shape\ar})^l \frac{ \scenariosshape{\shape}{l}{t+l} }{\strate^t}  \big)\big]
\\
&=&
% t : 1 -> p, l : 1 -> t 
% => t>=q 
% l : 1 -> p, t : l -> p
\frac{\ar}{\strate^{p+1}(\strate-\shape\ar)}
\big[ \binomialcoef{p+\shape+1}{\ \shape-1}+   \sum\nolimits_{l=1}^{p}(\frac{\ar}{\strate-\shape\ar})^l\sum\nolimits_{t=l}^{p} \binomialcoef{p-t+\shape+1}{\ \ \ \shape-1}  
    \scenariosshape{\shape}{l}{t+l} \big]
  \\
% t: l->p, ''m''=t+l: 2l->p+l=''n-2"
% => m:2l->p+l , t=m-l, 
&\refereq{\eqref{scenariosshapetrivial}}{=}&
\frac{\ar}{\strate^{p+1}(\strate-\shape\ar)}
 \scenariosshape{\shape}{1}{p+2}
\\&&\hfill
+
\frac{\ar}{\strate^{p+1}(\strate-\shape\ar)}
      \sum\nolimits_{l=1}^{p}(\frac{\ar}{\strate-\shape\ar})^l\big(\sum\nolimits_{m=2l}^{l+p} \binomialcoef{\shape+(p+l+2)-m-1}{\ \ \ \shape-1}  
    \scenariosshape{\shape}{l}{m}\big)
\\
&\refereq{\eqref{scenariosshapeidentity}}{=} &
%\frac{\ar}{\strate^{p+1}(\strate-\shape\ar)}\scenariosshape{\shape}{1}{p+2}+\frac{\ar}{\strate^{p+1}(\strate-\shape\ar)}      \sum\nolimits_{l=2}^{p+1}\big(\frac{\ar}{\strate-\shape\ar}\big)^{l  -1}    \scenariosshape{\shape}{l}{p+1+l}
\frac{1}{\strate^{p+1}}
      \sum\nolimits_{l=1}^{p+1}(\frac{\ar}{\strate-\shape\ar})^{l}
    \scenariosshape{\shape}{l}{p+1+l}
 \end{array}
\end{equation}
and~\eqref{erlspecialmatrixnk} holds for all~$k$ by induction.

% $(\frac{\complex}{\strate}+1)^\shape - \frac{\ar}{\strate}  \sum\nolimits_{k=0}^{\shape-1} \binomialcoef{\ \shape}{k+1} (\frac{\complex}{\strate})^k = 0$
 
 \end{proof}

}%
%%%%%%%%%%%%%%%%%%%%%%%%%%

\section{Computation of \aclp{wfunction}: examples}

%\section{Examples} 
\label{appendix:examples}

\begin{example}[Step cost function and identical service times]\label{example:piecewisecst}
The  cost function $\costfunction=\indicatorfunction{[\taufunction,\infty)}$  is  considered with constant service times $\stochst=\st>0$. The value function to this problem was derived in~\cite{hyytia-peva-2017} as a solution of~\eqref{deresult}.
We have $\costifunction{0}=0$, $\costifunction{1}=1$, $\laplacetransformfs{\costfunction}{\complex}=({1}/{\complex})\, e^{-\complex\taufunction}$, $\pwshiftedlaplacest{\complex}{\taufunction}={1}/{\complex}$, $\Polef{\laplacetransformf{\costfunction}}=\{0\}$, and $\dominantpoleX{\Wt} = -\ar\,   [1 + ({1}/{\ar\st})\, \Lambertnx{-1}{-\ar\st \exp{-\ar\st } } ]<0 $, as detailed in Appendix~\ref{appendix:rates}.

For $\bl\in(\taufunction,\infty)$, we find,  
\begin{equation}\notag
\textstyle
\rightderivative\purevaluex{\bl}
\refereq{\eqref{solutionFNPlarge}}{=}
\frac{\ar}{1-\ar\st}
\bigresiduefx{
 \frac{1-\ar\st } {\complex+ \ar (1-{\exp{\complex\st}}) } e^{\complex(\bl-\taufunction)} }{\complex=0} 
\refereq{\eqref{residue}}{=}
\frac{\ar}{1-\ar\st}.
\end{equation}

%Using the results of Appendix~\ref{appendix:rates}, (\refeq{bilaterallaplacetransformVF}) reduces to
%The two-sided Laplace transform of~$\extdpurevaluefunction$ is given by
\obsolete{
Using~\eqref{bilaterallaplacetransformVF}, \eqref{piecewisecost}, and~\eqref{PKformula}, we find
\begin{equation}\label{cstexampleBlaplace}
\textstyle
\Blaplacetransformfs{\extdpurevaluefunction}{\complex}
 = 
\frac{\ar}{1-\ar\st} 
\big(
\frac{(1-\ar\st )\complex} {\complex+ \ar (1-{\exp{\complex\st}}) }
\big)
\frac{\exp{-\complex\taufunction}}{\complex}
 .
\end{equation} 
%
%$ \frac{(1-\ar\st )\complex} {\complex- \ar (1-{\exp{-\complex\st}}) }$
%
}%

For~$\bl\in(0,\taufunction)$, we inspect the positions of the poles and set $\anchor\in(0,-\dominantpoleX{\Wt})$.
Decomposing~$\LSTWts{-\complex}$ as in~\eqref{trick}, we find $\Polef{\pwshiftedlaplacest{\cdot}{\taufunction}}=\{0\}$,
$\iPolex{\bl} = \{ - \ar \}$, and $\fPolex{\bl}=\emptyset$.
Since~$\costifunction{0}=0$ and~$\fPolex{\bl}$ is empty, the first and third terms in~\eqref{solutionINP} both vanish and~$ \fLSTWtxfunction{\bl}$ needs not be considered.
%$ \ROC{\laplacetransformf{\costfunction}} = \{ \complex\in\Complex \setst\realpart{\complex}>0\}$, 
%
%while $\dominantpoleX{\Wt} = -\ar  \left(1 + \frac{1}{\ar\st}\Lambertnx{-1}{-\ar\st \exp{-\ar\st } } \right)<0 $ is given by the branch~-1 of the product logarithm function (cf. Appendix~\ref{appendix:rates}). 
%Thus, $\ROC{\Blaplacetransformf{\extdpurevaluefunction}} = (0, - \dominantpoleX{\Wt} )$.
%
%Decomposing~$\LSTWts{-\complex}$ as in~(\refeq{trick}), we find
%Choosing $\anchor\in(0,-\dominantpoleX{\Wt})$ then yields 
%$\iPolex{\bl} = \{ \ar \}$ and $\fPolex{\bl}=\emptyset$.
%Since~$\fPolex{\bl}$ is empty, the third term in~\eqref{solutionINP} vanishes and~$ \fLSTWtxfunction{\bl}$ needs not be considered.
We find,
\begin{equation}\label{cstexampledbvfres}
\begin{array}{rl}
\dpurevaluex{\bl}  
& \refereq{\eqref{solutionINP}}{=}
 \sum_{\pole\in\{-\ar,0\}} \bigresiduefx{   
 \frac{\ar^{m+1} \, \exp{\complex (\bl+m\st-\taufunction)} } {(\complex+\ar)^m [\complex+ \ar (1- \exp{\complex\st})] }
}{\complex=\pole}  
\end{array} 
\end{equation}
%
%For~$\bl<\taufunction$, we decompose~$\LSTWts{-\complex}$ as in~(\refeq{trick}), and~(\refeq{cstexampleBlaplace}) yields 
%
%Taking $0<\anchor<\cut$, with~$\cut$ given by~(\ref{cstrate}), and using the contour~$\contoura{\radius}$, we find $\dpurevaluex{\bl}={\ar}/({1-\ar\st})   $ for $\bl\in(\taufunction,\infty)$. For~$\bl<\taufunction$, we decompose~$\LSTWts{-\complex}$ as in~(\refeq{trick}), and~(\refeq{cstexampleBlaplace}) yields 
\obsolete{
\begin{equation}\label{cstexampleBlaplace}
\begin{array}{c}
\Blaplacetransformfs{\extdpurevaluefunction}{\complex} \exp{\complex\bl}
 = 
\frac{\ar}{(\complex+ \ar)^{m}} \left( \frac{ \ar^{m}\exp{\complex (\bl+m\st-\taufunction)} } {\complex+ \ar (1-{\exp{\complex\st}}) } 
\right) 
%\\ \hspace{40mm}
+
\left(
\frac{ \complex}{\complex+ \ar} \sum\nolimits_{k=0}^{m-1} \frac{ \ar^k\exp{\complex k\st}}{(\complex+ \ar)^k}
\right)
\frac{\exp{\complex(\bl-\taufunction)}}{\complex} 
 ,
\end{array}
\end{equation} 
}%
where $m=\trickix{1}{\bl}=\ceil{({\taufunction-\bl})/{\st}}$, % is given by~(\refeq{cststtrickix}), 
%and  the second term can be ignored due to~(\refeq{jordanpiecewisenull}). %since~$\LSTstochsts{-\complex}$ is analytic. % and~(\ref{jordanpiecewisenull}) yields~$0$.
% \left(\frac{ \ar\exp{\complex\st}}{\complex+ \ar}\right)^{\trickix{1}{\bl}} \frac{(1-\ar\st )\complex} {\complex+ \ar (1-{\exp{\complex\st}}) } + \frac{\complex(1-\ar\st) }{\complex+ \ar} \sum\nolimits_{k=0}^{\trickix{1}{\bl}-1}  \left(\frac{ \ar\exp{\complex\st}}{\complex+ \ar}\right)^{k}
 %For~$j=1,\dots,\trickix{1}{\bl}$, let~$\barbli{j}=\bl+j\st-\taufunction$, and let
We let~$\pwynk{\nn}{k}=\bl+(m-1-\nn+k)\st-\taufunction$ for all~$\nn,k\in\Natural$, and set
\obsolete{
\begin{equation}\label{cstexamplepwKs}
\begin{array}{c}
\pwKs{\complex}
 = 
\frac{ 
%\ar^{m}
%\ar
\exp{\complex 
(\bl+m\st-\taufunction)
%\barbli{m}
} 
}{\complex+ \ar (1-{\exp{\complex\st}}) } .
\end{array}
\end{equation}   
}%
$
\pwKs{\complex}
 = 
{ 
%\ar^{m}
%\ar
\exp{\complex 
(\bl+m\st-\taufunction)
%\barbli{m}
} 
}/[{\complex+ \ar (1-{\exp{\complex\st}}) }]
$,
the derivatives of which can be computed by induction. At~$-\ar$, we find the derivatives
\begin{equation}\label{cstexamplederivativepwK}
\textstyle
\pwdnKs{\nn}{-\ar}
 = 
%- \frac{\factorial{\nn}}{\ar^{\nn+1}} \sum\nolimits_{k=0}^{\nn} \frac{(\ar\barbli{m-1-\nn+k})^{k}}{\factorial{k}} \exp{-\ar\barbli{m-1-\nn+k}}
%- \frac{\factorial{\nn}}{\ar^{\nn+1}}\sum\nolimits_{k=0}^{\nn} \frac{(\ar\pwynk{\nn}{k})^{k}}{\factorial{k}}  \exp{-\ar\pwynk{\nn}{k}},  
- ({\factorial{\nn}}/{\ar^{\nn+1}})\,\sum\nolimits_{k=0}^{\nn} [{(\ar\pwynk{\nn}{k})^{k}}/{\factorial{k}} ] \, \exp{-\ar\pwynk{\nn}{k}},
  \quad  (\nn\in\Natural)
 ,
\end{equation} 
%\storeequationcounter{cstexamplederivativepwK}%
\storecompoundcounter{equation}{cstexamplederivativepwK}%
%
%\setcounter{equationsection}{\value{section}}%
%\setcounter{equationbuffer}{\value{equation}}%
%\storecounter{equationbuffer}{store:equation:cstexamplederivativepwK}%
%\storecounter{equationbuffer}{equation:cstexamplederivativepwK}%
%\storecounter{equationsection}{store:section:cstexamplederivativepwK}%
%\storecounter{equationsection}{section:cstexamplederivativepwK}%
%
\obsolete{
The derivatives of~$\pwKs{\complex}$ at~$-\ar$ are given by\footnotemark{}
\begin{equation}\label{cstexamplederivativepwK}
\begin{array}{ll}
\pwdnKs{\nn}{-\ar}
 = 
%- \frac{\factorial{\nn}}{\ar^{\nn+1}} \sum\nolimits_{k=0}^{\nn} \frac{(\ar\barbli{m-1-\nn+k})^{k}}{\factorial{k}} \exp{-\ar\barbli{m-1-\nn+k}}
- \frac{\factorial{\nn}}{\ar^{\nn+1}}\sum\nolimits_{k=0}^{\nn} \frac{(\ar\pwynk{\nn}{k})^{k}}{\factorial{k}}  \exp{-\ar\pwynk{\nn}{k}},
&  \quad  (\nn\in\Natural)
 .
\end{array}
\end{equation} 
By integration of~$\Blaplacetransformfs{\extdpurevaluefunction}{\complex}\exp{\complex\bl}$ along the vertical axis at~$\anchor$, we find, using the contour~$\contoura{\radius}$, the residue theorem at the poles~$0$ and~$-\ar$, \eqref{jordanpiecewisenull} and~\eqref{cstexampleBlaplace}:
}%  
and~\eqref{cstexampledbvfres} reduces, for $\bl\in[0,\taufunction)$, to
\begin{equation}\label{cstexampledbvf}\notag
\begin{array}{ll}
\dpurevaluex{\bl} &
\refereq{\eqref{residue}}{=}
\lim\nolimits_{\complex\to 0}\complex \left( \frac{\ar^{m+1} }{(\complex+ \ar)^{m}} \pwKs{\complex} \right)
  + \frac{1}{\factorial{(m-1)}}\lim\nolimits_{\complex\to -\ar}  \left( \ar^{m+1}  \pwdnKs{m-1}{\complex} \right)
%\\
%+
%\frac{\ar^{m+1}}{\factorial{(m-1)}} \pwdnKs{m-1}{-\ar}
%&
\\
  &
  \refereq{\eqref{cstexamplederivativepwK}}{=} 
 %-  \ar \sum\nolimits_{k=0}^{m-1} \frac{(\ar\pwynk{m-1}{k})^{k}}{\factorial{k}} \exp{-\ar\pwynk{m-1}{k}}
\frac{\ar}{1-\ar\st}
 -  \ar \sum\nolimits_{k=0}^{m-1} \frac{(\ar( \bl+k\st-\taufunction))^{k}}{\factorial{k}} 
 \exp{-\ar( \bl+k\st-\taufunction)}
 .
% \qquad \forall\bl\in[0,\taufunction).
\end{array}
\end{equation}
Integrating the last expression from~$\taufunction$ to~$\bl$ gives, for~$\bl\in[0,\taufunction)$,
\begin{equation}\label{cstexamplebvf}\notag
\begin{array}{rcl}
\purevaluex{\bl}
&= &
\purevaluex{\taufunction} + \frac{\ar(\bl-\taufunction)}{1-\ar\st} - \ar \sum\nolimits_{k=0}^{m-1}  \int\nolimits_{\taufunction-k\st}^{\bl}  \frac{(\ar( \altbl+k\st-\taufunction))^{k}}{\factorial{k}} \exp{-\ar( \altbl+k\st-\taufunction)} \, d\altbl
\\
&= &
\purevaluex{\taufunction} + \frac{\ar(\bl-\taufunction)}{1-\ar\st}  - \ar \sum\nolimits_{k=0}^{m-1}  \int\nolimits_{0}^{\bl+k\st-\taufunction}  \frac{(\ar\xi)^{k}}{\factorial{k}} \exp{-\ar\xi} \, d\xi
\\
&= &
\purevaluex{\taufunction} + \frac{\ar(\bl-\taufunction)}{1-\ar\st} +  \sum\nolimits_{k=0}^{m-1}   \left( \exp{-\ar(\bl+k\st-\taufunction)} 
\sum\nolimits_{q=0}^{k}  \frac{(\ar(\bl+k\st-\taufunction))^{q}}{\factorial{q}} -1 \right)
\\
&= &
\purevaluex{\taufunction} + \frac{\ar(\bl-\taufunction)}{1-\ar\st} - \trickix{1}{\bl} +  \sum\nolimits_{k=0}^{\trickix{1}{\bl} -1}   \exp{-\ar(\bl+k\st-\taufunction)} 
\sum\nolimits_{q=0}^{k}  \frac{(\ar(\bl+k\st-\taufunction))^{q}}{\factorial{q}} ,
\end{array}
\end{equation}  
where $\purevaluex{\taufunction}=  \frac{\ar\taufunction}{1-\ar\st} + \trickix{1}{0} - \sum\nolimits_{k=0}^{\trickix{1}{0}-1}   \exp{-\ar(k\st-\taufunction)} 
\sum\nolimits_{q=0}^{k}  {(\ar(k\st-\taufunction))^{q}}/{\factorial{q}}$, and $\trickix{1}{\altbl}=\ceil{({\taufunction-\altbl})/{\st}}$. Our result is coherent with~\cite[Theorem~2]{hyytia-peva-2017}.
\obsolete{%\footnotetext{
Indeed, it is  straightforward to verify that~(\ref{cstexamplederivativepwK}) holds for~$\nn=0$ and~$\nn=1$. For~$\nn\geq 2$, we proceed by induction.
 %$\pwdnKs{0}{-\ar}= -\ar^{-1}\exp{-\ar \barbli{m-1}}$
%$\pwdnKs{1}{-\ar} = -\ar^{-2}\exp{-\ar\barbli{m-2}}  -  \ar^{-1} \barbli{m-1} \exp{-\ar \barbli{m-1}}   $
observe that \eqref{cstexamplepwKs} rewrites as $ (\complex+ \ar -\ar \exp{\complex\st} ) \pwKs{\complex}
 = \exp{\complex \pwynk{\nn-1}{\nn}} 
$, which gives, after~$\nn$ differentiations at~$-\ar$:
\obsolete{%
\begin{equation}\label{cstexampleimplicit}
\begin{array}{c}
-\ar \exp{-\ar\st} \pwdnKs{\nn}{-\ar}
+ \nn(1-\ar \st \exp{-\ar\st} ) \pwdnKs{\nn-1}{-\ar}
-\sum\nolimits_{k=0}^{\nn-2}
\binomialcoef{\nn}{k}
\ar \st^{\nn-k} \exp{-\ar\st} 
\pwdnKs{k}{-\ar}
 = \barbli{m}^\nn \exp{-\ar \pwynk{\nn-1}{\nn}} 
\end{array}
\end{equation} 
 }%
\begin{equation}\label{cstexampleimplicit}
\begin{array}{c}
 \pwdnKs{\nn}{-\ar}
 =  
 \obsolete{
\frac{\nn}{\ar} (\exp{\ar\st}-\ar \st  ) \pwdnKs{\nn-1}{-\ar}
-  \sum\nolimits_{k=0}^{\nn-2}
\binomialcoef{\nn}{k}
\st^{\nn-k} 
\pwdnKs{k}{-\ar}
- \frac{ \pwynk{\nn-1}{\nn}^\nn}{\ar} \exp{-\ar \pwynk{\nn}{\nn}} 
\\=
}
\frac{\nn}{\ar} \exp{\ar\st}\pwdnKs{\nn-1}{-\ar}
-  \sum\nolimits_{k=0}^{\nn-1}
\binomialcoef{\nn}{k}
\st^{\nn-k} 
\pwdnKs{k}{-\ar}
- \frac{ \pwynk{\nn-1}{\nn}^\nn}{\ar} \exp{-\ar  \pwynk{\nn}{\nn}} .
\end{array}
\end{equation}  
 \obsolete{
\begin{equation}\label{cstexampleimplicit}
\begin{array}{c}
 \pwdnKs{2}{-\ar}
=
   -2\frac{1}{\ar^3} \exp{-\ar\barbli{m-3}} 
 -  2  \frac{\barbli{m-2}}{\ar^2} \exp{-\ar \barbli{m-2}} 
-  \frac{\barbli{m-1}^2}{\ar}  \exp{-\ar \barbli{m-1}}
\end{array}
\end{equation} 
 }%
Assuming that~(\ref{cstexamplederivativepwK}) holds for~$\nn=0,1,\dots,{p}-1$, the second term of the second member of~\eqref{cstexampleimplicit} reduces for~$\nn={p}$ to
 \begin{equation}\label{cstexamplederivativepwKproof}
\begin{array}{rcl}
 \sum\nolimits_{k=0}^{{p}-1}
\binomialcoef{{p}}{k}
\st^{{p}-k} 
\pwdnKs{k}{-\ar}
\hspace{-30mm}&&
\\
&\refereq{\eqref{cstexamplederivativepwK}}{=}&
-
\sum\nolimits_{q=0}^{{p}-1}
\binomialcoef{{p}}{q}
\st^{{p}-q} 
%\pwdnKs{q}{-\ar}
%
 \frac{\factorial{q}}{\ar^{q+1}} \sum\nolimits_{l=0}^{q} \frac{(\ar\pwynk{q}{l})^{l}}{\factorial{l}} \exp{-\ar\pwynk{q}{l}}
 \\
& =&
-
\sum\nolimits_{q=0}^{{p}-1}
\frac{\factorial{{p}}}{\factorial{{p}-q}}
\st^{{p}-q} 
%\pwdnKs{q}{-\ar}
%
 \frac{1}{\ar^{q+1}} \sum\nolimits_{l=0}^{q} \frac{(\ar\pwynk{{p}}{{p}+l-q})^{l}}{\factorial{l}} \exp{-\ar\pwynk{{p}}{{p}+l-q}}
\\ % t={p}+l-q, q={p}+l-t   q:0->{p}-1,  l: 0->q  => t:1->{p}, l:0->t-1
 &=&
-\frac{\factorial{{p}}}{\ar^{{p}+1}}  
\sum\nolimits_{t=1}^{{p}}
%\frac{1}{\factorial{t}}
\left( \sum\nolimits_{l=0}^{t-1}
\frac{(\ar\st)^{t-l}}{\factorial{t-l}}
  \frac{(\ar\pwynk{{p}}{t})^{l}}{\factorial{l}} \right) \exp{-\ar\pwynk{{p}}{t}}
\\ % t={p}+l-q, q={p}+l-t   q:0->{p}-1,  l: 0->q  => t:1->{p}, l:0->t-1
& =&
-\frac{\factorial{{p}}}{\ar^{{p}+1}}  
\sum\nolimits_{t=1}^{{p}}
\frac{\ar^t}{\factorial{t}}
\left( \sum\nolimits_{l=0}^{t}
\binomialcoef{t}{l}
{\st^{t-l}}
{\pwynk{{p}}{t}^{l}} - \pwynk{{p}}{t}^{t}\right) \exp{-\ar\pwynk{{p}}{t}}
\\ % t={p}+l-q, q={p}+l-t   q:0->{p}-1,  l: 0->q  => t:1->{p}, l:0->t-1
 &=&
-\frac{\factorial{{p}}}{\ar^{{p}+1}}  
\sum\nolimits_{t=1}^{{p}}
\frac{\ar^t}{\factorial{t}}
\left( 
{(\st + \pwynk{{p}}{t} )^{t}}
- \pwynk{{p}}{t}^{t}\right) \exp{-\ar\pwynk{{p}}{t}}
\\
 &=&
-\frac{\factorial{{p}}}{\ar^{{p}+1}}  
\sum\nolimits_{t=1}^{{p}}
\frac{\ar^t}{\factorial{t}}
\left( 
{\pwynk{{p}-1}{t}^{t}}
- \pwynk{{p}}{t}^{t}\right) \exp{-\ar\pwynk{{p}}{t}}
 .
\end{array}
\end{equation} 
 Inserting~\eqref{cstexamplederivativepwKproof} into~\eqref{cstexampleimplicit} yields
 \begin{equation}\label{cstexampleprooftwo}
\begin{array}{rcl}
 \pwdnKs{{p}}{-\ar}
 &\refereq{\eqref{cstexamplederivativepwK}}{=} &
- 
 \frac{\factorial{{p}}}{\ar^{{p}+1}} \sum\nolimits_{t=0}^{{p}-1} \frac{(\ar\pwynk{{p}-1}{t})^{t}}{\factorial{t}} \exp{-\ar\pwynk{{p}}{t}}
\\
&&\hspace{15mm}
-  \sum\nolimits_{k=0}^{{p}-1}
\binomialcoef{{p}}{k}
\st^{{p}-k} 
\pwdnKs{k}{-\ar}
- \frac{(\ar\pwynk{{p}-1}{{p}})^{p}}{\ar^{{p}+1}} \exp{-\ar \pwynk{{p}}{{p}}} 
\\
&\refereq{\eqref{cstexamplederivativepwKproof}}{=} &
- 
 \frac{\factorial{{p}}}{\ar^{{p}+1}} \sum\nolimits_{t=0}^{{p}} \frac{(\ar\pwynk{{p}-1}{t})^{t}}{\factorial{t}} \exp{-\ar\pwynk{{p}}{t}}
 \\
 && \hspace{20mm}
 + \frac{\factorial{{p}}}{\ar^{{p}+1}}  
\sum\nolimits_{t=1}^{{p}}
\frac{\ar^t}{\factorial{t}}
\left( 
{\pwynk{{p}-1}{t}^{t}}
- \pwynk{{p}}{t}^{t}\right) \exp{-\ar\pwynk{{p}}{t}}
\\
&=&
  - \frac{\factorial{{p}}}{\ar^{{p}+1}}  
\sum\nolimits_{t=0}^{{p}}
\frac{(\ar\pwynk{{p}}{t})^t}{\factorial{t}}
  \exp{-\ar\pwynk{{p}}{t}}
\end{array}
\end{equation}  
 and~(\ref{cstexamplederivativepwK}) holds for all~$\nn$.
 %$\pwynk{q}{l}=\bl+(m-1-\nn+(\nn+l-q))\st-\taufunction = \pwynk{\nn}{\nn+l-q} $
 %
}%  %footenotetext ending
\remarkmarker
\end{example}

% Example~\ref{example:exponentialtwo}
\begin{example}[\acl{Wfunction} from a Taylor series]%[Exponentially distributed service times]
\label{example:exponentialtwo}
Assume that the service times for~$\bl>0$ follow the exponential distribution $\distriXx{\stochst}{x}= 1-\exp{-\strate x}$ discussed in Appendix~\ref{appendix:rates}, where~$\strate>\ar$ in order to satisfy Assumption~\ref{assumption:stability}, and~$\dominantpoleX{\Wt}=\ar-\strate$. Consider the cost function $\costx{\bl}=1-\exp{-\sing\bl}$, with $\realpart{\sing}<\strate-\ar$ (Assumption~\ref{assumption:complexintegrablecost}) and $\sing\neq 0$. This cost function, which is given much attention in~\cite{hyytia-peva-2020}, %{hyytia-vgf-itc-2017},
 is entire ($\eorder=1$) of exponential type~$\etype=\modulus{\sing}$. Theorem~\ref{theorem:convergence} claims that the derivation of the value function from a Taylor  series at~$0$ is possible if~$\modulus{\sing}<\modulus{\dominantpoleX{\Wt}}$. %, where~$\modulus{\dominantpoleX{\Wt}}=\strate-\ar$, and proscribes it if~$\modulus{\sing}>\strate-\ar$. 
 This can be verified. Using the notations of Section~\ref{section:entirecostfunctions}, % by computing the actual Taylor series. 
we find
$\costderivativevectornk{}{\nn}=\dirack{\nn}-(-\sing)^{\nn}$
for $\nn\in\Natural$ and, with the help of Appendix~\ref{appendix:rates},
\obsolete{
\begin{equation}\label{Ztransformsexample}\begin{array}{rcll}
\Ztransformfz{\costderivativevectorn{\infty}}{\zcomplex}
&=&
%\sum\nolimits_{k=0}^{\infty} (\dirack{k}-(-\sing)^{k}) \zcomplex^{-k}
%=
\frac{\sing}{\zcomplex+\sing}
, & \makeROC{\modulus{\zcomplex}>\modulus{\sing}}
, \\
\Ztransformfz{\impulsefunction}{\frac1\zcomplex}
&\refereq{\refeq{expspecialmatrixnk}}{=}&
  \frac{\ar(\zcomplex-\strate )}{\zcomplex-(\strate-\ar)}
, & \makeROC{\modulus{\zcomplex} < \strate-\ar}
, \\
\Ztransformfz{\valuederivativevectorn{\infty}}{\zcomplex} &\refereq{\eqref{convolution}}{=}&
%\frac{\ar\sing(\zcomplex-\strate )}{(\zcomplex+\sing)(\zcomplex-(\strate-\ar))}
%=
\frac{\ar\sing}{\sing+\strate-\ar}\big(\frac{\sing+\strate}{\zcomplex+\sing}-\frac{\ar}{\zcomplex-(\strate-\ar)}\big)
, \quad& 
\makeROC{
\modulus{\sing} < \modulus{\zcomplex} < \strate-\ar
}
,
\end{array}
\end{equation}
}%
\begin{equation}\label{Ztransformsexample}\notag
\textstyle
\Ztransformfz{\costderivativevectorn{\infty}}{\zcomplex}
\,{=}\,
\frac{\sing}{\zcomplex+\sing}
, \quad%\!\!
\Ztransformfz{\impulsefunction}{\frac1\zcomplex}
\refereq{\!\!\,\eqref{expspecialmatrixnk}\!\!\,}{=}
  \frac{\ar(\zcomplex-\strate )}{\zcomplex-(\strate-\ar)}
, \quad%\!\!
\Ztransformfz{\valuederivativevectorn{\infty}}{\zcomplex} \refereq{\!\!\,\eqref{convolution}\!\!\,}{=}
\frac{\ar\sing}{\sing+\strate-\ar}\big(\frac{\sing+\strate}{\zcomplex+\sing}-\frac{\ar}{\zcomplex-(\strate-\ar)}\big)
,
\end{equation}
with
$\ROC{\Ztransformf{\costderivativevectorn{\infty}}} = \{ \zcomplex\in\Complex \setst\modulus{\zcomplex}>\modulus{\sing} \}$, $\ROC{\Ztransformf{\impulsefunction}} = \{\zcomplex\in\Complex \setst \modulus{\zcomplex} > \strate-\ar \}$ and, in consequence, $\ROC{\Ztransformf{\valuederivativevectorn{\infty}}} = \{\zcomplex\in\Complex \setst \modulus{\sing} < \modulus{\zcomplex} < \strate-\ar \}$, which, as predicted, is nonempty
%
%
%where the region of convergence of the last expression is, as predicted, nonempty 
if~$\modulus{\sing}<{\strate-\ar}$ and empty if~$\modulus{\sing}>{\strate-\ar}$.
\obsolete{

Since $\expectation{\exp{-\sing\stochst}}={\strate}/{\strate+\sing}$ and

$ \frac{\ar\strate} {\strate-\ar}\LSTWts{\sing} 
=
\frac{\ar\strate}{\strate-\ar}
 \frac{(\strate-\ar)(\sing+\strate)} {\strate(\sing+\strate-\ar)}
 =
 \frac{\ar(\sing+\strate)} {\sing+\strate-\ar}
 ,$
 
\begin{equation}\label{Carsingexample}\begin{array}{ll}
\Carsing{\ar}{\sing}=\frac{\ar(\sing+\strate)}{\sing(\sing+\strate-\ar)},
&
\frac{\ar}{1-\load} = \frac{\ar\strate}{\strate-\ar},
\end{array}\end{equation}
}%
Picking~$\LSTWtfunction$ from Table~\ref{table:moments}, the inverse Z-transform of~$\Ztransformf{\valuederivativevectorn{\infty}}$ then gives, for~$\nn\in\Natural$, 
\begin{equation}\notag%\label{example:valuederivativek}
\textstyle
%\begin{array}{rcl}
\valuederivativek{\nn}
%&=&
%\frac{\ar\sing}{\sing+\strate-\ar}\left(\frac{\sing+\strate}{\zcomplex+\sing}-\frac{\ar}{\zcomplex-(\strate-\ar)}\right) 
%\\
=%&=&
\frac{\ar(\sing+\strate)}{\sing+\strate-\ar} (\dirack{\nn}-(-\sing)^{\nn}) + \frac{\ar^2\sing}{(\strate-\ar)(\sing+\strate-\ar)}
% \frac{1}{1-\zcomplex/(\strate-\ar)}
\dirack{\nn}
=%\refereq{(\refeq{expPKformula})}{=}
%\\
%&\refereq{(\refeq{expPKformula})}{=}&
\frac{\ar}{1-\load} \big(
\dirack{\nn}
-
 \LSTWts{\sing} 
(-\sing)^{\nn} 
\big) ,
%\\
%&=&
%\frac{d^{k+1}\purevaluex{0}}{d\bl^{k+1}}
%\end{array} 
\end{equation}
which is the ${\nn}$-th derivative at~$0$ of $%\dpurevaluex{\bl}=
 {\ar}/({1-\load}) (1- \LSTWts{\sing} \exp{-\sing\bl}
)$.
%where~$\purevaluex{\bl} = \Carsing{\ar}{\sing} (\exp{-\sing\bl}-1) + {\ar\bl}/({1- \load })$ is the \mainsolution{} given in~(\refeq{ex1:start}).
%
%
%
It follows that~\eqref{sumsumx}  converges for~$\bl\in\Realpluszero$, and we find, in accordance with Table~\ref{table:valuefunctions},
\begin{equation} \notag\label{expsumsumx}
\textstyle
\purevaluex{\bl}
 \refereq{\eqref{sumsumx}}{=}
%\frac{\ar}{1-\load} \big( \bl-\LSTWts{\sing} \frac{1-\exp{-\sing\bl}}{\sing}\big)
[{\ar}/({1-\load})]\, [ \bl-\LSTWts{\sing} ({1-\exp{-\sing\bl}})/{\sing}]
.
\end{equation}

\obsolete{
Conversely, if~$\purevaluefunction$ is  given by~\eqref{expsumsumx},
we find, successively,
\begin{equation}\label{invZtransformsexample}\begin{array}{rcll}
\Ztransformfz{\valuederivativevectorn{\infty}}{\zcomplex} 
&=&
\frac{\ar\strate}{\strate-\ar}-  \frac{\ar(\sing+\strate)}{(\sing+\strate-\ar)} \frac{\zcomplex}{\zcomplex+\sing}
\obsolete{
=
 \frac{\ar\sing\left(\ar\zcomplex+\strate(\sing+\strate-\ar)\right)}{(\strate-\ar)(\sing+\strate-\ar)(\zcomplex+\sing)}
=
\frac{\ar\sing}{\sing+\strate-\ar} \frac{\ar\zcomplex+\strate(\sing+\strate-\ar)}{(\zcomplex+\sing)(\strate-\ar)}
}
, &  \makeROC{\modulus{\zcomplex} > \modulus{\sing} }
, \\
( \Ztransformfz{\impulsefunction}{-\zcomplex} )^{-1}
%\Ztransformfz{\invimpulsefunction}{\zcomplex}
%&\refereq{(\refeq{expspecialmatrixnk})}{=}&
&=&
 % \frac{\zcomplex-(\strate-\ar)}{\ar(\zcomplex-\strate )}=
 \frac{1}{\ar} + \frac{1}{\zcomplex-\strate } 
, & \makeROC{ \modulus{\zcomplex} < \strate }
,\\
\Ztransformfz{\costderivativevectorn{\infty}}{\zcomplex}
&\refereq{\eqref{convolution}}{=}&
\frac{\sing}{\zcomplex+\sing}
+
  \frac{\ar\sing}{(\strate-\ar)(\sing+\strate-\ar)}  \frac{\zcomplex}{\zcomplex-\strate} 
, \quad& \makeROC{ \modulus{\sing}<\modulus{\zcomplex}<\strate }
.
% a=(\strate-\ar), b=\strate/\ar(\sing+\strate-\ar)
% c=\strate, d=\sing
%
% (z-a)(z+b)/(z-c)(z+d) = 1 + B/(z-c) + C/(z+d)
%= z^2+(d-c)z-cd + B z + Bd+Cz-Cc 
%=z^2+(d-c+B+C)z-cd+Bd-Cc
%= z^2+(b-a)z-ab
% =>
% d-c+B+C=b-a =>B=b-a+c-d-C
% -cd+(b-a+c-d-C)d-Cc=-ab => C=d(b-a-d)/(c+d)
% => B=( (b-a+c-d)(c+d)-d(b-a-d) )/(c+d)= c(b-a+c)/(c+d)
%
% c+d=\sing+\strate
% d(b-a-d) = \sing/\ar(\sing+\strate-\ar)(\strate-\ar)
% c(b-a+c) = \strate/\ar (\strate(\sing+\strate-\ar)+\ar^2)
\end{array}
\end{equation}
If~$\modulus{\sing}<\modulus{\dominantpoleX{\stochst}}=\strate$,  we recover
$%\begin{equation}\begin{array}{rcl}
\costderivativevectornk{\nn}{\nn}
=%&=&
\dirack{\nn}-(-\sing)^\nn 
%=\frac{d^{k+1}\costx{0}}{d\bl^{k+1}}
$ for all~$\nn$ %\end{array}\end{equation}
by inverting~$\Ztransformf{\costderivativevectorn{\infty}}$. %This result was predicted by Proposition~\ref{proposition:costtaylorseries}.
}%

\noindent
\emph{Interval bounds.}
Notice that $\costfunction \in \polycostnfunction{\nn} + \intervx{\remaindercostnfunction{\nn}} $ holds if we set $
%\intervxy{\remaindercostnfunction{\nn}}{\bl}= [0,\frac{\sing^{\nn+1}}{\factorial{(\nn+1)}}]\,{\bl^{\nn+1}}
\intervxy{\remaindercostnfunction{\nn}}{\bl}{\,=\,} [0,{\sing^{\nn+1}}/{\factorial{(\nn{\,+\,}1)}}]\,{\bl^{\nn+1}}
$ for~$\nn$ even, and $
%\intervx{\remaindercostnfunction{\nn}} = [-\frac{\sing^{\nn+1}}{\factorial{(\nn+1)}},0]\,{\bl^{\nn+1}}
\intervx{\remaindercostnfunction{\nn}} = [-{\sing^{\nn+1}}/{\factorial{(\nn+1)}},0]\,{\bl^{\nn+1}}
$ for~$\nn$ odd.
The resulting interval~$\intervx{\complexboundpurevaluenfunction{\nn}}$ follows by inspection of Table~\ref{table:valuefunctions}.
%Polynomial bounds for~$\purevaluefunction$ can be inferred from Proposition~\ref{proposition:polynomialbounds}.
Figure~\ref{figure:expexample} displays the interval bounds
$\intervx{\purevaluenfunction{\nn}} =
\polypurevaluenfunction{\nn}+ \intervx{\complexboundpurevaluenfunction{\nn}}$
obtained for~$\purevaluefunction$ for various real values of~$\sing$.
%, where the bounding intervals for the derivatives are set to $\intervxy{\realintervn{2\nn}}{\bl}=[0,\sing^{2\nn+1}]$, $\intervxy{\realintervn{2\nn+1}}{\bl}=[-\sing^{2(\nn+1)},0]$, and  $\intervxy{\imaginaryintervn{\nn}}{\bl}=0$ for~$\bl\in\Realpluszero$ and for all~$\nn$, and the residual's bounds are computed according to~$(\refeq{polybound})$.
%We see in Fig.~\refeq{figure:expexample} that 
The  sequence~$\{\purevaluenfunction{\nn}\}$ shows to converge towards~$\purevaluefunction$ %as~$\nn$ grows 
for as long as~$\sing<\strate-\ar$.  The generation of such a sequence is, however, impossible when ${\sing}\geq{\strate-\ar}$, as the limit coefficients~$\valuederivativek{k}$ are then infinite.
\remarkmarker
\end{example}

In the next two examples, we consider the piecewise cost function
\begin{equation}\label{Makepurevfcost}\notag
\begin{array}{ll}
\costx{\bl} =   \sum\nolimits_{j=0}^{\nn} \pwkjin{j}  \bl^{j} \, \indicator{[0,\taufunction)}{\bl}
  + \excx{\bl}  \, \indicator{[\taufunction,\infty)}{\bl}
,& \quad \forall \bl\in\Realpluszero,
\end{array}
\end{equation} 
where $\excx{\bl}= \pwkout  \bl^{\pwnout} \exp{-\sing\bl}$, $\nn,\pwnout\in\Natural$, and~$\sing\in\Complex$.
%
%The \acl{wfunction} of~$\costfunction$ is denoted by $\Makepurevf{\pwkjin{0} ,\dots,\pwkjin{\nn}}{\excfunction}$. %

 %%%
 %%%

 %
\begin{example}[Polynomial cost in an interval]\label{example:piecewiseanalyticexp}
Consider service times exponentially distributed with parameter~$\strate>\ar$, i.e $\distriXx{\stochst}{x}= 1-\exp{-\strate x}$, and 
the cost function~\eqref{approximatedcostfunction} with~$\excfunction\equiv 0$. For this problem we have $\costix{0}{\bl}= \sum\nolimits_{j=0}^{\nn} \pwkjin{j}  \bl^{j}$, $\costifunction{1}= 0$ and $\pwdeltafunction=-\costifunction{0}$.
%
%and~\eqref{pwLT} used $\nn+1$ times with~$\sing=0$ also gives us 
%$\pwshiftedlaplacest{\complex}{\taufunction} = \sum\nolimits_{j=0}^{\nn} \pwkjin{j}\factorial{j}  \sum\nolimits_{q=0}^{j} {\taufunction^q}/({\factorial{q}\complex^{j-q+1}})$.
Besides,
\begin{equation} \label{transforms:piecewiseanalyticexp}
\textstyle
\laplacetransformfs{\costifunction{0}}{\complex} 
= 
\sum\nolimits_{j=0}^{\nn} \pwkjin{j} \frac{\factorial{j}}{\complex^{j+1}}
,\qquad
\pwshiftedlaplacest{\complex}{\taufunction} 
\refereq{\eqref{pwLT}}{=} 
-
\sum\nolimits_{j=0}^{\nn} \pwkjin{j} \factorial{j}  \sum\nolimits_{q=0}^{j} \frac{\taufunction^q}{\factorial{q}\complex^{j-q+1}} .
\end{equation}

%Since~$\costifunction{1}=0$, $\rightderivative\purevaluex{\bl}=0$ for $\bl\in(\taufunction,\infty)$.
Since $\Polef{\laplacetransformf{\costfunction}}=\emptyset$, \eqref{solutionFNPlarge} gives $\rightderivative\purevaluex{\bl}=0$ for $\bl\in(\taufunction,\infty)$.

For $\bl\in(0,\taufunction)$: using~\eqref{transforms:piecewiseanalyticexp},  $   \Polef{\laplacetransformf{\costifunction{0}}}=\{0\}$,
 $\load=\ar/\strate$, and the expression for~$\LSTWtfunction$ given in Table~\ref{table:moments}, % \eqref{expPKformula} with $\Polef{\LSTWtfunction}=\{\ar-\strate\}$,
\eqref{solutionFNP} reduces, after straightforward computations, to
\begin{equation} \label{sequelexamplefprdpv}
\textstyle
\begin{array}{c}
\rightderivative\purevaluex{\bl}
= 
%\frac{\ar\strate}{\strate-\ar} \bigresiduefx{ \LSTWts{-\complex} \laplacetransformfs{\costifunction{0}}{\complex} \, e^{\complex\bl} }{\complex=0} \\ -\frac{\ar\strate}{\strate-\ar} \bigresiduefx{ \LSTWts{-\complex}\pwshiftedlaplacest{\complex}{\taufunction} \, e^{\complex(\bl-\taufunction)} }{\complex=\strate-\ar} 
%
%\frac{\ar\strate}{\strate-\ar} \sum\nolimits_{j=0}^{\nn} \pwkjin{j}\factorial{j} \, \bigresiduefx{ \LSTWts{-\complex}  \frac{e^{\complex\bl}}{\complex^{j+1}} }{\complex=0} \\ +\frac{\ar\strate}{\strate-\ar} \sum\nolimits_{j=0}^{\nn} \pwkjin{j} \factorial{j}  \sum\nolimits_{q=0}^{j} \frac{\taufunction^q}{\factorial{q}} \bigresiduefx{ \LSTWts{-\complex}\frac{ e^{\complex(\bl-\taufunction)}}{\complex^{j-q+1}}   }{\complex=\strate-\ar} 
%
\frac{\ar\strate}{\strate-\ar} \sum\limits_{j=0}^{\nn}\factorial{j} \pwkjin{j}  \big[ \bigresiduefx{ \LSTWts{-\complex}  \frac{e^{\complex\bl}}{\complex^{j+1}} }{\complex=0} +  \sum\nolimits_{q=0}^{j} \frac{\taufunction^q}{\factorial{q}} \bigresiduefx{ \LSTWts{-\complex}\frac{ e^{\complex(\bl-\taufunction)}}{\complex^{j-q+1}}   }{\complex=\strate-\ar} \big]
\\ \refereq{\eqref{residue}}{=}
\ar \sum\limits_{j=0}^{\nn} \pwkjin{j} \bl^{j} + \ar^2 \sum\limits_{j=0}^{\nn}  \frac{\factorial{j}\pwkjin{j} }{(\strate-\ar)^{j+1}}   \Big(     \sum_{q=0}^{j}\frac{[(\strate-\ar)\bl]^q}{\factorial{q}}      - e^{-(\strate-\ar)(\taufunction-\bl)}     \sum\nolimits_{q=0}^{j} \frac{[(\strate-\ar)\taufunction]^q}{\factorial{q}}   \Big)
.
\end{array}
\end{equation}
\obsolete{
$$
\LSTWts{-\complex}
=
\frac{(\strate-\ar)(\strate-\complex)} {\strate(\strate-\ar-\complex)}
=
\frac{(\strate-\ar)}{\strate}
[ 1 +
\frac{\ar} {(\strate-\ar-\complex)}
]
$$
$$
\frac{d^n}{d\complex^n}
\LSTWts{-\complex}
=
\frac{(\strate-\ar)}{\strate}[ \dirack{n} +\frac{\ar \factorial{n}} {(\strate-\ar-\complex)^{n+1}}]
$$
}%
Integration of~\eqref{sequelexamplefprdpv} yields, for $\bl\in\Realpluszero$,
\begin{equation}\label{vfexamplefprdpv}\notag
\begin{array}{l}
\purevaluex{\bl}
=
 %\int\nolimits_{0}^{\min(\bl,\taufunction)} 
 \ar
\sum\nolimits_{j=0}^{\nn}  \pwkjin{j}  
\frac{ [\min(\bl,\taufunction) ]^{j+1} }{{j+1}} 
% \,d\altbl
\\ \hfill\qquad 
+ \ar^2  
\sum\nolimits_{j=0}^{\nn}\frac{\factorial{j}\,  \pwkjin{j} }{(\strate-\ar)^{j+2}} 
\sum\nolimits_{q=0}^{j} \left\{
  \frac{  [(\strate-\ar)\min(\bl,\taufunction)]^{q+1}  }{\factorial{(q+1)}} 
-
\frac{\exp{(\strate-\ar)\min(\bl,\taufunction)}-1}{\exp{(\strate-\ar)\taufunction}}   \frac{[(\strate-\ar)\taufunction]^{q}}{\factorial{q}} 
 \right\} 
.
\end{array}
\end{equation} 
\end{example}

%%%

\begin{example}\label{example:piecewiseexp}
Consider service times exponentially distributed with parameter~$\strate>\ar$, and the cost function $\costx{\bl}=\bl^\pwnout\exp{-\sing\bl} \, \indicator{[\taufunction,\infty)}{\bl}%\stepx{\bl-\taufunction}
$, i.e.~\eqref{approximatedcostfunction} with~$\nn=0$, $\pwkjin{0}=0$.
We have $\costix{0}{\bl}= 0$ and $\costifunction{1}=\pwdeltafunction= \bl^\pwnout\exp{-\sing\bl} $, so that
\begin{equation} \label{transforms:piecewiseexp}
\textstyle
\laplacetransformfs{\costifunction{0}}{\complex} 
= 
0
,\qquad
\pwshiftedlaplacest{\complex}{\taufunction} 
\refereq{\eqref{pwLT}}{=} 
\factorial{\pwnout} \exp{-\sing\taufunction} \sum\nolimits_{q=0}^{\pwnout} \frac{\taufunction^q}{\factorial{q}(\complex+\sing)^{\pwnout-q+1}}  .
\end{equation}

For $\bl\in(\taufunction,\infty)$, we %use $\load=\ar/\strate$ and
use~Table~\ref{table:moments} %\eqref{expPKformula} 
and $\Polef{\laplacetransformf{\costfunction}}=\{-\sing\}$, and get   
\begin{equation}\label{example:piecewiseexplarge}
\nocolsep
\begin{array}{rcl}
\rightderivative\purevaluex{\bl}
&\refereq{\eqref{solutionFNPlarge}}{=}&
\frac{\factorial{\pwnout}\ar\strate \, \exp{-\sing\taufunction}}{\strate-\ar}\sum\nolimits_{q=0}^{\pwnout} \frac{\taufunction^{\pwnout-q}}{\factorial{(\pwnout-q)}} \bigresiduefx{ \LSTWts{-\complex} \frac{e^{\complex(\bl-\taufunction)}}{(\complex+\sing)^{q+1}}  }{\complex=-\sing} 
 \\
&\refereq{\eqref{residue}}{=} &
%\ar\bl^\pwnout e^{-\sing\bl} + \factorial{\pwnout}\ar^2 \, e^{-\sing\bl}\sum\nolimits_{q=0}^{\pwnout} \frac{\taufunction^{\pwnout-q}}{\factorial{\pwnout-q}}  \sum_{j=0}^{q} \frac{(\bl-\taufunction)^{q-j} }{\factorial{(q-j)}}     \frac{1 } {(\strate-\ar+\sing)^{j+1}}
%
%\ar\bl^\pwnout e^{-\sing\bl} + \factorial{\pwnout}\ar^2 \sum\nolimits_{q=0}^{\pwnout}  \frac{\bl^q } {\factorial{q}(\strate-\ar+\sing)^{\pwnout-q+1}} \, e^{-\sing\bl} .
%
\ar\bl^\pwnout e^{-\sing\bl} + \frac{ \factorial{\pwnout}\ar^2 }{(\strate-\ar+\sing)^{\pwnout+1}} \{ \sum\nolimits_{q=0}^{\pwnout}  \frac{1}{\factorial{q}} [(\strate-\ar+\sing)\bl]^q  \} \, e^{-\sing\bl} .
 \end{array}
\end{equation}
% j 0 -> q   j<= q
% q 0 -> k
% j 0 -> k
% q j -> k 
% (k-q) 0 -> (k-j)   
% (k-j) 0 -> k
Alternatively, \eqref{example:piecewiseexplarge} can be derived from~\eqref{bilaterallaplacetransformcostanfunction} with cost function~$\costifunction{1}$, or by inspection of Table~\ref{table:valuefunctions} for~$\costifunction{1}$ via computation of~\eqref{table:poweracoefnk}.
\obsolete{
%$ \rightderivative\purevaluex{\bl}$ can be obtained from Table~\refeq{table:valuefunctions} using the general algorithm~\eqref{table:poweracoefnk}. 
\begin{equation}
 \textstyle
 \rightderivative\purevaluex{\bl}
\refereq{\eqref{bilaterallaplacetransformcostanfunction}}{=}
%
%\frac{\factorial{\pwnout}\ar}{1-\load}   \sum\nolimits_{j=0}^{\pwnout}  \big( \frac{(-1)^{\pwnout-j}}{\factorial{(\pwnout-j)}}  \frac{d^{\pwnout-j} \LSTWts{\sing}}{d\complex^{\pwnout-j}} \big) \frac{ \bl^j}{\factorial{j}} \exp{-\sing\bl}
%
%\frac{\factorial{\pwnout}\ar\strate}{\strate-\ar}   \sum\nolimits_{j=0}^{\pwnout}  \big( \frac{1}{\factorial{(\pwnout-j)}}  \frac{d^{\pwnout-j} }{d\complex^{\pwnout-j}} [ \LSTWts{-\complex}]_{\complex=-\sing} \big) \frac{ \bl^j}{\factorial{j}} \exp{-\sing\bl}
%
%\factorial{\pwnout}\ar    \sum\nolimits_{j=0}^{\pwnout}  \frac{1}{\factorial{(\pwnout-j)}} \big(  \dirack{\pwnout-j} +\frac{\ar \factorial{(\pwnout-j)}} {(\strate-\ar+\sing)^{\pwnout-j+1}} \big) \frac{ \bl^j}{\factorial{j}} \exp{-\sing\bl}
%
\factorial{\pwnout}\ar  \frac{ \bl^\pwnout}{\factorial{\pwnout}} \exp{-\sing\bl} + \factorial{\pwnout}\ar^2   \sum\nolimits_{q=0}^{\pwnout} \frac{\bl^q}{\factorial{q}(\strate-\ar+\sing)^{\pwnout-q+1}}   \exp{-\sing\bl}
  \end{equation}
}%

For $\bl\in(0,\taufunction)$, we combine~\eqref{solutionFNP} with~\eqref{transforms:piecewiseexp},  $   \Polef{\laplacetransformf{\costifunction{0}}}=\emptyset$,
  and ~$\LSTWtfunction$ (Table~\ref{table:moments}) %~\eqref{expPKformula} with $\Polef{\LSTWtfunction}=\{\ar-\strate\}$ 
  to get
\begin{equation}\notag %\label{notusedexamplefprdpv}
\textstyle
\begin{array}{c}
\rightderivative\purevaluex{\bl}
= 
 -
\frac{\ar\strate}{\strate-\ar} \factorial{\pwnout} \exp{-\sing\taufunction} \sum\nolimits_{q=0}^{\pwnout}
\bigresiduefx{\LSTWts{-\complex}  \, \frac{\taufunction^q}{\factorial{q}(\complex+\sing)^{\pwnout-q+1}}  \, e^{\complex(\bl-\taufunction)} }{\complex=\strate-\ar} 
\\
=
%\ar^2\factorial{\pwnout} \exp{-\sing\taufunction} \sum\nolimits_{q=0}^{\pwnout}  \frac{\taufunction^q}{\factorial{q}(\strate-\ar+\sing)^{\pwnout-q+1}}  \, e^{(\strate-\ar)(\bl-\taufunction)} 
%
\frac{\factorial{\pwnout}\ar^2 \exp{-\sing\taufunction} }{ (\strate-\ar+\sing)^{\pwnout+1}} \{ \sum\nolimits_{q=0}^{\pwnout}  \frac{1}{\factorial{q}}[(\strate-\ar+\sing)\taufunction]^q \} \, e^{(\strate-\ar)(\bl-\taufunction)} 
%
% \frac{\factorial{\pwnout} \ar^2 \,e^{-(\strate-\ar+\sing)\taufunction} }{ (\strate-\ar+\sing)^{\pwnout+1}} \, \{ \sum\nolimits_{q=0}^{\pwnout}  \frac{1}{\factorial{q}}[(\strate-\ar+\sing)\taufunction]^q \} \, e^{(\strate-\ar)\bl} 
  .
\end{array}
\end{equation}
Hence, if $\sing\neq 0$,
\begin{equation}\label{expexamplebvf}\notag
\begin{array}{l}
\purevaluex{\bl}
=
\frac{ \factorial{\pwnout}\ar^2 \,e^{-(\strate-\ar+\sing)\taufunction} }{ (\strate-\ar) (\strate-\ar+\sing)^{\pwnout+1}} \, \{ \sum\nolimits_{q=0}^{\pwnout}  \frac{1}{\factorial{q}}[(\strate-\ar+\sing)\taufunction]^q \} \, (e^{(\strate-\ar)\min(\bl,\taufunction)}-1) 
\\
\ \ +
%\factorial{\pwnout}\ar \left[  \frac{  \frac{1}{\sing} \sum_{j=0}^{\pwnout} \left(\frac{\pwnout}{\sing}\right)^{\pwnout-j}[\taufunction^j e^{-\sing\taufunction}-[\max(\bl,\taufunction)]^j e^{-\sing\max(\bl,\taufunction)}]      }{\factorial{\pwnout}}  + \ar   \sum\nolimits_{q=0}^{\pwnout} \frac{\frac{1}{\sing} \sum_{j=0}^{q} \left(\frac{q}{\sing}\right)^{q-j}[\taufunction^j e^{-\sing\taufunction}-[\max(\bl,\taufunction)]^j e^{-\sing\max(\bl,\taufunction)}] }{\factorial{q}(\strate-\ar+\sing)^{\pwnout-q+1}} \right]  
% \frac{\ar}{\sing}\sum_{j=0}^{\pwnout}\Big[  \left(\frac{\pwnout}{\sing}\right)^{\pwnout-j}         + \factorial{\pwnout}\ar   \sum\nolimits_{q=0}^{\pwnout}  \frac{ \left(\frac{q}{\sing}\right)^{q-j} }{\factorial{q}(\strate-\ar+\sing)^{\pwnout-q+1}} \Big]  [\taufunction^j e^{-\sing\taufunction}-  [\max(\bl,\taufunction)]^j e^{-\sing\max(\bl,\taufunction)}] 
% \frac{ \factorial{\pwnout}\ar^2 }{(\strate-\ar)(\strate-\ar+\sing)^{\pwnout+1}} \{ \sum\nolimits_{q=0}^{\pwnout} [[(\strate-\ar+\sing)/\sing]^{\pwnout-q+1}-1 ]  \frac{1}{\factorial{q}} (\strate-\ar+\sing)^q   \, [\taufunction^q e^{-\sing\taufunction} -\max(\bl,\taufunction)^q e^{-\sing\max(\bl,\taufunction)} ]  \}
% \frac{ \factorial{\pwnout}\ar^2 }{\strate-\ar} \Big\{ \sum\nolimits_{q=0}^{\pwnout} \Big( \frac{ [(\strate-\ar+\sing)/\sing]^{\pwnout-q+1}-1  }{\factorial{q} \, (\strate-\ar+\sing)^{\pwnout-q+1}}  \Big)    \, [\taufunction^q e^{-\sing\taufunction} -\max(\bl,\taufunction)^q e^{-\sing\max(\bl,\taufunction)} ]  \Big\}
 \frac{ \factorial{\pwnout}\ar^2 }{\strate-\ar} \big\{ \sum\nolimits_{q=0}^{\pwnout}  \frac{1}{\factorial{q}}
[ \sing^{-(\pwnout-q+1)}-(\strate-\ar+\sing)^{-(\pwnout-q+1)} ] 
   \, [\taufunction^q e^{-\sing\taufunction} -\max(\bl,\taufunction)^q e^{-\sing\max(\bl,\taufunction)} ]  \big\}
 %
 % \frac{\ar}{\sing}\sum_{q=0}^{\pwnout}\Big[  \left(\frac{\pwnout}{\sing}\right)^{\pwnout-q}         + \factorial{\pwnout}\ar   \sum\nolimits_{j=0}^{\pwnout}  \frac{ \left(\frac{j}{\sing}\right)^{j-q} }{\factorial{j}(\strate-\ar+\sing)^{\pwnout-j+1}} \Big]  [\taufunction^q e^{-\sing\taufunction}-  [\max(\bl,\taufunction)]^q e^{-\sing\max(\bl,\taufunction)}] 
,
 %,&\forall\bl\in\Realpluszero
\end{array}
\end{equation} 
\obsolete{
\begin{equation}\label{expexamplebvfzero}
\begin{array}{l}
\purevaluex{\bl}
=
\frac{\ar^2 \,e^{-(\strate-\ar+\sing)\taufunction} }{ (\strate-\ar) (\strate-\ar+\sing)}  \, (e^{(\strate-\ar)\min(\bl,\taufunction)}-1) 
 +
 \frac{\ar}{\sing}\Big[  1         +     \frac{ \ar}{\strate-\ar+\sing} \Big]  [ e^{-\sing\taufunction}-  e^{-\sing\max(\bl,\taufunction)}] 
,
\end{array}
\end{equation} 
}%
and, if~$\sing=0$,
\begin{equation}\label{expexamplebvfsingzero}\notag
\begin{array}{l}
\purevaluex{\bl}
=
\frac{ \factorial{\pwnout}\ar^2 \,e^{-(\strate-\ar)\taufunction} }{  (\strate-\ar)^{\pwnout+2}} \, \{ \sum\nolimits_{q=0}^{\pwnout}  \frac{1}{\factorial{q}}[(\strate-\ar)\taufunction]^q \} \, (e^{(\strate-\ar)\min(\bl,\taufunction)}-1) 
\\
\quad +
%\frac{\ar}{\pwnout+1}([\max(\bl,\taufunction)]^{\pwnout+1}-\taufunction^{\pwnout+1})  + \frac{ \factorial{\pwnout}\ar^2 }{(\strate-\ar)^{\pwnout+2}} \big\{ \sum\nolimits_{q=0}^{\pwnout}  \frac{[(\strate-\ar)\max(\bl,\taufunction)]^{q+1}-[(\strate-\ar)\taufunction]^{q+1}}{\factorial{q}(q+1)}   \big\} 
\frac{\ar\taufunction^{k+1}}{\pwnout+1}([\max(\frac\bl\taufunction,1)]^{\pwnout+1}-1)  + \frac{ \factorial{\pwnout}\ar^2 }{(\strate-\ar)^{\pwnout+2}} \big\{ \sum\nolimits_{q=0}^{\pwnout}  \frac{[(\strate-\ar)\taufunction]^{q+1}  }{\factorial{(q+1)}} ([\max(\frac\bl\taufunction,1)]^{q+1}-1)  \big\} 
.
\end{array}
\end{equation} 

\obsolete{

$$\textstyle
\int \altbl^0 e^{-\sing\altbl} d\altbl
% q \altbl^{q-1}
% -\frac{e^{-\sing\altbl}}{\sing}
=
[- \frac{e^{-\sing\altbl}}{\sing}]
$$

$$\textstyle
\int \altbl^q e^{-\sing\altbl} d\altbl
% q \altbl^{q-1}
% -\frac{e^{-\sing\altbl}}{\sing}
=
[-\altbl^q \frac{e^{-\sing\altbl}}{\sing}]
+ \frac{q}{\sing}\int  \altbl^{q-1} e^{-\sing\altbl} d\altbl
=
\sum_{j=0}^{q} \left(\frac{q}{\sing}\right)^{q-j}[-\altbl^j \frac{e^{-\sing\altbl}}{\sing}]
$$

$$\textstyle
\int_{\taufunction}^{\max(\bl,\taufunction)} \altbl^q e^{-\sing\altbl} d\altbl
% q \altbl^{q-1}
% -\frac{e^{-\sing\altbl}}{\sing}
%=
%\sum_{j=0}^{q} \left(\frac{q}{\sing}\right)^{q-j}[-\altbl^j \frac{e^{-\sing\altbl}}{\sing}]_{\taufunction}^{\max(\bl,\taufunction)}
=
\frac{1}{\sing} \sum_{j=0}^{q} \left(\frac{q}{\sing}\right)^{q-j}[\taufunction^j e^{-\sing\taufunction}-[\max(\bl,\taufunction)]^j e^{-\sing\max(\bl,\taufunction)}]
$$

}%

\end{example}

%%%%%%%%%%%%%%%%%%%%%%%%%

%\section{Computations of the Fourier coefficients}
\section{Proofs and auxiliary results}
\label{appendix:complement}

\begin{proof}[Theorem~\ref{theorem:convergence}]
\ref{convergence:yes}
If~$\eorder$ is the order of growth of the entire cost function~$\costfunction$, and~$\etype$ is its type, then for any~$\epsilon>0$, there is~$k_\epsilon<\infty$ such that,~\cite[Lecture~1]{levin96},
\begin{subequations}\label{levinsup}
\begin{align}
\label{levinsuporder}\begin{array}{ll}
%\frac{1}{\factorial{{k}}}\Modulus{\frac{d^{{k}}\costx{0}}{d\bl^{{k}}}} 
\frac{1}{\factorial{{k}}}\modulus{\dncostx{{k}}{0}}
< \big(\frac{e (\eorder+\epsilon)}{{k}}\big)^{\frac{{k}}{\eorder+\epsilon}}
, & \quad \forall {k}>k_\epsilon,
\end{array}\\
\label{levinsuptype}\begin{array}{ll}
%\frac{1}{\factorial{{k}}}\Modulus{\frac{d^{{k}}\costx{0}}{d\bl^{{k}}}} 
\frac{1}{\factorial{{k}}}\modulus{\dncostx{{k}}{0}}
< \big(\frac{e (\etype+\epsilon) \eorder}{{k}}\big)^{\frac{{k}}{\eorder}}
, &\quad \forall {k}>k_\epsilon.
\end{array}
\end{align}
\end{subequations}
Consider the quantity $\valuederivativek{{k}}= \sum\nolimits_{q=0}^{\infty} \specialmatrixnk{q+1}{q} \, \dncostx{k+q}{0}$ introduced in~\eqref{valuederivativek}, as well as
\begin{equation}\label{valueabsolutederivativek}
\begin{array}{ll}
\valueabsolutederivativek{{k}}=\frac{\ar}{1-\load} \sum\nolimits_{q=0}^{\infty} \specialmatrixnk{}{q}
%\Modulus{\frac{d^{{k}+q}\costx{0}}{d\bl^{{k}+q}}}
 \Modulus{\dncostx{{k}+q}{0}}
,&\quad \forall{k}\in\Natural.
\end{array}\end{equation}
Recall from Proposition~\ref{proposition:analycity}-\eqref{analycity:centered} in Appendix~\ref{appendix:LST} that $\lim\nolimits_{k\to\infty} \specialmatrixnk{}{k+1}/\specialmatrixnk{}{k} = \modulus{\dominantpoleX{\Wt}}^{-1}$. Besides, it can be seen  (e.g. using Stirling's approximation for the factorial)  that
\begin{equation}\label{testresult}\begin{array}{ll}
\lim\nolimits_{{k}\to\infty} \frac{ \factorial{({k}+1)} \big(\frac{e s r}{{k}+{l}+1}\big)^{\frac{{k}+{l}+1}{r}} }{ \factorial{{k}} \big(\frac{e s r}{{k}+{l}}\big)^{\frac{{k}+{l}}{r}} } 
= \Bigg\lbrace\begin{array}{cl} 
0,&\textup{if } r<1
\\
s,&\textup{if } r=1
\\
\infty,&\textup{if } r>1
\end{array}\Bigg\rbrace
,& \quad \forall {l}\in\Natural.
\end{array}\end{equation}
%Under the assumptions of~(\refeq{convergence:yes}), and by taking~$\epsilon$ sufficiently small, 
Equations~\eqref{levinsuporder} and~\eqref{levinsuptype} tell us that, under the assumptions of~\ref{convergence:yes} and by taking~$\epsilon$ sufficiently small, one can find a dominant series for~$\valuederivativek{{k}}$ and~$\valueabsolutederivativek{{k}}$ that successfully passes the ratio test for convergence due to~\eqref{testresult}, so that
%the ratio test and~(\refeq{testresult}) tell us that 
both~$\valuederivativek{{k}}$ and~$\valueabsolutederivativek{{k}}$ are finite for all~${k}$. The finiteness of~$\valueabsolutederivativek{{k}}$ allows us to interchange the integration order in the computation of~$\valuederivativek{{k}}$. 
Noting that $\specialmatrixnk{}{q}=\expectation{\Wt^q}/ \factorial{q}$  for all~$q$  (cf. Proposition~\ref{proposition:analycity}-\eqref{analycity:centered}), % in Appendix~\ref{appendix:LST}), 
we apply \hidecalculustheorems{Fubini's theorem}{Corollary~\ref{corollary:sumexpectation}}  %(Appendix~\refeq{appendix:ar}) 
and find, for~${k}\in\Naturalnn$,
\begin{equation}\label{switchvaluederivativek}
\nocolsep
\begin{array}{rcl}
\valuederivativek{{k}}
&=&
\frac{\ar}{1-\load} \sum\nolimits_{q=0}^{\infty} \expectation{
%\frac{d^{{k}+q}\costx{0}}{d\bl^{{k}+q}} 
\dncostx{{k}+q}{0}
\frac{\Wt^q}{ \factorial{q}} }
%\\
%&\refereq{(\refeq{sumexpectation})}{=}&
\hidecalculustheorems{=}{\refereq{\eqref{sumexpectation}}{=}}
\frac{\ar}{1-\load} \expectation{\sum\nolimits_{q=0}^{\infty}
% \frac{d^{{k}+q}\costx{0}}{d\bl^{{k}+q}} 
\dncostx{{k}+q}{0}
\frac{\Wt^q}{ \factorial{q}} }
%\\
%&\refereq{(\refeq{entirecost})}{=}&
\refereq{\eqref{entirecost}}{=}
\frac{\ar}{1-\load} \expectation{
%\frac{d^{{k}}\costx{\Wt}}{d\bl^{{k}}} 
\dncostx{{k}}{\Wt}
}
.
\end{array}\end{equation}
Similarly, we introduce, for~${k}\in\Natural$,
\begin{equation}\label{switchvalueabsolutederivativek}
\nocolsep
\begin{array}{l}%{rcl}
\valuemixedderivativek{{k}}
=%&=&
\frac{\ar}{1-\load} \Expectation{\Modulus{
%\frac{d^{{k}}\costx{\Wt}}{d\bl^{{k}}}
\dncostx{{k}}{\Wt}
 }}
%\\
%&\refereq{(\refeq{entirecost})}{\leq}&
\refereq{\eqref{entirecost}}{\leq}
\frac{\ar}{1-\load} \Expectation{
\sum\nolimits_{q=0}^{\infty} \Modulus{
%\frac{d^{{k}+q}\costx{0}}{d\bl^{{k}+q}}
\dncostx{{k}+q}{0}
} \frac{\Wt^q}{ \factorial{q}} }
\hspace{29mm}\\\hfill
%&\refereq{(\refeq{sumexpectation})}{=}&
\hidecalculustheorems{=}{\refereq{\eqref{sumexpectation}}{=}}
\frac{\ar}{1-\load} \sum\nolimits_{q=0}^{\infty} \specialmatrixnk{}{q}\Modulus{
%\frac{d^{{k}+q}\costx{0}}{d\bl^{{k}+q}}
\dncostx{{k}+q}{0}
}
%\\
%&\refereq{(\refeq{valueabsolutederivativek})}{=}& 
\refereq{\eqref{valueabsolutederivativek}}{=}
\valueabsolutederivativek{{k}}
.
\end{array}\end{equation}
and~$\valuemixedderivativek{{k}}$ is finite as well.
\obsolete{
\begin{equation}\label{switchvalueabsolutederivativek}
\begin{array}{rcl}
\valueabsolutederivativek{{k}}
&\refereq{\eqref{sumexpectation}}{=}&
\frac{\ar}{1-\load} \Expectation{\Modulus{\frac{d^{{k}}\costx{\Wt}}{d\bl^{{k}}} }}
.
\end{array}\end{equation}
}%
Suppose now that $\modulus{({d^{{k}}}/ {d\bl^{{k}}})\costx{0}}  < {\factorial{{k}}} ({e s r}/{{k}})^{{{k}}/{r}}$ for ${k}>k_\epsilon$|in the case~\ref{convergence:yes}, this holds either for some~$r<1$ or for $r=1$ and some finite~$s$|, and consider the  sequence
\begin{equation}\label{boundingsequence}
%\begin{array}{ll}
\textstyle
\boundvdk{{k}}= \frac{\ar}{1-\load}    \sum\nolimits_{q=0}^{\infty} \factorial{(q+{k})} \, \specialmatrixnk{}{q}  \big(\frac{e s r}{q+{k}}\big)^{\frac{q+{k}}{r}}
, \quad \forall{k}\in\Natural. 
%\end{array}
\end{equation}
It is easy to see that the three  sequences $\sum\nolimits_{{k}=0}^{\infty}\valuederivativek{{k}} { \bl^{{k}+1}}/{\factorial{({k}+1)}} $, $\sum\nolimits_{{k}=0}^{\infty}\valueabsolutederivativek{{k}} { \bl^{{k}+1}}/{\factorial{({k}+1)}} $ and $\sum\nolimits_{{k}=0}^{\infty}\valuemixedderivativek{{k}} { \bl^{{k}+1}}/{\factorial{({k}+1)}} $  converge wherever $\sum\nolimits_{{k}=0}^{\infty} \boundvdk{{k}} { \bl^{{k}+1}}/{\factorial{({k}+1)}} $ is convergent. Besides,
\begin{equation}\label{boundingsequencebis}\begin{array}{c}
\boundvdk{{k}+1}\refereq{\eqref{boundingsequence}}{=} \frac{\ar}{1-\load}    \sum\nolimits_{q=0}^{\infty} \Bigg[ \frac{ 
\factorial{(q+{k}+1)}  \, \big(\frac{e s r}{q+{k}+1}\big)^{\frac{q+{k}+1}{r}}
}{
 \factorial{(q+{k})} \, \big(\frac{e s r}{q+{k}}\big)^{\frac{q+{k}}{r}}
} \Bigg]
\factorial{(q+{k})} \, \specialmatrixnk{}{q} \, \big(\frac{e s r}{q+{k}}\big)^{\frac{q+{k}}{r}} .
\end{array}\end{equation}
In the conditions of~\ref{convergence:yes}, we infer  from~\ref{testresult} that the expression between brackets in~\eqref{boundingsequencebis} tends to a finite quantity not larger than~$s$, so that, for any~$\nu>0$ one can find a~$k_\nu$ such that $\boundvdk{{k}+1}\leq (\boundvdk{k_\nu+1}-\boundvdk{k_\nu}) + (s+\nu) \boundvdk{{k}}  $ for ${k}>k_\nu$. 
%
%It follows from the ratio test that $\sum\nolimits_{{k}=0}^{\infty} \boundvdk{{k}} { \bl^{{k}+1}}/{\factorial{({k}+1)}} $ converges for $\bl\in\Realpluszero$, and so do $\sum\nolimits_{{k}=0}^{\infty}\valuederivativek{{k}} { \bl^{{k}+1}}/{\factorial{({k}+1)}} = \sumsumx{\bl}$ and $\sum\nolimits_{{k}=0}^{\infty}\valueabsolutederivativek{{k}} { \bl^{{k}+1}}/{\factorial{({k}+1)}} $.  Fubini's theorem  (Theorem~\refeq{appendix:ar}-\refeq{theorem:fubini}) yields $\sumsumx{\bl} = \int\nolimits_{{k}=0}^{\bl} \sum\nolimits_{0}^{\infty}\valuederivativek{{k}} { \xi^{{k}}}/{\factorial{{k}}} \, d\xi$. Then, using~(\refeq{switchvaluederivativek}), (\refeq{switchvalueabsolutederivativek}), and Corollary~\refeq{appendix:ar}-\refeq{corollary:sumexpectation}  again, we find, for every~$\bl\in\Realpluszero$,
It follows from the ratio test that $\sum\nolimits_{{k}=0}^{\infty} \boundvdk{{k}} { \xi^{{k}}}/{\factorial{{k}}} $ converges for $\xi\in\Realpluszero$, and so do $\sum\nolimits_{{k}=0}^{\infty}\valuederivativek{{k}} { \xi^{{k}}}/{\factorial{{k}}} = \sumsumx{\bl}$, $\sum\nolimits_{{k}=0}^{\infty}\valueabsolutederivativek{{k}} {\xi^{{k}}}/{\factorial{{k}}} $ and $\sum\nolimits_{{k}=0}^{\infty}\valuemixedderivativek{{k}} {\xi^{{k}}}/{\factorial{{k}}} $.  This last conclusion, together with~\eqref{switchvaluederivativek}, \eqref{switchvalueabsolutederivativek}, and %\refeq{appendix:ar}-
\hidecalculustheorems{Fubini's theorem applied to set of natural numbers  with the counting measure}{Corollary~\ref{corollary:sumexpectation}}, yields, for~$\bl\in\Realpluszero$,
\begin{equation}\label{switchtotal}\notag
\begin{array}{l}%{rcl}
\sumsumx{\bl} 
\refereq{\eqref{switchvaluederivativek}}{=}%&\refereq{(\refeq{switchvaluederivativek})}{=}&
\frac{\ar}{1-\load} \int\nolimits_{0}^{\bl} 
\sum\nolimits_{{k}=0}^{\infty}  \expectation{
%\frac{d^{{k}}\costx{\Wt}}{d\bl^{{k}}} 
\dncostx{{k}}{\Wt}
 \frac{ \xi^{{k}}}{\factorial{{k}}} } \, d\xi
%\\&\refereq{(\refeq{sumexpectation})}{=}&
\hidecalculustheorems{=}{\refereq{\eqref{sumexpectation}}{=}}
\frac{\ar}{1-\load} \int\nolimits_{0}^{\bl} 
\expectation{\sum\nolimits_{{k}=0}^{\infty} 
% \frac{d^{{k}}\costx{\Wt}}{d\bl^{{k}}}
\dncostx{{k}}{\Wt}
  \frac{ \xi^{{k}}}{\factorial{{k}}} } \, d\xi
%\\&\refereq{(\refeq{entirecost})}{=}&
\hspace{8mm}\\\hfill
\refereq{\eqref{entirecost}}{=}
\frac{\ar}{1-\load} \int\nolimits_{0}^{\bl} 
\Expectation{\costx{\xi+\Wt} } \, d\xi
%\\&\showhideproposition{\refereq{(\refeq{dwfunction})}{=}}{=}&
\showhideproposition{\refereq{\eqref{dwfunction}}{=}}{=}
\purevaluex{\bl} ,
\end{array}\end{equation}
where the last result follows 
\showhideproposition{%
from   Proposition~\ref{proposition:identities}\eqref{proposition:vfdifferentiability}. 
}{%
from~\cite[Proposition~1]{hyytia-vgf-itc-2017}.
}%
Since
\begin{equation}\label{absolutesumsumx}\notag
\textstyle
\sum\nolimits_{{k}=0}^{\infty}\valueabsolutederivativek{{k}} \frac{\xi^{{k}}}{\factorial{{k}}}
=
\frac{\ar}{1-\load}    %\int\nolimits_{0}^{\bl} 
\sum\nolimits_{{k}=0}^{\infty} \big(\sum\nolimits_{q=0}^{\infty}\modulus{  \specialmatrixnk{}{q} 
%\frac{d^{{k}+q}\costx{0}}{d\bl^{{k}+q}} 
\, \dncostx{{k}+q}{0}
\frac{ \xi^{{k}}}{\factorial{{k}}}}\big) %\d\xi 
<\infty
,   \quad
\forall \xi\in\Realpluszero,
\end{equation}
\hidecalculustheorems{Fubini's theorem applies}{} and one may  interchange the order of summation in~\eqref{sumsumx}:
\begin{equation}\label{switchsumsumx}\notag
\nocolsep
\begin{array}{rcl}
\sumsumx{\bl}
 &\hidecalculustheorems{=}{\!\!\refereq{\eqref{fubini}}{=}\!\!} &
\frac{\ar}{1-\load}   \int\nolimits_{0}^{\bl} \sum\nolimits_{q=0}^{\infty}\specialmatrixnk{}{q}\big( \sum\nolimits_{{k}=0}^{\infty}  
%\frac{d^{{k}+q}\costx{0}}{d\bl^{{k}+q}}
\dncostx{{k}+q}{0}
 \frac{ \xi^{{k}}}{\factorial{{k}}} \big) \, d\xi
=%\\
% &= &
\frac{\ar}{1-\load}   \int\nolimits_{0}^{\bl}  \sum\nolimits_{q=0}^{\infty}\specialmatrixnk{}{q} 
%\frac{d^{q}\costx{\xi}}{d\bl^{q}} 
\, \dncostx{q}{\xi}
\, d\xi
%\\
%&\refereq{(\refeq{sumx})}{=}&
\!\refereq{\eqref{sumx}}{=}\!\!
\sumx{\bl},
\end{array}\end{equation}
which holds for~$\bl\in\Realpluszero$.

\ref{convergence:no}
Similarly, for any~$\epsilon>0$, one can find growing  sequences of naturals~$\{\sspn{{k}}\}$ and~$\{\ssqn{{k}}\}$ such that, \cite[Lecture~1]{levin96},
\begin{subequations}\label{levininf}
\begin{align}
\label{levininforder}\begin{array}{ll}
%\frac{1}{\factorial{\sspn{{k}}}}\Modulus{\frac{d^{\sspn{{k}}}\costx{0}}{d\bl^{\sspn{{k}}}}}
\frac{\modulus{\dncostx{\sspn{{k}}}{0}} }{\factorial{\sspn{{k}}}}
 > \big(\frac{e (\eorder-\epsilon)}{\sspn{{k}}}\big)^{\frac{\sspn{{k}}}{\eorder-\epsilon}}
, & \quad ({k}\in\Naturalnn),
\end{array}\\
\label{levininftype}\begin{array}{ll}
%\frac{1}{\factorial{\ssqn{{k}}}}\Modulus{\frac{d^{\ssqn{{k}}}\costx{0}}{d\bl^{\ssqn{{k}}}}} 
\frac{\modulus{\dncostx{\ssqn{{k}}}{0}} }{\factorial{\ssqn{{k}}}}
> \big(\frac{e (\etype-\epsilon) \eorder}{\ssqn{{k}}}\big)^{\frac{\ssqn{{k}}}{\eorder}}
, & \quad ({k}\in\Naturalnn).
\end{array}
\end{align}
\end{subequations}
Recall the series~$\valuederivativek{{k}}$ defined  in~\eqref{valuederivativek}.
By taking~$\epsilon$ sufficiently small in~\eqref{levininforder}
  and~\eqref{levininftype} and using~\eqref{testresult}, we find that the asymptotic ratio between the moduli of two terms of~\eqref{valuederivativek} with respective indices~$\sspn{q}-{k},\sspn{q+1}-{k}$  (in the case~$\eorder>1$) or~$\ssqn{q}-{k},\ssqn{q+1}-{k}$ (in the case~$\eorder=1$,  $\etype>\modulus{\dominantpoleX{\Wt}}^{-1}$) is greater than one for~$q$ taken large enough. Hence, one can find a subsequence of terms of~\eqref{valuederivativek} which grows in modulus, and~$\valuederivativek{{k}}$ diverges for all~${k}$.
\qed\end{proof}

\storecounter{figure}{store:figure:contourintegral}%
\begin{figure}[!t]
\centering
\begin{tikzpicture}[scale=0.85]
\def\dotwidth{0.08}
\def\tickwidth{0.05}
\def\gap{-.1}
\def\margin{0.2}
\def\deltaangle{5}
\def\arrowlength{\tauf/4}
\def\arrowangle{70}
\def\Arrowangle{20}
\def\R{2.5}
\def\epsil{0.4}
\def\tauf{2.5}
\def\sig{\tauf*.5+\R*.75}
\def\ballsize{.1}
\def\pole{\R*0.44}
\def\npole{9}
\coordinate  (tau)  at  (0,\tauf);
\coordinate  (sig)  at  (0,\sig);
\draw[fill=black!05!white,thick]  (\tauf*.5+\R,0) arc (0:365:\R) -- cycle;
\draw[->,thin]  (\tauf*.5,0) ++(45-\Arrowangle*0.5:\R+\gap) arc (45-\Arrowangle*0.5:45+\Arrowangle*0.5:\R+\gap) ;
\node  at (\tauf/2+ \R*cos 45  +2* \margin*cos 45,\R * sin 45 +\margin*sin 45 ) {\scalebox{0.8}{$\contoura{1/\epsilon}$}};
\fill[fill=white]  (0,0) ++ (\deltaangle:\epsil) arc (\deltaangle:360-\deltaangle:\epsil) -- cycle;
\draw[thick]  (0,0) ++ (\deltaangle:\epsil) arc (\deltaangle:360-\deltaangle:\epsil) ;
\draw[->,thin]  (0,0) ++ (180+\arrowangle*.5:\epsil-\gap) arc (180+\arrowangle*.5:180-\arrowangle*.5:\epsil-\gap) ;
\node  at (0  +2* \margin,-\epsil -\margin ) {\scalebox{0.8}{$\contouras{\epsilon}{0}$}};
\node  at (\tauf  +2* \margin,-\epsil -\margin ) {\scalebox{0.8}{$\contouras{\epsilon}{\taufunction}$}};
\fill[fill=white]  (\tauf,0) ++ (-180+\deltaangle:\epsil) arc (-180+\deltaangle:180-\deltaangle:\epsil) -- cycle;
\draw[thick]  (\tauf,0) ++ (-180+\deltaangle:\epsil) arc (-180+\deltaangle:180-\deltaangle:\epsil) ;
\draw[->,thin]  (\tauf,0) ++ (\arrowangle*.5:\epsil-\gap) arc (\arrowangle*.5:-\arrowangle*.5:\epsil-\gap) ;
\fill[fill=white]  (0,\epsil*sin \deltaangle) -- (\tauf,\epsil*sin \deltaangle) -- (\tauf,-\epsil*sin \deltaangle) -- (0,-\epsil*sin \deltaangle)  -- cycle;
\draw[thick]  (\epsil*cos \deltaangle,\epsil*sin \deltaangle) -- (\tauf-\epsil*cos \deltaangle,\epsil*sin \deltaangle) ;
\draw[thick] (\tauf-\epsil*cos \deltaangle,-\epsil*sin \deltaangle) -- (\epsil*cos \deltaangle,-\epsil*sin \deltaangle)  ;
\draw[->,thin] (\tauf/2-\arrowlength/2,\epsil*sin \deltaangle-\gap) --  (\tauf/2+\arrowlength/2,\epsil*sin \deltaangle-\gap) ;
\draw[->,thin] (\tauf/2+\arrowlength/2,-\epsil*sin \deltaangle+\gap) --  (\tauf/2-\arrowlength/2,-\epsil*sin \deltaangle+\gap) ;
\makedot{(0,0)}
\makedot{(\tauf,0)}
\makedot{(\sig,0)}
\foreach \i in {1,...,\npole}
{
%\pgfmathsetmacro{\CosValue}{cos(90+180*\i/\npole)} 
%\pgfmathsetmacro{\SinValue}{sin(90+180*\i/\npole)} 
\pgfmathsetmacro{\CosValue}{cos(180/\npole+360*(\i-1)/\npole)}
\pgfmathsetmacro{\SinValue}{sin(180/\npole+360*(\i-1)/\npole)} 
\makedot{(\pole*\CosValue,\pole*\SinValue)}
}
\node  at (-\margin*3,-\pole+0.2*\margin) {{\scalebox{0.8}{$\Polef{\allpolycostnfunction{k}}$}}};
%\coordinate  (C)  at  (0,0);
\draw[very thin] (\tauf*.5-\R-\margin,0) -- (\tauf,0) node[below=0]{\scalebox{0.8}{$\taufunction$}}-- (\sig,0) node[below right=-0.10]{\scalebox{0.8}{$\frac{1}{\polylamb}$}} --(\tauf*.5+\R+3*\margin,0) node[below=0]{$\realpart{\complex}$};
\draw[very thin] (0,-\R-\margin)  -- (0,0) node[below left=-0.07]{\scalebox{0.8}{$0$}}-- (0,\R+\margin) node[right=0]{$\imaginarypart{\complex}$};
\lengthmark{\tauf}{0}{135}{\epsil}{\epsil/4}{\scalebox{0.8}{$\epsilon$}}
\lengthmark{\tauf/2}{0}{90}{\R}{\epsil*0.75}{\scalebox{0.8}{$1/\epsilon$}}
%\lengthmark{\tauf/2}{ - \epsil*sin \deltaangle}{90}{\epsil*sin \deltaangle * 2}{0.3}{\tiny{$\zero{\epsilon}$}}
%\draw[very thin] (C) node[left=0]{$0$} -- (Cy) node[above=0]{$\plusw\plusvnorm{\setofinteresti{\iset}}{\plusvix{\iset}{\llamb}}$};
%\draw  (3,0.5) --  node {hello};
%\draw[black!10!white,thick] (0.5,3) -- (1.5,0) -- (0,0
\end{tikzpicture}
\caption{%
Singularities of  $\allpolycostnx{k}{\complex}  \,  \complexexponent{\complex}{-{1}/{2}}{-\pi} \complexexponent{(\complex-\taufunction)}{-{1}/{2}}{-\pi}\,  (1-\polylamb \complex)^{-\polydegreek{k}}$ and computation of~$\realfouriercoefklamb{k}{\polylamb}$ by contour integration.
} \label{figure:contourintegral}
\end{figure}%
%
% Lemma~\ref{lemma:fouriercoefsforquotientofpolynomials} in Appendix~\ref{appendix:complement}
% (\refeq{solutionrealfouriercoefk})
\begin{lemma}[Coefficients~$\{%\realpart{\fouriercoefk{k}}
\fouriercoefk{k}
\}$ for quotients of polynomials] \label{lemma:fouriercoefsforquotientofpolynomials}
Let~$\npolycostnfunction{\npolydegree}$ and~$\dpolycostnfunction{\dpolydegree}$ be polynomials of degrees~$\npolydegree$ and~$\dpolydegree$, and consider
\begin{equation}\label{polycostfunction}\notag
\textstyle
\costx{\bl}
=
\frac{\npolycostnx{\npolydegree}{\bl}}{\dpolycostnx{\dpolydegree}{\bl}},
\quad \forall \bl \in \Realpluszero.
\end{equation}
For $\taufunction>0$, recall~\eqref{polycostnx} and define $\allpolycostnx{k}{\complex} = 
\costx{\complex} \, \polcosnx{k}{{2\complex}/{\taufunction}-1}%\npolycostnx{\npolydegree}{\complex} \polcosnx{k}{{2\complex}/{\taufunction}-1}/\dpolycostnx{\dpolydegree}{\complex}
$ 
under the assumption $\Polef{\allpolycostnfunction{k}}\cap[0,\infty)=\emptyset$.
%assume that the set of poles of~$\allpolycostnfunction{k}$, $\Pole$, has no intersection with~$[0,\infty)$. 
%, for $\sing\in\Complex\setminus\Pole$ $q=0,\dots,\polydegree-1$, 
The Fourier coefficients~\eqref{realfouriercoefficientstwo}  %~\eqref{realfouriercoefficientstwo} 
of~$\costfunction$ satisfy, for~$k\geq 0$, 
\begin{equation} \label{solutionrealfouriercoefk}\begin{array}{c}
\fouriercoefk{k}
%\realpart{\fouriercoefk{k}}
=
\sqrt{\pi} \sum\nolimits_{q=0}^{\polydegreek{k}}     \frac{ \laurentcoef{-q}\,  (-\taufunction)^{q} }{ \factorial{q} \Gammax{\frac{1}{2}-q}} 
- 
%\frac{1}{2\i\pi}
\sum\nolimits_{\sing\in\Polef{\allpolycostnfunction{k}}} \Residuefx{
\allpolycostnx{k}{\complex}  \,  \complexexponent{\complex}{-\frac{1}{2}}{-\pi} \complexexponent{(\complex-\taufunction)}{-\frac{1}{2}}{-\pi}
}{\complex=\sing},
\end{array}
\end{equation}
\storecompoundcounter{equation}{solutionrealfouriercoefk}%
where~$\polydegreek{k}=\max(0,\npolydegree-\dpolydegree+k)$ is   the largest nonnegative integer~$\polydegree$ such that $\lim\nolimits_{\complex\to 0}
\complex^{\polydegree} \Allpolycostnx{k}{{1}/{\complex}} $ is finite, and~$\{\laurentcoef{q}\}$ are the  coefficients of the Laurent series at~$+\infty$ of the analytic continuation of~$\allpolycostnfunction{k}$, i.e.,
\begin{equation}\label{laurentcoef}\begin{array}{ll} 
\laurentcoef{q} = 
\frac{1}{\factorial{\left(\polydegreek{k}+q\right)}} \lim\nolimits_{\complex\to 0}
\frac{d^{\polydegreek{k}+q}}{d\complex^{\polydegreek{k}+q}} \left[\complex^{\polydegreek{k}} \Allpolycostnx{k}{\frac{1}{\complex}} \right]
, &\quad  (q=-\polydegreek{k},\dots,\infty).
\end{array}\end{equation}
\end{lemma}
\storecompoundcounter{lemma}{lemma:fouriercoefsforquotientofpolynomials}%
A suggestion for deriving the coefficients~$\{%\realpart{\fouriercoefk{k}}
\fouriercoefk{k}
\}$ in Lemma~\ref{lemma:fouriercoefsforquotientofpolynomials} is to consider in the complex domain the contour integral  
\begin{equation} \label{realfouriercoefficientspolyIklamb} \notag
\begin{array}{l}
%\hspace{-5mm}
\integralklamb{k}{\polylamb}
=
\frac{1}{\pi} \ointctrclockwise\nolimits_{\contourfunction} \Big(\frac{\npolycostnx{\npolydegree}{\complex} \, \Polcosnx{k}{\frac{2\complex}{\taufunction}-1}}{\dpolycostnx{\dpolydegree}{\complex} \, (1-\polylamb \complex)^%
{\polydegreek{k}}%{\npolydegree-\dpolydegree+k}
}\Big) \,  \complexexponent{\complex}{-\frac{1}{2}}{-\pi} \, \complexexponent{(\complex-\taufunction)}{-\frac{1}{2}}{-\pi} \, d\complex  
 \hspace{27mm}\\\hfill
 =
 \frac{1}{\pi} \Big( \ointctrclockwise\nolimits_{\contoura{1/\epsilon}} \!  +   \varointclockwise\nolimits_{\contouras{\epsilon}{0}} \!  + \int\nolimits_{\epsilon}^{\taufunction-\epsilon} \!  + \varointclockwise\nolimits_{\contouras{\epsilon}{\taufunction}} \!  + \int\nolimits_{\taufunction-\epsilon}^{\epsilon} \Big) \frac{\allpolycostnx{k}{\complex}  \,  \complexexponent{\complex}{-\frac{1}{2}}{-\pi} \complexexponent{(\complex-\taufunction)}{-\frac{1}{2}}{-\pi}}{ (1-\polylamb \complex)^{\polydegreek{k}}} \, d\complex ,
 \end{array}\end{equation}
where $\complexexponent{\complex}{\alpha}{-\pi}=\exp{\alpha(\ln{\modulus{\complex}}+\i\complexarg{\complex}{-\pi})}$ denotes the principal branch of the complex exponentiation, and the circles~$\contoura{1/\epsilon}$, $\contouras{\epsilon}{0}$, and~$\contouras{\epsilon}{\taufunction}$ are understood as in Figure~\ref{figure:contourintegral} with $\epsilon>0$ chosen small enough so that~$1/\polylamb$ and the poles of~$\allpolycostnfunction{k}$ all lie  between the outer contour~$\contoura{\epsilon}$ and the inner contour.
\obsolete{
\begin{proof}
We would like to compute \begin{equation}\label{realfouriercoefficientspoly} \begin{array}{c}%{ll}
\fouriercoefk{k}
%\realpart{\fouriercoefk{k}}
\refereq{\eqref{realfouriercoefficientstwo} }{=}
 \frac{1}{\pi} \int\nolimits_{0}^{\taufunction} \frac{\npolycostnx{\npolydegree}{\bl} \, \Polcosnx{k}{\frac{2\bl}{\taufunction}-1}}{\dpolycostnx{\dpolydegree}{\bl} \, \sqrt{\bl(\taufunction-\bl)}} \, d\bl,
 %& (k\in\Naturalnn)
\end{array}
\end{equation}
for any~$k\in\Naturalnn$.
For~%
%$\polylamb\gg \taufunction$, 
$\polylamb\ll 1/\taufunction$, 
we define %, for reasons that will become clear later, 
the altered coefficient
\begin{equation}\label{realfouriercoefficientspolylamb} \begin{array}{c}
\realfouriercoefklamb{k}{\polylamb}
=
 \frac{1}{\pi} \int\nolimits_{0}^{\taufunction} \frac{\npolycostnx{\npolydegree}{\bl} \, \Polcosnx{k}{\frac{2\bl}{\taufunction}-1}}{\dpolycostnx{\dpolydegree}{\bl} \, (1-\polylamb \bl)^{\polydegree} \, \sqrt{\bl(\taufunction-\bl)}} \, d\bl,
  %& (k\in\Naturalnn)
\end{array}
\end{equation}
%where~$\polylamb \gg \taufunction$ is a parameter, 
which has the property to converge to~$\realpart{\fouriercoefk{k}}%\realfouriercoefk{k}
$ as~$\polylamb\downarrow 0$.  Indeed, since by assumption~$\costfunction$ and~$\polcosnfunction{k}$ are bounded on~$[0,\taufunction]$, the integrand of~\eqref{realfouriercoefficientspoly} is absolutely integrable on the interval. As soon as~$\polylamb\leq 1/(2\taufunction)$, one  has $\modulus{(1-\polylamb \bl)^{-\polydegree}}\leq 2^\polydegree$ and the conditions of Lebesgue's dominated convergence theorem are met.
%
%To compute~(\refeq{realfouriercoefficientspoly}), 
Now, consider the contour integral in the complex domain 
\begin{equation} \label{realfouriercoefficientspolyIklamb} 
\begin{array}{l}
%\hspace{-5mm}
\integralklamb{k}{\polylamb}
=
\frac{1}{\pi} \ointctrclockwise\nolimits_{\contourfunction} \big(\frac{\npolycostnx{\npolydegree}{\complex} \, \Polcosnx{k}{\frac{2\complex}{\taufunction}-1}}{\dpolycostnx{\dpolydegree}{\complex} \, (1-\polylamb \complex)^{\npolydegree-\dpolydegree+k}}\big) \,  \complexexponent{\complex}{-\frac{1}{2}}{-\pi} \, \complexexponent{(\complex-\taufunction)}{-\frac{1}{2}}{-\pi} \, d\complex  
 \hspace{27mm}\\\hfill
 =
 \frac{1}{\pi} \Big( \ointctrclockwise\nolimits_{\contoura{1/\epsilon}} \!  +   \varointclockwise\nolimits_{\contouras{\epsilon}{0}} \!  + \int\nolimits_{\epsilon}^{\taufunction-\epsilon} \!  + \varointclockwise\nolimits_{\contouras{\epsilon}{\taufunction}} \!  + \int\nolimits_{\taufunction-\epsilon}^{\epsilon} \Big) \frac{\allpolycostnx{k}{\complex}  \,  \complexexponent{\complex}{-\frac{1}{2}}{-\pi} \complexexponent{(\complex-\taufunction)}{-\frac{1}{2}}{-\pi}}{ (1-\polylamb \complex)^{\polydegree}} \, d\complex ,
 \end{array}\end{equation}
where $\complexexponent{\complex}{\alpha}{-\pi}=\exp{\alpha(\ln{\modulus{\complex}}+\i\complexarg{\complex}{-\pi})}$ denotes the principal branch of the complex exponentiation, and the circles~$\contoura{1/\epsilon}$, $\contouras{\epsilon}{0}$, and~$\contouras{\epsilon}{\taufunction}$ are understood as in Figure~\ref{figure:contourintegral} with $\epsilon>0$ chosen small enough so that~$1/\polylamb$ and the poles of~$\allpolycostnfunction{k}$ all lie  between the outer contour~$\contoura{\epsilon}$ and the inner contour.
%
%$\contoura{\epsilon}=\{\complex\setst\modulus{\complex}=1/\epsilon\}$
%
%$\contouras{\epsilon}{0},\contouras{\epsilon}{\taufunction}$
%

%We let $\polyphinlambx{k}{\polylamb}{\complex} = \allpolycostnx{k}{\complex}  \,  \complexexponent{\complex}{-\frac{1}{2}}{-\pi} \complexexponent{(\complex-\taufunction)}{-\frac{1}{2}}{-\pi}\,  (1-\polylamb \complex)^{-\polydegree} $ and 
We proceed to compute~$\integralklamb{k}{\polylamb}$ term by term. Let~$\polydegree=\max(0,\npolydegree-\dpolydegree+k)$. First notice that
\begin{equation}\label{grandesencochesprelim}\begin{array}{c}
\lim\nolimits_{\complex\to\infty} \left(\complex -\frac{\taufunction}{2} \right)
\left( \frac{\allpolycostnx{k}{\complex}  }{ (1-\polylamb \complex)^{\polydegree}} \right)   \complexexponent{\complex}{-\frac{1}{2}}{-\pi} \complexexponent{(\complex-\taufunction)}{-\frac{1}{2}}{-\pi} 
%\polyphinlambx{k}{\polylamb}{\complex}
\refereq{\eqref{laurentcoef}}{=} \frac{ \laurentcoef{-\polydegree} }{(-\polylamb)^{\polydegree}}  
%\\ \hspace{70mm}
%\refereq{(\refeq{laurentcoef})}{=} 
%\left| \begin{array}{ll}
%\frac{ \Coscoef{k}{\Floor{\frac k2}} \left(\frac{2}{\taufunction} \right)^k }{(-\polylamb)^{\polydegree}}     
%&\textup{if } \npolydegree-\dpolydegree+k \geq 0,
%\\
%0
%&\textup{otherwise}.
%\end{array}\right.
\end{array}\end{equation}
 where we consider that~$ \laurentcoef{-\polydegree} = 0$   whenever $ \npolydegree-\dpolydegree+k < 0$.
By using  Jordan's second lemma~\cite[\S{}3.1.4, Theorem~2]{mitrinovic84} % grandes encoches
(or, equivalently, by computing the residue at~$\infty$), we  find
\begin{equation}\label{grandesencoches}\begin{array}{c}
\lim\nolimits_{\epsilon\to 0}
 \frac{1}{\pi} \ointctrclockwise\nolimits_{\contoura{1/\epsilon}} \frac{\allpolycostnx{k}{\complex}  \,  \complexexponent{\complex}{-\frac{1}{2}}{-\pi} \complexexponent{(\complex-\taufunction)}{-\frac{1}{2}}{-\pi}}{ (1-\polylamb \complex)^{\polydegree}}
 %\polyphinlambx{k}{\polylamb}{\complex}
  \, d\complex 
=
\frac{2\i \, \laurentcoef{-\polydegree} }{(-\polylamb)^{\polydegree}} .
\end{array}\end{equation}
Besides,  $\lim\nolimits_{\epsilon\to \sing}(\complex -\sing) \allpolycostnx{k}{\complex}  \,  \complexexponent{\complex}{-\frac{1}{2}}{-\pi} \complexexponent{(\complex-\taufunction)}{-\frac{1}{2}}{-\pi}\,  (1-\polylamb \complex)^{-\polydegree}
%\polyphinlambx{k}{\polylamb}{\complex}
 = 0$ for~$\sing=0,\taufunction$ as a consequence of the assumption $0,\taufunction\notin\Polef{\allpolycostnfunction{k}}$. It follows from
Jordan's first lemma~\cite[\S{}3.1.4, Theorem~1]{mitrinovic84} % petites encoches
that
\begin{equation}\label{petitesencoches}\begin{array}{c}
\lim\nolimits_{\complex\to 0,\taufunction}
 \frac{1}{\pi} \left(  \varointclockwise\nolimits_{\contouras{\epsilon}{0}}   + \varointclockwise\nolimits_{\contouras{\epsilon}{\taufunction}}  \right) \frac{\allpolycostnx{k}{\complex}  \,  \complexexponent{\complex}{-\frac{1}{2}}{-\pi} \complexexponent{(\complex-\taufunction)}{-\frac{1}{2}}{-\pi}}{ (1-\polylamb \complex)^{\polydegree}}
 %\polyphinlambx{k}{\polylamb}{\complex}
  \, d\complex 
= 0 .
\end{array}\end{equation}
Lastly, by inspection of~$\complexexponent{\complex}{-\frac{1}{2}}{-\pi} \complexexponent{(\complex-\taufunction)}{-\frac{1}{2}}{-\pi}$ right above and below the segment~$(0,\taufunction)$, it can be seen that
\begin{equation}\label{limitcomplexexponentiation}\begin{array}{c}
\lim\nolimits_{\epsilon\to 0}
 \frac{1}{\pi} \left(  \int\nolimits_{\epsilon}^{\taufunction-\epsilon}   + \int\nolimits_{\taufunction-\epsilon}^{\epsilon}  \right) 
 \frac{\allpolycostnx{k}{\complex}  \,  \complexexponent{\complex}{-\frac{1}{2}}{-\pi} \complexexponent{(\complex-\taufunction)}{-\frac{1}{2}}{-\pi}}{ (1-\polylamb \complex)^{\polydegree}} 
% \polyphinlambx{k}{\polylamb}{\complex}
 \, d\complex 
= -2\i\, \realfouriercoefklamb{k}{\polylamb} .
\end{array}\end{equation}

On the other hand the residue theorem gives
\begin{equation}\label{realfouriercoefficientspolyIklambresidues} \begin{array}{c}
\hspace{-5mm}
\integralklamb{k}{\polylamb}
= \left(\frac{1}{\pi} \right) \, {2\i\pi} \sum\nolimits_{\sing\in\Polef{\allpolycostnfunction{k}}\cup\{\frac{1}{\polylamb}\}} \Bigresiduefx{
%\polyphinlambfunction{k}{\polylamb}
\frac{\allpolycostnx{k}{\complex}  \,  \complexexponent{\complex}{-\frac{1}{2}}{-\pi} \complexexponent{(\complex-\taufunction)}{-\frac{1}{2}}{-\pi}}{ (1-\polylamb \complex)^{\polydegree}} 
}{\complex=\sing}
\end{array}\end{equation}
We draw our attention to the residue at~$1/\polylamb$. 
Using the Taylor development of $(1-x \taufunction )^{-{1}/{2}-j}$ at~$x=0$ we find, for~$\polylamb<\taufunction$ and $t\in\Naturalnn$,
\begin{equation}\label{difficultderiv} \begin{array}{l}
\lim\nolimits_{\complex\to\frac{1}{\polylamb}}
\frac{d^{t}}{d\complex^{t}} \left[   \complexexponent{\complex}{-\frac{1}{2}}{-\pi} \complexexponent{(\complex-\taufunction)}{-\frac{1}{2}}{-\pi} \right] 
%&
%\\ \hspace{15mm}
 =\polylamb^{t+1} \sum\nolimits_{j=0}^{t} \binomialcoef{t}{j}\frac{  \Gammax{\frac{1}{2}}^2 (1-\polylamb\taufunction)^{\frac{1}{2}-j} }{  \Gammax{\frac{1}{2}-j} \Gammax{\frac{1}{2}-t+j} }
%&
\quad\\ \hfill%\hspace{15mm}
=
 \polylamb^{t+1} \sum\nolimits_{q=0}^{\infty} \Big[ \sum\nolimits_{j=0}^{t} \binomialcoef{t}{j} \frac{ \Gammax{\frac{1}{2}}^2 }{  \Gammax{\frac{1}{2}-q-j} \Gammax{\frac{1}{2}-t+j} } \Big]
\frac{(-\polylamb\taufunction)^{q}}{\factorial{q}}
.%,
%& (t\in\Naturalnn).
\end{array}\end{equation}
For  large~$\complex$, the value of~$\allpolycostnx{k}{\complex} $ is given by the Laurent series expansion of~$\allpolycostnfunction{k}$ at~$+\infty$, i.e.,
\begin{equation}\label{laurentallpolycostnx}\begin{array}{l}
\allpolycostnx{k}{\complex} = \sum\nolimits_{q=-\polydegree}^{\infty} \laurentcoef{q} \, \complex^{-q},
\qquad (0 \ll \modulus{\complex}<\infty  ).
\end{array}\end{equation}
Then,
\obsolete{
\begin{equation}\begin{array}{ll}
 \Residuefx{
%\polyphinlambfunction{k}{\polylamb}
\frac{\allpolycostnx{k}{\complex}  \,  \complexexponent{\complex}{-\frac{1}{2}}{-\pi} \complexexponent{(\complex-\taufunction)}{-\frac{1}{2}}{-\pi}}{ (1-\polylamb \complex)^{\polydegree}} 
}{\complex=\frac{1}{\polylamb}} \hspace{-49mm} &
\\
&=
\lim\nolimits_{\complex\to\frac{1}{\polylamb}} \frac{1}{\factorial{(\polydegree-1)}} \frac{d^{\polydegree-1}}{d\complex^{\polydegree-1}} \left[  \left( \complex - \frac 1 \polylamb\right)^{\polydegree} \left(\frac{\allpolycostnx{k}{\complex}  \,  \complexexponent{\complex}{-\frac{1}{2}}{-\pi} \complexexponent{(\complex-\taufunction)}{-\frac{1}{2}}{-\pi}}{ (1-\polylamb \complex)^{\polydegree}} \right) \right]
\\
& \refereq{\eqref{difficultderiv}}{=}
%\frac{1}{\factorial{(\polydegree-1)}(-\polylamb)^\polydegree} \sum\nolimits_{t=0}^{\polydegree-1} \binomialcoef{\polydegree-1}{\ t}  \Taylorcoef{\frac{1}{\polylamb}}{\polydegree-1-t} \factorial{(\polydegree-1-t)} \pi \polylamb^{t+1} \sum\nolimits_{q=0}^{\infty} \left[ \sum\nolimits_{j=0}^{t} \frac{\binomialcoef{t}{j} }{  \Gammax{\frac{1}{2}-q-j} \Gammax{\frac{1}{2}-t+j} } \right]\frac{(\polylamb\taufunction)^{q}}{\factorial{q}}
%\\&=
\frac{\pi (-1)^\polydegree}{\factorial{(\polydegree-1)}} \sum\nolimits_{t=0}^{\polydegree-1} \frac{ \binomialcoef{\polydegree-1}{\ t}  \Taylorcoef{\frac{1}{\polylamb}}{\polydegree-1-t} \factorial{(\polydegree-1-t)}
 }{ \polylamb^{\polydegree-1-t} } \sum\nolimits_{q=0}^{\infty} \left[ \sum\nolimits_{j=0}^{t} \frac{\binomialcoef{t}{j} }{  \Gammax{\frac{1}{2}-q-j} \Gammax{\frac{1}{2}-t+j} } \right]
\frac{(\polylamb\taufunction)^{q}}{\factorial{q}}
\\&=
\frac{\pi (-1)^\polydegree}{\factorial{(\polydegree-1)}}  \sum\nolimits_{q=0}^{\infty}  \frac{(\polylamb\taufunction)^{q}}{\factorial{q}} \sum\nolimits_{t=0}^{\polydegree-1} \frac{ \binomialcoef{\polydegree-1}{\ t}  \Taylorcoef{\frac{1}{\polylamb}}{\polydegree-1-t} \factorial{(\polydegree-1-t)}
 }{ \polylamb^{\polydegree-1-t} }  \sum\nolimits_{j=0}^{\polydegree-1} \frac{\binomialcoef{t}{j} }{  \Gammax{\frac{1}{2}-q-j} \Gammax{\frac{1}{2}-t+j} } 
\\&=
\frac{\pi (-1)^\polydegree}{\factorial{(\polydegree-1)}}  \sum\nolimits_{q=0}^{\infty}  \frac{(\polylamb\taufunction)^{q}}{\factorial{q}}   \sum\nolimits_{j=0}^{\polydegree-1} \frac{1}{\Gammax{\frac{1}{2}-q-j} }  \sum\nolimits_{t=j}^{\polydegree-1} \frac{ \binomialcoef{\polydegree-1}{\ t} \binomialcoef{t}{j}  \Taylorcoef{\frac{1}{\polylamb}}{\polydegree-1-t} \factorial{(\polydegree-1-t)}
 }{ \polylamb^{\polydegree-1-t} \Gammax{\frac{1}{2}-t+j} }  
%\\&=
%\frac{\pi (-1)^\polydegree}{\factorial{(\polydegree-1)}}  \sum\nolimits_{q=0}^{\infty}  \frac{(\polylamb\taufunction)^{q}}{\factorial{q}}   \sum\nolimits_{j=0}^{\polydegree-1} \frac{1}{\Gammax{\frac{1}{2}-q-j} }  \sum\nolimits_{t=j}^{\polydegree-1} \frac{ \binomialcoef{\polydegree-1}{\ j} \binomialcoef{\polydegree-1-j}{\ t-j}  \Taylorcoef{\frac{1}{\polylamb}}{\polydegree-1-t} \factorial{(\polydegree-1-t)} }{ \polylamb^{\polydegree-1-t} \Gammax{\frac{1}{2}-t+j} }  
\\&=
\frac{\pi (-1)^\polydegree}{\factorial{(\polydegree-1)}}  \sum\nolimits_{q=0}^{\infty}  \frac{(\polylamb\taufunction)^{q}}{\factorial{q}}   \sum\nolimits_{j=0}^{\polydegree-1} \frac{\binomialcoef{\polydegree-1}{\ j}}{\Gammax{\frac{1}{2}-q-j} }  \sum\nolimits_{t=0}^{\polydegree-1-j} \frac{  \binomialcoef{\polydegree-1-j}{\ \, t}  \Taylorcoef{\frac{1}{\polylamb}}{\polydegree-1-j-t} \factorial{(\polydegree-1-j-t)} }{ \polylamb^{\polydegree-1-t-j} \Gammax{\frac{1}{2}-t} }  
%\\& \refereq{(\refeq{laurentallpolycostnx})}{=}
\\&=
\frac{\pi (-1)^\polydegree}{\factorial{(\polydegree-1)}}  \sum\nolimits_{q=0}^{\infty}  \frac{(\polylamb\taufunction)^{q}}{\factorial{q}}   \sum\nolimits_{j=0}^{\polydegree-1} \frac{\binomialcoef{\polydegree-1}{\ j}}{\Gammax{\frac{1}{2}-q-j} }\left( \frac{ \frac{d^{\polydegree-1-j}}{dx^{\polydegree-1-j}} \left[ 
 %\allpolycostnx{k}{x} x^{-\frac{1}{2}}
 \frac{\sum\nolimits_{t=0}^{\polydegree-1} \Taylorcoef{\frac 1 \polylamb}{t} \left( x-\frac 1 \polylamb \right)^t}{\sqrt{x}}
 \right]_{x=\frac{1}{\polylamb}} }{\sqrt{\pi}   \polylamb^{\polydegree-\frac{1}{2}-j}} \right)
\\&=
\frac{ \sqrt{\pi} \left(-\frac 1 \polylamb\right)^{\polydegree}}{\factorial{(\polydegree-1)} }  \sum\nolimits_{q=0}^{\infty}  \frac{\taufunction^{q}}{\factorial{q} \, \Gammax{\frac{1}{2}-q}} 
\\  &\hspace{10mm} \sum\nolimits_{t=0}^{\polydegree-1} \Taylorcoef{\frac 1 \polylamb}{t}  \sum\nolimits_{j=0}^{\polydegree-1}\binomialcoef{\polydegree-1}{\ j}  
%\frac{\frac{d^{j}}{(dx)^{-j}} \left[  x^{-\frac{1}{2}-q} \right]_{x=\frac{1}{\polylamb}} }{\Gammax{\frac{1}{2}-q}}
\left(\frac{  \Gammax{\frac{1}{2}-q} \left(\frac 1 \polylamb\right)^{-\frac{1}{2}-q-j}  }{\Gammax{\frac{1}{2}-q-j} } \right)
 \frac{d^{\polydegree-1-j}}{dx^{\polydegree-1-j}} \left[ 
 %\allpolycostnx{k}{x} x^{-\frac{1}{2}}
 \frac{\left( x-\frac 1 \polylamb \right)^t}{\sqrt{x}}
 \right]_{x=\frac{1}{\polylamb}} 
\\&=
\frac{ \sqrt{\pi}}{\factorial{(\polydegree-1)} (-\polylamb)^\polydegree}    \sum\nolimits_{q=0}^{\infty}  \frac{\taufunction^{q}}{\factorial{q}\, \Gammax{\frac{1}{2}-q}}  \sum\nolimits_{t=0}^{\polydegree-1} \Taylorcoef{\frac 1 \polylamb}{t}  \frac{d^{\polydegree-1}}{dx^{\polydegree-1}} \left[  %x^{-\frac{1}{2}-q} \allpolycostnx{k}{x} x^{-\frac{1}{2}}
 \frac{\left( x-\frac 1 \polylamb \right)^t}{x^{q+1}}
 \right]_{x=\frac{1}{\polylamb}} 
\\ &  \refereq{\eqref{laurentallpolycostnx}}{=}
\end{array}\end{equation}
}%
\begin{equation}\label{Residuefxseries}\begin{array}{ll}
&%\frac{1}{2\i\pi}
 \Bigresiduefx{
%\polyphinlambfunction{k}{\polylamb}
\frac{\allpolycostnx{k}{\complex}  \,  \complexexponent{\complex}{-\frac{1}{2}}{-\pi} \complexexponent{(\complex-\taufunction)}{-\frac{1}{2}}{-\pi}}{ (1-\polylamb \complex)^{\polydegree}} 
}{\complex=\frac{1}{\polylamb}} %\hspace{-53mm} &
\\
&=
\lim\nolimits_{\complex\to\frac{1}{\polylamb}} \frac{1}{\factorial{(\polydegree-1)}} \frac{d^{\polydegree-1}}{d\complex^{\polydegree-1}} \Big[  ( \complex - \frac 1 \polylamb)^{\polydegree} \Big(\frac{\allpolycostnx{k}{\complex}  \,  \complexexponent{\complex}{-\frac{1}{2}}{-\pi} \complexexponent{(\complex-\taufunction)}{-\frac{1}{2}}{-\pi}}{ (1-\polylamb \complex)^{\polydegree}} \Big) \Big]
\\
& \! \! \!\, \refereq{\eqref{difficultderiv}}{=} \! \!\,
\frac{ (-\frac 1 \polylamb )^\polydegree}{\factorial{(\polydegree-1)}} \sum\nolimits_{t=0}^{\polydegree-1} \binomialcoef{\polydegree-1}{\ t}  \frac{\dqAllpolycostnx{k}{\polydegree-1-t}{\frac{1}{\polylamb}} }{\
(\frac 1 \polylamb )^{t+1} }
  \sum\nolimits_{q=0}^{\infty} \Big[ \sum\nolimits_{j=0}^{t} \frac{\binomialcoef{t}{j} \, \Gammax{\frac{1}{2}}^2 }{  \Gammax{\frac{1}{2}-q-j} \Gammax{\frac{1}{2}-t+j} } \Big]
\frac{(-\polylamb\taufunction)^{q}}{\factorial{q}}
\\&=
\frac{(-\frac 1 \polylamb )^\polydegree}{\factorial{(\polydegree-1)}}  \sum\nolimits_{q=0}^{\infty}  \frac{(-\polylamb\taufunction)^{q}}{\factorial{q}} \sum\nolimits_{t=0}^{\polydegree-1} \frac{ \binomialcoef{\polydegree-1}{\ t} \dqAllpolycostnx{k}{\polydegree-1-t}{\frac{1}{\polylamb}} 
 }{ (\frac 1 \polylamb )^{t+1} }  \sum\nolimits_{j=0}^{\polydegree-1} \frac{\binomialcoef{t}{j} \, \Gammax{\frac{1}{2}}^2  }{  \Gammax{\frac{1}{2}-q-j} \Gammax{\frac{1}{2}-t+j} } 
\\&=
\frac{ (-\frac 1 \polylamb )^\polydegree}{\factorial{(\polydegree-1)}}  \sum\nolimits_{q=0}^{\infty}  \frac{(-\polylamb\taufunction)^{q}}{\factorial{q}}   \sum\nolimits_{j=0}^{\polydegree-1} \frac{\Gammax{\frac{1}{2}} }{\Gammax{\frac{1}{2}-q-j} }  \sum\nolimits_{t=j}^{\polydegree-1} \frac{ \binomialcoef{\polydegree-1}{\ t} \binomialcoef{t}{j}  \Gammax{\frac{1}{2}} \dqAllpolycostnx{k}{\polydegree-1-t}{\frac{1}{\polylamb}}
 }{ (\frac 1 \polylamb )^{1+t} \Gammax{\frac{1}{2}-t+j} }  
\\&=
\frac{ (-\frac{1}{\polylamb})^\polydegree}{\factorial{(\polydegree-1)}}  \!\!\, \sum\nolimits_{q=0}^{\infty} \!\!\,  \frac{(-\taufunction)^{q}}{\factorial{q}}  \!\!\,  \sum\nolimits_{j=0}^{\polydegree-1} \!\!\, \frac{\binomialcoef{\polydegree-1}{\ j} \Gammax{\frac{1}{2}}   (\frac 1 \polylamb )^{-\frac 1 2 - q -j} }{\Gammax{\frac{1}{2}-q-j}  }  \!\!\, \sum\nolimits_{t=0}^{\polydegree-1-j} \!\!\,  \frac{  \binomialcoef{\polydegree-1-j}{\ \, t} \Gammax{\frac{1}{2}}   \dqAllpolycostnx{k}{\polydegree-1-j-t}{\frac{1}{\polylamb}} }{ (\frac 1 \polylamb )^{\frac 1 2 + t} \Gammax{\frac{1}{2}-t} }  
%\\& \refereq{(\refeq{laurentallpolycostnx})}{=}
\\&=
\frac{(-\frac 1 \polylamb )^\polydegree}{\factorial{(\polydegree-1)}}  \sum\nolimits_{q=0}^{\infty}  \frac{\Gammax{\frac{1}{2}} (-\taufunction)^{q}}{\factorial{q} \Gammax{\frac{1}{2}-q}}   \sum\nolimits_{j=0}^{\polydegree-1} 
\binomialcoef{\polydegree-1}{\ j}
\frac{d^{\polydegree-1-j}}{dx^{\polydegree-1-j}} \left[ 
 \frac{x^{-q}}{\sqrt{x}}
 \right]_{x=\frac{1}{\polylamb}}
\frac{d^{\polydegree-1-j}}{dx^{\polydegree-1-j}} \left[ 
 \frac{\allpolycostnx{k}{x}}{\sqrt{x}}
 \right]_{x=\frac{1}{\polylamb}}
\\&=
\frac{(-\frac 1 \polylamb )^\polydegree}{\factorial{(\polydegree-1)}}  \sum\nolimits_{q=0}^{\infty}  \frac{\Gammax{\frac{1}{2}} (-\taufunction)^{q}}{\factorial{q} \Gammax{\frac{1}{2}-q}}   
\frac{d^{\polydegree-1}}{dx^{\polydegree-1}} \left[ 
x^{-(q+1)}\allpolycostnx{k}{x}
 \right]_{x=\frac{1}{\polylamb}} .
 \end{array}\end{equation}
 For small~$\polylamb$, the function $\allpolycostnfunction{k}$ is analytic in a neighborhood of~${1}/{\polylamb}$, and so is~$\complex^{-(q+1)}\allpolycostnx{k}{\complex}$ for any~$q$. 
 It follows that the derivation in~\eqref{Residuefxseries} applies term by term to the Laurent series~\eqref{laurentallpolycostnx}, and we find
 \begin{equation}\label{towardsfubini}\begin{array}{ll}
%\frac{1}{2\i\pi} 
\Bigresiduefx{
%\polyphinlambfunction{k}{\polylamb}
\frac{\allpolycostnx{k}{\complex}  \,  \complexexponent{\complex}{-\frac{1}{2}}{-\pi} \complexexponent{(\complex-\taufunction)}{-\frac{1}{2}}{-\pi}}{ (1-\polylamb \complex)^{\polydegree}} 
}{\complex=\frac{1}{\polylamb}}
 \hspace{-47mm} &
\\
& \refereq{\eqref{laurentallpolycostnx}}{=}
\frac{(-\frac 1 \polylamb )^\polydegree}{\factorial{(\polydegree-1)}}  \sum\nolimits_{q=0}^{\infty}   \Big( \frac{\Gammax{\frac{1}{2}} }{ \Gammax{\frac{1}{2}-q}} \Big) \frac{(-\taufunction)^{q}}{\factorial{q} }   
\frac{d^{\polydegree-1}}{dx^{\polydegree-1}} \left[ 
 \sum\nolimits_{j=-\polydegree}^{\infty} \laurentcoef{j} \, x^{-(j+q+1)} 
 \right]_{x=\frac{1}{\polylamb}}
\\
\obsolete{
& =
\frac{\laurentcoef{\polydegree}}{(-  \polylamb)^\polydegree} 
 -
  \frac{1}{\factorial{(\polydegree-1)}}
 \sum\nolimits_{q=0}^{\infty}   \left( \frac{\Gammax{\frac{1}{2}} }{ \Gammax{\frac{1}{2}-q}} \right) \frac{(-\taufunction)^{q}}{\factorial{q} }  \sum\nolimits_{j=-\min(q,\polydegree)}^{\infty} 
%\frac{d^{\polydegree-1}}{(dx)^{\polydegree-1}} \left[  x^{t-q-1} \right]_{x=\frac{1}{\polylamb}}
%\frac{d^{\polydegree-1}}{(dx)^{\polydegree-1}} \left[  x^{t-q-1} \right]_{x=\frac{1}{\polylamb}}
 \frac{\factorial{(\polydegree-1+q+j)}}{\factorial{(q+j)}}
 %\binomialcoef{\polydegree-1+q-t}{\ \, q-t}
\laurentcoef{j} \, \left(\frac{1}{\polylamb}\right)^{-(j+q)}
\\
}%
& =
\frac{\laurentcoef{-\polydegree}}{(-  \polylamb)^\polydegree} 
 -
 \sum\nolimits_{q=0}^{\infty} \Big( \frac{\Gammax{\frac{1}{2}} }{ \Gammax{\frac{1}{2}-q}} \Big) \frac{(-\polylamb\taufunction)^{q}}{\factorial{q} }   \sum\nolimits_{j=-\min(q,\polydegree)}^{\infty} 
 \binomialcoef{\polydegree-1+q+j}{\ \, \polydegree-1}
\laurentcoef{j} \, \polylamb^{j}.
\end{array}\end{equation}
Observe on the other hand  that, for $q+j\geq0$,
 \begin{equation}\begin{array}{c}
 \Modulus{\binomialcoef{\polydegree-1+q+j}{\ \, \polydegree-1}
}
\leq
\frac{(\polydegree-1+q+j)^{\polydegree-1}}{\factorial{(\polydegree-1)}}
\leq
\sum\nolimits_{t=0}^{\infty}\frac{(\polydegree-1+q+j)^{t}}{\factorial{t}}
\leq
\exp{\polydegree-1+q+j}.
\end{array}\end{equation}
Hence, for small~$\polylamb$,
 \begin{equation}\label{intbound}\begin{array}{c}
 \sum\nolimits_{j=-\min(q,\polydegree)}^{\infty} 
 \Modulus{\binomialcoef{\polydegree-1+q+j}{\ \, \polydegree-1}
\laurentcoef{j} \, \polylamb^{j}}
\leq
\exp{\polydegree-1+q}
 \sum\nolimits_{j=-\polydegree}^{\infty} 
 \modulus{\laurentcoef{j}} \, (\polylamb e)^{j}
 \leq 
 \intbound{\polylamb}{\polydegree} \, \exp{\polydegree-1+q} ,
\end{array}\end{equation}
with $\intbound{\polylamb}{\polydegree} <\infty$ by absolute convergence of the above series. Consequently,
\begin{equation}\begin{array}{l} 
 \sum\limits_{q=0}^{\infty} \sum\limits_{\ j=-\min(q,\polydegree)}^{\infty}  \Modulus{ \Big( \frac{\Gammax{\frac{1}{2}} }{ \Gammax{\frac{1}{2}-q}} \Big) \frac{(-\polylamb\taufunction)^{q}}{\factorial{q} }   
 \binomialcoef{\polydegree-1+q+j}{\ \, \polydegree-1}
\laurentcoef{j} \, \polylamb^{j} }
%\\
%\hspace{40mm}
\refereq{\eqref{intbound}}{\leq} \,
\intbound{\polylamb}{\polydegree}\exp{\polydegree-1}
 \sum\limits_{q=0}^{\infty} \Big( \frac{\Gammax{\frac{1}{2}} }{ \Gammax{\frac{1}{2}-q}} \Big) \frac{(\polylamb\taufunction)^{q}}{\factorial{q} }  \exp{q}\, 
\end{array}\end{equation}
is a finite quantity as it passes the ratio test for~$\polylamb<1 / \taufunction$|this can be shown using Stirling's formula.
It follows from Fubini's theorem \hidecalculustheorems{}{(Theorem~\ref{theorem:fubini}% in Appendix~\refeq{appendix:ar}
) }% 
that the summation order in~\eqref{towardsfubini} can be permuted. By setting~$t=j+q$ we find
 \begin{equation}\label{afterfubini}\begin{array}{rcl}
 %\frac{1}{2\i\pi} 
 \Bigresiduefx{
%\polyphinlambfunction{k}{\polylamb}
\frac{\allpolycostnx{k}{\complex}  \,  \complexexponent{\complex}{-\frac{1}{2}}{-\pi} \complexexponent{(\complex-\taufunction)}{-\frac{1}{2}}{-\pi}}{ (1-\polylamb \complex)^{\polydegree}} 
}{\complex=\frac{1}{\polylamb}}
 \hspace{-40mm} &&
\\
& \hidecalculustheorems{=}{\refereq{\eqref{fubini}}{=}} &
\frac{\laurentcoef{-\polydegree}}{(-  \polylamb)^\polydegree} 
 -
  \sum\nolimits_{q=0}^{\infty} \sum\nolimits_{t=\max(0,q-\polydegree)}^{\infty}\Big[ \laurentcoef{t-q}  \binomialcoef{\polydegree-1+t}{\ \, \polydegree-1} \Big( \frac{\Gammax{\frac{1}{2}} }{ \Gammax{\frac{1}{2}-q}} \Big) \frac{(-\taufunction)^{q}}{\factorial{q} } \Big]   \polylamb^{t} 
\\
& =&
\frac{\laurentcoef{-\polydegree}}{(-  \polylamb)^\polydegree} 
 -  \sum\nolimits_{q=0}^{\polydegree}  \laurentcoef{-q}  \Big( \frac{\Gammax{\frac{1}{2}} }{ \Gammax{\frac{1}{2}-q}} \Big) \frac{(-\taufunction)^{q}}{\factorial{q} }
+ \magnitude{\polylamb}
.
\end{array}\end{equation}
We eventually obtain~\eqref{solutionrealfouriercoefk} by combining~\eqref{realfouriercoefficientspolyIklamb} and~\eqref{realfouriercoefficientspolyIklambresidues}, together with the intermediate results~\eqref{grandesencoches}, \eqref{petitesencoches}, \eqref{limitcomplexexponentiation}, and~\eqref{afterfubini},  and letting~$\polylamb\to 0$.
\end{proof}
}%
%
%\begin{remark}
%\noindent

The computation of the residues in~\eqref{solutionrealfouriercoefk} is straightforward for every pole in~$\Polef{\allpolycostnfunction{k}}$. The final result can be stated as a function of the derivatives of~$\allpolycostnfunction{k}$ and of the function defined by $\irrationalx{\complex} = \complexexponent{\complex}{-{1}/{2}}{-\pi} \complexexponent{(\complex-\taufunction)}{-{1}/{2}}{-\pi}$.
The successive derivatives of~$\irrationalfunction$ can be obtained by induction on~$t\geq 2$, using
\obsolete{
$$\sum\nolimits_{q=0}^{t} \binomialcoef{t}{q}   \dnirrationalx{q}{\complex} \dnirrationalx{t-q}{\complex}  = 
(-1)^t \factorial{t}\sum\nolimits_{j=0}^{t} 
%\binomialcoef{t}{j} \factorial{j}\factorial{(t-j)} 
\complex^{-(j+1)}(\complex-\taufunction)^{-(t-j+1)}  
$$

$$\dnirrationalx{t}{\complex}    = \frac{1}{2}\left(
(-1)^t \factorial{t}\sum\nolimits_{j=0}^{t} 
%\binomialcoef{t}{j} \factorial{j}\factorial{(t-j)} 
\complex^{-(j+1)}(\complex-\taufunction)^{-(t-j+1)}
- \sum\nolimits_{q=1}^{t-1} \binomialcoef{t}{q}   \dnirrationalx{q}{\complex} \dnirrationalx{t-q}{\complex}  
\right) \irrationalx{\complex}^{-1}
$$
}%
\begin{equation}\notag\begin{array}{rcl}
\dnirrationalx{1}{\complex}  &=&- \big( \frac{2\complex-\taufunction }{2 \complex^{2}(\complex-\taufunction)^{2}} \big) \Complexsqrt{\complex}{-\pi} \Complexsqrt{\complex-\taufunction}{-\pi}  ,
\\
\dnirrationalx{t}{\complex}    &=& \frac{1}{2}\big\{
(-1)^t \factorial{t}\sum\nolimits_{j=0}^{t} 
\complex^{-(j+1)}(\complex-\taufunction)^{-(t-j+1)}
- \sum\nolimits_{q=1}^{t-1} \binomialcoef{t}{q}   \dnirrationalx{q}{\complex} \dnirrationalx{t-q}{\complex}  
\big\} 
\hspace{10mm}
\\&&\hfill
\times \Complexsqrt{\complex}{-\pi} \Complexsqrt{\complex-\taufunction}{-\pi},
\end{array}\end{equation}
which follows from the derivation of~$\irrationalx{\complex} ^2=\complex^{-1}(\complex-\taufunction)^{-1}$ using Leibniz's product rule.
%\end{remark}

%The method of solution of Result~\ref{result:fouriercoefsforquotientofpolynomials} extends to other types of cost functions with   Laurent series expansions at~$\infty$ for as long as the integration~(\refeq{grandesencoches}) along the contour~$\contoura{1/\epsilon}$  can be sorted out.

%%%%%%%%%%%%%%

\storecounter{section}{store:section:lastappendix}%
\storecounter{section}{section:lastappendix}%
 
%\vfill \
% \pagebreak
 %\includepdf[pages={1-},scale=1.0]{QUESTA20_SupplementaryMaterial.pdf}

\end{document}